\documentclass[12pt, a4paper]{report}
\usepackage[numbers]{natbib}
\usepackage{amsmath,amsfonts,amssymb,amsthm,graphics,graphicx,epsfig,bbm}
\usepackage[colorlinks=true,citecolor=blue,linkcolor=blue,urlcolor=blue]{hyperref}
\usepackage[usenames]{color}
\usepackage{graphicx,epsfig,xcolor}
\usepackage{graphicx}
\usepackage{lipsum}
\usepackage{subfigure}
\usepackage{amsmath}
\usepackage{nomencl}
\usepackage{float} 
\usepackage{epsfig}
\usepackage{dcolumn}
\usepackage{bm}
\usepackage{color}
\usepackage{enumerate}
\usepackage{times}
\usepackage{epstopdf}
\usepackage{amssymb}
\usepackage{amstext}
\usepackage{latexsym}
\usepackage{hyperref}
\usepackage{amsfonts}
\usepackage{psfrag}
\usepackage{soul,xcolor}
\usepackage[normalem]{ulem}
\usepackage{dsfont}
\usepackage{txfonts}
\usepackage{physics}
\usepackage{footnote}
\usepackage{multirow}
\usepackage{appendix}
\usepackage{mathtools}
\usepackage{xspace}
\usepackage[utf8]{inputenc}
\usepackage{amsmath}
\usepackage{amsfonts}
\usepackage{amssymb}
\usepackage{graphicx}
\usepackage{hyperref}
\usepackage{fancyhdr}
\usepackage{geometry}
\title{Quantum Information Scrambling, Chaos, Sensitivity, and Emergent State Designs}
\author{Varikuti Naga Dileep}
\date{June 2024}

\makenomenclature

\begin{document}

\begin{minipage}{0.4\textwidth}
    \quad\includegraphics[width=0.5\textwidth]{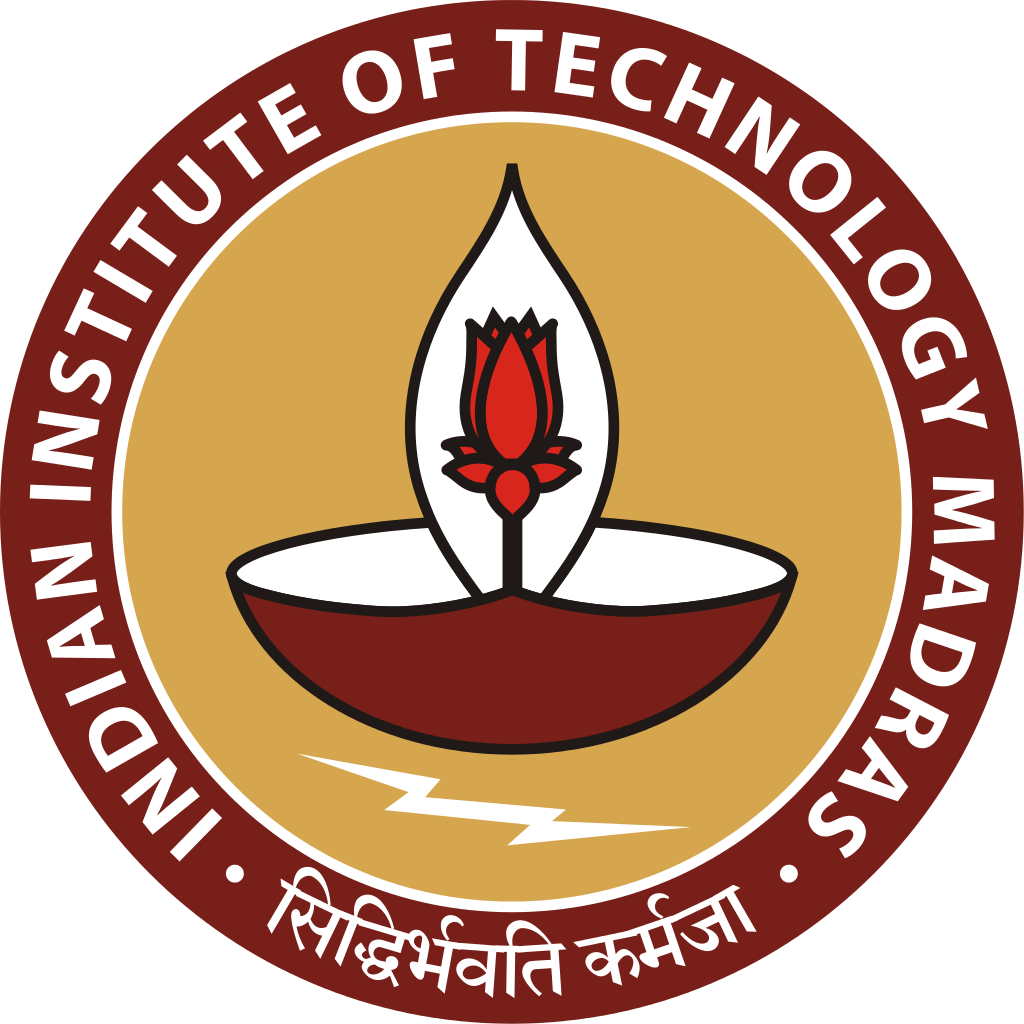}
\end{minipage}
\hfill
\begin{minipage}{0.6\textwidth}
    DEPARTMENT OF PHYSICS
    \par
    INDIAN INSTITUTE OF TECHNOLOGY MADRAS
    \par
    CHENNAI -- 600036
\end{minipage}
\par
\makebox[\textwidth][l]{\rule{\paperwidth}{1mm}}
\vspace{1cm}
\begin{center}
\vspace*{\parskip}
\setlength{\fboxrule}{2pt}
\setlength{\fboxsep}{10pt}
        \begin{minipage}{0.9\textwidth}
        \begin{center}
            \bfseries\LARGE{Quantum Information Scrambling, Chaos, Sensitivity, and Emergent State Designs}
        \end{center}
        \end{minipage}
\end{center}

\begin{center}
    \includegraphics[scale=0.4]{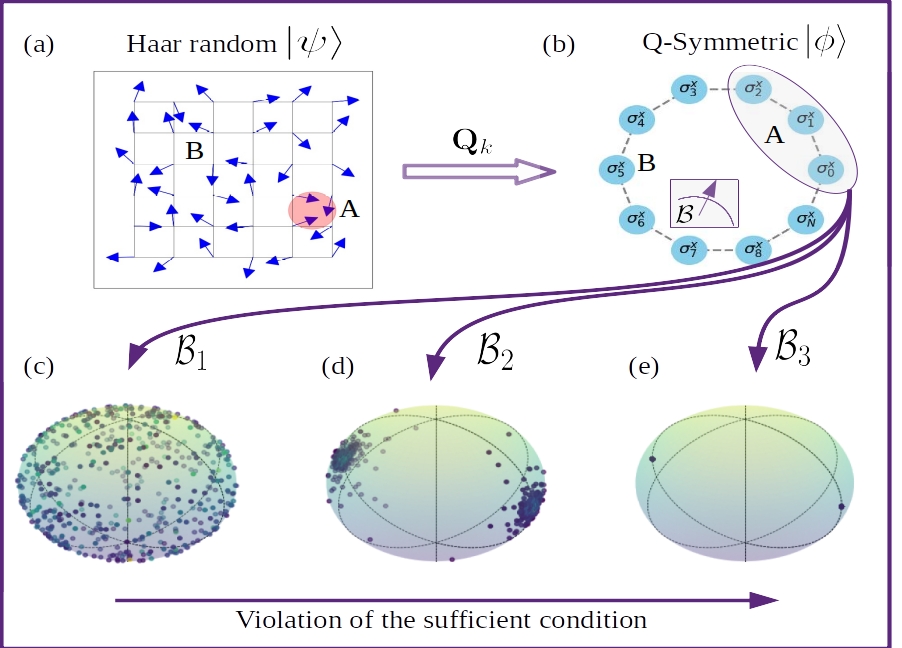}
\end{center}

\vspace{\parskip}

\begin{center}
    \textit{A Thesis}
\end{center}
\begin{center}
    \textit{Submitted by}
\end{center}
\begin{center}
    \textbf{VARIKUTI NAGA DILEEP}
\end{center}
\vspace{1cm}
\begin{center}
    \textit{For the award of the degree}
\end{center}
\begin{center}
    \textit{Of}
\end{center}
\begin{center}
    \textbf{DOCTOR OF PHILOSOPHY}
\end{center}
\begin{center}
    June, 2024
\end{center}
\vspace*{\fill}
\textcopyright\ 2024\ Indian Institute of Technology Madras

%

\chapter*{}
\vspace*{\fill}
 \newcommand{\Hquad}{\hspace{0.5em}} 

\begin{center}
\begin{minipage}{0.5\textwidth}
\itshape
\begin{eqnarray*}
&&Yoga-sthaḥ\Hquad kuru\Hquad karm\overline{a}ṇi\Hquad sa\dot{n}ga\dot{m}\Hquad tyaktv\overline{a}\Hquad dhana\Tilde{n}jaya\Hquad |\nonumber\\
&&siddhy-asiddhyoḥ\Hquad samo\Hquad bh\overline{u}tv\overline{a}\Hquad samatva\dot{m}\Hquad yoga\Hquad uchyate\Hquad||
\end{eqnarray*}
\end{minipage}
\end{center}
\hspace*{0.75\textwidth} {\bfseries}

\textbf{Translation}: O Arjun, maintain an unwavering dedication to your duty, relinquishing attachment to both success and failure. Such a state of equanimity is called Yog.
\vspace*{\fill}
\newpage

\chapter*{}
\vspace*{\fill}

\begin{center}
\begin{minipage}{0.75\textwidth}
\itshape

To my mom, dad, our Guru ``Sree Venkanna Babu (Govinda Swamy)", and my brother. 
\end{minipage}
\end{center}

\vspace*{\fill}
\newpage

\chapter*{Thesis Certificate}

This is to undertake that the Thesis titled \textbf{QUANTUM INFORMATION
SCRAMBLING, CHAOS, SENSITIVITY, AND EMERGENT STATE DESIGNS},
submitted by me to the Indian Institute of Technology Madras, for the award of \textbf{Doctor
of Philosophy}, is a bona fide record of the research work done by me under the
supervision of \textbf{Dr. Vaibhav Madhok}. The contents of this Thesis, in full or in parts,
have not been submitted to any other Institute or University for the award of any degree
or diploma.

\vspace{2cm}

\textbf{Chennai 600036\hspace{8cm} Varikuti Naga Dileep}\\

\textbf{June 2024}

\vspace{2cm}

\textbf{\hspace{11cm} Dr. Vaibhav Madhok}\newline
\text{\hspace{12.5cm} Research advisor}\newline
\text{\hspace{12.2cm}Associate Professor}\newline
\text{\hspace{11.7cm}Department of Physics}\newline
\text{\hspace{13.6cm}IIT Madras}\newline

\vspace*{\fill}
\textcopyright\ 2024\ Indian Institute of Technology Madras
\par

\chapter*{List of Publications}
\subsection*{Refereed journals based on thesis}
\begin{enumerate}
    \item \textcolor{blue}{Naga Dileep Varikuti}, Abinash Sahu, Arul Lakshminarayan, and Vaibhav Madhok, “Probing dynamical sensitivity of a non-kolmogorov-arnold-moser system through out-of-time-order correlators,” Phys. Rev. E \textbf{109}, 014209 (2024). \par

    \item \textcolor{blue}{Naga Dileep Varikuti}, and Vaibhav Madhok, “Out-of-time ordered
    correlators in kicked coupled tops: Information scrambling in mixed phase space and the role of conserved quantities,” Chaos \textbf{34}, 063124 (2024).\par

    \item \textcolor{blue}{Naga Dileep Varikuti} and Soumik Bandyopadhyay, ``Unraveling the emergence of quantum state designs in systems with symmetry," Quantum \textbf{8}, 1456 (2024).

\end{enumerate}

\subsection*{Refereed journals (others)}

\begin{enumerate}    
    \item Sreeram Pg, \textcolor{blue}{Naga Dileep Varikuti}, and Vaibhav Madhok, ``Exponential speedup in measuring out-of-time-ordered correlators and gate fidelity with a single bit of quantum information," Physics Letters A, \textbf{397}:127257, 2021. \par

    \item Abinash Sahu, \textcolor{blue}{Naga Dileep Varikuti}, Bishal Kumar Das, and Vaibhav Madhok, ``Quantifying operator spreading and chaos in Krylov subspaces with quantum state reconstruction," Phys. Rev. B, \textbf{108}, 224306 (2023).\par

    \item Abinash Sahu, \textcolor{blue}{Naga Dileep Varikuti}, and Vaibhav Madhok, ``Quantum tomography under perturbed hamiltonian evolution and scrambling of errors–a quantum signature of chaos," arXiv preprint arXiv:2211.11221 (2022). [under review]

\end{enumerate}

\chapter*{Acknowledgments}
\noindent
Firstly, I express my gratitude to Dr. Vaibhav Madhok, my supervisor, for his consistent guidance and support during my Ph.D. journey. Vaibhav, your physics insights and occasional remarks have been invaluable to me at every step. Thank you for always being reassuring and encouraging me to foster independent thinking while emphasizing the collaborative nature of research. I am especially grateful for the opportunity you've provided me to explore independent research ideas. Thanks a ton for always being available to discuss physics at any time of the day.

I am indebted to Prof. Arul Lakshminarayan for the invaluable hours spent engaging in insightful discussions. I particularly thank him for always being available to discuss physics despite his extremely busy schedule. Special thanks to Soumik Bandyopadhyay for his time and devotion at various stages of our collaborative projects and for being a good friend. I want to thank my colleagues, friends, and collaborators --- Sreeram PG, Abinash Sahu, and Bishal Kumar Das, with whom I have had countless insightful blackboard discussions. I am grateful to the participants of the quantum journal club meetings. These meetings were instrumental in generating several core ideas that were subsequently developed into significant parts of the current thesis. 

I am grateful to the members of my doctoral committee, Prof. S. Lakshmi Bala (chairperson, retd), Prof. Rajesh Narayanan (current chairperson), Prof. Arul Lakshminarayan, Dr. Prabha Mandayam, and Prof. Andrew Thangaraj for their continuous assessment and evaluation of my research. I sincerely thank IIT Madras for the HTRA fellowship and adequate research facilities, including computing resources, during my Ph.D. tenure. I also thank the Centre for Quantum Information, Communication and Computing (CQuICC), IIT Madras, for their generous financial support at various stages of my doctoral research.  

Thanks to our current and previous postdocs --- Ramdas, Balakrishnan Viswanathan, Shrikant Utagi, Aravinda, and Subhrajit Modak. I thank my fellow Ph.D. students --- Dhrubajyoti Biswas, Sourav Manna, Praveen, Vikash, Debjyoti Biswas, Sourav Dutta, and Bhavesh. I also thank my friends in the quantum group --- Akshaya, Suhail, Vishnu, and Soumyabrata Paul. Special thanks to the MSc project student and my dear friend, Shraddha, for the invaluable wisdom and support you have imparted to me. Many thanks to other project students --- Shiva Prasad, Dinesh, Bidhi, and Arun. I also want to express my heartfelt thanks to my friends at IITM — Minati, Jatin, Shakshi, Urvashi, Dipanvita, Bubun, Anshul, Nihar, Koushik, Devendar, Sarvesh, Santhosh, Himanshu, Chandra Shekhar Tiwari, and Sonam — for the wonderful times we shared. I am grateful to my former master's companions --- Prathmesh and Alok, as well as my bachelor's companion, China Anjaneyulu, for always being available for quick chats. Special thanks to my lifelong friends, Rama Krishna, Gorre Gopi Krishna, and Gautham, for always believing in me. Your friendship means a lot to me.

Finally, my heartfelt gratitude extends to my family --- mom, dad, brother, sister-in-law, Rudhra (my little nephew), Prabhu, and all my other cousins. Your unwavering presence through every challenge and setback, support during my toughest moments, and boundless love have been truly invaluable. I am profoundly grateful and blessed to have each and every one of you in my life.

\begin{abstract}
\noindent
\begin{tabular}{@{} >{\bfseries}p{0.2\textwidth} p{0.77\textwidth} @{}}
\MakeUppercase{Keywords} &
Quantum chaos; Kolmogorov-Arnold-Moser theorem; kicked harmonic oscillator; out-of-time ordered correlators; quantum Fisher information; kicked coupled tops; random matrix theory; quantum state designs; projected ensembles\\
\end{tabular}
\par
\vspace{1cm}
Understanding quantum chaos is of profound theoretical interest and carries significant implications for various applications, from condensed matter physics to quantum error correction. Recently, out-of-time ordered correlators (OTOCs) have emerged as a powerful tool to characterize and quantify chaos in quantum systems. For a given quantum system, the OTOCs measure incompatibility between an operator evolved in the Heisenberg picture and an unevolved operator. In the first part of this thesis, we employ OTOCs to study the dynamical sensitivity of a perturbed non-Komogorov-Arnold-Moser (non-KAM) system in the quantum limit as the parameter that characterizes the \textit{resonance} condition is slowly varied. For this purpose, we consider a quantized kicked harmonic oscillator (KHO) model that displays stochastic webs resembling Arnold's diffusion that facilitates large-scale diffusion in the phase space. Although the maximum Lyapunov exponent of the KHO at resonances remains close to zero in the weak perturbative regime, making the system weakly chaotic in the conventional sense, the classical phase space undergoes significant structural changes. Motivated by this, we study the OTOCs when the system is in resonance and contrast the results with the non-resonant case. At resonances, we observe that the long-time dynamics of the OTOCs are sensitive to these structural changes where they grow quadratically as opposed to linear or stagnant growth at non-resonances.
On the other hand, our findings suggest that the short-time dynamics remain relatively more stable and show the exponential growth found in the literature for unstable fixed points. The numerical results are backed by analytical expressions derived for a few special cases. The OTOC analysis is followed by a study of quantum Fisher information (QFI) at the resonances and a comparison with the non-resonance cases. We shall show that scaling of the QFI in time is greatly enhanced at the resonances, making the dynamics of the non-KAM systems good candidates for quantum sensing. 

In the second part, we study the OTOCs in a system of kicked coupled tops. We address connections between inter-system coupling, chaos, and information scrambling in a system where the mechanism of chaos and coupling is the same hyper-fine interaction between the spins. Due to the conservation law, the system admits a decomposition of the dynamics into distinct invariant subspaces. Under strong coupling, the OTOC growth rate in the largest subspace correlates remarkably well with the classical Lyapunov exponent. In the mixed phase space, we use Percival's conjecture to partition the eigenstates of the Floquet map into ``regular" and ``chaotic." Using these states as the initial states, we examine how their mean phase space locations affect the growth and saturation of the OTOCs. Beyond the largest subspace, we study the OTOCs across the entire system, including all other smaller subspaces. For certain initial operators, we analytically derive the OTOC saturation using random matrix theory (RMT).

The last part of the thesis is devoted to the study of the emergence of quantum state designs as a signature of quantum chaos and the role of symmetries in this phenomenon. Recently proposed projected ensemble framework utilizes quantum chaos as a resource to construct approximate higher-order state designs. Under this framework, projective measurements are performed on one part of a quantum state obtained from a sufficiently chaotic evolution. Then, the ensemble of pure states of the other part tends to approximate higher-order state designs. Despite being ubiquitous, the effects of symmetries on the emergence of quantum state designs remain under-explored. While quantum chaos randomizes, symmetries enforce order in a system. We thoroughly investigate this by demonstrating the interplay between symmetries and measurements in constructing approximate state designs. At the core of our work, we shed light on what projective measurements to consider in the presence of symmetries. We do so by analytically deriving a sufficient condition on the measurement basis relative to a symmetry operator. The condition, if satisfied, guarantees the emergence of state designs. It significantly simplifies the search for an appropriate measurement basis as we identify instances where a basis is incompatible. We numerically corroborate these results for random quantum states with different symmetries and states obtained from the chaotic evolution of the transverse field Ising model. Considerable violation of the sufficient condition yields non-convergence to the designs, as also observed in the numerical simulations, a crucial aspect to understanding the necessity of the condition. In summary, this work provides insights into appropriate measurement bases for constructing quantum state designs in systems with symmetries. 
\end{abstract}

\tableofcontents

\listoffigures

\listoftables

\chapter*{ABBREVIATIONS}

\begin{itemize}
    \item \textbf{BGS} \hspace{1cm} Bohigas-Giannoni-Schmit.
    \item \textbf{COE} \hspace{1cm} Circular orthogonal ensemble.
    \item \textbf{CSE} \hspace{1cm} Circular symplectic ensemble.
    \item \textbf{CUE} \hspace{1cm} Circular unitary ensemble.
    \item \textbf{ETH} \hspace{1cm} Eigenstate thermalization hypothesis.
    \item \textbf{GOE} \hspace{1cm} Gaussian orthogonal ensemble.
    \item \textbf{GSE} \hspace{1cm} Gaussian symplectic ensemble.
    \item \textbf{GUE} \hspace{1cm} Gaussian unitary ensemble.
    \item \textbf{KAM} \hspace{1cm} Kolmogorov-Arnold-Moser.
    \item \textbf{KCT} \hspace{1cm} Kicked coupled tops.
    \item \textbf{KHO} \hspace{1cm} Kicked harmonic oscillator.
    \item \textbf{LE} \hspace{1cm} Lyapunov exponent.
    \item \textbf{OBC} \hspace{1cm} Open boundary conditions.
    \item \textbf{OTOCs} \hspace{1cm} Out-of-time-ordered correlators.
    \item \textbf{PBC} \hspace{1cm} Periodic boundary conditions.
    \item \textbf{RMT} \hspace{1cm} Random matrix theory.
    \item \textbf{TI} \hspace{1cm} Translation invariant.
\end{itemize}\newpage

\nomenclature{$U$, $U_{\tau}$}{Unitary operator/Floquet operator}
\nomenclature{$\mathbb{I}$}{Identity operator}
\nomenclature{$\tr$}{Trace}
\nomenclature{$\mathcal{O}(t)$}{Time evolved operator}
\nomenclature{$\textbf{I}$}{Angular momentum vector}
\nomenclature{$\textbf{J}$}{Angular momentum vector}
\nomenclature{$\sigma^x$, $\sigma^y$, $\sigma^z$}{Pauli matrices}
\nomenclature{$\ket{\theta,\phi}$}{Spin coherent state}
\nomenclature{$\ket{\delta\theta,\thinspace \delta\phi}$}{Projection of tensor product of spin coherent state on an invariant subspace}
\nomenclature{$\mathcal{H}^{d}$}{Hilbert space of the dimension $d$}
\nomenclature{$C_{AB}(t)$}{OTOC at time $t$ for initial operators $A$ and $B$. }
\nomenclature{$\mathcal{E}_{\text{Haar}}$}{Ensemble of Haar random states}
\nomenclature{$\mathcal{E}$}{An ensemble of quantum states}
\nomenclature{$\varepsilon$}{A small positive real number}
\nomenclature{$\parallel \cdot\parallel_{p} $}{Schatten-p norm}
\nomenclature{$\Delta^{(t)}_{\mathcal{E}}$}{Trace distance between the $t$-th moments of $\mathcal{E}$ and $\mathcal{E}_{\text{Haar}}$}
\nomenclature{$\pi_{j}$}{An arbitrary permutation operator}
\nomenclature{$\mathbf{\Pi}_{t}$}{Projector onto the permutation symmetric subspace of $t$-replicas of a Hilbert space}
\nomenclature{$S_{t}$}{Permutation group of $t$-elements}
\nomenclature{$\mathbf{T}_{k}$}{Projector onto a translation symmetric subspace with the momentum charge $k$.}
\nomenclature{$\mathbf{Z}_{k}$}{Projector onto a $Z_2$-symmetric subspace with the momentum charge $k$.}
\nomenclature{$U(d)$}{Unitary group of dimension-$d$}
\nomenclature{$U_{\text{TI}}(d)$}{Unitary subgroup of dimension-$d$ containing all the translation symmetric unitaries}
\nomenclature{$T$}{Translation operator}
\nomenclature{$\mathbb{E}(\cdot)$}{Expectation value of a quantity}
\nomenclature{$\mathcal{E}^{k}_{\text{TI}}$}{Ensemble of translation invariant states with momentum charge $k$}
\nomenclature{$\mathcal{B}$}{A basis of states}
\nomenclature{$\text{gcd}(N, j)$}{Greatest common divisor of two numbers $N$ and $j$}
\nomenclature{$\Delta(\mathbf{Q}_{k}, \mathcal{B})$}{Quantifier of violation of sufficient condition shown by a measurement basis with respect to the symmetry operator $\mathbf{Q}$ and charge $k$}
\nomenclature{$S_{AB}$}{Swap operator between the Hilbert spaces $\mathcal{H}_{A}$ and $\mathcal{H}_{B}$}
\nomenclature{$H$}{Hamiltonian}
\printnomenclature

\setcounter{page}{1}     
\pagenumbering{arabic}   

\chapter{Introduction}

\section{Introduction}
The chaos of the universe, often nothing short of mesmerizing, reminds us that even a small butterfly can bring unexpected changes when we least expect them. One can start by asking the question, what is chaos? In classical physics, a bounded dynamical system exhibiting exponential sensitivity to initial conditions is said to be chaotic \cite{ott2002chaos, strogatz2018nonlinear, devaney2018introduction}.
This means two nearby phase space trajectories separate exponentially 
at a rate given by the characteristic Lyapunov exponent (LE) of the system. As a result, small disturbances in the initial conditions of a perfectly deterministic system may lead to completely different trajectories in the phase space as the system evolves over a sufficient length of time. This is what is popularly known as the ``butterfly effect" --- the flap of wings of a butterfly in Brazil can cause a tornado in China \cite{lorenz1963deterministic}. Therefore, the system loses its long-term predictability even though the equations describing its dynamics are completely deterministic. In such a scenario, the trajectories typically do not settle on fixed points, periodic orbits, or limit cycles \cite{ott2002chaos, strogatz2018nonlinear, devaney2018introduction}. Despite extensive study of chaos in dynamical systems, no universally accepted mathematical definition exists. According to  Devaney's popular text on chaos \cite{devaney2018introduction}, a system can be considered chaotic if it meets any two of the following three criteria: (i) sensitive dependence on initial conditions, (ii) transitivity and (iii) existence of a dense set of periodic orbits. Infact, meeting any two of the above automatically implies the third one \cite{banks1992devaney}. Hence, a transitive system with a sensitive dependence on initial conditions can be considered a chaotic system. 

 


While classical physics can explain the macroscopic world, quantum mechanics is essential for a completely accurate description of the world at the atomic scale. The double slit experiment, quantum tunneling, photon statistics, and electron diffraction are a few phenomena that have no classical explanations \cite{scully1997quantum, shankar2012principles}. Then, a natural pivotal question to ask is the presence of chaos at the atomic level.
This question has been a much-debated issue due to the fact that the unitary evolution of quantum mechanics preserves the inner product between two initial state vectors. Therefore, the initial disturbance doesn't increase exponentially. Instead, it stays constant throughout the evolution. However, suppose one describes classical mechanics in terms of the evolution of probability densities. In that case, one can show
that even classical dynamics preserve the overlap of initial probability densities throughout evolution \cite{koopman1931hamiltonian, neumann1932operatorenmethode, jordan1961lie}. This seems reasonable as the analogue of a wave function in quantum mechanics is a probability density in classical phase space. Chaos, as characterized by the sensitivity to initial conditions, refers to the exponential departure of nearby \textit{trajectories} and not \textit{probabilities}. Although this explanation motivates one to look beyond the overlap of state vectors in quantum mechanics, this does not answer several pertinent questions: How far does our classical understanding of chaos lead us into the quantum world? Are there ways to characterize chaos in quantum systems the way LEs do for classically chaotic systems? How does chaos emerge from the underlying quantum dynamics in the classical world? If there exists a notion of quantum chaos, what is the route of chaos starting from an integrable system? Classically, this problem has been extensively studied in the last century, starting from Poincaré and culminating in the celebrated Kolmogorov-Arnold-Moser (KAM) theorem \cite{arnold2009proof, kolmogorov1954conservation, moser1962invariant, moser1967convergent, dumas2014kam, Moser+2001, poschel2009lecture}. Therefore, another relevant question one can ask is about the role of the KAM stability \footnote{Here, KAM stability refers to the robustness of the classical phase space in integrable systems under weak, generic perturbations. For further details, see Sec. \ref{kamts}.} of integrable systems in the quantum limit \cite{burgarth2021kolmogorov}.

It must be emphasized here that searching for a proper characterization of quantum chaos is not simply fixing a ``definition". There are profound consequences that are intimately tied to quantum chaos. Ehrenfest correspondence principle (also known as classical-quantum correspondence) tells us the length of time till which quantum expectation values of observables are going to follow classical equations of motion. The characteristic ``break-time" when these two depart is exponentially small if the underlying dynamics of the system in the classical limit are chaotic as opposed to regular \cite{berman1978condition, toda1987quantal, gu1990evidences}. This line of thought was picked up by Zurek, who reached provocative conclusions by suggesting that Hyperion, one of the moons of Saturn, should deviate from its trajectory within 20 years \cite{zurek1998decoherence}! His solution to why this has not happened was the role played by environmental decoherence \footnote{In quantum mechanics, decoherence can be naively understood as the process by which a system loses or shares its information with the environment due to the inherent coupling between the two.}, which suppresses the quantum effects responsible for the breakdown of Ehrenfest correspondence. The role of chaos in the rate of decoherence was studied by Zurek and Juan Pablo Paz, who showed that LEs should determine how rapidly the system decoheres \cite{paz2002environment}. Furthermore, the connections of non-linear dynamics, chaos, and thermalization are at the cornerstone of statistical mechanics. How rapidly and under what conditions an isolated quantum system thermalizes is connected to the underlying chaos in the system dynamics.


Quantum chaos also has a crucial role to play in quantum information processing. This is because, though quantum systems do not show sensitivity to initial states, they can show sensitivity to small changes in the Hamiltonian parameters \cite{peres1984stability}. This sensitivity presents challenges in simulating the quantum chaotic systems on various quantum platforms, especially in the face of hardware errors \cite{chinni2022trotter}. Moreover, in quantum control protocols, one has to reach a target state despite many-body chaos and unavoidable fluctuations in the control dynamics \cite{tomsovic2023controlling}. Therefore, the design of control protocols must consider the consequences of chaos. It is also worth noting that the sensitivity of quantum systems can be a resource in quantum parameter estimation protocols due to its relation with the quantum Fisher information.
Hence, understanding quantum chaos can lead to harnessing the properties of quantum systems to develop superior information processing devices for computation, communication, and metrology. On the other hand, quantum chaos can address several foundational questions regarding the quantum-to-classical transition, defining the quantum-classical border, and the emergence of classical chaos from the underlying quantum mechanics.

This thesis aligns with the overarching context we have outlined so far. The \textit{out-of-time ordered correlators} (OTOCs) play an analogous role in the quantum regime as the LEs in the classical limit \cite{larkin}. Two-fourths of this thesis are devoted to studying the OTOCs in two different systems, namely, kicked harmonic oscillator (KHO) \cite{dileep2024} and kicked coupled tops (KCTs) \cite{varikuti2022out}. The former system has the special property that it is a non-KAM system in the classical limit, meaning that it is highly sensitive to parameter variations. Such changes are also reflected in the corresponding quantum dynamics, making quantum simulations of these systems on quantum computing platforms highly challenging. Through a comprehensive analysis of the OTOCs, we provide a way to rigorously characterize the dynamical sensitivity of this system in the quantum regime. Due to structural instabilities, the non-KAM class of systems has recently been shown to display large-scale trotter errors \footnote{The Trotter error arises when approximating the quantum evolution by splitting the total Hamiltonian into simpler parts and evolving each part separately over small time steps. This error decreases as the time steps become smaller.} on digital quantum computers \cite{chinni2022trotter}. 
We expect that our analysis will pave the way for the development of efficient simulation techniques for these systems. Due to their exceptionally high sensitivity, these systems can be used as quantum sensors. In this thesis, we provide a comprehensive analysis of the quantum Fisher information in the quantum KHO model. 
On the other hand, the KCT system does not show such sensitivity. The OTOC analysis in this system aims to understand the role of its rich mixed-phase space dynamics in the scrambling of information. In the latter part of this thesis, we depart from the conventional OTOC picture of quantum chaos. We examine quantum chaos in many-body quantum systems through a measurement-induced phenomenon called \textit{emergent state designs} \cite{cotler2023emergent} in the presence of symmetries \cite{varikuti2024unraveling}.

\section{Classical Integrability to chaos --- via KAM theorem}
\label{kamts}
Integrability is a centuries-old concept and cornerstone of classical physics. A classical Hamiltonian system of $d$-degrees of freedom can be characterized by a $2d$-number of first-order differential equations involving pairs of canonically conjugate variables. The integrability is traditionally referred to as the exact solvability of these equations by means of standard integration techniques, thereby writing the evolution of phase space variables in terms of elementary mathematical functions of time \cite{dumas2014kam}. Whenever a direct integration is not feasible, the conventional approach is to seek solutions by changing variables via canonical transformations. However, whether and when such a transformation exists for the given Hamiltonian is a hard question to answer. In 1853, J. Liouville first provided a rigorous answer to this question, showing that the integrability requires at least $d$-number of independent constants of motion (COMs), which are in complete involution. The COMs are smooth functions over the phase space that do not change with time along the phase space trajectories, and the involution means that their Poisson brackets mutually vanish. It is to be noted that a system of $2d$ independent linear differential equations would require $2d-1$ COMs for complete solvability. Thanks to the symplectic structure of the Hamiltonian dynamics, only $d$-COMs are needed \cite{babelon2003introduction, arnol2013mathematical}. 
Then, the Hamiltonian can be transformed into a new set of action-angle variables, where the COMs play the role of the actions, and the angle coordinates remain cyclic. If $\{I_i\}_{i=1}^{d}$ denote the action variables and $\{\theta_i\}_{i=1}^{d}$ are the corresponding angle variables, Hamilton's equations of motion are given by
$I_i(t)=I_i(0)=\text{const}$, and $\theta_i(t)=(\omega_i t+\theta_0)\mod 1 $, where $\omega_i=\dfrac{\partial H(I_1, I_2, \cdots, I_d)}{\partial \theta_i}$ are interpreted as characteristic frequencies. Fixing each of the COM to a constant value results in a $d$-dimensional surface in the $2d$-dimensional phase space. This surface has the properties of a $d$-torus and is usually termed an invariant torus --- any trajectory set out on this torus traverses on it forever with the characteristic frequencies. These tori are said to be resonant if the trajectories always return to their initial positions after a certain time. Otherwise, if the trajectories never return to their initial conditions, such tori are called non-resonant. The corresponding trajectories are usually referred to as quasi-periodic. When the Hamiltonian is non-degenerate, meaning the characteristic frequencies are non-linear functions of the actions, resonant tori are distributed among non-resonant ones in a manner analogous to the distribution of rational numbers among irrational numbers on the real line $(\mathbb{R})$. Non-degeneracy is crucial in establishing the stability of the integrable systems against small generic perturbations. While integrable systems are generally scarce in nature, they still accurately describe much of the physical world as we know it. Notable examples of these systems include Kepler's two-body problem, the harmonic oscillator, and Euler tops \cite{reichl2021transition}.

\subsection{Kolmogorov-Arnold-Moser stability}

In the early twentieth century, mathematicians and physicists actively investigated the robustness of integrable systems in the presence of weak perturbations. Suppose $H(I)$ denotes the initial integrable Hamiltonian in the actin-angle coordinates. Then, the perturbed Hamiltonian will take the form $H'(I, \theta)=H(I)+\varepsilon f(I, \theta, t)$, where $\varepsilon$ is the perturbation strength and $f(I, \theta, t)$ denotes the perturbation. Then, this pertinent question was asked: \textit{What happens to the conserved quantities when an integrable system is subjected to a smooth generic and weak perturbation that preserves the symplectic structure of Hamilton's equations of motion?}. This is equivalent to asking if the perturbed system admits new COMs derived from the perturbative expansion of the original system's COMs. 
This problem was initially tackled by Poincaré, who observed that the presence of the resonant tori leads to the divergence of the perturbative expansion \cite{poschel2009lecture}. The sketch of this result involves seeking a generator $(\mathcal{\eta})$ of a symplectic transformation, which aims to convert the perturbed Hamiltonian into a new set of action-angle variables $\{(J_i, \phi_i)\}_{i=1}^{d}$ such that the final Hamiltonian remains free of $\phi_i$s up to the first order in the perturbation strength. The result further involves expanding both the generator and the perturbation in the Fourier series, subsequently leading to a relation between the Fourier coefficients of the same: $\eta_{K}(J)=-\dfrac{f_{k}(J)}{(2\pi i)(K.\hat{\omega})}$, where $K\in\mathbb{Z}^{d}-\{\hat{0}\}$ and $\hat{\omega}$ is the vector of characteristic frequencies. The Fourier coefficients diverge whenever \( K \cdot \hat{\omega} = 0 \), which is indeed the case for resonant tori. Additionally, even when the tori are not resonant, there are situations where \( K \cdot \hat{\omega} \) takes very small values, resulting in the non-convergence of the Fourier series expansion of the generator. This result became widely known as the \textit{small-devisor problem}. A naive interpretation of this result without further examination implies that integrable systems are generally unstable under perturbations and tend to become ergodic even when the perturbation is small. Moreover, it indicates that, in general, any generic system is non-integrable \cite{genecand1993transversal, markus1974generic, zehnder1973homoclinic}. However, this interpretation was later proven incorrect, although Poincaré's result itself was not false but rather an incomplete solution. Later advancements due to Kolmogorov, Arnold, and Moser revealed that the integrable systems, under certain conditions, are indeed stable. A definitive set of results, collectively known as the Kolmogorov-Arnold-Moser (KAM) theorem, was formulated to establish this stability, ultimately resolving the problem of the stability of integrable systems \cite{arnold2009proof, kolmogorov1954conservation, moser1962invariant, moser1967convergent, dumas2014kam, Moser+2001, poschel2009lecture}.

The first result of the KAM theorem was due to Kolmogorov, who adopted a different approach to study the small-divisor problem. Recall that the non-resonant tori occupy the phase space with a measure of one in non-degenerate systems, while the resonant tori only constitute a measure zero set. Instead of directly examining the generator function over the entire phase space, Kolmogorov was more interested in studying the stability of individual non-resonant tori and, consequently, the quasiperiodic trajectories over them. This prompted him to study a special set of invariant tori, excluding the resonant and the near-resonant tori. This set is characterized by the frequency vectors as $\mathcal{M}\equiv\{\omega\in \mathbb{R}^{d}| |K\cdot \omega|\geq \alpha/|K|^{\tau}\text{ for all }K\in\mathbb{Z}^{d} \text{ and }K\neq \hat{0}\}$, where $|K|=|K_1|+\cdots +|K_{d}|$, $\alpha>0$, and $\tau>n-1$. It turned out that the invariant tori corresponding to this set are highly non-resonant. Then, the Fourier series expansion of the generating function converges on this set, establishing the stability of the non-resonant tori in the phase space. Under weak perturbations, these tori get slightly deformed. This indicates that the non-degenerate integrable systems generally are stable under perturbations \cite{arnold2009proof, kolmogorov1954conservation, moser1962invariant, moser1967convergent, poschel2009lecture, reichl2021transition}.
Moreover, this set constitutes the full measure over the phase space while the complementary set with resonant and near-resonant tori constitutes measure zero in the phase space. This result was further refined and alternative proofs were provided by Arnold and Moser. The KAM theorem enables the application of perturbation theory to study the dynamics of perturbed integrable systems, which is often applicable to real-world systems.

\subsection{Routes to chaos}
\textbf{KAM route to chaos:}
When the perturbation strength is weak enough, the KAM theorem ensures the regularity of the typical (non-degenerate) near-integrable systems. As the perturbation strength increases, the resonant tori break first. The non-resonant tori follow the resonant ones as the perturbation is further increased. Larger perturbation strengths inevitably lead to the breaking of all invariant tori, rendering the system completely non-integrable. Most of the physical systems existing in nature are generally non-integrable. Examples include the famous three-body problem, double pendulum, and anharmonic oscillators. When no conserved quantities are present, these systems can exhibit ergodic properties. This means that a single phase-space trajectory can fill the entire phase-space region as it evolves indefinitely over time. However, it is important to note that the non-integrability itself is not a sufficient condition for ergodicity, and the latter is slightly a stronger condition for chaos than the former. The following hierarchy characterizes various levels of ergodicity in the dynamical systems \cite{cornfeld2012ergodic, halmos2017lectures, sep-ergodic-hierarchy}:
\begin{eqnarray*}
\text{Bernoulli}\subset\text{Kolmogorov}\subset\text{Strong mixing}\subset\text{Weak mixing}\subset\text{Ergodic}.
\end{eqnarray*}
Recall that the sensitive dependence on the initial conditions is one of the three main characteristics of chaotic systems. Let $\hat{X}(0)$ denote an initial condition of a phase space trajectory corresponding to an arbitrary chaotic system, and $\hat{X}(t)$ be its location after $t$-time steps. Then, another trajectory with a slightly different initial condition $\hat{X}(0)+\delta \hat{X}$ shows an exponential separation from the previous trajectory as it evolves, i.e., $\delta \hat{X}(t)/\delta \hat{X}(0)\sim e^{\lambda_{\text{LE}} t}$, where $\lambda_{\text{LE}}$ is the maximum positive LE. For chaotic systems, LE is always positive. In the hierarchy, the Bernoulli and Kolmogorov systems display chaotic behavior.

\textbf{Non-KAM route to chaos:}
Although the KAM stability is guaranteed for generic integrable systems, the validity of the KAM theorem rests upon a few fundamental assumptions. For example, the theorem presupposes that the unperturbed Hamiltonian is non-degenerate, which means that when expressed in the action-angle variables, the Hamiltonian takes the form of a nonlinear function involving only the action variables while the angle variables remain cyclic. In addition, the characteristic frequency ratios must be sufficiently irrational for the phase space tori to survive the perturbations. Upon failing to meet these assumptions, the tori will likely break immediately at any arbitrary perturbation, leading to structural instabilities. Near the structural instabilities (also called resonances), the integrability of the system gets completely lost to the perturbations. Most integrable systems generally satisfy the KAM conditions. There is, however, a family of non-KAM systems that do not follow the usual KAM route to non-integrability \cite{sankaranarayanan2001quantum, sankaranarayanan2001chaos}. At the resonances, the non-KAM systems display large-scale structural changes in the presence of perturbations. In the classical phase space, the resonances are generally associated with breaking the invariant phase space tori via the creation of stable and unstable phase space manifolds. Such a mechanism results in diffusive chaos in the phase space even when the perturbation is arbitrarily small. Thus, the non-KAM systems show high sensitivity to the small changes in the system parameters at the resonances. The harmonic oscillator model is a well-studied example of a non-KAM integrable system. 

\textbf{Mixed phase space dynamics}
Mixed phase space dynamics refers to the coexistence of regular and chaotic motions within the phase space of the given dynamical system \cite{trail2008entanglement, varikuti2022out}. In such systems, the phase space displays regions where the trajectories are regular, bounded, and predictable alongside regions where the trajectories show chaotic behavior, exhibiting sensitive dependence on initial conditions. This can arise whenever the system has hyperbolic fixed points as well as stable fixed points. This behavior is usually observed in systems that show a transition between integrability, non-integrability, and chaos as a parameter that characterizes the chaos in the system varies. Examples of systems include the standard textbook models for quantum chaos, such as kicked tops, kicked coupled tops, and kicked rotors. Despite coexistence, these regions are perfectly separable in the phase space.

\section{Many facets of quantum chaos}
Quantum chaos aims to understand the emergence of chaos and unpredictability in quantum systems. As initially proposed by Michael Berry, one way to tackle this problem is to study quantum systems that are chaotic in the classical limit \cite{berry1989quantum}. Quantum mechanically, the evolution of an isolated system is given by an arbitrary unitary operator. When two different quantum states evolve under the same unitary, the distance between them remains constant throughout the evolution. This renders a naive generalization of the definition of chaos to the quantum regime meaningless. Hence, searching for suitable signatures of quantum chaos through other means became necessary. Asher Peres, in his seminal paper, introduced Loschmidt echo \footnote{Note that an experimental implementation of the Loschmidt echo was done way before it found applications in quantum chaology \cite{hahn1950spin}.} as a diagnostic tool for chaos in quantum systems \cite{peres1984stability}. The Loschmidt echo measures the overlap between a quantum state that evolved forward in time under a given Hamiltonian ($H$) and then evolved backward in time under a perturbed or modified Hamiltonian ($H'$), expressed as $|\langle\psi |e^{iH't}e^{-iHt} | \psi\rangle|^2$. Subsequently, this quantity has been extensively studied across various physical settings, including classical limits, to understand its behavior in both chaotic and regular systems \cite{jalabert2001environment, jacquod2001golden, cerruti2003uniform, gorin2006dynamics, goussev2012loschmidt, prosen2002stability}. Around the time Asher Peres proposed the Loschmidt echo, Bohigas, Giannoni, and Schmit (BGS) conjectured that quantum systems with chaotic classical limits adhere to Wigner-Dyson statistics \cite{bohigas1984characterization} \footnote{For more details refer to Chapter \ref{chap-back}}. on the other hand, for the quantum systems with regular classical limits, it was conjectured that the level spacings follow Poisson statistics due to the absence of level repulsions \cite{berry1976closed}. 
Shortly afterward, a dynamic measure called spectral form factor (SFF) was studied as a probe for quantum chaos \cite{haake1991quantum} \footnote{The spectral form factor is given by \( \text{SFF} = |\text{Tr}(U)|^2 \), where \( U \) represents the unitary evolution operator of the system.}, which complemented the BGS conjecture in examining quantum chaos. The SFF is expressed as $\langle |\text{Tr}(e^{-iH})|^2 \rangle$, where the quantity within the angular brackets is averaged over statically similar system Hamiltonians. Like the level spacing statistics, the SFF shows universal behavior for chaotic quantum systems. In addition, various other related measures, including the average level spacing ratio \cite{atas2013distribution}, have been extensively used to investigate the onset of chaos in both single-body and many-body quantum systems, regardless of the existence of classical limits. While extremely useful in probing quantum chaos, their applicability is often limited due to the presence of symmetries and other aspects. For instance, one must ensure the system is desymmetrized while dealing with spectral statistics in systems with symmetries \cite{tekur2020symmetry}. On the other hand, the spectral form factor is not self-averaging. Hence, physicists have sought to identify various alternative probes of quantum chaos. With the advent of quantum information theory, the dynamical generation of entanglement in many-body systems became another useful quantum chaos indicator in many-body systems \cite{bandyopadhyay2002testing, bandyopadhyay2004entanglement}. Recently, a powerful probe called \textit{out-of-time ordered correlators} (OTOCs) has been introduced as a tool to investigate quantum chaos \cite{chaos1}. The OTOCs quantify the compatibility between a time-evolved observable and a stationary observable. The OTOCs, unlike other probes, have a direct correspondence to the maximum classical LE. In many-body systems, the OTOCs emerge as natural diagnostics of scrambling of initially localized information. This caused a large-scale interest in the quantum chaos community. Interestingly, the OTOCs were proposed in 1969 by Larkin and Ovchinnikov in the context of superconductivity \cite{larkin}. However, they have been recently revived in the study of the black hole information paradox \cite{pawan, chaos1}.

\section{Emergence of state designs --- a new signature of quantum chaos}
Chaos in many-body quantum systems is closely related to the idea of generating random quantum states through the scrambling mechanism. This renders quantum chaos a crucial resource in various protocols like quantum state tomography \cite{madhok2016characterizing, madhok2014information, sreeram2021quantum, sahu2022quantum, sahu2022effect}, information recovery in the Hayden–Preskill protocol \cite{rampp2023hayden}, and coherent state targeting through quantum chaos control \cite{tomsovic2023controlling, tomsovic2023controlling1}. On the other hand, preparing random quantum states and operators is an essential ingredient to explore a variety of quantum protocols, such as randomized benchmarking \cite{benchmarking1, knill2008randomized, benchmarking2}, randomized measurements \cite{vermersch2019probing, elben2023randomized}, circuit designs \cite{harrow2009random, brown2010convergence}, quantum state tomography \cite{smith2013quantum, merkel2010random}, etc., and has vast applications ranging from quantum gravity \cite{sekino2008fast}, information scrambling \cite{styliaris2021information, pawan}, quantum chaos \cite{haake1991quantum}, information recovery \cite{hayden2007black, yoshida2017efficient}, machine-learning \cite{huang2020predicting, huang2022quantum, holmes2021barren} to quantum algorithms \cite{tilly2022variational}. Quantum state $t$-designs were introduced to answer the pertinent question: \textit{How can one efficiently sample a Haar random state from the given Hilbert space?}
To this end, state designs correspond to finite ensembles of pure states uniformly distributed over the Hilbert space, replicating the behavior of Haar random states to a certain degree \cite{renes2004symmetric, klappenecker2005mutually, benchmarking2}. 
However, generating such states in experiments is a challenging task since it requires precise control over the targeted degrees of freedom with fine-tuned resolution \cite{morvan2021qutrit, proctor2022scalable, boixo2018characterizing}.

Motivated by the recent advances in quantum technologies \cite{gross2017quantum, blatt2012quantum, browaeys2016experimental, gambetta2017building}, the \textit{projected ensemble} framework has been introduced as a natural avenue for the emergence of state designs from quantum chaotic dynamics \cite{cotler2023emergent, choi2023preparing}. Under this framework, one employs projective measurements on the larger subsystem (bath) of a single bi-partite state undergoing quantum chaotic evolution. For each outcome on the bath subsystem, the smaller subsystem gets projected onto a pure state. These pure states, together with the outcome probabilities, are referred to as the projected ensemble. The projected ensembles remarkably converge to a state design when the measured part of the system is sufficiently large. This phenomenon, dubbed as \textit{emergent state design}, has been closely tied to a stringent generalization of regular quantum thermalization. Under the usual framework, predominately characterized through the Eigenstate Thermalization Hypothesis (ETH) \cite{deutsch1991quantum, srednicki1994chaos, ETH_ansatz_expt, d2016quantum, deutsch_18, eth_nonherm}, the bath degrees of freedom are traced out, and the thermalization is retrieved at the level of local observables. Whereas the projected ensemble retains the memory of the bath through measurements such that thermalization is explored for the sub-system wavefunctions. This generalization has been referred to as \textit{deep thermalization} \cite{ho2022exact, ippoliti2022solvable, ippoliti2023dynamical}.
The emergence of higher-order state designs has been explicitly studied in recent years under various physical settings \cite{cotler2023emergent, ho2022exact, ippoliti2022solvable, lucas2023generalized, shrotriya2023nonlocality, versini2023efficient}, including dual unitary circuits \cite{claeys2022emergent, ippoliti2023dynamical} and constrained physical models \cite{bhore2023deep} with applications to classical shadow tomography \cite{mcginley2023shadow} and benchmarking quantum devices \cite{choi2023preparing}. 
In the case of chaotic systems without symmetries, arbitrary measurement bases can be considered to witness the emergence of state designs. The presence of symmetries is expected to influence this property.

Symmetries in quantum systems are associated with discrete or continuous group structures. Their presence causes the decomposition of the system into charge-conserving subspaces. This results in constraining the dynamical \cite{ope4, ope5, friedman2019spectral} and equilibrium properties \cite{yunger2016microcanonical} of many-body systems \cite{nakata2023black, bhattacharya2017syk, balachandran2021eigenstate, kudler2022information, chen2020many, paviglianiti2023absence, agarwal2023charge, varikuti2022out}. 
When a generic system displays symmetry, ETH is known to be satisfied within each invariant subspace \cite{deutsch1991quantum, srednicki1994chaos, d2016quantum}. Deep thermalization, on the other hand, depends non-trivially on the specific measurement basis \cite{cotler2023emergent, bhore2023deep}.   
Motivated by this, we ask the following intriguing question in the later part of this thesis:
\textit{What's the general choice of measurement basis for the emergence of $t$-designs when the generator state abides by a symmetry?}
In order to address this question, we first adhere our analysis to generator states with translation symmetry. In particular, we consider the ensembles of the random translation invariant (or shortly T-invariant) states and 
investigate the emergent state designs within the projected ensemble framework. 
We then elucidate the generality of our findings by extending its applicability to other discrete symmetries.   



\section{Structure of the thesis}

This thesis is structured as follows:

* In the following chapter, we outline several essential mathematical tools and frameworks that will be beneficial for understanding the subsequent chapters of this thesis. 

* In Chapter \ref{nonKAMchap}, we employ OTOCs to study the dynamical sensitivity of a perturbed non-KAM system in the quantum limit as the parameter that characterizes the \textit{resonance} condition is slowly varied. In particular, we study the OTOCs when the system is in resonance and contrast the results with the non-resonant case. We support the numerical results with analytical expressions derived for a few special cases. We will then extend our findings concerning the non-resonant cases to a broad class of near-integrable KAM systems.

* In Chapter \ref{nonKAMchapsen}, we examine the quantum Fisher information generated by the KHO model and contrast the resonances with the non-resonances. We provide analytical arguments to complement the numerical results. 

* In Chapter \ref{KCTchap}, We study operator growth in a system of kicked coupled tops using the OTOC. In this work, along with the globally chaotic dynamics, we explore scrambling behavior in the mixed phase space. In the mixed phase space, we invoke Percival's conjecture to partition the eigenstates of the Floquet map into ``regular" and ``chaotic" and show that the scrambling rate for the mixed phase space dynamics can be predicted by OTOCs calculated with respect to the chaotic eigenstates. In addition to the largest subspace, we study the OTOCs across the entire system, encompassing all other subspaces.

* In Chapter \ref{deepthchap}, we examine the emergence of state designs from the random generator states exhibiting symmetries. Leveraging on translation symmetry, we analytically establish a sufficient condition for the measurement basis leading to the state $t$-designs. Subsequently, we inspect the violation of the sufficient condition to identify bases that fail to converge. We further demonstrate the emergence of state designs in a physical system by studying the dynamics of a chaotic tilted field Ising chain with periodic boundary conditions. To delineate the general applicability of our results, we extend our analysis to other symmetries.
    
* We conclude this thesis in Chapter \ref{conclusionthesis}, offering perspectives on future directions.

\newtheorem{theorem}{Theorem}[section]
\chapter{Background}\label{chap-back}
In this chapter, we provide essential tools and frameworks necessary for understanding analyses in the subsequent chapters. For better accessibility, mathematical concepts are elucidated with relevant examples. In the following, we briefly cover RMT, OTOCs, and quantum Fisher information, delve into representation theory tools, and comprehensively discuss quantum state designs with applications. Readers familiar with these concepts can skip this chapter.

\section{Random matrix theory and universal ensembles}
In the 1950s, while studying complex heavy nuclei, Eugene Wigner noticed that their spectral fluctuations could be modeled using the eigenvalues of random symmetric matrices with elements drawn independently from the Gaussian distribution $\sim\mathcal{N}(0, 1)$. Essentially, he suggested that the spectra of the complex nuclei share statistical features with the spectra of the random Gaussian symmetric matrices. This idea sparked a widespread development of RMT in subsequent years \cite{mehta2004random}, which finds applications across various branches of physics including atomic and nuclear physics \cite{weidenmuller2009random, mitchell2010random}, statistical mechanics \cite{forrester2010log}, condensed matter physics \cite{beenakker1997random}, and, notably, quantum chaos \cite{haake1991quantum} and quantum information theory \cite{collins2016random}. In the context of quantum chaos, this field was further advanced with the works of Dyson, who formalized the notion of Gaussian ensembles and classified them into three different ensembles, namely, Gaussian unitary ensemble (GUE), Gaussian orthogonal ensemble (GOE), and Gaussian symplectic ensemble (GSE) \cite{porter1965statistical}. The corresponding probability distribution functions of these ensembles are given by
\begin{equation}
  \setlength{\arraycolsep}{0pt}
  P(H) = \left\{ \begin{array}{ l l }
    &{} e^{-n\text{Tr}(H^2)/2}/\mathcal{N}_{\text{GUE}}\quad \text{($H$ is complex Hermitian)}\\
    &{} e^{-n\text{Tr}(H^2)/4}/\mathcal{N}_{\text{GOE}}\quad \text{($H$ is real symmetric)} \\
    &{} e^{-n\text{Tr}(H^2)}/\mathcal{N}_{\text{GSE}} \quad \text{($H$ is Hermitian quaternionic)}
  \end{array} \right.
\end{equation}
The probability measure over GUE is invariant under arbitrary unitary transformations, i.e., $P(H)=P(u^{\dagger}Hu)$ for any $u\in U(d)$, where $U(d)$ is the $d$-dimensional unitary group. Similarly, GOE and GSE are invariant under orthogonal and symplectic transformations, respectively. The matrices drawn from these ensembles display universal spectral characteristics. In particular, the eigenvalues of these matrices repeal each other so that the probability of finding a pair of eigenvalues close to each other is very small. This is evident from the joint distribution function of the eigenvalues, which takes the following form \cite{mehta2004random}:
\begin{eqnarray}
P(\lambda_{1}, \lambda_{2}, \cdots, \lambda_{d})=\dfrac{1}{\mathcal{N}_{\beta}}\exp\left\{-\dfrac{1}{2}\sum_{i=1}^{d}\lambda^2_{i}\right\}\prod_{j<k}\left| \lambda_{j}-\lambda_{k}  \right|^{\beta} ,    
\end{eqnarray}
where
\begin{eqnarray}
\mathcal{N}_{\beta}=(2\pi)^{d/2}\prod_{j=1}^{d} \dfrac{\Gamma(1+j\beta/2)}{\Gamma(1+\beta/2)}.     
\end{eqnarray}
and $\beta=1, 2$ and $4$ correspond to GOE, GUE and GSE. The eigenvalues in the above expression are arranged in descending order, i.e., $\lambda_1\geq \lambda_2\cdots\geq\lambda_d$. The probability of having two identical eigenvalues is clearly zero. While the original purpose of RMT was to study the properties of heavy nuclei, it was later conjectured that quantum chaotic systems display spectral properties resembling those of one of these ensembles, depending on whether the system possesses time-reversal symmetry \cite{bohigas1984characterization}. To be more specific, if $s$ denotes a random variable corresponding to the distance between two nearest neighbor eigenvalues of a matrix chosen randomly from the Gaussian ensembles, its probability density is given by $P(s)\propto s^{\beta}e^{-A_\beta s^2}$, where $A_\beta=\left[\dfrac{\Gamma((\beta+2)/2)}{\Gamma((\beta+1)/2)}\right]^2$. If the given Hamiltonian is not time-reversal symmetric, the spectral statistics follow that of the GUE. In contrast, the GOE ($T^2=1$) and the GSE ($T^2=-1$) admit the time-reversal symmetry, where $T$ is the time-reversal operator.

Following the above classification, Dyson further introduced three analogous unitary ensembles \cite{dyson1962threefold}, namely circular unitary ensemble (CUE), circular orthogonal ensemble (COE), and circular symplectic ensemble (CSE) based on their properties under time-reversal operation. While the Gaussian ensembles are suitable for modeling time-independent systems, the circular ensembles are necessary for modeling time-dependent systems, particularly Floquet systems. The CUE encompasses all unitary matrices from the unitary group $U(d)$, i.e., CUE $\cong U(d)$, and hence a group. On the other hand, the COE consists of the symmetric unitaries, i.e., for any $u\in\text{COE}(d)$, we have $u^{T}=u$. it is to be noted that, unlike the CUE, the COE does not form a group. Moreover, for any $u\in \text{CUE}(d)$, $u^{T}u\in \text{COE}(d)$. The circular ensembles provide an accurate description of the chaotic time-dependent quantum systems. 

RMT description of quantum chaos has been immensely successful. Besides, the systems whose classical limits give rise to regular dynamics have also been studied extensively by the quantum chaos community in RMT context. In their seminal paper, Berry and Tabor \cite{berry1976closed} conjectured that the eigenvalues of the quantum systems with regular classical limits show no correlations among them. They behave as if they are identical and independently distributed (iid) random variables. If $s$ is a random variable describing the spacing distribution between two adjacent levels of the Hamiltonian, the probability density of $s$ follows $P(s)=e^{-s}$. In the case of time-dependent systems (Floquet), the eigenphases of the unitary behave as if they were drawn uniformly at random from the unit complex circle.

\subsection{Out-of-time ordered correlators}
The OTOCs were first introduced in the context of superconductivity \cite{larkin} and have been recently revived in the literature to study information scrambling in many-body quantum systems \cite{ope2, ope1, ope4, ope5, lin2018out, shukla2022out}], quantum chaos [\cite{chaos1, pawan, seshadri2018tripartite, lakshminarayan2019out, shenker2, moudgalya2019operator, omanakuttan2019out, manybody2, chaos2, prakash2020scrambling, prakash2019out, varikuti2022out, markovic2022detecting}, many-body localization \cite{manybody3, manybody4, manybody1, huang2017out} and holographic systems \cite{shock1, shenker3}. Given two operators $A$ and $B$, the out-of-time-ordered commutator function in an arbitrary quantum state $|\psi\rangle$ is given by
\begin{eqnarray}\label{commutator}
C_{|\psi\rangle,\thinspace AB}(t)=\langle\psi| \left[A(t), B\right]^{\dagger}\left[A(t), B\right]|\psi\rangle,
\end{eqnarray}
where $A(t)=\hat{U}^{\dagger}(t)A\hat{U}(t)$ is the Heisenberg evolution of $A$ governed by the Hamiltonian evolution of the system. When expanded, the function $C_{|\psi\rangle,\thinspace AB}(t)$ contains two-point and four-point correlators. Since the time ordering in these four-point correlators is non-sequential, they are usually referred to as the OTOCs. The behavior of $C_{AB}(t)$ depends predominantly on the four-point correlators. Hence, the terms OTOC and the commutator function are often used interchangeably to denote the same quantity $C_{|\psi\rangle, AB}(t)$.

To understand how the OTOCs diagnose chaos, consider the phase space operators $\hat{X}$ and $\hat{P}$ in the semiclassical limit ($\hbar\rightarrow 0$), where the Poisson brackets replace the commutators. It can be readily seen that $\{ X(t), P \}^2=(\delta X(t)/\delta X(0))^2\sim e^{2\lambda t}$, where $\lambda$ is the Lyapunov exponent of the system under consideration, which is positive for chaotic systems. The correspondence principle then establishes that the OTOCs of a quantum system, whose classical limit is chaotic, grow exponentially until a time known as Ehrenfest’s time $t_{\text{EF}}$ that depends on the dynamics of the system \cite{schubert2012wave, rozenbaum2020early, jalabert2018semiclassical, chen2018operator}. For the single particle chaotic systems, $t_{\text{EF}}$ scales logarithmically with the effective Planck constant and inversely with the corresponding classical Lyapunov exponent --- $t_{\text{EF}}\sim \ln(1/\hbar_{\text{eff}})/\lambda$. For $t>t_{\text{EF}}$, the correspondence breaks down due to the non-trivial $\hbar$ corrections arising from the phase space spreading of the initially localized wave packets. While many recent works have revealed that the early-time growth rate of OTOCs correlates well with the classical Lyapunov exponent for the chaotic systems, it is, however, worthwhile to note that the exponential growth may not always represent true chaos in the system \cite{pappalardi2018scrambling, hashimoto2020exponential, pilatowsky2020positive, xu2020does, hummel2019reversible, steinhuber2023dynamical}. Nonetheless, by carefully treating the singular points of the system, one can show that the OTOCs continue to serve as a reliable diagnostic of chaos \cite{wang2022statistical, wang2021quantum}. 

\begin{figure}
    \centering
    \includegraphics[width=\textwidth,height=8.5cm]{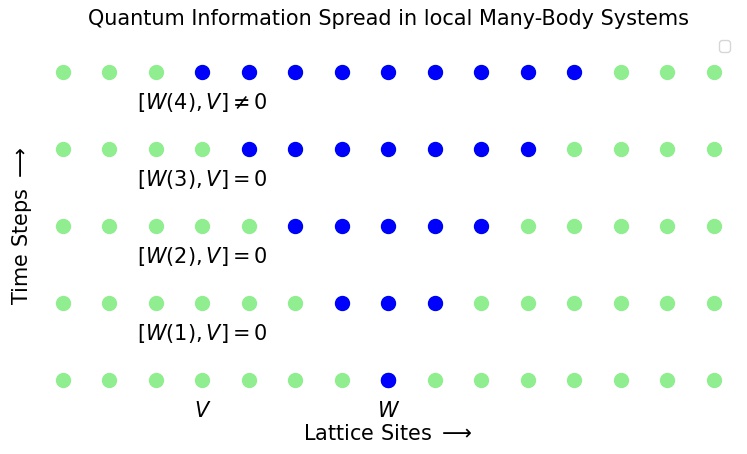}
    \caption{Schematic illustrating how information propagates in a local many-body quantum system with nearest neighbor interactions. }
    \label{fig:scramb}
\end{figure}

Recently, the OTOCs have been found to show intriguing connections with other probes of quantum chaos such as tripartite mutual information \cite{pawan}, operator entanglement \cite{styliaris2021information, zanardi2021information}, quantum coherence \cite{anand2021quantum} and Loschmidt echo \cite{yan2020information} to name a few. Moreover, the OTOCs have been investigated in the deep quantum regime and observed that the signatures of short-time exponential growth can still be found in such systems \cite{sreeram2021out}. Also, see Ref. \cite{pg2022witnessing} for an interesting comparison of the OTOCs with observational entropy, a recently introduced quantity to study the thermalization of closed quantum systems \cite{vsafranek2019quantum, vsafranek2019quantum1}.

Another intriguing feature of OTOC is that in multipartite systems, the growth of OTOC is related to spreading an initially localized operator across the system degrees of freedom. A representative figure is shown in Fig. \ref{fig:scramb}. The figure illustrates how local perturbations, caused by the applications of local unitaries in many-body systems with nearest-neighbor interactions, spread across the entire system. The commutator-Poisson bracket connection gives an analog of the classical separation of two trajectories with quantum mechanical operators replacing the classical phase space trajectories. Under chaotic many-body dynamics, an initially simple local operator becomes increasingly complex and will have its support grown over many system degrees of freedom, thereby making the initially localized quantum information at later times hidden from the local measurements \cite{pawan, hayden2007black, moudgalya2019operator, ope1}. As the operator evolves in the Heisenberg picture under chaotic dynamics, it will start to become incompatible with any other operator with which it commutes initially. Thus, the OTOCs provide a perfect platform to diagnose information scrambling in many-body quantum systems.

\section{Quantum Fisher Information}
This section briefly outlines classical and quantum Fisher information and its connection to quantum sensing accuracy. Quantum sensing involves the estimation of an unknown parameter $\theta$ \cite{helstrom1969quantum}. It mainly constitutes three steps: (i) preparation of an optimal state $|\psi_{\theta}\rangle$, (ii) learning $|\psi_\theta\rangle$ through measurements in a suitable basis, and (iii) inferring $\theta$ from the measurement data by employing an appropriate classical estimation strategy. The Cram\'{e}r-Rao inequality lower bounds the variance of the unbiased estimators of $\theta$, thereby providing a maximum derivable precision in estimating $\theta$ \cite{rao1992information, cramer1999mathematical}: $\Delta^2\theta\geq1/nI(\theta)$, where $n$ is the number of times the experiment is repeated. The classical Fisher information (CFI), as denoted by $I(\theta)$, is defined as the variance of the score function: $I(\theta)=\mathbb{E}_\theta\left[\partial_{\theta}(\ln P_{\theta}(X))\right]^2$, where $P_{\theta}(X)$ denotes probability distribution associated with the measurement output $X$. Therefore, the CFI is measurement dependent. When optimized over all the measurements (POVMs), the CFI attains a maximum value, which is commonly termed quantum Fisher information (QFI) \cite{braunstein1994statistical, braunstein1996generalized}. In a quantum sensing protocol, if $m$ quantifies the resources available, such as the number of probes, then QFI, in general, will scale as $I(\theta)\sim m$ (standard quantum limit). However, invoking quantum effects such as entanglement in step-(i) can lead to the Heisenberg limited sensing: $I(\theta)\sim m^2$ \cite{giovannetti2004quantum, giovannetti2006quantum}.

Geometrically, the CFI can be identified with the only monotone Riemannian metric on a smooth statistical manifold \cite{chentsov1982statiscal, morozova1991markov}. As a result, given a parametric probability distribution $P_\theta$, the second derivative of all monotonic distances $d(P_\theta, P_{\theta+\epsilon})$ with respect to $\epsilon$ are identical to the CFI upto a constant multiplicate factor, i.e., $\partial^2_{\epsilon}d(P_\theta, P_{\theta+\epsilon})|_{\epsilon\rightarrow 0}\equiv C I(\theta)$. In other words, the CFI measures the sensitivity of a probability distribution to infinitesimal variations in its parameters \cite{meyer2021fisher}. Analogously, given $|\psi_{\theta}\rangle=\hat{U}_{\theta}|\psi\rangle$, QFI can be defined as the sensitivity of the state $|\psi_{\theta}\rangle$ to small variations in $\theta$ \cite{helstrom1969quantum, braunstein1994statistical}:
\begin{equation}
I(\theta)= 4\lim_{\epsilon\rightarrow 0}\left( \dfrac{1-|\langle\psi_{\theta+\epsilon}|\psi_\theta\rangle|^2}{\epsilon^2} \right)
=\left. -2\partial^2_{\epsilon}|\langle \psi_{\theta+\epsilon}|\psi_\theta \rangle|^2\right|_{\epsilon=0}
\end{equation}
where, in the second equality, the fidelity metric $|\langle\psi_{\theta+\epsilon}|\psi_\theta\rangle|^2$ is assumed to be smooth and differentiable with respect to $\theta$. Physically, the parameter $\theta$ carries Hamiltonian information such as coupling strengths, magnetic or electric field strengths, etc. Note that the fidelity metric is given by the Fubini-Study metric for pure states, whereas the Bures metric is used for mixed states. Due to the divergent susceptibilities in the ground states, the quantum critical transitions have been proposed as a promising tool in quantum metrology \cite{zanardi2008quantum, invernizzi2008optimal, ivanov2013adiabatic, tsang2013quantum, macieszczak2016dynamical, rams2018limits, frerot2018quantum, chu2021dynamic, garbe2020critical, garbe2022critical, gietka2022understanding}. On the other hand, recent studies have used quantum chaotic dynamics as a resource in quantum metrology \cite{fiderer2018quantum, liu2021quantum}.

\section{Tools from representation theory}
This section provides a concise yet brief overview of essential mathematical techniques from representation theory, such as Haar integrals, schur-weyl duality, and Levy's lemma. These tools facilitate both analytical and numerical calculations in the subsequent chapters. However, note that this chapter merely provides an outline of these tools. For a detailed treatment, readers are encouraged to refer to \cite{brouwer1996diagrammatic, collins2006integration, zhang2014matrix}, as well as the references therein.
In the theory of quantum information and computation, one often encounters the integrals of the form $\int_{u\in U(d)}d\mu(u)f(u)$, where $d\mu(u)$ is the normalized measure associated with the unitary group $U(d)$, and $f$ denotes a measurable function compactly supported over $U(d)$ [see for e.g. \cite{zanardi2000entangling, popescu2006entanglement, zhang2014matrix} and references therein]. The notion of measure associates an `invariant area' element to the topological groups, meaning that the area element remains left or right invariant under the group action. 
To be more precise, consider a topological group as denoted by $G$. Then for any $g\in G$, $\mu(g)$ denotes the measure over $G$. If $\mu(g)=\mu(gh)$ for any $h\in G$, then $\mu(g)$ is said to be right-invariant. Likewise, $\mu(g)=\mu(hg)$ implies that the measure is left-invariant. Then, the Haar measure is the one which is both left and right-invariant under the group action. This happens when $G$ is a locally compact group. The compactness ensures that the group is bounded and closed, thereby making certain measurable functions over these groups finite, particularly in the context of Haar integrals. Examples of compact groups include the unitary group $U(d)$, special unitary group $SU(d)$, and orthogonal group $O(d)$. Its worthwhile to note that the latter two are two subgroups of the former. Since the unitary group $U(d)$ is compact, one can naturally associate it with a Haar measure, i.e., for a compactly supported function $f$, the following integral exists and is finite:
\begin{eqnarray}
\int_{u\in U(d)}f(u)d\mu(u)=\int_{u\in U(d)}f(ug)d\mu(ug)=\int_{u\in U(d)}f(gu)d\mu(gu)    
\end{eqnarray}
for any $g\in U(d)$. Moroever, the Haar measure satisfies $\int_{u\in U(d)}1d\mu(u)=1$ and for any $S\subset U(d)$, we have $\int_{u\in S}1d\mu(u)\geq 0$. 
Furthermore, the Haar measure formalizes the notion of sampling elements from compact groups or spaces uniformly at random. Under certain conditions, these integrals are amenable to exact solutions. These solutions are aided by the famous Schur-Weyl duality. In the following, we briefly outline the same with a few illustrative examples. 

\subsection{Schur-Weyl duality} 
The Schur-Weyl duality is a powerful tool to solve certain problems involving Haar integrals over compact lie groups. Here, we mainly focus on the integrals over the unitary group $U(d)$. For a detailed review and treatment of the Schur-Weyl duality for the unitary group with applications, refer to \cite{zhang2014matrix, mele2023introduction}. Also, for solving Haar integrals using diagrammatic calculus, refer to \cite{brouwer1996diagrammatic}. 

\textbf{Definition.} (Commutant) Let $\mathcal{L}(\mathbb{C}^{d})$ be the algebra of operators acting on $\mathcal{H}^{d}$. Then, for some $S\subseteq \mathcal{L}(\mathbb{C}^{d})$, its $t$-th order commutant can be defined as 
\begin{eqnarray}
 \text{Comm}(S, t)\equiv \{P\in \mathcal{L}(\mathbb{C}^{d})^{\otimes t}/ \left[P, Q^{\otimes t}\right]=0\quad \forall\thinspace Q\in S \}   
\end{eqnarray}

\textbf{Definition.} (Permutation operators) Let $S_t$ denote the permutation group over $t$-replicas of the Hilbert space ($\mathcal{H}^{d}$), each characterized by the indices $i_1, i_2, \cdots i_t$, then for any $\hat{\pi}\in S_{t}$, its action on the product state can be uniquely defined as 
\begin{eqnarray}
\hat{\pi} \left(|\phi_{i_1}\rangle \otimes |\phi_{i_2}\rangle \otimes\cdots \otimes |\phi_{i_t}\rangle\right) = |\phi_{\pi^{-1}(i_1)}\rangle \otimes |\phi_{\pi^{-1}(i_2)}\rangle \otimes\cdots \otimes |\phi_{\pi^{-1}(i_t)}\rangle, 
\end{eqnarray}
where the set $\{\pi^{-1}(i_1), \cdots\pi^{-1}(i_t)\}$ denotes the permutation of the indices $\{i_1, \cdots , i_t\}$ corresponding to the operator $\hat{\pi}$.

\begin{theorem}{\normalfont (Schur-Weyl duality)}
Let $U(d)$ be the unitary group supported over the Hilbert space $\mathcal{H}^{d}$. Then, for any $t\in\mathbb{Z}^{+}$, the commutant of $U(d)$, $\text{Comm}(U(d), t)$, is given by the $t$-th order permutation group $S_t$, i.e., 
\begin{eqnarray}
\normalfont{\text{Comm}(U(d), t)}\equiv \text{span}\left( \pi_{i}\in S_t\right).
\end{eqnarray}
\end{theorem}
The proof of this theorem is omitted here. An immediate consequence of this theorem is that an operator $P\in \mathcal{L}(\mathbb{C}^{d})^{\otimes t}$ commutes with all operators in the unitary group $U(d)$ if and only if $P$ is a linear combination of the elements of $S_{t}$ [\cite{roberts2017chaos}], i.e., 
\begin{eqnarray}
   P=\sum_{i=1}^{t!}a_i \hat{\pi_i}\quad\Longleftrightarrow  \quad u^{\dagger \otimes t}P u^{\otimes t}=P \quad \forall \thinspace u\in U(d). 
\end{eqnarray}

\subsubsection{Examples:}
We now consider a few illustrative examples to demonstrate the applicability of the Schur-Weyl duality in solving certain Haar integrals over the unitary group. 

(\textbf{1}). We first solve the integral $L=\int_{u\in U(d)}d\mu(u) \left(u^{\dagger}\otimes u\right)$. We first write individual elements of the operator $L$ as 
\begin{eqnarray}
 L_{ij, kl}&=&\int_{u\in U(d)}d\mu(u)\langle ij|u^{\dagger}\otimes u|kl\rangle \nonumber\\
 &=&\int_{u\in U(d)}d\mu(u) \left( \langle i|u^{\dagger}|k\rangle \langle j|u|l\rangle \right)\nonumber\\
 &=&\langle i|\left[ \int_{u\in U(d)}d\mu(u)\left( u^{\dagger} |k\rangle\langle j| u \right)  \right] |l\rangle
\end{eqnarray}
Clearly, the operator $M=\int_{u\in U(d)}d\mu(u) \left( u^{\dagger}|k\rangle\langle j| u \right)$ commutes with any $v\in U(d)$, i.e., 
\begin{eqnarray}
  v^{\dagger}Mv=\int_{u\in U(d)}d\mu(u)\left( v^{\dagger}u^{\dagger}|k\rangle\langle j| uv \right)  =\int_{u\in U(d)}d\mu(u)\left(u^{\dagger}|k\rangle\langle j| u \right)=M, 
\end{eqnarray}
where the second equality follows from the right invariance of the Haar measure under the group action. Therefore, the integral can be solved by writing it as a linear combination of the elements of $S_1$, i.e., 
\begin{eqnarray}
M=\int_{u\in U(d)}d\mu(u)\left(u^{\dagger}|k\rangle\langle j| u \right)=\alpha\mathbb{I},  \end{eqnarray}
where $\alpha$ can be obtained by equating the traces of the operators on both sides. It then follows that $M=\dfrac{\mathbb{I}}{d}\delta_{kj}$. This implies that $L_{ij, kl}=\dfrac{\delta_{kj}\delta_{il}}{d}$. It is now straightforward to get the final expression for $L$:
\begin{eqnarray}
   L=\sum_{i, j, k, l=0}^{d-1}|ij\rangle\langle kl| L_{ij, kl}= \dfrac{1}{d}\sum_{i, j, k, l=0}^{d-1}|ij\rangle\langle kl| \delta_{il}\delta_{jk}=\dfrac{1}{d}\sum_{ij}|ij\rangle\langle ji| =\dfrac{F}{d}, 
\end{eqnarray}
where $F$ is the swap operator. 

(\textbf{2}). We now consider an example involving moment operators of the unitary group. For an arbitrary $A\in \mathcal{L}(\mathbb{C}^{d})^{\otimes t}$, the moment operator of order $t$ can be written as 
\begin{eqnarray}
M_{t}(A)=  \int_{u\in U(d)}d\mu(u)\left( u^{\dagger \otimes t} A u^{\otimes t} \right).   
\end{eqnarray}
The right invariance of the Haar measure implies that $M(A)$ commutes with $v^{\otimes t}$ for all $v\in U(d)$. Hence, the solution can be written using the linear combination of the operators from $S_t$. For $t=1$, $M_{1}(A)=\text{Tr}(A)\mathbb{I}_{d}/d$. For $t=2$, the solution is given by
\begin{eqnarray}
M_2(A)=    \left(\dfrac{\text{Tr}(A)}{d^2-1}-\dfrac{\text{Tr}(AF)}{d(d^2-1)}\right)\mathbb{I}_{d^2}+\left( \dfrac{\text{Tr}(AF)}{d^2-1}-\dfrac{\text{Tr}(A)}{d(d^2-1)} \right)F
\end{eqnarray}

\subsection{Levy's lemma}
\textbf{Definition} (Lipschitz continuous functions). A function $f: X\rightarrow Y$ is Lipschitz continuous with Lipschitz constant $\eta$, if for any $x_1, x_2\in X$, it holds that 
\begin{eqnarray}
d_y(f(x_1), f(x_2))\leq \eta d_x(x_1, x_2),     
\end{eqnarray}
where $d_x$ and $d_y$ indicate the distance metrics associated with the spaces $X$ and $Y$,  respectively. The Lipschitz continuity is a stronger form of the uniform continuity of $f$ [\cite{o2006metric}], and $\eta$ upper bounds the slope of $f$ in $X$ [\cite{milman1986asymptotic, ledoux2001concentration}]. 

\begin{theorem}$\normalfont\text{\cite{milman1986asymptotic, ledoux2001concentration}}$
Let $f: \mathbb{S}^{d-1} \rightarrow \mathbb{R}$ be a Lipschitz function defined over a $(d-1)$-sphere $\mathbb{S}^{d-1}$, equipped with a natural Haar measure. Suppose a point $x \in \mathbb{S}^{d-1}$ is drawn uniformly at random from $\mathbb{S}^{d-1}$. Then, for any $\varepsilon > 0$, the following concentration inequality holds:
\begin{eqnarray}
\text{Pr}\left[ \left|f(x)-\mathbb{E}_{x\in \mathbb{S}^{d-1}}(f(x))\right|\geq \varepsilon \right]\leq 2\exp\left\{ \dfrac{-d\varepsilon^2}{9\pi^3\eta^2} \right\}    , 
\end{eqnarray}
where $\eta$ is the Lipschitz constant of $f$ and $c$ is a positive constant.
\end{theorem}
A proof of this theorem can be found in Ref. \cite{gerken2013measure}. Levy's lemma guarantees that the value of a Lipschitz continuous function at a typical $x\in \mathbb{S}^{d-1}$ is always close to its mean value, as given by $\mathbb{E}_{x\in \mathbb{S}^{d-1}}(f(x))$. The difference between the mean and a typical value is exponentially suppressed with the Hilbert space dimension.

\section{Quantum designs} 
The notion of \textit{designs} has its roots in the theory of numerical integration of polynomial functions over continuous domains \cite{hardin1996mclaren, bondarenko2013optimal}. Consider the problem of integrating a polynomial function over a three-dimensional unit sphere ($2$-sphere). Then, for an arbitrary $f(r, \theta, \phi)$, obtaining a closed form expression of the integral $\int f(r, \theta, \phi)r^2dr\sin\theta d\theta d\phi$ is typically beyond reach. This prompts one to rely on numerical techniques, which can often be challenging due to the continuum of points over which integration needs to be performed. The concept of designs addresses this pertinent question: How can one efficiently sample a finite set of points from a continuous domain, ensuring that the weighted sum of the polynomial at these points effectively equates to the integral? In the context of $n$-sphere, such a set is called a \textit{spherical design}. A spherical $t$-design represents a finite set of points over which the weighted sum of any $t$-th or less than $t$-th degree polynomial matches the average integration over the entire sphere. For a $2$-sphere, the vertices of a regular tetrahedron embedded in it form a $2$-spherical design [see Fig. \ref{fig:designs}a]. This means that the average value of any second-degree polynomial function can be calculated by evaluating the function only at four vertices of the tetrahedron. Since pure quantum states in Hilbert space, $\mathcal{H}^{d}$, can be regarded as points on a hypersphere of dimension $2d-1$, one can generalize the notion of spherical designs to the quantum states \cite{renes2004symmetric}.  

\begin{figure}
    \centering
    \includegraphics[scale=0.6]{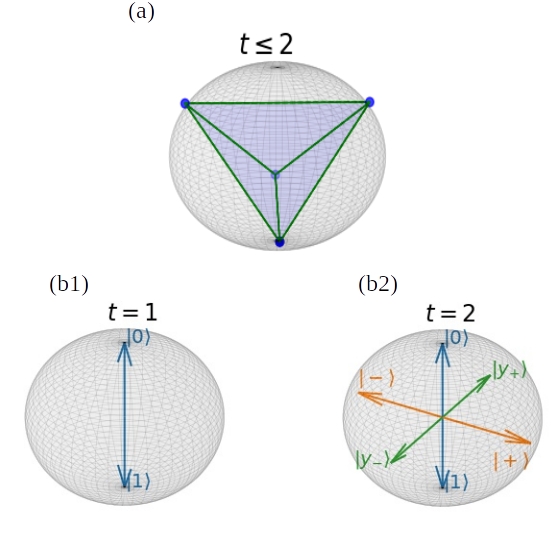}
    \caption{Schematic representation of designs. (a). Vertices of a regular tetrahedron embedded inside a $2$-sphere form a spherical $2$-design. (b1) and (b2) illustrate quantum state $1$ and $2$-designs on the bloch sphere.}
    \label{fig:designs} 
\end{figure}

A $t$-th order quantum state design ($t$-design) is an ensemble of pure quantum states that reproduces the average behavior of any quantum state polynomial of degree $t$ or less over all possible pure states, represented by the Haar average. An ensemble $\mathcal{E}\equiv\{p_i, |\psi_{i}\rangle\}$ is an exact $t$-design if and only if its moments match those of the Haar ensemble up to order $t$, i.e., 
\begin{eqnarray}
\sum_{i=0}^{|\mathcal{E}|-1}p_i\left(|\psi\rangle\langle\psi |\right)^{\otimes t}=\int_{|\psi\rangle}d\psi \left(|\psi\rangle\langle\psi|\right)^{\otimes t}. 
\end{eqnarray}
One can use Schur-Weyl duality to deduce an exact expression for the Haar integral on the right-hand side \cite{renes2004symmetric} in terms of the linear combination of permutation operators acting on $t$-replicas of the Hilbert space. The same is given by 
\begin{eqnarray}
\int_{|\psi\rangle\in\mathcal{E}_{\text{Haar}}}d\psi \left(|\psi\rangle\langle\psi|\right)^{\otimes t}= \dfrac{\mathbf{\Pi}_{t}}{\mathcal{D}}, \end{eqnarray}
where the Haar measure over $\mathcal{E}_{\text{Haar}}$ is denoted by $d\psi$, $\mathcal{D}=d(d+1)\cdots (d+t-1)$, and 
$\mathbf{\Pi}_{t}$ is the projector onto the permutation symmetric subspace of the replica Hilbert spaces $\mathcal{H}^{\otimes N}\otimes\mathcal{H}^{\otimes N}\otimes\cdots t$-times, i.e., $\mathbf{\Pi}_{t}=\sum_{\pi_i\in S_t}\pi_i$. It is to be noted that similar to the quantum state designs, unitary operator designs have also been investigated with active interest in recent years. Applications of quantum designs span various quantum protocols such as quantum state tomography \cite{madhok2016characterizing, madhok2014information, sreeram2021quantum, sahu2022quantum}, randomized benchmarking \cite{benchmarking1, knill2008randomized, benchmarking2}, randomized measurement protocols \cite{vermersch2019probing, elben2023randomized}, and many more [see Fig. \ref{fig:applications}]. In the following, we illustrate the importance of state designs in randomized measurement protocols with application to measuring the OTOCs.

\begin{figure}
    \centering
    \includegraphics[scale=0.4]{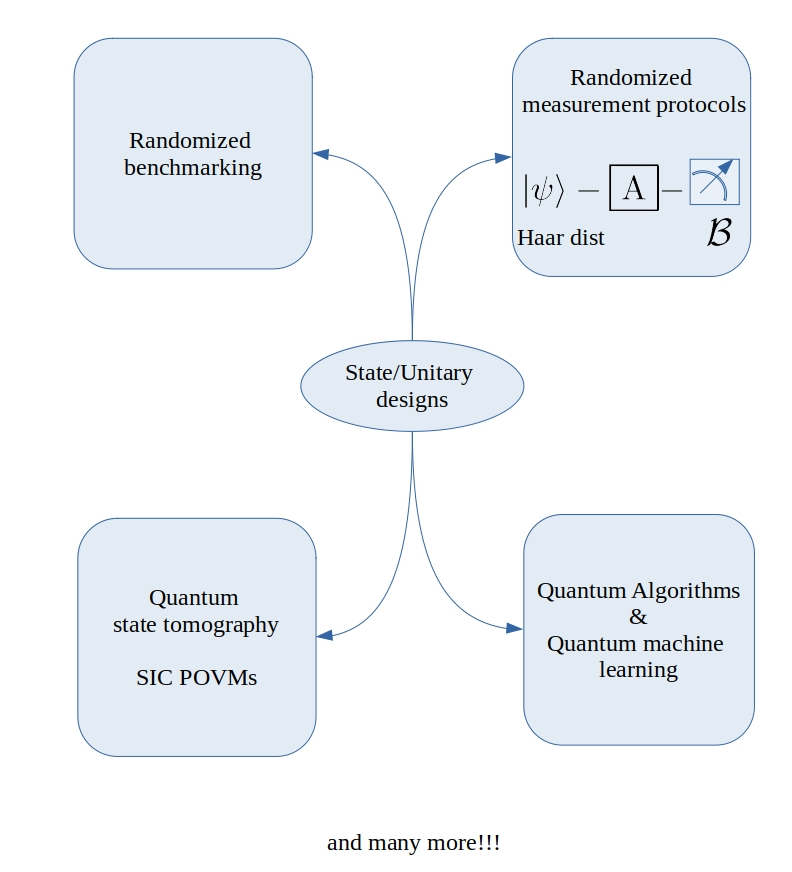}
    \caption{Schematic illustrating a few applications of quantum designs. }
    \label{fig:applications}
\end{figure}

\subsection{Application: Randomized measurement protocols for OTOC}
In quantum protocols, employing many randomized measurements can simplify the computation of certain quantities \cite{elben2023randomized}. This is achieved by leveraging the properties of the Haar measure, specifically by utilizing the Haar random states and unitaries. This approach harnesses the mathematical structure provided by the Haar integrals, leading to more manageable measurements of several information-theoretic quantities. The random measurement protocol has been first demonstrated for computing the traces of the powers of density matrices $\text{Tr}(\rho^n)$ \cite{van2012measuring}. Later, this technique has been extended to calculate the R{\'e}nyi entropies \cite{elben2018renyi}, OTOCs \cite{vermersch2019probing}, and multipartite entanglement \cite{knips2020multipartite, ketterer2019characterizing}. Here, we demonstrate the protocol for the OTOC computations and exemplify the significance of quantum state designs in this protocol. Computing the OTOCs in a physical experiment can be challenging as it requires time-reversal operations on the evolution of the given Hamiltonian. Moreover, if the system is chaotic, imperfect time reversal operations can lead to a significant accumulation of errors in the OTOC measurements. The randomized measurement protocol for the OTOCs completely bypasses the need for time reversals in such scenarios. To demonstrate, we consider a traceless Hermitian operator $W$ and a unitary operator $V$ as the initial operators for the OTOC measurements. Then, the four-point OTOC is given by $C(t)\propto\text{Tr}\left(W(t)V^{\dagger}W(t)V\right)$. It then follows that
\vspace{-0.5cm}
\begin{eqnarray}
\text{Tr}(\underbrace{W(t)}_{P}\underbrace{V^{\dagger}W(t)V}_{Q})&=&
\begin{minipage}{0.38\textwidth}
\centering
\vspace{0.3cm}
\includegraphics[scale=0.525]{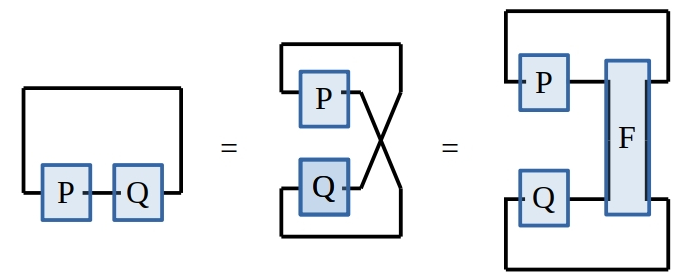}
\end{minipage}\nonumber\\
&=&\text{Tr}\left[F\left( W(t)\otimes V^{\dagger}W(t)V \right)\right]+\underbrace{\text{Tr}\left( W(t)\otimes V^{\dagger}W(t)V \right)}_{=0}\nonumber\\
&=&\text{Tr}\left[ \left( F+\mathbb{I} \right)\left( W(t)\otimes V^{\dagger}W(t)V \right) \right]\nonumber\\
&\propto&\text{Tr}\left[ \int_{d\psi}d\psi \left(|\psi\rangle\langle\psi|\right)^{\otimes 2}\left( W(t)\otimes V^{\dagger}W(t)V \right) \right] \nonumber\\
&=& \int_{\psi}d\psi \left[\langle\psi |W(t)|\psi\rangle \langle\psi | V^{\dagger}W(t)V |\psi\rangle\right].
\end{eqnarray}
In the first equality, we have used tensor network techniques to show that $\text{Tr}(PQ)=\text{Tr}(F(P\otimes Q))$. In the second equality, we have added an extra term $\text{Tr}(W(t)\otimes V^{\dagger}W(t)V)$, which is equal to $0$ as the operator $W$ is trace-less. In the fourth equality, we have replaced $(F+\mathbb{I})$ with the Haar integral $\int_{\psi}d\psi \left(|\psi\rangle\langle\psi|\right)^{\otimes 2}$. The final expression is a Haar integral over a second-degree quantum state polynomial. Since the Haar average requires that the experiment be performed infinitely many times, the second-order state designs can be used to reduce the complexity of the experiment. It is worth noting that the union of all mutually unbiased bases form state $2$-designs \cite{klappenecker2005mutually}.

\subsection{Projected ensemble framework}
\begin{figure}
\begin{center}
\includegraphics[scale=0.45]{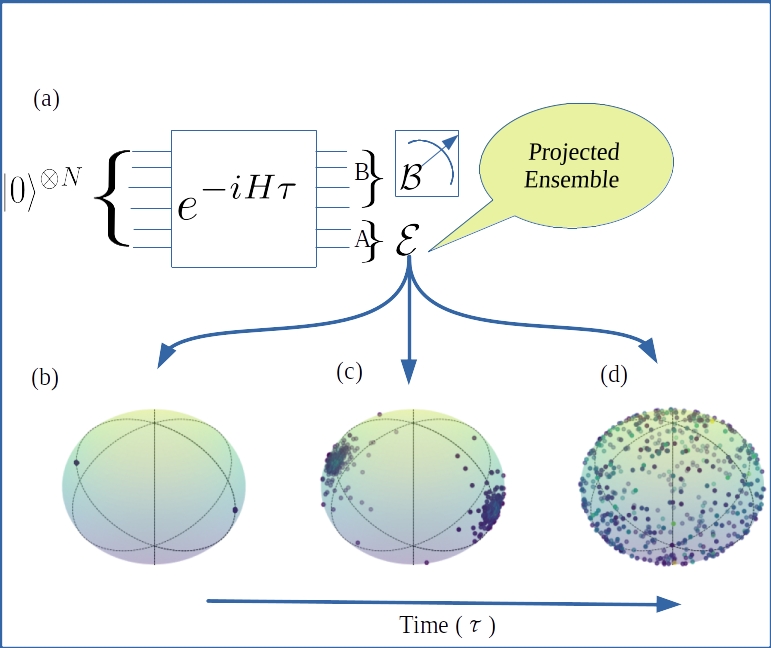}
\end{center}
\caption{\label{fig:sch} Schematic illustrating the projected ensemble framework for generating approximate higher-order state designs. (a) A chaotic Hamiltonian acts upon a trivial initial state. A subsystem of the final state is projectively measured. The panels (b)-(d) are representatives of the distribution of the pure states in the projected ensemble at different evolution times. }   
\end{figure}

Here, we briefly outline the projected ensemble framework \cite{cotler2023emergent}. The projected ensemble framework aims to generate approximate quantum state designs from a single chaotic or random many-body quantum state. The protocol involves performing local projective measurements on part of the system. First, consider a generator quantum state $|\psi\rangle\in\mathcal{H}^{\otimes N}$, where $\mathcal{H}$ denotes the local Hilbert space of dimension $d$ and $N=N_A+N_B$ denotes the size of the system constituting subsystems-$A$ and $B$. Then, projectively measuring the subsystem-$B$ gives a statistical mixture of pure states (or projected ensemble) corresponding to the subsystem-$A$. When the measurement basis is $\mathcal{B}\equiv\{|b\rangle\}$ supported over the subsystem-$B$, the resultant state corresponding to the subsystem-A can be written as follows:
\begin{eqnarray}\label{unnor}
|\Tilde{\psi}(b)\rangle =\left(\mathbb{I}_{2^{N_A}}\otimes|b\rangle\langle b|\right)|\psi\rangle \quad \text{ (unnormalized) },  
\end{eqnarray}
with the probability $p_b=\langle\Tilde{\psi}(b)|\Tilde{\psi}(b)\rangle=\langle\psi |b\rangle\langle b|\psi\rangle$. Since the projective measurements disentangle the subsystems, we can safely disregard the subsystem-$B$ and focus on the quantum state of subsystem-$A$. After normalizing the post-measurement state, we obtain 
\begin{eqnarray}
|\phi(b)\rangle=\dfrac{|\Tilde{\psi}(b)\rangle}{\sqrt{p_b}}=\dfrac{\left(\mathbb{I}_{2^{N_A}}\otimes\langle b|\right)|\psi\rangle}{\sqrt{\langle\psi |b\rangle\langle b|\psi\rangle}}.    
\end{eqnarray}
These states, together with the born probabilities, given by $\mathcal{E}(|\psi\rangle, \mathcal{B})\equiv\left\{p_b, |\phi(b)\rangle\right\}$ are collectively called a projected ensemble. The projected ensembles approximate higher-order quantum state designs if $|\psi\rangle$ is generated by a chaotic evolution \cite{ho2022exact, cotler2023emergent}, i.e., for $t\geq 1$, 
\begin{eqnarray}\label{Deltat}
\Delta^{(t)}_{\mathcal{E}}\equiv\left\|\sum_{|b\rangle\in\mathcal{B}}\dfrac{\left(\langle b|\psi\rangle\langle\psi |b\rangle\right)^{\otimes t}}{\left(\langle\psi |b\rangle\langle b|\psi\rangle\right)^{t-1}}-\int_{\phi\in\mathcal{E}^{A}_{\text{Haar}}}d\phi \left[|\phi\rangle\langle\phi|\right]^{\otimes t}\right\|_1\leq\varepsilon.
\end{eqnarray}
In the above equation, the trace norm (or Schatten $1$-norm) of an operator $W$, denoted as $\|W\|_{1}$, is defined as $\|W\|_{1}=\text{Tr}(\sqrt{W^{\dagger}W})$, which is equivalent to the sum of singular values of the operator. The two terms inside the trace norm are the $t$-th moments of the projected ensemble and the ensemble of Haar random states supported over $A$, respectively. 
The trace distance $\Delta^{(t)}_{\mathcal{E}}$ in Eq. (\ref{Deltat}) vanishes if and only if the ensemble $\mathcal{E}(|\psi\rangle, \mathcal{B})$ forms an exact $t$-design \cite{renes2004symmetric, klappenecker2005mutually, benchmarking2}. If the generator state $|\psi\rangle$ is Haar random, the trace distance $\Delta^{(t)}_{\mathcal{E}}$ exponentially converges to zero with $N_B$ for any $t\geq 1$, as demonstrated in Ref. \cite{cotler2023emergent}. In this case, the measurement basis can be arbitrary, and the behavior is generic to the choice of basis.  
On the other hand, for the generator state abiding by symmetry, the choice of measurement basis becomes crucial.

\section{Takeaways from this chapter}
This chapter outlines the basics of RMT, OTOCs, quantum Fisher information, and quantum state designs. In addition, we have provided a few essential mathematical tools from the representation theory, such as Haar integrals and Schur-Weyl duality, to help the reader understand the contents of the following chapters. For the sake of comprehensiveness, a few examples have been provided to demonstrate the applicability of the aforementioned tools. 

\newpage

\chapter{Out-of-time Ordered correlators in a non-KAM system}
\label{nonKAMchap}
\section{Introduction}
\label{section-1}
Quantum chaos is the study of quantum systems whose classical counterparts are chaotic. An overwhelming majority of such studies have considered Hamiltonian chaos in classical systems, where the celebrated Kolmogorov-Arnold-Moser (KAM) theorem is applicable and studied the signatures of classical chaos in the quantum domain. Recall that the KAM theorem states that if an integrable Hamiltonian system is subjected to a weak generic perturbation, most invariant tori in the phase space will persist with slight deformations \cite{arnold2009proof, kolmogorov1954conservation, moser1962invariant, moser1967convergent}. The chaos in such systems manifests through the gradual destruction of the invariant tori. However, the validity of the KAM theorem rests upon a few fundamental assumptions. For example, it presupposes that the unperturbed Hamiltonian is non-degenerate, meaning that when expressed in the action-angle variables, the Hamiltonian takes the form of a nonlinear function involving only the action variables while the angle variables remain cyclic. In addition, the characteristic frequency ratios must be sufficiently irrational for the phase space tori to survive the perturbations. Upon failing to meet these assumptions, the tori will likely break immediately at any arbitrary perturbation, leading to the emergence of widespread chaos \cite{poschel2009lecture}. Most realistic physical systems satisfy the KAM conditions. There is, however, a family of non-KAM systems that do not follow the usual KAM route to chaos.

In this chapter, we address the following question: How sensitive is the information scrambling to the perturbations in a quantum system whose classical counterpart is non-KAM? We tackle this question by studying the scrambling at classical resonances, which are the salient features of a non-KAM system. At the resonances, the non-KAM systems display large-scale structural changes in the presence of perturbations. In the classical phase space, the resonances are generally associated with breaking the invariant phase space tori via the creation of stable and unstable phase space manifolds. Such a mechanism results in diffusive chaos in the phase space even when the perturbation is arbitrarily small. Thus, the non-KAM systems show high sensitivity to the small changes in the system parameters at the resonances. Earlier works exemplified the dynamics of non-KAM systems by studying the systems that transit from being discontinuous to continuous depending upon the values of the appropriate parameters \cite{sankaranarayanan2001quantum, sankaranarayanan2001chaos, paul2016barrier, santhanam2022quantum}. We instead focus on a system that exhibits non-KAM behavior as a consequence of the classical degeneracy. With this objective in mind, we adopt the kicked harmonic oscillator (KHO) model as a paradigm to study information scrambling. We use out-of-time-ordered correlators (OTOCs) to diagnose and investigate the sensitivity of scrambling at the resonances and the non-resonances of the quantum KHO.

This chapter is structured as follows. In Sec. \ref{section-2}, we review some basic features of the KHO model, including resonances and non-resonances in both classical and quantum domains. We analyze the behavior of OTOCs in Sec. \ref{section-3} with a special emphasis given to the short-time dynamics in Sec. \ref{resonanceOTOC}. Thereafter, we focus on the asymptotic time dynamics of the OTOCs in Sec. \ref{non-resonanceOTOC} and show how the OTOCs distinguish the resonances from the non-resonances. In Sec. \ref{analtytical}, we analytically derive the OTOCs for a few special cases of the quantum KHO model. Then, in Sec. \ref{OTOCXP}, we provide a brief overview of the OTOCs for the phase space operators. We finally conclude this text in Sec. \ref{section-4} with a few remarks on the relevance of this work to the stability of quantum simulators.

\section{Model: Kicked harmonic oscillator}
\label{section-2}
We consider the harmonic oscillator model with a natural frequency $\omega$, subjected to periodic kicks by a nonlinear position-dependent field, having the following Hamiltonian \cite{chernikov1989symmetry, billam2009quantum, berman1991problem, kells2005quantum, afanasiev1990width, reichl2021transition, gardiner1997quantum, rechester1980calculation, ichikawa1987stochastic, ishizaki1991anomalous, daly1994classical, borgonovi1995translational, engel2007quantum}:
\begin{eqnarray}
H=\dfrac{P^2}{2m}+\dfrac{1}{2}m\omega^2X^2+ K\cos(kX)\sum_{n=-\infty}^{\infty}\delta\left(t-n\tau\right),
\end{eqnarray}
where $X$ is the position, $P$ is the momentum, $m$ is the mass of the oscillator and $k$ is the wave vector. The strength of the kicking is denoted by $K$. The time interval between two consecutive kicks is given by $\tau$. For simplicity, here and throughout this chapter, we shall take $m=k=1$. The system is parity invariant, i.e., $H(P, X)=H(-P, -X)$. 

In the action-angle coordinates, the Hamiltonian of the harmonic oscillator ($H_0$) scales linearly with the action coordinate --- the canonical transformation to the action-angle variables as given by $(X, P)= (\sqrt{2I/\omega}\cos\theta, \sqrt{2I\omega}\sin\theta)$ yields $H_0=\omega I$. Since $H_0$ is linear in $I$, the characteristic frequency of the phase space tori $(\partial H_0/\partial I=\omega)$ turns out to be independent of $I$, which violates one of the KAM assumptions that the unperturbed Hamiltonian must be non-linear in $I$ (the non-degeneracy condition). Hence, the harmonic oscillator is classified as non-KAM integrable. In the following, we briefly discuss the dynamical aspects of the KHO model in both the classical and quantum limits. 

\subsection{Classical dynamics}
\begin{figure}
\begin{center}
\includegraphics[scale=0.58]{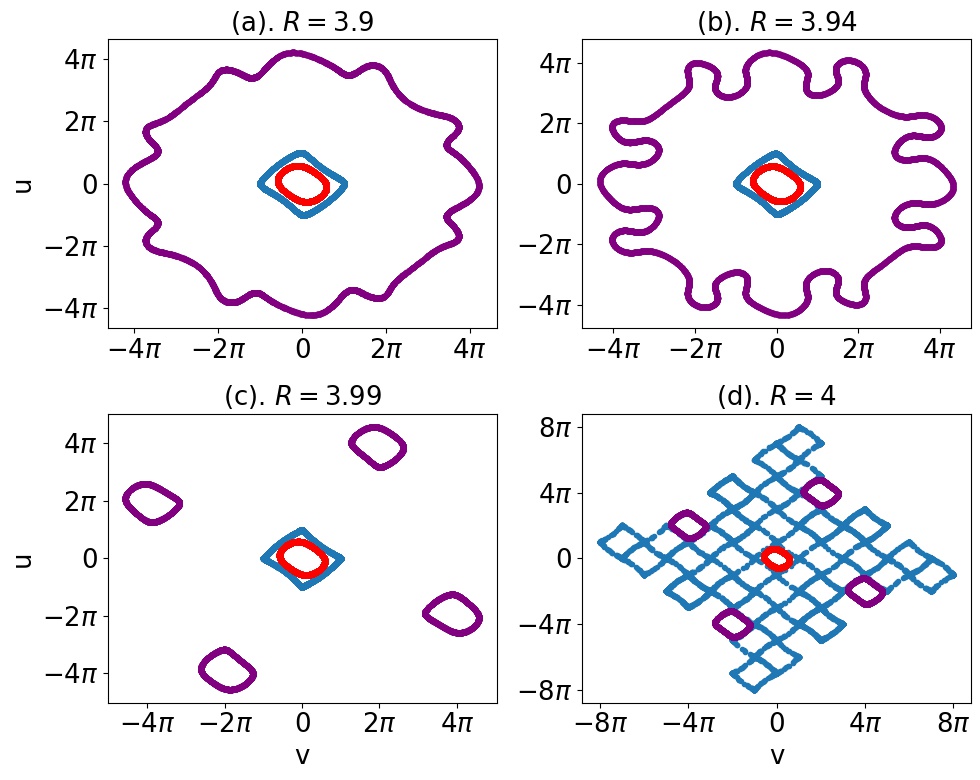}
\end{center}
\caption{\label{fig:poincare}The figure illustrates the classical phase space trajectories of the KHO model for three randomly chosen initial conditions in the vicinity of $R=4$, and each evolved for $10^4$ time steps. Here, we set $K=1$ and $\tau=1$. When $R=4$, for the given set of parameters, the phase space mainly constitutes three regions: the stochastic web, period four islands, and a period one regular island surrounding the origin. Each initial condition in the figure corresponds to one of these regions. When the system is fully non-resonant, the phase space is mostly regular. With $R$ varying from $3.9$ to $4$ in the sequential order depicted in panels (a)-(d), the trajectories become increasingly distorted until they completely break apart at the resonance. The periodic islands feature stable periodic points at their centers and unstable periodic points on their exterior. The boundaries of these islands collectively constitute the stochastic web that facilitates large-scale phase space diffusion.}
\end{figure}
In the classical limit, the dynamics of the KHO model can be visualized by the following two-dimensional map:
\begin{eqnarray}\label{dynmap}
u_{n+1} &=&(u_n+\epsilon\sin v_n)\cos(\omega\tau) +v_n \sin(\omega\tau),\nonumber\\
v_{n+1} &=&-(u_n+\epsilon\sin v_n)\sin(\omega\tau) +v_n \cos(\omega\tau),
\end{eqnarray}
where $u=P/\omega$, and $v=X$ denote scaled momentum and position variables, respectively and $\epsilon=K/\omega$ is the effective strength of perturbation. In the remainder of this chapter, we adopt the notation $\omega\tau=2\pi/R$. The system is considered classically resonant (or simply in resonance) whenever $R$ assumes integer values. For non-integer $R$, the system is non-resonant. The KHO map can be physically realized through the motion of a charged particle in a constant magnetic field, with the particle being kicked by the wave packets of an electric field that propagates perpendicular to the direction of the magnetic field \cite{zaslavsky2007physics}.

Under the rescaling of the variables $(X, P/\omega)\rightarrow (v, u)$, the unperturbed harmonic oscillator is rotationally invariant in the phase space. Besides, the kicking potential, $\cos X$, is invariant under the translations along the $X$-direction by integer multiples of $2\pi$. Thus, when $K$ is non-zero, there will be a natural competition between the translational and rotational symmetries of the phase space. This competition becomes more prominent for $R\in R_c\equiv\{1, 2, 3, 4, 6\}$ \cite{chernikov1989symmetry, afanasiev1990width}. The system is exactly solvable for $R=\{1, 2\}$. In the remaining cases, the phase space consists of periodic stochastic webs. These webs display translation and rotational invariance in the phase space. In particular, the cell structure of these webs closely resembles tessellations. For $R=4$, the web appears as a square lattice, and for $R=3$ and $6$, it is a Kagome lattice with hexagonal symmetry. In this work, we focus on the vicinity of $R=4$ for the studies of information scrambling. 

The stochastic webs resemble Arnold's diffusion in systems with more than two degrees of freedom. Their thickness is exponentially small in $\epsilon$ \cite{afanasiev1990width}. These webs arise from the exteriors of the periodic islands and mainly consist of the separatrices of the KHO. Any trajectory that sets out on the separatrices will eventually diffuse towards the infinity --- $\langle r_n\rangle\sim \epsilon\sqrt{n}$, where $r_n=\omega\sqrt{P^2+X^2}$ is the distance traversed by an average phase space trajectory after $n$ time steps. As a result, the mean energy of the system grows linearly at any finite perturbation --- $\langle E_n\rangle\sim\epsilon^2 n$. The diffusion, however, is suppressed for weak perturbations when $R$ takes non-integer values. In the latter case, the diffusion coefficient remains close to zero for small perturbations \cite{kells2005quantum}. Nevertheless, the differences between the resonances and the non-resonances become less apparent as the perturbation increases. Figure \ref{fig:poincare} shows the phase space trajectories of the KHO system for a few randomly chosen initial conditions in the vicinity of $R=4$ for $K=1$. The corresponding separatrix equation is given by $v=\pm(u+\pi)+2l\pi$, $l\in\mathbb{Z}$ \cite{afanasiev1990width}. The phase space is regular with distorted circles when the system is non-resonant. However, it can be seen from the figure that the trajectories get increasingly deformed as $R$ approaches $4$. At $R=4$, the phase space undergoes significant changes due to the creation of period-$4$ orbits. Such behavior has applications in the chaotic electron transport in semiconductor superlattices \cite{fromhold2001effects, fromhold2004chaotic}.

\subsection{Quantum dynamics}
The existence of stochastic webs in the classical phase space can have far-reaching consequences on the corresponding quantum dynamics, which we briefly discuss here to set the ground for the OTOC analysis in the next section. As the system is being kicked at periodic intervals of time, the time-evolution is given by the following Floquet operator:
\begin{equation}
\hat{U}_\tau=\exp\left\{-\dfrac{2\pi i}{R} \hat{a}^\dagger \hat{a}\right\}\exp\left\{-\dfrac{iK}{\hbar}\cos\hat{X} \right\},
\end{equation}
where $\hat{a}$ and $\hat{a}^{\dagger}$ are the annihilation and creation operators corresponding to the particle trapped in the harmonic potential, respectively. The position operator is $\hat{X}=\sqrt{\hbar/2\omega}(\hat{a}+\hat{a}^{\dagger})$. The irrelevant global phase $e^{-i\pi /R}$ is ignored. The quantum chaos in the KHO model has been extensively studied over many years \cite{berman1991problem, shepelyansky1992quantum, daly1994classical, kells2005quantum, daly1996non, engel2007quantum, kells2004dynamical}. Experimental proposals have also been put forth to realize the dynamics of quantum KHO using ion traps, and Bose-Einstein condensates \cite{gardiner1997quantum, carvalho2004web, gardiner2000nonlinear, duffy2004nonlinear, billam2009quantum}.

Recall that the classical KHO displays translational invariance whenever $R\in R_{c}$, which, in the quantum limit, translates into the existence of commuting groups of translations. In particular, under the translation invariance, the $R$-th powers of $\hat{U}_{\tau}$ commute with either one or two parameter groups of translations or displacement operators \cite{borgonovi1995translational}. As a result, the system admits extended Floquet states in the phase space, leading to an unbounded growth of the mean energy $\langle(a^\dagger a)(t)\rangle$. Moreover, these states also facilitate the dynamical tunneling of the localized coherent states \cite{carvalho2004web}. As we shall see later, the translation invariance is also crucial to obtaining certain analytical expressions of the OTOCs, which are otherwise intractable.

For the non-integer $R$ values, previous studies argued that quantum localization takes place, and energy growth will be stopped after some time \cite{borgonovi1995translational}. More specifically, when $R$ is an irrational number, the quantum KHO model can be mapped to a tight-binding approximation in the limit of $K\lesssim \hbar\pi$, which explains the localization of the quantum dynamics \cite{frasca1997quantum, kells2005quantum}. The effects of localization are also prevalent in the scrambling dynamics. We will explore this in more detail in the coming sections.

\section{ Scrambling in quantum KHO model}
\label{section-3}
In this section, we address the central goal of this chapter, which is to contrast the dynamics of information scrambling at resonances with that of non-resonances of the quantum KHO. The OTOCs are natural candidates to quantify the scrambling. While the OTOCs have been extensively studied in finite-dimensional systems, including those with time dependence \cite{prakash2020scrambling, sreeram2021out, borgonovi2019timescales, shen2017out, zamani2022out}, the studies on continuous variable systems have primarily focused on time independent settings \cite{zhuang2019scrambling, hashimoto2020exponential}. However, the system considered in this work is both continuous variable and time-dependent, leading to possible unbounded orbits in phase space and consequently to unbounded growth of OTOCs.

To study the OTOCs in the continuous variable systems, an appropriate choice of initial operators would be the canonical pair of position ($\hat{X}$) and momentum ($\hat{P}$) operators. However, for the reasons that become clear later, we instead consider the bosonic ladder operators $\hat{a}$ and $\hat{a}^{\dagger}$ as the initial operators. For an arbitrary state $|\psi\rangle$, we are interested in evaluating the quantity
\begin{eqnarray}\label{sf}
C_{|\psi\rangle,\thinspace \hat{a}\hat{a}^{\dagger}}(t)=\langle\psi| [\hat{a}(t), \hat{a}^\dagger]^\dagger [\hat{a}(t), \hat{a}^\dagger]|\psi\rangle,  
\end{eqnarray}
where $\hat{a}(t)=\hat{U}_{\tau}^{\dagger t} \hat{a} \hat{U}_{\tau}^t$ is the Heisenberg evolution of $\hat{a}$ under the dynamics of KHO and $t$ is the total number of time steps. Before analysing $C_{|\psi\rangle,\thinspace \hat{a}\hat{a}^{\dagger}}(t)$, it is helpful first to examine how $\hat{a}(t)$ depends on $t$. A single application of $\hat{U}_\tau$ on $\hat{a}$ gives
\begin{eqnarray}
\hat{a}(1)=\hat{U}^{\dagger}_{\tau}\hat{a}\hat{U}_{\tau}=e^{-2\pi i/R}\left[\hat{a}+\dfrac{iK}{\sqrt{2\hbar\omega}}\sin\hat{X}\right].
\end{eqnarray}
After $t$ recursive applications of $\hat{U}_{\tau}$, the Heisenberg evolution of $\hat{a}$ reads as
\begin{eqnarray}\label{bosonic-evolution1}
\hat{a}(t)e^{2\pi i t/R}=\hat{a}+\dfrac{iK}{\sqrt{2\hbar\omega}}\sum_{j=0}^{t-1}e^{2\pi ij/R}\sin\hat{X}(j),
\end{eqnarray}
where $\hat{X}(j)=\hat{U}^{\dagger j}_{\tau}\hat{X}\hat{U}^j_{\tau}$ and $\sin\hat{X}=[e^{i\hat{X}}-e^{-i\hat{X}}]/2i$. From Eq. (\ref{bosonic-evolution1}), we make the following immediate observations: (i) the total number of terms on the right-hand side scales linearly with $t$ and (ii) the operator $\sin\hat{X}(j)$ is always bounded, i.e., $\|\sin\hat{X}(j)|\psi\rangle\|\leq|\|\psi\rangle\|$ for any $|\psi\rangle$. As a result, the asymptotic growth of  $\hat{a}(t)$ can be at most linear. Long-time dynamics of several other quantities, such as mean energy growth, follow directly from this result. For instance, in an arbitrary Fock state $|n\rangle$, one can use Eq. (\ref{bosonic-evolution1}) to show that the mean energy is always bounded above by a quadratic function of $t$, i.e.,
\begin{equation}
\langle n|(a^\dagger a)(t)|n\rangle\leq n+\sqrt{\dfrac{2n}{\hbar\omega}}Kt+\dfrac{K^2t^2}{2\hbar\omega}.
\end{equation}

In a similar way, we use Eq. (\ref{bosonic-evolution1}) to learn the behavior of $C_{|\psi\rangle,\thinspace \hat{a}\hat{a}^{\dagger}}(t)$. To do so, we take
\begin{eqnarray}\label{com}
\left[\hat{a}(t), \hat{a}^{\dagger}\right]e^{2\pi it/R}=1+\dfrac{iK}{\sqrt{2\hbar\omega}}\sum_{j=0}^{t-1}e^{2\pi ij/R}\left[\sin\hat{X}(j), \hat{a}^{\dagger}\right].
\end{eqnarray}
For $K=0$, the system is just an integrable harmonic oscillator and the OTOC remains a constant in any given state $|\psi\rangle$ for all $t\geq 0$ --- $C_{|\psi\rangle,\thinspace \hat{a}\hat{a}^{\dagger}}(t)=1$. This means that the initial operator $\hat{a}$ remains fully regular, i.e., the operator $\hat{a}(t)$ retains the diagonality in the coherent state basis. For $K\neq 0$, as the system evolves, the operator will start to scramble into the operator Hilbert space via the mixing of eigenstates, which is hinted at by the positive growth of the OTOC. For example, after one time step, we can calculate explicitly that
\begin{equation}\label{first_step}
C_{|\psi\rangle,\thinspace \hat{a}\hat{a}^{\dagger}}(t=1)=1+\dfrac{K^2}{4\omega^2}\langle\psi|\cos^2\hat{X}|\psi\rangle\geq 1\text{ (for all }|\psi\rangle).
\end{equation}
For $t>1$, in general, a closed form expression for $C_{|\psi\rangle,\thinspace \hat{a}\hat{a}^{\dagger}}(t)$ is out of reach. Hence, we resort to numerical methods to probe the OTOCs. To be precise, we numerically compute the OTOCs by considering a weak ($K\lesssim 1$) and a moderately strong ($K\sim O(1)$) kicking strength under both resonance ($R=4$) and non-resonance ($R=3.9$) conditions. We choose two initial quantum states: the vacuum state $|0\rangle$ and a coherent state $|\alpha\rangle$. In what follows, we shall first examine the early-time behavior of the OTOCs, then proceed to analyze the long-time dynamics.
\begin{figure*}
\includegraphics[width=\textwidth,height=6.5cm]{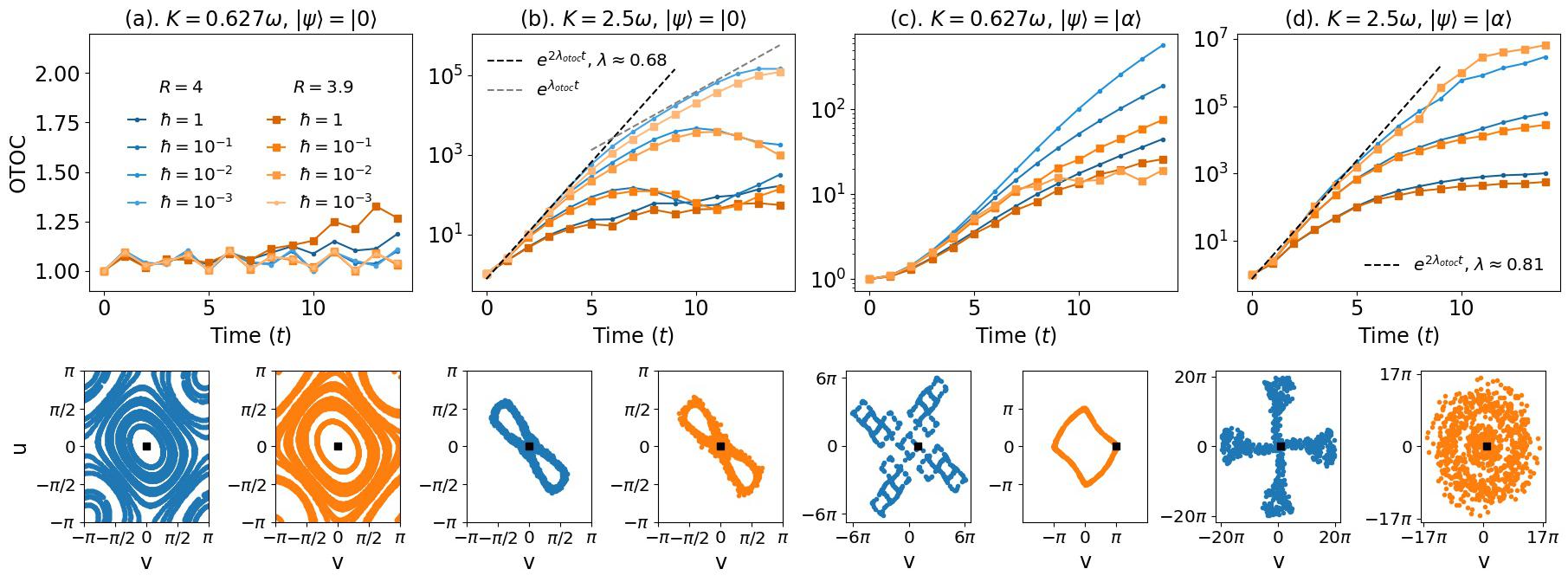}
\caption{\label{fig:protoc} The figure illustrates the early time OTOC behavior near $R=4$ in two different perturbative regimes of the quantum KHO, considering two different initial states, namely, the vacuum state $|0\rangle$ and a coherent state $|\alpha\rangle$, where $\Re(\alpha)=\sqrt{\omega/2\hbar}\langle\hat{X}\rangle$ and $\Im(\alpha)=\langle\hat{P}\rangle/\sqrt{2\omega\hbar}$, centered at an unstable period four point given by $(\langle\hat{P}\rangle, \langle\hat{X}\rangle)=(0, \pi)$. The kicking period is fixed at $\tau=1$, which automatically sets $\omega=2\pi/R$. While the top panels demonstrate the OTOC calculations, the bottom panels show the classical phase space dynamics in the vicinity of the corresponding initial conditions. The square-shaped black dot in the bottom panels represents the corresponding classical initial condition. Also, note the color coding --- we use blue to denote the resonant case ($R=4$) and orange for the non-resonant case ($R=3.9$). The darker colors denote the deep quantum regime, and the lighter colors indicate the semiclassical domain. (a)-(b) The system is initialized in the vacuum state $|0\rangle$. In (a), the kicking strength is fixed at $K=0.627\omega$, corresponding to the weak perturbative regime. Moderately strong perturbation is considered in (b). The plot shows early-time dynamics on the semi-log scale. Black and grey colored dashed lines are plotted here to illustrate the two-step early-time exponential growth. Furthermore, we find no visible differences between the cases of $R=4$ and $R=3.9$. (c)-(d) Here, we repeat the same calculations as before by replacing the vacuum state with the coherent state [see the main text for details].}
\end{figure*}

\subsection{Early-time dynamics}
\label{resonanceOTOC}
Classically, the vacuum state corresponds to a fixed point of the dynamical map given in Eq. (\ref{dynmap}). In contrast, the coherent state is chosen to be centered at a point on the classical stochastic web associated with the resonant case $R=4$, satisfying the equation $v=\pm(u+\pi)+2l\pi$, $l\in\mathbb{Z}$. In the phase space, these two initial conditions give rise to different dynamics altogether. Thus, it is essential to see if the differences are also reflected in the behavior of the OTOCs over short periods.
\subsubsection{Vacuum state (fixed point) OTOC}
We first discuss the resonant scenario and then contrast it with the non-resonant one. Due to the correspondence principle, the early time growth of the vacuum state OTOCs will have a close correspondence with the fixed point behavior of $(u, v)=(0, 0)$ in the classical phase space. To illustrate, we consider the Jacobian matrix of the classical map evaluated at $(0, 0)$. 
\begin{equation}\label{jacob}
J_{R=4}=
\left.\begin{bmatrix}
0 & 1 \\
-1 & -\dfrac{K\cos{v_n}}{\omega}
\end{bmatrix} \right|_{u_n=0, v_n=0}=
\begin{bmatrix}
0 & 1 \\
-1 & -\dfrac{K}{\omega}
\end{bmatrix},
\end{equation}
whose eigenvalues are
\begin{math}
\gamma_\pm =(-K\pm\sqrt{K^2-4\omega^2})/2\omega .
\end{math}
In this work, we always assume that $K$ is positive. Therefore, for $K<2\omega$, the eigenvalues are a pair of complex conjugates with unit modulus, i.e.,  $|\gamma_{\pm}|=1$, which implies that the fixed point is stable. The phase space trajectories in this region will wrap around $(0, 0)$ in closed elliptically-shaped orbits. Consequently, the vacuum state OTOCs are expected to remain stagnant until the Ehrenfest time due to the vanishing Lyapunov exponent. The map then undergoes a bifurcation at $K=2\omega$ with the emergence of two other stable fixed points. For $K>2\omega$, the point $(0, 0)$ becomes a saddle --- in this case, the trajectories with the initial conditions located slightly off $(0, 0)$ will diverge exponentially from one another with the saddle point exponent $\lambda_s$ given by $\sim\text{max}\{\log(|\gamma_+|), \log(|\gamma_-|)\}$. Accordingly, the OTOCs in the vacuum state $|0\rangle$ or any coherent state $|\alpha\rangle$ that lives close by are expected to display short-time exponential growth whenever $K>2\omega$. 

Figure \ref{fig:protoc}a and \ref{fig:protoc}b illustrate the vacuum state OTOCs for two different kicking strengths, namely, $K=0.627\omega$ and $K=2.5\omega$ in the vicinity of $R=4$. All the blue curves (with point markers) correspond to $R=4$, and the orange-colored curves (with square markers) represent the case of $R=3.9$. First, when $R=4$ and $K$ is small, the function $C_{|0\rangle,\thinspace \hat{a}\hat{a}^{\dagger}}(t)$ shows an initial stagnant behavior, which is demonstrated in Fig. \ref{fig:protoc}a. The lack of growth at initial times is reminiscent of the classical elliptic stability of the point $(0, 0)$. On the other hand, for $K=2.5\omega$, we observe the initial growth to be quadratic in the deep quantum regime ($\hbar=1$) [see Fig. \ref{fig:protoc}b]. This is because, in the deep quantum regime, the Ehrenfest time $\tau_{\text{EF}}\sim\ln(1/\hbar_{\text{eff}})/2\lambda_s$ \cite{schubert2012wave} remains very small, which makes it hard to observe the exponential growth. However, one can witness genuine exponential growth by slowly moving towards the semiclassical regime. This can be done by tuning the Planck constant $\hbar$. While for $K=0.627\omega$, the semiclassical OTOCs under the resonance remain stagnant for much longer as they should be, the case of $K=2.5\omega$ shows a clear exponential scaling. In the latter case, when $\hbar=10^{-3}$, the quantum Lyapunov exponent extracted from the OTOC through an exponential fitting of the first six data points is $\lambda_{\text{otoc}}\approx 0.68$. This value aligns well with the classical saddle point exponent of the origin, $\lambda_s=\ln(2)\approx 0.693$. In Ref. \cite{steinhuber2023dynamical}, it has been observed that in locally hyperbolic systems, the OTOCs exhibit a two-step early-time exponential growth. Due to the hyperbolicity of $(0, 0)$, a similar behavior is expected in the vacuum state OTOC for $K > 2\omega$. We indeed observe in Fig. \ref{fig:protoc}b that the initial growth of $\sim e^{2\lambda_{\text{otoc}}t}$ is followed by a subsequent regime scaling as $\sim e^{\lambda_{\text{otoc}}t}$.

To strengthen the correspondence between the classical and the quantum exponents, we analytically extract the quantum exponent from the vacuum state OTOC in the semiclassical limit ($\hbar\rightarrow 0$). Let $l \in \mathbb{R}^+$ and $l>1$, then reducing $\hbar$ by the factor of $l$ is equivalent to increasing both $K$ and $\omega$ in the second term of the Floquet operator by the same factor $l$ while keeping $\hbar$ fixed. This, however, does not affect the first term $\exp\{-i2\pi/R \hat{a}^{\dagger}\hat{a}\}$. Since $K$ is now increased, one plausible effect would be enhanced scrambling. Then again, for small $t$, with the initial state $|0\rangle$, it holds that $\cos(\hat{X}/\sqrt{l})\approx 1-\hat{X}^2/2l$ for $l\gg 1$. Therefore, we have $e^{-iKl\cos(\hat{X}/\sqrt{l})}\approx e^{-iKl}e^{iK\hat{X}^2/2}$. The phase term $e^{-iKc}$ can be ignored here. At $R=4$, under the modified Floquet evolution, the bosonic operators evolve according to the following transformation:
\begin{eqnarray}
\begin{bmatrix} \hat{a}(t) \\ \hat{a}^{\dagger}(t) \end{bmatrix}
 =
\begin{pmatrix}
   y-i & y \\
   y & y+i
\end{pmatrix}^{t}
\begin{bmatrix} \hat{a} \\ \hat{a}^{\dagger} \end{bmatrix}, 
\end{eqnarray}
where $y=K/2\omega$. 
The eigenvalues of the above linear transformation are given by $\lambda_{\pm}=(-K\pm\sqrt{K^2-4\omega^2})/2\omega$, which also turn out to be the eigenvalues of the Jacobian in Eq. (\ref{jacob}). Hence the semiclassical vacuum state OTOC ($C_{|0\rangle ,\hat{a}\hat{a}^{\dagger}}(t)$) grows exponentially in time with $\text{max}\{\log(|\lambda_+|), \log(|\lambda_-|)\}$ being the rate of growth. For $K>2\omega$, the growth remains exponential. For $K=2\omega$, the classical bifurcation point, $|\lambda_{\pm}|=1$, leading to the stagnant behavior. When $K$ is below $2\omega$, oscillatory behavior is expected. This establishes a strong classical-quantum correspondence in the limit of $\hbar\rightarrow 0$. At this point, the time scale for which the above approximation remains valid is unclear. Nevertheless, given the nature of classical fixed point behavior, we anticipate its validity within the Ehrenfest regime.

Interestingly, for $R=3.9$, we observe that the short-time growth of the OTOCs almost coincides with that of the resonant counterpart, regardless of the strength of kicking. This can be intuitively understood by examining the fixed point nature of ($0, 0$) under the non-resonance condition. The general condition for the bifurcation of $(0, 0)$ is given by
\begin{equation}\label{bif}
\dfrac{K}{\omega}=\dfrac{2\cos\omega\tau\pm 2}{\sin\omega\tau}, \text{ where }\omega\tau=\dfrac{2\pi}{R}.
\end{equation}
Accordingly, for $R=3.9$, the bifurcation point is located at $K= 1.921\omega$, slightly off that of $R=4$. Therefore, for $K<1.921\omega$, the point $(0, 0)$ remains stable. Also, when the perturbation is more than what is required for the bifurcation, the origin becomes a saddle with the emergence of two other stable fixed points. Thus, the bifurcation mechanism resembles that of the resonance, with the only difference being a slight change in the bifurcation location along the parameter axis of $K$. Consequently, at any given perturbation strength, the phase trajectories in the vicinity of the origin do not acquire any major changes as $R$ is moved from $4$ to $3.9$. The phase space plots presented in Fig. \ref{fig:protoc}a support this conclusion. Therefore, we expect the short-time growth of the vacuum state OTOCs for $R=3.9$ and $R=4$ to be qualitatively identical, irrespective of other system parameters, which is confirmed by the numerical results. Similar to the case of $R=4$, the two-step exponential growth is also observed for $R=3.9$. This behavior is evident in Fig. \ref{fig:protoc}b, where the blue and orange curves exhibit nearly identical growth rates.

\subsubsection{Coherent state (on the web) OTOCs}
As for the coherent state OTOCs, at $R=4$, the short-time dynamics rely mainly on the nature of the classical stochastic web. Recall that for $K<2\omega$, the origin ($0, 0$) is the only classical fixed point. There exist, however, an infinite number of period four points, each located at $(u, v)=(p\pi, q\pi)$ with $p, q\in\mathbb{Z}$ for any arbitrary small $K$ \cite{kells2005quantum}. The period four orbits these points generate are stable if $p+q$ is even and unstable otherwise. Thus, the presence of alternate stable and unstable manifolds causes the movement of any trajectory set out from the neighborhood of an unstable orbit to be highly complex. Such a complex motion persists even in the limit of small $k$, where the Lyapunov exponent approaches zero --- for example, see the phase space trajectories in Fig. \ref{fig:poincare}d. We are interested in seeing if such behavior is also reflected in coherent state OTOCs. Here, for the OTOC calculations, we focus on a specific coherent state with the mean coordinates located at $(\langle\hat{P}\rangle, \langle\hat{X}\rangle)=(0, \pi)$, which corresponds to an unstable period four point. 

The results are plotted in Fig. \ref{fig:protoc}c and \ref{fig:protoc}d for the same perturbation strengths as before. We first discuss the case of resonance. When $K=0.627\omega$, the OTOC in the deep quantum regime exhibits an approximate quadratic growth. In this case, the corresponding classical phase space trajectories diverge at a rate given by $\lambda_{\text{cl}}\approx 0.06$, the maximum Lyapunov exponent associated with the unstable period four point $(0, \pi)$. Due to the smallness of the exponent, it's difficult to observe the exponential growth even when $\hbar$ is small. On the other hand, when $K=2.5\omega$, the corresponding classical exponent is given by $\lambda_{\text{cl}}\approx 0.9$. In this case, the OTOC in the deep quantum regime ($\hbar =1$) still displays an initial algebraic growth. However, upon reducing $\hbar$, a clean exponential growth can be observed. For $\hbar=10^{-2}$, the rate of growth of the OTOC is extracted to be $\lambda_{\text{otoc}}\approx 0.81$, which agrees closely with the classical exponent, highlighting a good quantum-classical correspondence. The corresponding plots are shown in Fig. \ref{fig:protoc}d. We also notice that the coherent state OTOCs for smaller $\hbar$ overshoot those for larger $\hbar$ in both panels \ref{fig:protoc}c and \ref{fig:protoc}d. Moreover, since the stochastic webs fill a significant portion of the phase space at resonances, it is expected that nearly all coherent states overlapping with these webs will exhibit similar behavior as shown in Fig. \ref{fig:protoc}c and \ref{fig:protoc}d.

\begin{figure*}
\includegraphics[width=\textwidth,height=4.cm]{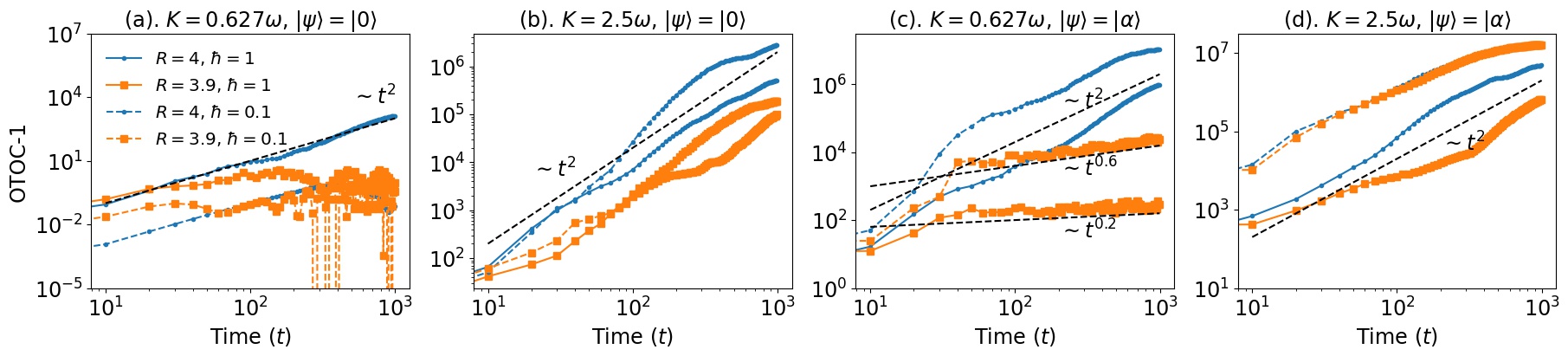}
\caption{\label{fig:protoc_ext} The figure illustrates the OTOC growth over a period of $10^3$ time steps for both the initial states $|0\rangle$ and $|\alpha\rangle$. We consider the interval between two successive data points $10\tau$ to avoid computational overhead due to the unbounded nature of the system. The parameters across the panels are fixed to be the same as those in Fig. \ref{fig:protoc}. The black-colored dashed lines following various scalings are drawn to contrast the growth under the resonance and non-resonance conditions. In the strong perturbative regime ($K=2.5\omega$), the quadratic growth remains the dominant behavior for all the parameters as shown in \ref{fig:protoc_ext}b and \ref{fig:protoc_ext}d. The differences between the resonant and non-resonant cases become apparent when $K$ is small [see the main text for details. Also see Fig. \ref{fig:protoc} for the corresponding short-time growth]. Note that as the time interval is large, the oscillations may not be fully visible in Fig. \ref{fig:protoc}a.}
\end{figure*}

To contrast the above results with the non-resonant case, we repeat the OTOC calculations by setting $R=3.9$ for the same coherent state $(\langle \hat{P}\rangle, \langle\hat{X}\rangle)=(0, \pi)$. First, under small perturbations, the classical phase space remains primarily regular. Hence, the initial time growth of OTOCs is expected not to contain any exponential scaling. We confirm this from Fig. \ref{fig:protoc}c, where the initial growth (as shown by orange curves with square markers) for all the cases of $\hbar$ is algebraic and shows no exponential growth. At initial times, the orange curves nearly coincide with the blue curves. However, a clear departure from one another can be observed as time progresses. The algebraic growth can be reasoned with the following argument: Although the system is in non-resonance, the closeness of $R$ to the resonance condition suggests a competition between the diffusive and regular dynamics. The diffusive dynamics usually dominate the early-time behavior. Hence, we see a coincidence between the initial growths of resonant and non-resonant cases. On the other hand, when $K$ is sufficiently large, the classical phase space trajectories display chaotic diffusion. One might expect exponential scaling over a short time in this case. The numerics indicate that the OTOCs in the deep quantum regime still display algebraic growth. Nevertheless, the exponential growth appears in the semiclassical regime [see Fig. \ref{fig:protoc}d], where the orange curves show identical growth as the blue curves. Recall that reducing $\hbar$ leads to a competition between the enhanced scrambling and the oscillatory behavior due to an effective increase in $K$ and $\omega$. Here, for $\hbar=0.1$, the OTOC grows at a faster rate than the case of $\hbar=1$ for both values of $K$. There is, however, no reason to expect similar behavior as $\hbar$ further varies. For instance, in Fig. \ref{fig:protoc}c, the growth for $\hbar=0.01$, as shown by the lighter orange curve, seems to be suppressed after some initial time. Intuitively, the correspondence principle implies oscillatory behavior in the OTOC as $\hbar\rightarrow 0$ as the phase space is stable under non-resonances. Hence, for very small $\hbar$, the enhanced scrambling effect is suppressed, and the oscillatory behavior takes over. When the perturbation is large, the semiclassical OTOCs outperform those in the deep quantum regime [see Fig. \ref{fig:protoc}d].

\subsection{Long-time dynamics}
\label{non-resonanceOTOC}
Although the OTOCs do not appear to distinguish the non-resonances from the resonances over a short time, the long-time dynamics produce significant differences between the two. These differences are more pronounced when the kicking strength is small ($K\lesssim 1$). Let us first examine the resonant cases to see how long-time dynamics stand out. Recall from Eq. (\ref{com}) that the total number of terms in the equation grows linearly with time. At each time step, the right-hand side grows by one more term, which is given by $[\sin\hat{X}(j), \hat{a}^{\dagger}]e^{2\pi ij/R}$. The resonances will then lead to a \textit{coherent} addition of all the terms, thus rendering the asymptotic growth of the OTOC a quadratic function of time, for a typical initial state --- $C(t)\sim K^2t^2$. When $K$ is too small ($K\ll 1$), the operator $\sin\hat{X}(j)$ grows only by a negligible amount. Then, the OTOC can be explicitly shown to exhibit quadratic growth by ignoring the terms of order $O(K^3)$ and higher in the time-evolved operator $\hat{a}(t)$. Refer to Appendix \ref{appendix:b} for more details, where we give an explicit derivation for the same. However, since $\sin\hat{X}(j)$ is bounded, the quadratic growth is expected to persist even when $K$ is large.

To verify numerically, we refer to Fig. \ref{fig:protoc_ext}, where the long-time dynamics have been illustrated for $\sim 10^3$ time steps. Due to the unbounded nature of the system, a sufficiently large Hilbert space is needed to perform these numerical simulations. Here, we analyze the resonant cases in the figure. The figure demonstrates that in the deep quantum regime ($\hbar=1$), for both the initial states $|0\rangle$ and $|\alpha\rangle$, the long-time dynamics always scale quadratically as long as the resonance condition is satisfied regardless of the strength of perturbation. While the same result seems to hold in the semiclassical limit for the coherent state, the vacuum state OTOC shows anomalous oscillatory behavior as shown in Fig. \ref{fig:protoc_ext}a. This behavior is not surprising as it is expected due to the classical-quantum correspondence. The classical stability of $(0, 0)$ implies a longer $t_{\text{EF}}$ for the vacuum state OTOC in the semiclassical limit compared to other cases in the figure. The numerics suggest that this behavior persists beyond $10^3$ time steps. The quadratic scaling may emerge as the OTOC picks up non-trivial $\hbar$ effects. Moreover, for large $K$, the quadratic growth persists for both $\hbar=1$ and $\hbar=0.1$ as shown in Fig. \ref{fig:protoc_ext}b and \ref{fig:protoc_ext}d. It is worth noting that the degree of scrambling, as quantified by the magnitude of the OTOC, always depends on the specific choice of the state vector despite both states exhibiting long-time quadratic growth. For example, the coherent states with the mean coordinates located on the stochastic web delocalize quickly in the phase space when acted upon by the Floquet unitary $\hat{U}_{\tau}$. In these states, the initial operators are relatively more prone to get scrambled compared to the vacuum state $|0\rangle$.

Contrary to the resonant cases, the non-resonant $R$ comprises a dense set of rational numbers (excluding integers) with measure zero and irrational numbers with measure one. When $R$ is irrational, the equidistribution property implies that in the long-time limit, the phases $\{e^{2\pi ij/R}\}_{j=0}^{t}$ will tend to behave as though they were drawn uniformly at random from the unit circle in the complex plane. This leads to \textit{incoherent} summations, causing various terms in Eq. (\ref{com}) to interfere destructively. Therefore, we expect that the growth of the OTOC will be suppressed. As we shall show in the next section, the irrationality of $R$, on average, induces linear growth in the OTOC [see also Appendix \ref{appendix:b} for more details]. Let us also point out that the OTOC growth in a typical pure state may generally vary from the linear behavior. Moreover, the rational $R$ values also induce subdued growth, though not as much as the irrationals. The results demonstrated in Fig. \ref{fig:protoc_ext} indicate that the growth is suppressed for $R=3.9$ in all the cases considered (shown in the orange-colored curves). When $K$ is small, while the vacuum state OTOCs oscillate as shown in Fig. \ref{fig:protoc_ext}a, the coherent state OTOCs follow power laws given by $\sim t^{0.2}$ and $\sim t^{0.6}$ respectively for $\hbar=1$ and $\hbar=0.1$ [see Fig. \ref{fig:protoc_ext}c]. As $K$ becomes large, the OTOCs at non-resonances will also grow quadratically as the system becomes fully chaotic.    
\section{Analytically Tractable cases}
\label{analtytical}
In the preceding section, our analysis primarily relied on numerical investigations to compare the operator growth in both the resonant and the non-resonant scenarios. Given that the kicking potential is highly non-linear, the exact solutions of the quantum KHO are generally intractable. There is, however, a narrow window of opportunity involving a few special cases, such as the symmetries and \textit{quantum resonances} that allow for the analytical treatment of the OTOCs. In this section, we exploit the translation invariance of the quantum  KHO by considering a few selective resonant cases to obtain the OTOCs as explicit functions of $t$. In particular, we consider the trivial cases, $R=1$ and $2$, followed by a more intricate case of $R=4$. Additionally, we provide an analytical approximation for the OTOC averaged over the space of pure states, which we shall refer to as the average state-OTOC, by considering small $K$ ($K\ll 1$) and an irrational $R$. 
\subsection{Case-1: $R=1$ and $2$}
Here, we focus on the commutator function $C_{|\psi\rangle,\thinspace \hat{a}\hat{a}^{\dagger}}(t)$ for the case of $R=2$. We do not consider $R=1$ separately since the resulting expressions for the OTOCs are the same in both cases. In the present case, stochastic webs do not constitute the classical phase space due to the exact solvability of the classical map, which is given by
\begin{eqnarray}\label{cla-R2}
v_{n} &=&(-1)^nv_0\nonumber\\ 
u_{n} &=&(-1)^n\left[u_0+n\varepsilon\sin v_0\right].
\end{eqnarray} 
It is worth noting that the time-evolved operators $\hat{X}(t)$ and $\hat{P}(t)$ admit the same functional form as the classical variables $v_n$ and $u_n$, indicating a persistent classical-quantum correspondence. 

When $R=2$, the phase operator $e^{-(2\pi i/R) \hat{a}^{\dagger}\hat{a}}$ possesses alternating $+1$ and $-1$ on its diagonal. This enables us to express $\hat{a}(t)$ explicitly as 
\begin{equation}
(-1)^t\hat{a}(t)=\hat{a}+\dfrac{iKt}{\sqrt{2\hbar\omega}}\sin\hat{X},
\end{equation}
which shows a clear linear dependence on $t$ accompanied by an oscillatory behavior originating from the term $(-1)^t$. By noting that $[\sin\hat{X}, \hat{a}^{\dagger}]=\sqrt{\hbar/2\omega}\cos\hat{X}$, we finally obtain
\begin{eqnarray}\label{otocatR2}
C_{|\psi\rangle,\thinspace \hat{a}\hat{a}^{\dagger}}(t)=1+\dfrac{K^2t^2}{4\omega^2}\langle\psi|\cos^2\hat{X}|\psi\rangle,
\end{eqnarray}
which is a quadratic function of $t$. The oscillating term $(-1)^t$ disappears as the OTOCs involve absolute squares of the commutators. Moreover, the commutator function always retains quadratic growth except when $\langle\psi|\cos^2\hat{X}|\psi\rangle=0$. As an example, we here take $|\psi\rangle=|0\rangle$, then $\langle 0|\cos^2\hat{X}|0\rangle=e^{-\hbar/2\omega}\cosh(\hbar/2\omega)$. We can therefore write
\begin{equation}
C_{|0\rangle,\thinspace \hat{a}\hat{a}^{\dagger}}(t)=1+\dfrac{K^2t^2}{4\omega^2}e^{-\hbar/2\omega}\cosh\left(\dfrac{\hbar}{2\omega}\right).
\end{equation}

To obtain the OTOC as a state-independent quantity, one can average $C_{|\psi\rangle,\thinspace \hat{a}\hat{a}^{\dagger}}(t)$ over the space of pure states or a suitable set of vectors forming a continuous variable $1$-design. In continuous variable systems, the Fock state and the coherent state bases are known to form $1$-designs \cite{blume2014curious, iosue2022continuous}. Note that in finite-dimensional systems, the OTOCs are often evaluated in the maximally mixed states, equivalent to averaging over the Haar random pure states. On the contrary, the notion of maximally mixed states in continuous variable systems is not well-defined due to the diverging traces induced by the infinite-dimensional Hilbert spaces. Despite this limitation, the averaging procedure can still provide insights into the nature of scrambling in a typical pure state.

For $R=2$, one can readily evaluate the average state-OTOC as follows:
\begin{eqnarray}
\overline{C_{|\psi\rangle,\thinspace \hat{a}\hat{a}^{\dagger}}}(t)=1+\dfrac{K^2t^2}{4\omega^2}\overline{\cos^2\hat{X}}.
\end{eqnarray}
This expression can be simplified by recalling that $\cos\hat{X}$ is related to the displacement operators.
\begin{equation}
\cos\hat{X}=\dfrac{1}{2}\left[D\left(i\sqrt{\dfrac{\hbar}{2\omega}}\right)+D\left(-i\sqrt{\dfrac{\hbar}{2\omega}}\right)\right].
\end{equation}
From Eq. (\ref{dispavg}) and (\ref{dispavgfinal}), it then follows that
\begin{eqnarray}
\overline{C_{|\psi\rangle,\thinspace \hat{a}\hat{a}^{\dagger}}}(t)&=&1+\dfrac{K^2t^2}{16\omega^2}\left[\overline{D\left(i\sqrt{\dfrac{2\hbar}{\omega}}\right)+D\left(-i\sqrt{\dfrac{2\hbar}{\omega}}\right)+2D(0)}\right]\nonumber\\
&=&1+\dfrac{K^2t^2}{8\omega^2}.
\end{eqnarray}
In the second equality, we used the result from Appendix \ref{appendix:a} that $\overline{D(\beta)}=\delta_{\Re(\beta), 0}\delta_{\Im(\beta)}$ for any $\beta\in\mathbb{C}$. We also assumed a finite non-zero value for $\hbar$. The above expression for $\overline{C_{|\psi\rangle,\thinspace \hat{a}\hat{a}^{\dagger}}}(t)$ is exact and valid for all kicking strengths.

\subsection{Case-2: $R=4$}
Here, we set $R=4$ and derive the analytical expression for the OTOC owing to certain restrictions on $\omega$. Our analysis can also be extended to the other cases when $R=3$ and $6$. Recall from Ref. \cite{borgonovi1995translational} and Section \ref{section-2} that under the translational invariance, $R$-th powers of the unitary evolution, $U_{\tau}^R$, commute with either one-parameter or two-parameter groups of translations depending on the values $\omega$ assumes. Since $\hat{a}(t)$ contains the terms $\sin\hat{X}(j)$ for $j=0, ..., t-1$, we are interested in finding a suitable $\omega$ that ensures the commutation relation $[\sin{\hat{X}}, \hat{U}^4]=0$. In order to proceed, it is useful to write $U^4_{\tau}$ in a compact form as follows:
\begin{equation}
\hat{U}^4_{\tau}=\left(e^{-i(K/\hbar)\cos\left(\hat{P}/\omega\right)}e^{-i(K/\hbar)\cos{\hat{X}}}\right)^2,
\end{equation}
where the Floquet operator within the parentheses can be identified with the kicked Haper model [\cite{kells2005eigensolutions}]. Since $\sin\hat{X}$ always commutes with $\cos\hat{X}$, we now only require to find $\omega$, that allows for the commutation between $\sin\hat{X}$ and $\cos(\hat{P}/\omega)$, i.e, $[\sin\hat{X}, \cos(\hat{P}/\omega)]=0$.

One can verify that the commutator $[\sin\hat{X}, \cos(\hat{P}/\omega)]$ indeed vanishes whenever $\omega=\hbar/2k\pi$ with $k\in\mathbb{Z}^+$. To see this, we consider the following:
\begin{eqnarray}
\left[e^{i\hat{X}}, e^{\pm i\hat{P}/\omega}\right]&=&e^{i\hat{X}}e^{\pm i\hat{P}/\omega}-e^{\pm i\hat{P}/\omega}e^{i\hat{X}} \nonumber\\
&=&e^{i\hat{X}\pm (i\hat{P}/\omega)\mp (i\hbar/2\omega)}-e^{i\hat{X}\pm (i\hat{P}/\omega)\pm (i\hbar/2\omega)}\nonumber\\
&=&(-1)^k\left[e^{i\hat{X}\pm 2ki\pi\hat{P}}-e^{i\hat{X}\pm 2ki\pi\hat{P}}\right]\hspace{0.3cm}\text{for }\omega=\dfrac{\hbar}{2k\pi}\nonumber\\
&=&0.
\end{eqnarray}
The second equality follows from the BCH formula. In the third equality, we have made the substitution that $e^{\pm ik\pi}=(-1)^k$. Likewise, $e^{-i\hat{X}}$ can also be shown to commute with $\cos(\hat{P}/\omega)$. Alternatively, one can also show that the operator $e^{\pm i\hat{X}}$ belongs to a one-parameter group of translations whenever the above condition on $\omega$ is satisfied \cite{borgonovi1995translational}. As a result, we have $\hat{U}^{\dagger 4}_{\tau}\sin\hat{X}\hat{U}^{4}_{\tau}=\sin\hat{X}$. Therefore, for some $j\in \{0, 1, 2, 3\}$, the translational invariance implies $\sin\hat{X}(j)=\sin\hat{X}(t+j)$, for $t$ being an integer multiple of $4$. Furthermore, due to $\hbar$ dependence, this is also called the quantum resonance condition \cite{billam2009quantum}. In this case, the operator $\hat{a}(t)$ at a total time $t=4s$, where $s$ is a non-negative integer, can be expressed as an explicit function of $t$:
\begin{eqnarray}\label{R4bos}
\hat{a}(t)=\hat{a}+\dfrac{iKt}{4\sqrt{2\hbar\omega}}\sum_{j=0}^3 e^{ij\pi/2}\sin\hat{X}(j).
\end{eqnarray}
This allows us to write the OTOC as follows: 
\begin{align}\label{otocat4res}
C_{\hat{a}\hat{a}^{\dagger}}(t)=1+&\dfrac{Kt}{4\sqrt{2\hbar\omega}}\left[\sum_{j=0}^3i^{j+1}\left[\sin\hat{X}(j), \hat{a}^{\dagger}\right]+\textbf{h.c.}\right]\nonumber\\
&+\dfrac{K^2t^2}{32\hbar\omega}\left[\sum_{j, j'=0}^{3}i^{j'-j}\left[\sin\hat{X}(j), \hat{a}^{\dagger}\right]^{\dagger}\left[\sin\hat{X}(j'), \hat{a}^{\dagger}\right]\right],
\end{align}
where $\textbf{h.c.}$ denotes the Hermitian conjugate of the operator $i^{j+1}[\sin\hat{X}(j), \hat{a}^{\dagger}]$.
\begin{figure}
\begin{center}
\includegraphics[scale=0.38]{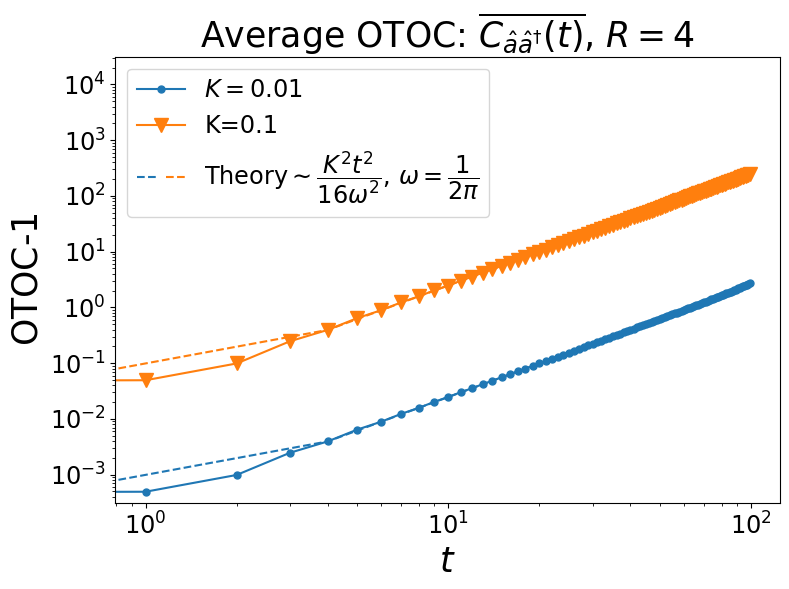}
\end{center}
\caption{\label{fig:otoct} The averaged commutator function $\overline{C_{\hat{a}\hat{a}^{\dagger}}(t)}$ for the translation invariant case at $R=4$. The figure illustrates the growth for two values of $K$, namely, $K=0.01$ and $0.1$. We have taken $\hbar=1$. While the solid lines with the markers denote numerical results, the dashed lines correspond to the theory as given in Eq. (\ref{otocat4}). The solid and dashed lines almost coincide for all $t>0$. For the numerical results, we perform the average over two thousand number states.}
\end{figure}
Although Eq. (\ref{otocat4res}) implies that the growth of the OTOC is a quadratic polynomial of time, the presence of operators such as $\{\sin\hat{X}(j)\}_{j=1}^3$, which are not amenable to the exact treatment, makes the OTOC only quasi-exact solvable. Nevertheless, in the limit of weak perturbations, we can still obtain a good approximate analytical solution up to the leading orders in $K$. In particular, when $K$ is small ($K\ll\ 1$) and at sufficiently long times, the terms of the order $O(Kt)$ and $O(K^2t^2)$ dominate over the remaining others, such as $O(K^mt)$ with $m\geq 2$ and $O(K^nt^2)$ with $n\geq 3$. Therefore, we only require approximating the terms $\{\sin\hat{X}(j)\}$ to the zeroth order in $K$ --- $\sin\hat{X}(1)\approx i\sin (\hat{P}/\omega)\approx -\sin\hat{X}(3)$ and $\sin\hat{X}(2)\approx -\sin\hat{X}$. Thus, ignoring all the other insignificant contributions, we finally obtain
\begin{equation}\label{R4comf}
C_{\hat{a}\hat{a}^{\dagger}}(t)\approx 1+\dfrac{K^2t^2}{16\omega^2}\left[\cos\hat{X}+\cos\left(\dfrac{\hat{P}}{\omega}\right)\right]^2.
\end{equation}
Eq. (\ref{R4comf}) can be evaluated in any arbitrary initial state. We calculate the average state-OTOC to better understand its behavior in a typical quantum state. 
\begin{equation}\label{otocat4}
\overline{C_{|\psi\rangle,\thinspace \hat{a}\hat{a}^{\dagger}}(t)}= 1+ct^2\hspace{0.2cm}\text{for }t=4s, \hspace{0.1cm}s\in\mathbb{Z}^+\cup\{0\}, 
\end{equation}
where $c\approx K^2/16\omega^2$ for $K\ll 1$.

Figure \ref{fig:otoct} contrasts Eq. (\ref{otocat4}) with the numerically computed average state-OTOC for two different kicking strengths. The quantum resonance condition is invoked by fixing the frequency at $\omega=\hbar/2\pi$, with $\hbar=1$. We find an excellent agreement between the numerical results and Eq. (\ref{otocat4}). Note, however, that the approximation breaks down at large values of $K$ as the higher order terms in $K$ become significant and can no longer be ignored.

One can generalize this analysis to the other translationally invariant cases by imposing suitable conditions on $\omega$. For instance, an analogous condition for $R=3$ in the deep quantum regime reads $\omega=\sqrt{3}/4k\pi$, where $k\in\mathbb{Z}^+$, thereby demonstrating quadratic growth of the OTOC when observed at $t=3s$, with $s$ being a non-negative integer.
An upshot of this analysis is that the quantum resonances hinder the early-time exponential behavior, irrespective of the choice of other free parameters, such as the kicking strength $K$. It is worth noting that in finite-dimensional chaotic systems (or systems with torus boundary conditions in the classical limit), the OTOCs typically saturate to values predicted by the random matrix theory \cite{GMata2023}. In contrast, finding an unbounded operator growth is often possible in infinite-dimensional systems like the KHO. Quantum resonances are one such example. Here, we leveraged their solvability to show the indefinite operator growth.

\subsection{Small $K$ and irrational $R$}
\label{irr-der}
In the last section, we argued that whenever $R$ takes irrational values, various terms in Eq. (\ref{bosonic-evolution1}) destructively interfere, given that $K$ is small. As a result, the corresponding dynamics are suppressed. Here, we extend this argument by providing an analytical expression for the average state-OTOC, $\overline{C_{|\psi\rangle,\thinspace \hat{a}\hat{a}^{\dagger}}}(t)$, in the weak perturbative regime by assuming irrational values for $R$. The complete derivation is presented in Appendix \ref{appendix:b}. Retaining only the leading order terms in $K$ up to the order of $O(K^2)$, the Heisenberg evolution of $\hat{a}$ can be written as follows:
\begin{align}
\hat{a}(t)e^{\frac{2\pi it}{R}}\approx\hat{a}+\dfrac{iK}{\sqrt{2\omega}}\sum_{j=0}^{t-1}e^{\frac{2\pi ij}{R}}\sin\left(\hat{X}_{\frac{2\pi j}{R}}\right)-\dfrac{K^2}{\sqrt{2\omega}}\sum_{j=0}^{t-1}\sum_{n=0}^{j-1}e^{\frac{2\pi ij}{R}}\left[\cos\left(\hat{X}_{\frac{2\pi n}{R}}\right),\thinspace \sin\left(\hat{X}_{\frac{2\pi j}{R}}\right) \right],\nonumber\\
\end{align}
where $\hat{X}_{\theta}=(\hat{a}e^{-i\theta}+\hat{a}^{\dagger}e^{i\theta})/\sqrt{2\omega}$, the quadrature operator with the phase $\theta$. We have taken $\hbar=1$ for the sake of simplicity. In the next step, we plug this into Eq. (\ref{com}). This is then followed by the absolute squaring of the commutator --- $[\hat{a}(t), \hat{a}^{\dagger}]^{\dagger}[\hat{a}(t), \hat{a}^{\dagger}]$, which is given in Eq. (\ref{irr-comm-comp}) of Appendix \ref{appendix:b}. After averaging over the pure states, we finally obtain
\begin{figure}
\begin{center}
\includegraphics[scale=0.4]{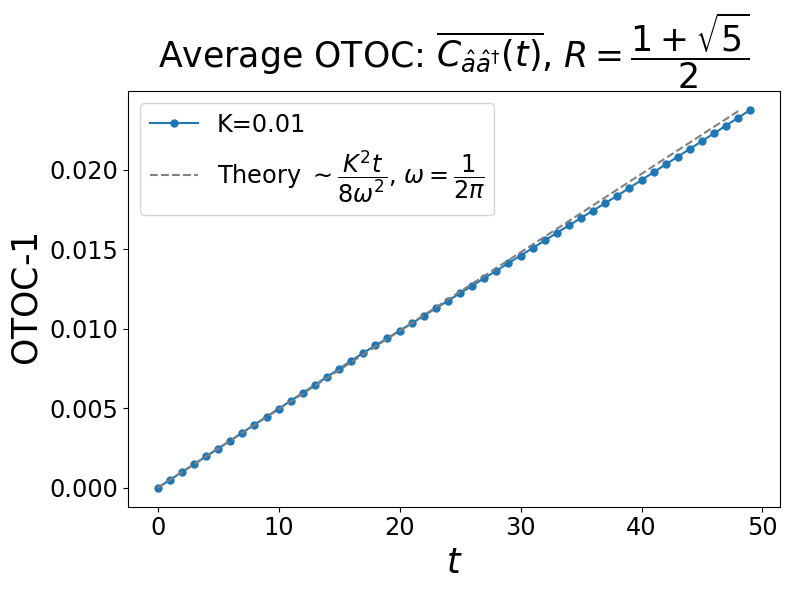}
\end{center}
\caption{\label{fig:otocatgm} The averaged commutator function $\overline{C_{\hat{a}\hat{a}^{\dagger}}(t)}$ in a weak perturbative regime when $R$ is irrational. We fix $R$ as the golden mean number, known to be the most irrational number. The figure compares the numerically computed average OTOC with the theoretical prediction given in Eq. (\ref{irr-otoc}). For the numerics, we average the OTOC over two thousand Fock states.}
\end{figure}
\begin{equation}\label{irr-otoc}
\overline{C_{|\psi\rangle,\thinspace \hat{a}\hat{a}^{\dagger}}}(t)\approx 1+\dfrac{K^2 t}{8\omega^2},\hspace{0.5cm}\text{for }K\ll 1
\end{equation}
To test Eq. (\ref{irr-otoc}), we numerically evaluate the commutator function averaged over the Fock state basis when $R$ assumes an irrational number. For the numerical calculations, we take $R$ to be the golden mean number --- $R=(1+\sqrt{5})/2$. Figure. \ref{fig:otocatgm} compares the numerical result with Eq. (\ref{irr-otoc}). We see a good agreement between the both.

We assert that the above result is not limited to the KHO model under the irrationality of $R$ but rather encompasses a broader range of KAM systems, including that of the finite dimensions. To be precise, the linear growth here is a direct consequence of (i) the uncorrelated eigenphases of the unperturbed evolution $e^{-(2\pi i/R)\hat{a}^{\dagger}\hat{a}}$, and (ii) the initial operators being conserved up to a phase under the unperturbed evolution, i.e., $e^{(2\pi i/R)\hat{a}^{\dagger}\hat{a}}\hat{a}e^{-(2\pi i/R)\hat{a}^{\dagger}\hat{a}}=\hat{a}e^{-2\pi i/R}$. One can readily verify that the average state-OTOC in a typical integrable quantum system perturbed by a weak generic time-dependent potential, in general, exhibits linear growth until the saturation as long as the above two conditions are satisfied --- $\overline{C(t)}\sim \varepsilon^2t$ with $\varepsilon$ being the strength of the perturbation. Let us also mention that, in general, the first condition holds for any typical KAM integrable system owing to the Berry-Tabor conjecture \cite{Berry375, Berry77a, berry1976closed}. In Appendix \ref{appendix:c}, we provide a detailed derivation for the linear growth of the OTOCs in finite-dimensional quantum systems by incorporating the aforementioned conditions.

\section{OTOC for phase space operators}
\label{OTOCXP}
For completeness, we here provide a bird's-eye view of the position-momentum OTOC in an arbitrary state $|\psi\rangle$ given by
\begin{eqnarray}
C_{XP}(t)=\langle\psi|\left[\hat{X}(t), \hat{P}\right]^{\dagger}\left[\hat{X}(t), \hat{P}\right]|\psi\rangle. 
\end{eqnarray}
Followed by Eq. (\ref{bosonic-evolution1}), the Heisenberg evolution of $\hat{X}$ can be readily obtained as 
\begin{equation}\label{Xevol}
\hat{X}(t)=u^{\dagger t}\hat{X}u^t +\dfrac{K}{\omega}\sum_{j=0}^{t-1}\sin\left[\dfrac{2\pi}{R}\left(t-j\right)\right]\sin\hat{X}(j),
\end{equation}
where $u=e^{(2\pi it/R) \hat{a}^{\dagger}\hat{a}}$, the evolution under the harmonic oscillator Hamiltonian with an effective frequency $\omega\tau=2\pi/R$. Due to the presence of $u$, the OTOC exhibits persistent oscillations in time. On the other hand, the second term takes into account the effects of the kicking potential on operator growth.
\begin{figure}
\begin{center}
\includegraphics[scale=0.35]{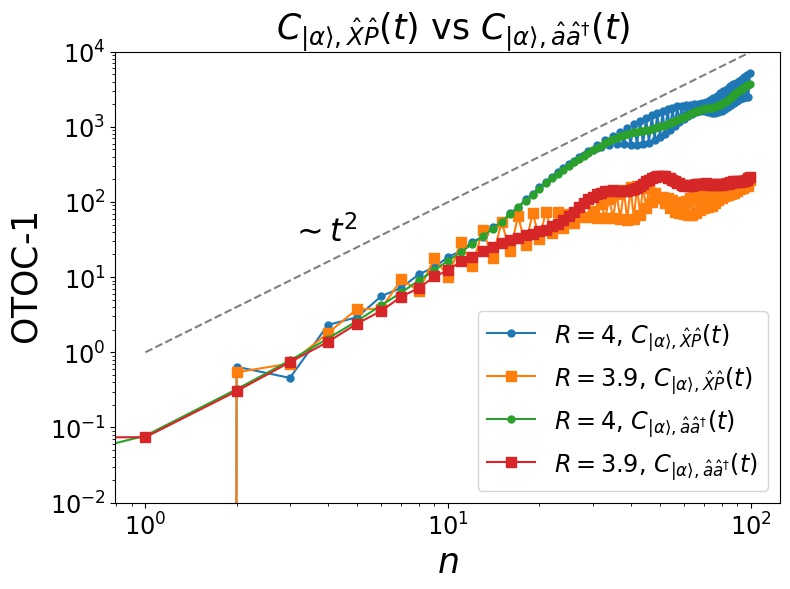}
\end{center}
\caption{\label{fig:xpotoc} Contrasting the OTOCs corresponding to the phase space operators with the ladder operators in the vicinity of $R=4$. The system parameters are kept fixed at $K=1$ and $\tau=1$. The coherent state centered at $(\hat{P}, \hat{X})=(0, \pi)$ is once again chosen to be the initial state for the OTOC calculations. The figure illustrates that the $\hat{X}\hat{P}$-OTOC and the $\hat{a}\hat{a}^{\dagger}$-OTOC exhibit strikingly similar behaviors in both the resonant and non-resonant cases, with the only difference being the presence of oscillations in the former.}
\end{figure}

When $K=0$, the OTOC in any given state $|\psi\rangle$ is $C_{|\psi\rangle,\thinspace \hat{X}\hat{P}}(t)=\hbar^2 \cos^2(2\pi t/R)$. As $K$ varies from zero, the initial operator $\hat{X}$ starts to scramble into the operator Hilbert space, which is also accompanied by the oscillatory behavior.
We numerically compute the $\hat{X}\hat{P}$-OTOC and compare the results with the ladder operator OTOC. The findings are shown in Fig. \ref{fig:xpotoc}. The figure demonstrates that the qualitative nature of the OTOCs for both choices of the initial operators is identical, except for the presence of oscillations in the $\hat{X}\hat{P}$-OTOC due to the harmonic evolution. These oscillations may pose challenges while probing the operator growth and quantum Lyapunov exponents. Nevertheless, a comparison between Equation (\ref{bosonic-evolution1}) and Equation (\ref{Xevol}) reveals that the oscillating component is absent in the time evolved ladder operator $\hat{a}(t)$. As a result, we observe steady growth without oscillations in the ladder operator OTOC, prompting us to study it as an alternative to the $\hat{X}\hat{P}$-OTOC.

\section{Summary and Discussion}
\label{section-4}


In this chapter, we have examined how the resonances that arise when a degenerate classical system is perturbed by a weak time-dependent potential affect the dynamics of information scrambling in the quantum domain. For this purpose, we considered the kicked harmonic oscillator model. Classically, this system exhibits very complex dynamics. Under the resonance condition, the system undergoes unusual structural changes in the phase space. These unusual changes can be traced to the emergence of stable and unstable periodic orbits, leading to diffusive chaos in the phase space for any finite kicking strength. Conversely, under the non-resonance condition, the phase space trajectories of the harmonic oscillator maintain their regularity with slight deformations, thereby effectively suppressing the chaos. Motivated by this peculiar classical behavior, our work examined how these effects manifest in the dynamics of operator growth and scrambling in the quantum limit.     

To study the information scrambling, we considered the OTOC with the bosonic ladder operators as the initial operators. We mainly focussed on two specific initial states: the vacuum state $|0\rangle$ that corresponds to a fixed point of the KHO map and a coherent state $|\alpha\rangle$ with its mean coordinates located on the stochastic web. We numerically studied the OTOC, particularly emphasizing the early-time and the asymptotic dynamics. In the semiclassical limit, the early-time dynamics correlate very well with the classical dynamics of the KHO. In this region, the quantum Lyapunov exponents extracted from the OTOC match very well with the corresponding classical Lyapunov exponents. To strengthen the correspondence between the classical and quantum exponents, we analytically derived the quantum exponent from the vacuum state OTOC in the semiclassical limit. Furthermore, we observed that in the weak perturbative regime, the transitions between resonance and non-resonance conditions do not affect the early-time growth, while the asymptotic dynamics remain highly sensitive. On the contrary, when the perturbation is strong, the differences become less visible. 

The numerical and analytical evidence presented in this work suggests a long-time quadratic growth for the OTOC when the system is in resonance for a generic quantum state, irrespective of the strength of the perturbation. In contrast, the growth under the non-resonance condition is largely suppressed. We also identified different scalings for the non-resonant OTOCs in the coherent state. To support the numerical results, we have argued based on the coherent and incoherent additions of the terms in the expansion of $[\hat{a}(t), \hat{a}^{\dagger}]e^{2\pi it/R}$ that demonstrate the distinct qualitative behavior exhibited by the resonances and the non-resonances. 

Following the numerical results, we provided analytical treatment of the OTOC for a few exceptional cases of $R\in R_c$ by utilizing the translation invariance of the KHO. For $R=1$ and $2$, the results are exact. We have shown that the corresponding commutator function grows quadratically. At $R=4$, we utilized the quantum resonance condition to obtain the quasi-exact expressions of the OTOCs. The quantum resonances largely impede the early-time exponential growth of the OTOCs. We also provided analytical derivation for the linear growth of the average state-OTOC whenever $R$ takes irrational values, given that the kicking strength is sufficiently small. 

One central focus in quantum many-body physics is to simulate quantum systems on a quantum device and benchmark these simulations in the presence of hardware errors \cite{sahu2022quantum, trivedi2022quantum}. In the semiclassical limit, the assurance will be provided for the stability of simulation if the classical counterpart of the target Hamiltonian is KAM and the error that scales extensively with the number of particles is small enough \cite{bulchandani2022onset}. However, for systems that have no classical analogue, the quantum KAM theorem remains elusive. Nevertheless, recent progress in this direction studied the stability analysis for the symmetries of quantum systems --- see for instance, Refs. \cite{brandino2015glimmers, burgarth2021kolmogorov} and the references therein. Quantum simulations of systems whose classical limit is non-KAM can pose a significant challenge to experimental implementations. Any slight perturbation in the form of hardware noise can give rise to dynamics far from the target dynamics. Moreover, digital quantum simulation of these systems is also challenging due to structural changes near resonances even when the hardware error is negligible \cite{chinni2022trotter}. Hence, more careful methods must be devised to simulate this class of systems. We hope our study paves the way for exploring these intriguing directions.   
\chapter{Quantum sensing using non-KAM systems}
\label{nonKAMchapsen}

\section{Introduction}
Quantum metrology (also known as quantum sensing) holds great potential for quantum-enhanced measurements and is mainly concerned with estimating unknown parameters of quantum Hamiltonians \cite{giovannetti2004quantum, giovannetti2006quantum}. These parameters could be magnetic field strength, electric field strength, or the frequency of a harmonic oscillator. Quantum sensing is intimately related to the distinguishability of quantum states \cite{braunstein1994statistical, braunstein1996generalized}. Phenomena, like quantum critical transitions and quantum chaos have been known to give rise to a high degree of sensitivity in quantum systems. The sensing protocols based on critical transitions, where the system under consideration exhibits high sensitivity to variations in the order parameters, have been extensively studied over many years \cite{zanardi2008quantum, invernizzi2008optimal, ivanov2013adiabatic, tsang2013quantum, macieszczak2016dynamical, rams2018limits, frerot2018quantum, garbe2020critical, chu2021dynamic, garbe2022critical}. Likewise, quantum chaotic systems, with their hypersensitivity to parameter variations, display better sensing capabilities than integrable systems \cite{fiderer2018quantum, liu2021quantum}.

Classical systems that do not obey the Kolmogorov–Arnold–Moser theorem (also known as non-KAM systems) can show exceptionally high susceptibility to small time-dependent perturbations. In this chapter, we depart from the studies of conventional quantum sensors and examine the non-KAM systems in the quantum limit for sensing applications. Under certain conditions (such as resonances), these systems become highly sensitive to the variations in the parameters despite being weakly chaotic \cite{sankaranarayanan2001quantum, sankaranarayanan2001chaos, dileep2024}. This motivates us to study the dynamics of these systems carefully for sensing applications. For this purpose, we consider the quantized dynamics of the kicked harmonic oscillator (KHO) system.

In the previous chapter (see chapter \ref{nonKAMchap} and also Ref. \cite{dileep2024}), we rigorously probed the OTOCs in the quantum KHO model and contrasted the resonant cases with the non-resonant cases. The Hamiltonian of this system is given by
\begin{eqnarray}
 H=\underbrace{\dfrac{P^2}{2m}+\dfrac{1}{2}m\omega^2 X^2}_{H_0}+\underbrace{K\cos \left( kX \right)}_{H_1} \sum_{n=-\infty}^{\infty}\delta(t-n\tau),  
\end{eqnarray}
where $X$ is the position, $P$ is the momentum, $m$ is the mass of the oscillator and $k$ is the wave vector. The strength of the kicking is denoted by $K$. The time interval between two consecutive kicks is given by $\tau$. Like the previous chapter, we fix $m=k=1$ throughout this chapter. The resonances are particular occurrences in this system where the perturbation is time-periodic, and the frequencies of the unperturbed harmonic oscillator ($H_0$) and the perturbation ($H_1$) commensurate --- $\dfrac{\omega}{(2\pi/\tau)}=\dfrac{1}{R}$, $R\in \mathbb{Z}^{+}$, where $\omega$ and $2\pi/\tau$ are the frequencies of $H_0$ and $H_1$, respectively. The non-resonances are characterized by non-integer $R$ values. At resonances, the dynamics become highly susceptible to variations in $R$. In the quantum limit, the corresponding Floquet operator of the system is given by 
\begin{eqnarray}\label{floq}
\hat{U}_{\omega}=\exp\left\{-\dfrac{2\pi i}{R} \hat{a}^\dagger \hat{a}\right\}\exp\left\{-\dfrac{iK}{\hbar}\cos\hat{X} \right\},    
\end{eqnarray}
where $\hat{a}$ and $\hat{a}^{\dagger}$ denote the bosonic annihilation and creation operators and $\hat{X}=\dfrac{\hat{a}+\hat{a}^{\dagger}}{\sqrt{2\omega}}$. The classical non-KAM effects at the resonances translate into sharp changes in eigenvalues and eigenvector statistics of the quantum dynamics. These sharp changes render the system highly suitable for sensing applications.
Using the quantum KHO as a toy model and the quantum Fisher information (QFI) as a figure of metric, we demonstrate that the sensitivity of the non-KAM systems near resonances can be utilized in sensing (metrology) applications.

Let us now revisit the details of QFI from chapter \ref{chap-back}: QFI quantifies the sensitivity of a parametrized quantum state $|\psi_{\theta}\rangle$ to the variations in $\theta$. Here, we are interested in evaluating QFI associated with the parameter $\omega$, the natural frequency of the unperturbed harmonic oscillator. Then, the QFI of a given state having the imprint of $\omega$ is given by
\begin{equation}
I\left(\omega ; {|\psi_{\omega}\rangle}\right)=-2\left.\partial_{\varepsilon}^2|\langle\psi_{\omega+\varepsilon}|\psi_{\omega}\rangle|^2\right|_{\varepsilon=0}.
\end{equation}
We encode $\omega$ onto a known initial quantum state by using the dynamics of the quantum KHO. This means that the given initial state is acted upon by the time evolution unitary of the quantum KHO: $|\psi_\omega (t)\rangle=\hat{U}^{t}_{\omega}|\psi\rangle$. Under this setting, the evolution of QFI over time is given by 
\begin{eqnarray}\label{quench}
I(\omega ;|\psi\rangle, t)= 4\lim_{\varepsilon\rightarrow 0}\left( \dfrac{1-|\langle\psi|\hat{U}_{\omega+\varepsilon}^{\dagger t}\hat{U}^{t}_{\omega}|\psi\rangle|^2}{\varepsilon^2} \right)\equiv 4\left(\Delta^2_{|\psi\rangle} \hat{h}(\omega ; t)\right),
\end{eqnarray}
where $\hat{h}(\omega ; t) =i \hat{U}^{\dagger t}_{\omega}\partial_{\omega}\hat{U}^{t}_{\omega}$, and $\Delta^2$ denotes variance with respect to the initial state $|\psi\rangle$ \cite{gu2010fidelity}. The operator $\hat{h}(\omega ; t)$, when acts on $|\psi_{\omega}(t)\rangle$, generates infinitesimal translations along $\omega$ in the parameter space. In this work, we evaluate QFI in different quantum states evolved under the dynamics of the quantum KHO and contrast the results corresponding to the resonances with the non-resonant cases. We show that the resonances allow for huge enhancements in the time scaling of QFI compared to the quadratic scaling at the non-resonances. We shall also provide arguments to show that QFI growth is directly related to the mean energy growth in the given quantum state.

The rest of the chapter is structured as follows: The main results are presented in Sec. \ref{sen-Results}. In Sec. \ref{Mean energy growth vs. QFI}, we demonstrate a connection between the mean energy growth and QFI in the quantum KHO model. We then discuss the behavior of the Loschmidt echo (Loschmidt echo) in Sec. \ref{Loschmidt Echo}. Then, in Sec. \ref{QFI in the quantum KHO model}, we contrast QFI at non-resonances with the non-resonances of the KHO model. We discuss QFI under special conditions called quantum resonances in Sec. \ref{quantum resonances}. Finally, we conclude this chapter in Sec. \ref{sen-discussion}.

\begin{section}{Results}
\label{sen-Results}
This work only considers the pure initial states for the parameter encoding. 
When the initial state is pure, QFI associated with $\omega$ is related to the variance of the generator $\hat{h}(\omega; t)$ in that state. The analysis is slightly intricate for the mixed initial states, and we do not carry it out here. We now seek to obtain the generator corresponding to the quantum KHO model. For $t=1$, this operator takes the following form:
\begin{eqnarray}\label{generator_exp}
\hat{h}\left(\omega ; 1\right)=i\hat{U}_{\omega}^\dagger \partial_{\omega}\hat{U}_{\omega}
=\tau\hat{n}-\dfrac{K}{2\omega}\hat{A}\sin \hat{A},
\end{eqnarray}
where $\hat{n}=\hat{a}^{\dagger}\hat{a}$ is the number operator and $\hat{A}=(\hat{a}e^{-i\omega\tau}+\hat{a}^\dagger e^{i\omega\tau})/\sqrt{2\omega}$. It then follows that the generator corresponding to $\hat{U}^{t}_{\omega}$, the Floquet operator after $t$-number of steps, can be written as
\begin{eqnarray}\label{timegen}
\hat{h}(\omega ; t)&=&\sum_{j=0}^{t-1}\hat{U}^{\dagger j}_{\omega}\hat{h}(\omega ;1)\hat{U}^{j}_{\omega}\nonumber\\ 
&=&\sum_{j=0}^{t-1}\left[\tau\hat{n}(j)-\dfrac{K}{2\omega}\hat{A}(j)\sin \hat{A}(j)\right],
\end{eqnarray}
where $\hat{n}(j)$ and $\hat{A}(j)$ denote the Heisenberg evolution of the respective operators for $j$-number of time steps under the dynamics of the quantum KHO. The time evolution of QFI is then given by 
\begin{eqnarray}\label{gen-time}
I\left(\omega ; |\psi\rangle, t\right)=\Delta^{2}_{|\psi\rangle}\left[\sum_{j=0}^{t-1}\left(\tau\hat{n}(j)-\dfrac{K}{2\omega}\hat{A}(j)\sin \hat{A}(j)\right)\right]
\end{eqnarray}
This equation implies that the evolution of QFI strongly correlates with the mean energy growth in a given quantum state under the dynamics of the quantum KHO. In what follows, we shall examine how the mean energy growth in a state affects QFI in the quantum KHO model. Then, we shall proceed with a numerical analysis of QFI.

\subsection{Mean energy growth vs. QFI}
\label{Mean energy growth vs. QFI}
In atom-optical models, the evolution of QFI is strongly correlated with the mean energy growth in the system \cite{pang2017optimal, giovannetti2006quantum, garbe2022critical}. In these systems, it has been shown that if the mean energy grows as $\langle\hat{n}(t)\rangle\sim t^\alpha$, where $\alpha\geq 0$, then the maximum attainable scaling of QFI follows: $I(\omega ;|\psi\rangle, t)\sim t^{2\alpha+2}$. We shall show that this result also holds for the quantum KHO model. To see this, we consider the following inequality:
\begin{eqnarray}\label{bound}
I(\omega; |\psi\rangle, t)=4\Delta^2\left(\sum_{j=0}^{t-1}\hat{h}(\omega ; j)\right)\leq 4\left(\sum_{j=0}^{t-1}\sqrt{\Delta^2\hat{h}(\omega ; j)}\right)^2,
\end{eqnarray}
where $\sqrt{\Delta^2 \hat{h}(\omega, j)}$ is the standard deviation of the time evolved generator $\hat{h}(\omega, j)$ in the initial state $|\psi\rangle$. This quantity can be bounded from above as follows:  
\begin{eqnarray}
\sqrt{\Delta^2 \hat{h}(\omega ;j)}\leq\sqrt{\tau^2\Delta^2\hat{n}(j)}+\sqrt{\dfrac{K^2}{4\omega^2}\Delta^2 \left(\hat{A}(j)\sin{\hat{A}}(j)\right)},
\end{eqnarray}
which leads to the following inequality:
\begin{equation}\label{fin-ineq}
I(\omega; |\psi\rangle, t)\leq 4 \left[\sum_{j=0}^{t-1}\tau\Delta \hat{n}(j)+\dfrac{K}{2\omega}\Delta \left(\hat{A}(j)\sin \hat{A}(j)\right)\right]^2.
\end{equation}
The above expression suggests a connection between the maximum possible growth of QFI and the mean energy fluctuations in the system. From the inequality, one can heuristically argue that in a generic quantum state $|\psi\rangle$, if the mean energy grows as $t^\alpha$, where $\alpha\geq 0$, then $\Delta^2 \hat{n}(t)\sim t^{2\alpha}$. Accordingly, the growth of QFI in time can be at most a hexic function of time as implied by the inequality in Eq. (\ref{fin-ineq}) --- $I(\omega; |\psi\rangle, t)\sim \left[\sum_{j=0}^{t-1} O(j^\alpha)\right]^2$. From the known inequality $\int_{0}^{t-1}j^{\alpha}dj<\sum_{j=0}^{t-1} j^\alpha <\int_{0}^{t}j^{\alpha}dj$, it follows that $\sum_{j=0}^{t-1} j^\alpha\sim t^{\alpha+1}$. As a result, QFI scales as $I(\omega; t)\sim O(t^{2\alpha +2})$. Also, note that the second term on the right-hand side of the inequality in Eq. (\ref{fin-ineq}) grows at most with the scaling $t^{\alpha +2}$. Thus, the maximum possible growth of QFI is solely determined by the growth of the mean energy.

We now examine how the mean energy grows in the quantum KHO model. For that, we first analyze the dynamical evolution of the ladder operators. 
\begin{eqnarray}\label{bosonic-evolution}
\hat{a}(t)e^{i\omega\tau t}=\hat{a}+\dfrac{iKe^{i\omega\tau}}{\sqrt{2\omega}}\sum_{j=0}^{t-1}e^{ij\omega\tau}\sin \hat{A}(j).
\end{eqnarray}
As time progresses, the number of terms in the expansion of $\hat{a}(t)$ grows linearly. Besides, all the terms that are accumulated are bounded --- for any $0\leq j< t$ and for any $|\psi\rangle$, we have $\|\sin\left(\hat{X}(j)\right) |\psi\rangle\|\leq\||\psi\rangle\|$. This indicates that under KHO dynamics, the quadrature operators can grow linear at maximum, i.e., $\langle X\rangle\sim t$ and $\langle P\rangle\sim t$. Moreover, using Eq. (\ref{bosonic-evolution}), one can bound the mean energy of the quantum KHO in an arbitrary state. For instance, in an arbitrary Fock state $|l\rangle$, the mean energy can be bounded from above as follows:
\begin{eqnarray}\label{energy-bound}
\langle l|\hat{n}(t)|l\rangle\leq l+\dfrac{\sqrt{2l}Kt}{\sqrt{\omega}}+\dfrac{K^2t^2}{2\omega}. 
\end{eqnarray}
It is known that under the resonance condition, in particular, when the system has translation invariance, the mean energy of the quantum KHO can grow quadratically with time \cite{borgonovi1995translational}. Consequently, long-time growth of QFI can display hexic scaling with time ($\sim t^6$) at maximum.

\subsection{Loschmidt echo}
\label{Loschmidt Echo}
The distance metric in the evaluation of QFI, $|\langle\psi|U^{\dagger t}_{\theta+\varepsilon}U^t_{\theta}|\psi\rangle|^2$, also known as Loschmidt echo is an established diagnostic of quantum chaos. The Loschmidt echo was initially studied by Asher Peres to investigate the stability of quantum mechanical trajectories owing to the perturbations in the Hamiltonian or time evolution operator \cite{peres1984stability}. Since its inception, the Loschmidt echo has been studied extensively in systems of a broader spectrum ranging from regular to chaotic \cite{prosen2002stability, Goussev:2012}. The KHO is an unbounded system, leading to the infinite-dimensional Hilbert space. Here, we briefly outline a few established results of the Loschmidt echo for infinite dimensional systems. In weakly perturbative regimes (with perturbation strength $\varepsilon$), the Loschmidt echo for the regular systems has been shown to display short-time parabolic decay ($\sim 1-\varepsilon^2t^2/\hbar^2$) followed by a Gaussian decay. The parabolic decay appears for time scales of order $O(1/\varepsilon)$, allowing the floquet operator to be approximated by its second-order Taylor series expansion. On the other hand, the initial decay has been shown to be linear for chaotic systems. Late-time exponential decay follows the initial linear decay in these systems. Since QFI is related to the Loschmidt echo in the limit of vanishing perturbation $(\varepsilon\rightarrow 0)$, we try to extrapolate the results of Loschmidt echo to QFI of the quantum KHO model in the regular regions. It is worth noting that the aforementioned studies of Loschmidt echo are only applicable when both the unperturbed ($U_\theta$) and the perturbed ($U_{\theta+\varepsilon}$) systems are either regular or chaotic \cite{prosen2002stability}. Meanwhile, the KHO system remains dominantly regular when it is non-resonant. However, the system displays large-scale structural changes at resonance points and hence can be a resource in quantum metrology.

\begin{figure}
\center
\includegraphics[scale=0.5]{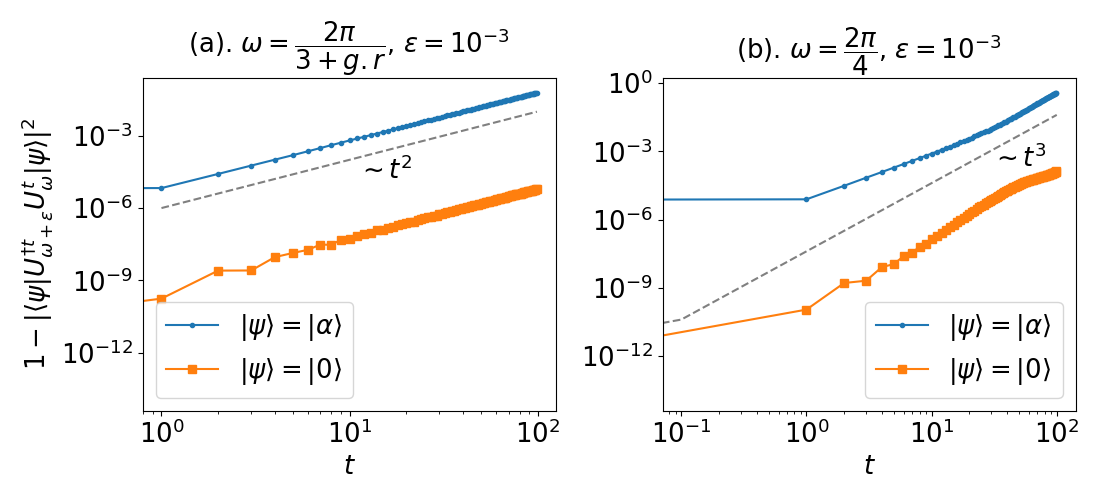}
\caption{\label{fig:Loschmidt} The evolution of $1$-Loschmidt echo for the quantum KHO for both non-resonant (R = 3 + g.r) and resonant (R = 4) cases is illustrated in panels (a) and (b), respectively. Two initial states, namely, the vacuum state and a coherent state with mean coordinates on the phase space $(\langle\hat{X}\rangle, \langle\hat{P}\rangle)=(\pi, 0)$, are considered. The perturbation is taken to be $\varepsilon=10^{-3}$. Other parameters are fixed at $K=0.1$ and $\tau=1$. The dashed lines in both panels contrast the growth of $1-$Loschmidt echo with the quadratic and cubic laws.}
\end{figure}

We here present a brief numerical analysis of the Loschmidt echo for the quantum KHO with respect to the variations in $\omega$. The other parameters are kept fixed at $\tau=1$ and $K=0.1$. The results are depicted in Fig. \ref{fig:Loschmidt}. The figure shows the quantity $1-|\langle\psi|U_{\omega+\varepsilon}^t\hat{U}_{\omega}^t|\psi\rangle |^2$ plotted against the number of time steps (t) for both resonant and non-resonant cases by taking two initial states, namely, the vacuum state and a coherent state centered in the phase space at $(\langle \hat{X}\rangle ,\langle \hat{P}\rangle )=(\pi, 0)$. The perturbation strength is $\varepsilon=10^{-3}$. When $\omega=\dfrac{2\pi}{3+g.r}$, where $g.r$ is the golden ratio, the system remains strongly non-resonant \cite{frasca1997quantum}, and the classical phase space retains an entirely regular structure. As a result, the quantity $1-|\langle\psi|U_{\omega+\varepsilon}^t\hat{U}_{\omega}^t|\psi\rangle |^2$ displays quadratic growth for both the initial states, which fits well with the earlier studies \cite{prosen2002stability}. On the other hand, at the resonance point $R=4$, the growth of $1-|\langle\psi|U_{\omega+\varepsilon}^t\hat{U}_{\omega}^t|\psi\rangle |^2$ exceeds the quadratic scaling for both initial states. This indicates that QFI at resonance points can demonstrate remarkable time-scaling enhancements. At a more detailed level, as discussed earlier, these enhancements are deeply rooted in the nature of the mean energy growth \cite{pang2017optimal}. Motivated by the Loschmidt echo results, we now calculate the time evolution of QFI and study the relevant scalings of QFI. We investigate the behavior of QFI at resonance points with particular interest.

\subsection{QFI in the quantum KHO model}
We first examine the evolution of QFI for $\omega=2\pi/R$ with non-integer $R$. This automatically sets the time interval of the kicking to be $\tau=\dfrac{2\pi}{R\omega}$. We also fix $K=0.1$. Recall that the quantum diffusion is suppressed whenever $R$ takes non-integer values, particularly highly irrational numbers. Hence, the system dynamics can be considered completely regular. From known results of the Loschmidt echo \cite{prosen2002stability, Goussev:2012}, the dynamical behavior of QFI, in this case, can be inferred to be quadratic ($\sim t^2\Delta^2 \hat{h}(\omega; 1)$). This is confirmed by the numerical results depicted in Fig. \ref{fig:fisherfull} and cross-validated by the Loschmidt echo calculations shown in Fig. \ref{fig:Loschmidt}a. Moreover, the mean energy ($\langle \hat{a}^{\dagger}\hat{a}(t) \rangle$), in this case, remains stagnated due to the localization effects --- $\langle\psi|\hat{a}^\dagger \hat{a}(t)|\psi\rangle\sim c t^0$, $c$ is a constant \cite{borgonovi1995translational}. Hence, the mean energy viewpoint supports the above inference that QFI grows quadratically. The numerical results in Fig. \ref{fig:fisherfull}a and \ref{fig:fisherfull}b illustrate the evolution of QFI at two non-resonance points ($R=3+g.r$, $5.2$) for the same initial states, as considered in the Loschmidt echo computation. Note that the smaller $K$ values render the system regular for non-integer $R$ values. However, for larger $K$ values, the system eventually becomes chaotic irrespective of the values of $R$. In such a case, the results will be modified according to the mean energy growth. 

\label{QFI in the quantum KHO model}
\begin{figure}
\center
\includegraphics[scale=.5]{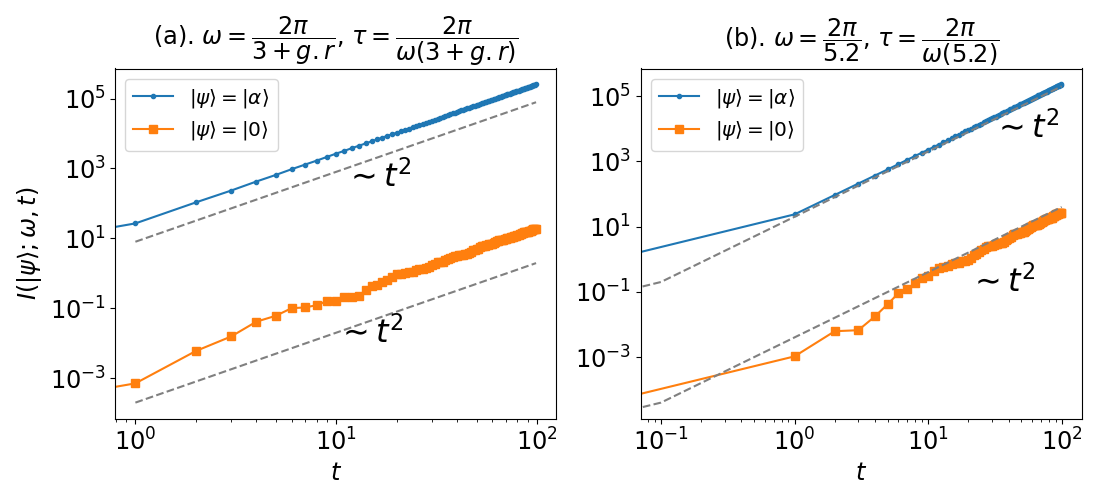}
\caption{\label{fig:fisherfull}Evolution of QFI for two arbitrary non-resonance conditions: (a). $R=3+$g.r, g.r is the golden ratio number, and (b). $R=5.2$. The initial states are fixed as before. The other parameters are $\tau=1$ and $K=0.1$. In all the cases considered, QFI follows a clear quadratic scaling. 
}
\end{figure}

We now examine how the resonances affect the growth of QFI. When the system is in resonance, QFI can exhibit an enhanced growth rate. Here, we deal with the resonance points $R\notin R_c\equiv\{1, 2, 3, 4, 6\}$ and $R\in\mathbb{Z}^+$, which do not display the translational invariance property in the classical phase space. As a representative case, we take $R=5$, where the classical KHO shows quasi-crystalline structures in the classical phase space. For this choice of $R$, the quantum KHO undergoes localization to delocalization transition as the kicking strength $K$ varies from zero \cite{shepelyansky1992quantum}. Alternatively, one can observe similar behavior in the limit of large effective Planck constant $\hbar$, even if the kicking strength $K$ is too small. As the system becomes delocalized, the mean energy $\langle\psi(t)| \hat{a}^\dagger \hat{a}|\psi(t)\rangle$ grows atmost linearly \cite{borgonovi1995translational}. This can lead to a maximum $t^4$ scaling in QFI. Figure \ref{fig:fisherfull5} shows the corresponding numerical results. We again fix the kicking strength at $K=0.1$. In \ref{fig:fisherfull5}a, we illustrate QFI behavior by taking $\omega=\dfrac{2\pi}{5}$. For these parameters, the quantum KHO shows localization properties. This results in quadratic growth of QFI in both initial states. On the other hand, when $\omega=\dfrac{2}{\pi}$, the system becomes delocalized and the mean energy grows linearly. As a result, QFI grows as $\sim t^4$, which is depicted in Fig. \ref{fig:fisherfull5}b.

\end{section}

\begin{figure}
\center
\includegraphics[scale=.5]{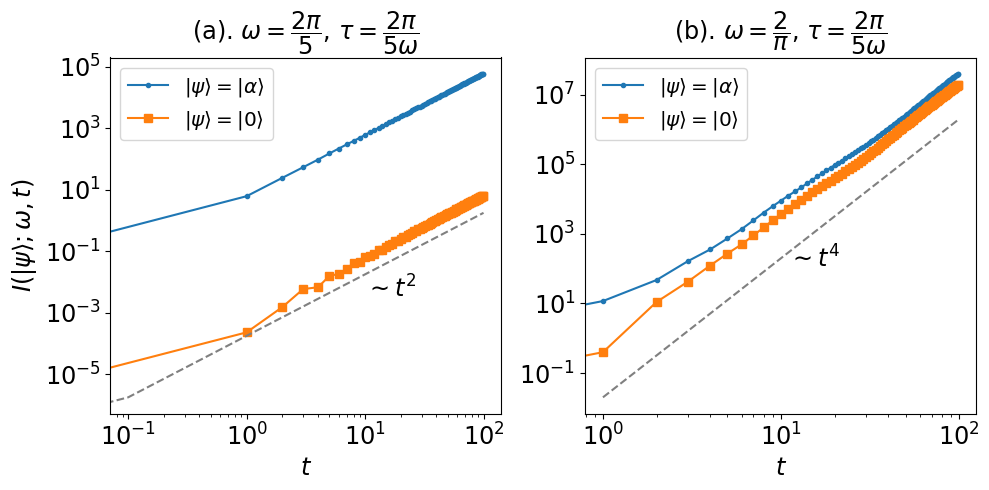}
\caption{\label{fig:fisherfull5}The evolution of QFI for $R=5$: (a). $\omega=\dfrac{2\pi}{5}$, and (b). $\omega=\dfrac{2}{\pi}$. The kicking interval $\tau$ is automatically fixed as $\tau=\dfrac{2\pi}{\omega R}$. The initial states are fixed as before. The kicking strength is $K=0.1$. In panel (a), the system remains regular, leading to a quadratic growth of QFI for both initial states. In (b), due to the linear growth of the mean energy $\langle (\hat{a}^{\dagger}\hat{a})(t)\rangle$, QFI can be seen to follow a quartic scaling.  
}
\end{figure}

\section{Role of translational invariance and quantum resonances}
\label{quantum resonances}
\label{section-4}
Recall that for $R\in R_c$, the classical KHO displays translational invariance in the phase space. In the quantum limit, this translates into the existence of symmetry groups of displacement operators --- $\hat{U}^{R}_{\omega}$ commutes with either one-parameter or two-parameter group of translations \cite{borgonovi1995translational}. The translational invariance leads to extended Floquet states in the phase space and, consequently, to possible unbounded growth of the mean energy. 
Here, we study QFI for $R\in R_c$. In particular, we consider $R=1$, $2$, and $4$, where the former two are special cases amenable to the analytical treatment. In the latter case, particular instances exist in the parameter space of $\omega$ known as quantum resonances, where the mean energy growth can be analytically shown to be quadratic.

\begin{figure}
\center
\includegraphics[scale=0.5]{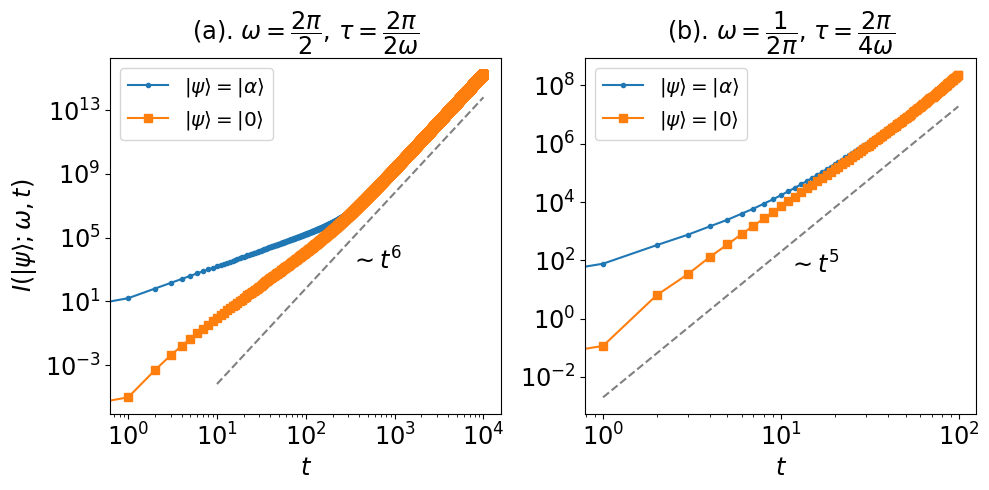}
\caption{\label{fig:fis-2-trans} The plot shows the time evolution of QFI under translational invariance. (a) QFI growth is shown for $10^4$ time steps. Initially, the coherent state and the vacuum state display different scaling. At around $\sim10^{2.5}$ time steps, QFI for both states displays a crossover to $t^6$ scaling. In (b), QFI growth for over $100$ time steps under a quantum resonance condition at $R=4$ is illustrated. The growth follows $t^5$ scaling. }
\end{figure}

We obtain an analytical expression of the generator $\hat{h}_{t}(\omega)$ for these special cases of $R$ as they are amenable to the exact computations. Let us recall the generator corresponding to $t$-number of time steps:
\begin{eqnarray}\label{tevolgen}
h(\omega, t)=\sum_{j=0}^{t-1}U^{\dagger j}h(\omega, 1) U^{j},
\end{eqnarray}
The cases where $R=1$ and $2$ are trivial with regards to the translational invariance, for which the Heisenberg evolution of the bosonic operators admit a simple form, which is given by,
\begin{eqnarray}\label{timeevolvedbosonic}
a(t)e^{i\omega t}=a+i\dfrac{Kt}{\sqrt{2\omega}}\sin\left( \hat{X} \right), 
\end{eqnarray}
where $\omega\tau=\dfrac{2\pi}{R}$. The expression holds true for both $R=2$ and $R=1$. Here, we shall analyze the case for $R=2$, and the analysis for the other case follows a similar approach. Equation (\ref{timeevolvedbosonic}) suggests that the Heisenberg evolution of the ladder operators has a linear time dependence for the above-specified frequencies. It is also worth noting that Eq. (\ref{timeevolvedbosonic}) holds true whenever $\omega\tau=n\pi$ for $n\in \mathbb{Z}^{+}$. The linear time dependence implies an unbounded quadratic growth of the mean energy. For instance, if the initial state is a fock state of a quantum harmonic oscillator, the mean energy grows as per the following expression:
\begin{eqnarray}
\langle n|(\hat{a}^\dagger \hat{a} )(t)|n\rangle=n+\dfrac{K^2t^2}{4\omega}\left(1-e^{-1/\omega}L_n^0\left(\dfrac{2}{\omega}\right)\right),
\end{eqnarray}
where $L_n^0$ is an associated Laguerre polynomial. As the mean energy growth is quadratic ($\alpha =2$), QFI will grow algebraically with an exponent $\lesssim 2\alpha+2=6$. Through numerical results and analytical arguments, we indeed confirm that after a sufficiently long time, the evolution of QFI converges to the scaling $\sim t^6$. We show this by obtaining an analytical expression for the operator $\hat{h}(\omega; t)$ for $R=2$.

To proceed further, we first write down $\hat{h}(\omega; 1)$, the generator after one time step, in the following.  
\begin{eqnarray}\label{gen}
\hat{h}(\omega; 1)&=&\tau \hat{a}^\dagger \hat{a}-\dfrac{K}{2\omega}\left( \dfrac{\hat{a}e^{-i\omega\tau}+\hat{a}^\dagger e^{i\omega\tau}}{\sqrt{2\omega}} \right)\sin\left( \dfrac{\hat{a}e^{-i\omega\tau}+\hat{a}^\dagger e^{i\omega\tau}}{\sqrt{2\omega}} \right)\nonumber\\
&=&\tau \hat{a}^\dagger \hat{a}-\dfrac{K}{2\omega}\hat{X}\sin\hat{X},
\end{eqnarray}
where we used $\omega\tau=2\pi/R$. Then, for some $j\in \mathbb{Z}^+$, 
\begin{eqnarray}
U^{\dagger j}\hat{h}(\omega, 1)U^j&=&\tau \hat{a}^\dagger \hat{a}(j)-\dfrac{K}{2\omega}\hat{X}(j)\sin\hat{X}(j)\nonumber\\
&=&\tau\left\{ \hat{a}^\dagger \hat{a} -i\dfrac{Kj}{\sqrt{2\omega}}\sin(\hat{X}) \hat{a}+i\dfrac{Kj}{\sqrt{2\omega}}\hat{a}^\dagger\sin\hat{X} +\dfrac{K^2j^2}{2\omega}\sin^2\hat{X} \right\}-\dfrac{K}{2\omega}\hat{X}\sin\hat{X}.\nonumber\\
\end{eqnarray}
It then follows that
\begin{eqnarray}\label{timeevolvedgen}
\hat{h}(\omega, t)=\tau\left\{t\hat{a}^\dagger \hat{a} -\dfrac{iKS_1(t)}{\sqrt{2\omega}}\left(\sin (\hat{X}) \hat{a}-\hat{a}^\dagger\sin\hat{X}\right)+\dfrac{K^2S_2(t)}{2\omega}\sin^2 \hat{X} \right\}-\dfrac{Kt}{2\omega}\hat{X}\sin\hat{X} ,\nonumber\\
\end{eqnarray}
where 
\begin{eqnarray*}
S_1(t)=\dfrac{t(t-1)}{2}\hspace{0.5cm}\text{and}\hspace{0.5cm}S_2(t)=\dfrac{t(t-1)(2t-1)}{6}. 
\end{eqnarray*}
From Eq. (\ref{timeevolvedgen}), QFI growth in any arbitrary initial state $|\psi\rangle$ can be straightforwardly obtained by evaluating the variance of $\hat{h}(\omega; t)$, i.e.,  
\begin{eqnarray}
I(|\psi\rangle; \omega, t)=4\left( \langle \psi|\hat{h}^2(\omega, t)|\psi\rangle-\langle \psi|\hat{h}(\omega, t)|\psi\rangle^2 \right) 
\end{eqnarray}
One can see that the scalar coefficients in the expression of the generator $\hat{h}(\omega, t)$ are polynomials in $t$ with degrees ranging from linear to cubic. This will imply that $I_{|\psi\rangle}(\omega; t)\sim O(t^6)$. We numerically evaluate QFI for the two initial states --- the vacuum state $|0\rangle$ and the coherent state $|\alpha\rangle$. The results are shown in Fig. \ref{fig:fis-2-trans}a. For both the initial states, the hexic scaling is clearly visible beyond $t\sim 10^3$ time steps.

We now focus on QFI growth under quantum resonance conditions. For $R=4$, and $\omega=l/2\pi$, where $l\in\mathbb{Z}^{+}$, it was shown in Ref. \cite{borgonovi1995translational} that $\hat{U}^4$ commutes with a two-parameter group of displacement operators in the phase space of the form
\begin{eqnarray}
\mathcal{D}_{n_1, n_2}\equiv D\left(2\pi\sqrt{\dfrac{\omega}{2}}\left(n_1-in_2\right)\right), 
\end{eqnarray}
where $n_1, n_2\in \mathbb{Z}$ are the free parameters and $D(\beta)$ represents a displaccement operator with the complex argument $\beta$. In this case, the mean energy was shown to grow quadratically in the long-time limit. This can be made slightly more rigorous with an example. We consider $\omega=1/2\pi$, which automatically sets $\tau={\pi}^2$. For an arbitrary choice of other free parameters, one can show that $[\hat{U}^{4}_{\omega}, e^{\pm i\hat{P}/\omega}]=0$. This, for $t=4s$, where $s\in\mathbb{Z}^{+}\cup\{0\}$, results in an explicit linear growth in the Heisenberg evolution of the bosonic ladder operators:
\begin{eqnarray}
\hat{a}(t)=\hat{a}-\dfrac{Kt}{4\sqrt{2\omega}}\sum_{j=0}^{3}i^{j}\sin\left(\dfrac{\hat{P}(j)}{\omega}\right),  
\end{eqnarray}
where $\hat{P}=-i\sqrt{\dfrac{\omega\hbar}{2}}\left(\hat{a}-\hat{a}^{\dagger}\right)$ is the momentum operator. Accordingly, the mean energy operator $(\hat{a}^{\dagger}\hat{a})(t)$ displays a clear quadratic dependence over $t$ irrespective of the choice of the initial state, i.e., 
\begin{eqnarray}
\left(\hat{a}^{\dagger}\hat{a}\right)(t)&=& \hat{a}^{\dagger}\hat{a}-\dfrac{Kt}{4\sqrt{2\omega}}\sum_{j=0}^{3}\left[ (-i)^j\sin\left( \dfrac{\hat{P}(j)}{\omega} \right)\hat{a}+i^j\hat{a}^{\dagger}\sin\left( \dfrac{\hat{P}(j)}{\omega} \right) \right]+\dfrac{K^2t^2}{32\omega}\left[\sum_{j=0}^{3} i^{j}\sin\left(\dfrac{\hat{P}(j)}{\omega}\right) \right]^2.\nonumber\\
\end{eqnarray}
The mean energy growth shows a crossover to the quadratic scaling beyond the time that scales like $t\sim 1/K^2$ in any generic initial state. Therefore, one can expect that QFI shows at most hexic scaling. We numerically evaluate QFI by considering the same initial quantum states as before. The results are plotted in Fig. \ref{fig:fis-2-trans}b. The figure demonstrates that QFI shows $\sim t^5$ scaling for both the initial states. We expect a crossover to the $\sim t^6$ scaling after the time scale $T\sim O(1/K^2)$.

\begin{section}{Summary and Discussion}
\label{sen-discussion}

Quantum Fisher information is a central quantity in the quantum parameter estimation theory, which measures the sensitivity of a parametrized quantum state to an infinitesimally small perturbation. In this regard, under certain conditions, the systems that do not respect the assumptions of the KAM theorem (also known as non-KAM systems) display higher sensitivity to the perturbations than a typical chaotic or regular system. In this work, through numerical analysis in conjunction with analytical results, we have examined the metrological performance of the quantum KHO dynamics, a well-known non-KAM system, as a quantum sensor. We have used QFI as a figure of metric. The frequency of the unperturbed oscillator has been taken as the parameter to be estimated or sensed. Moreover, we have used the unitary encoding protocol to imprint the parameter onto the known initial states.

First, we examined the connection between mean energy growth and QFI in the KHO model. We have argued that if the mean energy in a state grows as $t^{\alpha}$, QFI will grow as $t^{2\alpha +2}$. Under the non-resonance condition, it is known that the mean energy remains nearly constant due to the localization effects in the phase space. In this case, we have numerically shown that QFI grows quadratically over time. In contrast, the mean energy grows quadratically at resonances, implying nearly $t^6$ scaling in QFI growth. We have numerically verified these predictions for the vacuum state and a coherent state.
Moreover, when $R\in R_c$, the system dynamics are amenable to the quasi-exact solutions. By considering $R=2$, we have analytically obtained the generator of translations in $\omega$ at any non-zero time. Since this generator has a $t^3$ dependence, QFI shows a clear $t^6$ dependence, which we have verified numerically. Finally, we have also considered a particular case called the quantum resonance condition. At quantum resonances, the mean energy can be analytically shown to be a quadratic function of time in a typical quantum state. As a result, QFI, in this case, can show $t^6$ scaling. In addition, we have discussed the behavior of the Loschmidt echo under the resonance and non-resonance conditions. 

\end{section}
\chapter{Information scrambling in kicked coupled tops system}
\label{KCTchap}

\section{Introduction}
Chaos in classical physics is closely related to non-integrability, ergodicity, complexity, entropy production, and thermalization, which are central to the study of many-body physics and classical statistical mechanics. The quantum mechanics of systems whose classical counterparts are chaotic, generally known as quantum chaos, aims to extend these ideas to the quantum domain. In classical physics, the sensitive dependence on initial conditions implies chaos. A naive generalization of classical chaos to the quantum domain fails due to the unitarity of quantum evolutions. Hence, it is necessary to find the signatures of chaos in quantum systems through various other means. Many quantities, such as level spacing statistics \cite{haake1991quantum}, entanglement entropy \cite{bandyopadhyay2002testing, bandyopadhyay2004entanglement}, out-of-time ordered correlators (OTOCs) \cite{chaos1}, and tri-partite mutual information \cite{pawan}, have emerged as powerful tools to characterize the chaos in quantum systems. OTOCs are particularly interesting due to their usefulness in characterizing various features of quantum many-body systems. Originally introduced in the theory of superconductivity \cite{larkin}, the OTOCs are being studied with renewed interest in the context of quantum many-body systems \cite{ope2, ope1, ope4, ope5, lin2018out}, quantum chaos \cite{chaos1, pawan, seshadri2018tripartite, lakshminarayan2019out, shenker2, moudgalya2019operator, manybody2, chaos2, cotler2017chaos}, many-body localization \cite{manybody3, manybody4, manybody1, huang2017out, pg2021exponential} and holographic systems \cite{shock1, shenker3}. The early time growth rate of OTOCs, in particular, is being actively studied \cite{rozenbaum2017lyapunov, prakash2020scrambling, jalabert2018semiclassical, lakshminarayan2019out, ope5, garcia2018chaos, chen2018operator, moudgalya2019operator, haehl2019classification, chen2017out, omanakuttan2019out, alonso2019out, borgonovi2019timescales, yan2020information, rozenbaum2020early, rozenbaum2019universal, lerose2020bridging}, which is a quantum counterpart of the classical Lyapunov exponent (LE).

In order to define the OTOC, consider two local Hermitian or unitary operators $A$ and $B$ acting on two disjoint local subsystems of a given system with the dimensions $d_A$ and $d_B$, respectively. Then, a function of the commutator, namely the squared commutator, is given by 
\begin{equation}\label{commutator}
C(t)=\frac{1}{2}\langle\psi| [A(t), B]^{\dagger}[A(t), B]|\psi\rangle,
\end{equation}
where $A(t)=U^{\dagger}(t)A(0)U(t)$ and $U$ represents system's time evolution operator. For the sake of simplicity and experimental feasibility, the state $|\psi\rangle$ is usually taken to be the maximally mixed state --- $\mathbb{I}/d_Ad_B$, where $d_A$ and $d_B$ denote subsystem dimensions. In this text, we consider Hermitian operators for OTOC calculations. Then, the commutator function becomes $C(t)=C_{2}(t)-C_{4}(t)$, where
\begin{eqnarray}
C_2(t)=\frac{\tr(A^2(t)B^2)}{d_Ad_B}\quad\text{and}\quad C_4(t)=\frac{\tr(A(t)BA(t)B)}{d_Ad_B}.
\end{eqnarray}
Here, $C_2(t)$ is a time-ordered two-point correlator, and $C_4(t)$ represents the four-point correlator function. $C_4(t)$ possesses an unusual time ordering, which is why it is referred to as OTOC (Out-of-Time-Ordered Correlator). The four-point correlator is the dominant factor driving the total commutator function $C(t)$. Hence, we use the terms OTOC and commutator function interchangeably to denote the same quantity $C(t)$.
 
In this chapter, we study the dynamics of OTOCs in the kicked coupled tops (KCT), a bipartite system with a well-defined classical limit. In prior studies, many authors have considered systems of two (or more) degrees of freedom with time-dependent Hamiltonians of the form $H_{12}(t) = H_1(t) + H_2(t) + H_{12}(t)$, where the classical dynamics generated by $H_1$ and $H_2$ can exhibit chaos for each degree of freedom separately, and the coupling interaction $H_{12}$ between them can be independently varied. For such systems, the chaoticity parameter and coupling constant play different roles. To this effect, a system of coupled kicked rotors has been previously studied \cite{prakash2020scrambling, prakash2019out} where two kicked rotors that independently exhibit chaos were weakly coupled. In this work, we study a system where the coupling strength is the chaoticity parameter, and chaos occurs due to this mechanism rather than separately in the two systems.

In this chapter, our primary goal is to examine operator growth as quantified by the OTOCs in the system of kicked coupled tops. In the context of classical-quantum correspondence, works thus far have largely focussed on the OTOCs in quantum systems with chaotic classical limits. However, the behavior of operator growth in the mixed-phase space remains poorly understood with only a few studies \cite{mondal2021dynamical, notenson2023classical, bergamasco2019out, roy2021entanglement, kidd2021saddle, kidd2020thermalization}. This prompts us to study the role of the mixed-phase space dynamics on operator growth along with the fully chaotic dynamics. Studies in the context of mixed-phase space are interesting for the following reasons: Firstly, the random matrix theory (RMT) has been successful in characterizing the quantum systems that are chaotic in the classical limit. RMT can accurately explain level spacing distributions as well as other statistical properties like the saturation value of entanglement for a globally chaotic system. Furthermore, the regular systems (with Poisson statistics) can be well described using the diagonal unitaries with the diagonal elements chosen uniformly at random from the unit complex circle. However, in the case of mixed systems, the RMT can only capture the universal properties associated with the ratio of the phase space volumes occupied by the regular and chaotic regions \cite{rosenzweig1960repulsion, berry1984semiclassical, prosen1994semiclassical, robnik2000topics, kravtsov2015random}. Secondly, classical systems in mixed-phase space enforce a perfect separation between chaotic and regular trajectories. While the former live in the chaotic sea and have a positive Lyapunov exponent, the latter reside in the regular islands and show no chaos. However, quantum mechanically, the time-evolved operators corresponding to the mixed phase space will have support over both regular and chaotic regions of the phase space, making direct quantum-classical correspondence extremely challenging.

The works thus far have studied the mixed phase space scrambling in a few different models such as standard map \cite{notenson2023classical}, coupled cat maps \cite{bergamasco2019out}, coupled kicked rotors \cite{prakash2020scrambling}, and Bose-Hubbard models \cite{kidd2020thermalization, kidd2021saddle}. 
It has been argued that in these systems, the short-time growth provides an ambiguous indicator of transition to chaos \cite{kidd2021saddle}. This is partly because of the ambiguity associated with the states lie at the boundaries between the regular and chaotic regions. These states, despite being regular, can show characteristics (such as entanglement entropy) similar to the chaotic states \cite{lombardi2011entanglement, madhok2015comment}. Nevertheless, the long-time dynamics of the OTOCs provide a much clearer indicator of chaos in these cases \cite{garcia2018chaos}. Moreover, in fully chaotic systems, the OTOCs approach saturation exponentially, while mixed phase space leads to a multi-step relaxation of OTOCs towards saturation \cite{notenson2023classical}. In this work, to study the scrambling in mixed-phase space, we approach a different route. In particular, we invoke Percival's conjecture \cite{percival1973regular} and partition the eigenstates of the Floquet map into ``regular" and ``chaotic" and examine the behavior of OTOCs in those states. 

Apart from being chaotic, the system we consider also displays a conservation law. The conserved quantities have been shown to slow the information scrambling in various systems \cite{chen2020quantum, ope4, ope5, kudler2021information, cheng2021scrambling, balachandran2021eth}. These works typically explore random unitary circuits or spin chains as physical models. These models, however, lack a smooth classical limit. In the latter part of this work, we explore how the conserved quantity constrains the growth of OTOCs as the system approaches a classical limit, highlighting the interplay between symmetries and chaos.

This chapter is structured as follows. In Sec. \ref{Model}, we review some of the details of the KCT model. Section \ref{Information scrambling in the KCT model} is devoted to analyzing information scrambling in the KCT model. In Sec. \ref{OTOC in the largest invariant subspace} and \ref{Scrambling in the mixed phase space}, we examine scrambling in the largest subspace for the fully chaotic and mixed phase space regimes, respectively. In Sec. \ref{the IzJzOTOC} and \ref{the IxJxOTOC}, we consider the entire system KCT system and study the scrambling for two different initial operator choices. Section \ref{Scrambling, operator entanglement and coherence} is devoted to the cases of random operators, wherein we show that the OTOC for the random operators is connected to the operator entanglement and coherent generating power of the time evolution operator. Finally, we conclude the text in Sec. \ref{SUMMARY AND DISCUSSIONS}. 

\section{Model: Kicked coupled top}\label{Model}
The kicked coupled top (KCT) model, originally introduced in Ref. \cite{trail2008entanglement}, is governed by the following Hamiltonian:
\begin{equation}
H=\dfrac{\alpha}{\sqrt{|\mathbf{I}||\mathbf{J}|}} \mathbf{I.J} +\beta\sum_{n=-\infty}^{\infty}\delta(t-n\tau) J_{z},
\end{equation}
where $\alpha$ and $\beta$ denote coupling strength and kicking strength, respectively. For simplicity, here and throughout, we take $|\mathbf{I}|=|\mathbf{J}|=J$. The system evolution can be understood as alternating rotations of $\mathbf{J}$ around $z$-axis, followed by a precession of $\mathbf{I}$ and $\mathbf{J}$ about $\mathbf{F=I+J}$ by an angle proportional to $\alpha |\mathbf{F}|$. Total magnetization along $z$-axis ($F_z=\hat{I}_{z}+\hat{J}_{z}$) is the only known constant motion in this system. The absence of enough conserved quantities makes the system non-integrable and also chaotic when $\alpha$ is sufficiently large. We note that a coupled top system has been studied earlier \cite{feingold1983regular}, and a time-independent variant of this model has also been studied \cite{fan2017quantum}. 
In our work, the motivation for choosing the particular system is that this Hamiltonian describes the hyperfine interaction between nuclear spin and total electron angular momentum with a magnetic field that has a negligible effect on the nucleus. While not extending deeply into the semiclassical regime, this realization remains suitable for investigating non-trivial mesoscopic regimes in large atoms with heavy nuclei and many valence shell electrons \cite{trail2008entanglement}. In contrast, an earlier study \cite{mondal2021dynamical} on coupled tops had a Hamiltonian that does not have an isotropic interaction. Moreover, the external field in their model acts on both spins, and their main objective is to use ``Fidelity OTOCs" (FOTOCS) to detect scars.

\begin{figure}
\center
\includegraphics[scale=0.6]{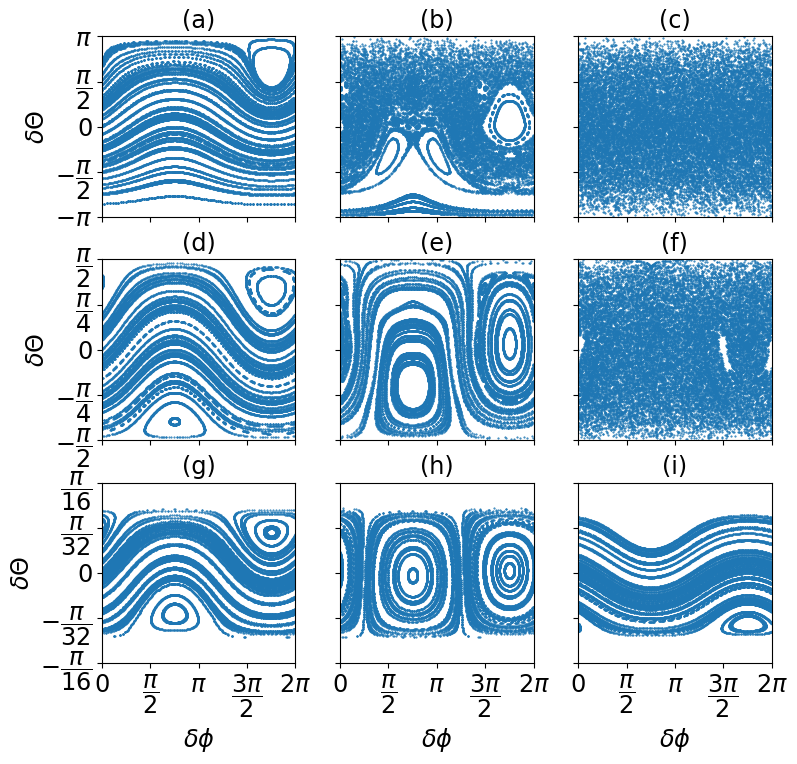}
\caption{\label{fig1} Poincaré surface of sections corresponding to different coupling strengths (along the rows) and different sectors of $F_z$ (along the columns). The kicking strength $\beta$ is kept constant at $\pi/2$. Panels (a)-(c) illustrate the case where $F_z=0$. Along the row, the coupling strengths are $\alpha=1/2$, $3/2$, and $6$ from left to right, respectively. In panels (d)-(f), $F_z=1$, while panels (g)-(i) illustrate $F_z=1.99$. For $\alpha=1/2$, the phase space in the sectors close to $F_z=0$ remains mostly regular across all $F_z$ sectors. However, $\alpha=3/2$ and $\alpha=6$ correspond to mixed and fully chaotic phase spaces. The phase space maintains regularity for all the couplings in the smaller sectors that are close to $|F_z|=2$.}
\end{figure}

Initially, the classical KCT model exhibits six degrees of freedom, constrained by $|\mathbf{I}|=|\mathbf{J}|=J$. This constraint results in four effective degrees of freedom. Conservation of $F_z$ further reduces the phase space from four to three dimensions. Additionally, fixing $F_z$ can effectively reduce the phase space to two dimensions. To actually visualize this, we write the Hamiltonian in terms of canonical conjugate pairs of sum and difference coordinates given by $(F_{z}=I_{z}+J_{z}, \bar{\phi}=\phi_{I}+\phi_{J})$ and $(\delta F_{z}=I_{z}-J_{z}, \delta \phi = \phi_{I}-\phi_{J})$. Then, the Hamiltonian becomes
\begin{eqnarray}\label{classical_eq}
H&=&\alpha\left[I_{z}J_{z}+|\mathbf{I}||\mathbf{J}|(\sin(\phi_{I})\sin(\phi_{J})+\cos(\phi_{I})\cos(\phi_{J}))\right]+\beta \sum_{n=-\infty}^{\infty}\delta(t-n\tau) J_{z}\nonumber\\
&=&\alpha\left(\frac{F_z^2-\delta F_z^2}{4}+ |\mathbf{I}||\mathbf{J}|\cos(\delta\phi)\right)+\beta\sum_{n=-\infty}^{\infty}\delta(t-n\tau)\left( \frac{F_{z}-\delta F_{z}}{2} \right), 
\end{eqnarray}
where we have defined $I_x=|\mathbf{I}|\cos(\phi_{I})$ and $I_y=|\mathbf{I}|\sin(\phi_{I})$ and similarly, $J_x=|\mathbf{J}|\cos(\phi_{J})$ and $J_y=|\mathbf{J}|\sin(\phi_{J})$.
Since $F_z$ is conserved, $\overline{\phi}$ remains a cyclic coordinate and does not appear in the Hamiltonian. Therefore, to examine an invariant Poincaré section corresponding to a constant $F_z$, we only need the variables $(\delta F_z, \delta\phi)$. For $F_z=0$, the corresponding Hamiltonian is
\begin{equation}\label{fzzeoHam}
H={\alpha}\left[\frac{-\delta F_{z}^2}{4}+|\mathbf{I}||\mathbf{J}| \cos(\delta \phi)\right]-\beta\sum_{n=-\infty}^{\infty}\delta(t-n\tau) \frac{\delta F_{z}}{2} .
\end{equation}
By writing $\delta F_z= \cos\theta_I -\cos\theta_J =\sin(\delta\theta/2)\sin(\overline{\theta}/2)$, where $\overline{\theta}=\theta_I+\theta_J=\pi$ and $\delta\theta=\theta_I-\theta_J$, the phase space can be visualized in $(\delta\theta, \delta\phi)$ variables for various initial conditions.

By noting that the rotations of the classical vectors can be implemented by the SO($3$) operators, the phase space of the KCT model can be easily visualized. To be specific, the evolution of the classical angular momentum vectors $\mathbf{I}$ and $\mathbf{J}$ can be written as follows:    
\begin{eqnarray}\label{In}
\mathbf{I}(n+1) = \exp\left\{ \alpha\left( F_x\hat{L}_x+F_y\hat{L}_{y}+F_zL_z \right) \right\}\mathbf{I}(n)  
\end{eqnarray}
and 
\begin{eqnarray}\label{Jn}
\mathbf{J}(n+1)=\exp\left\{ \beta L_z \right\} \exp\left\{ \alpha\left( F_x\hat{L}_x+F_y\hat{L}_{y}+F_zL_z \right) \right\}\mathbf{J}(n), 
\end{eqnarray}
where $L_{x, y, z}$ are the generators of the SO$(3)$ group, and the sum of the two vectors is given by $\mathbf{F}=\mathbf{I}+\mathbf{J}=[F_x, F_y, F_z]$. For convenience, in the above equations, we have normalized the classical vectors to $1$. The Poincaré sections corresponding to three different $F_z$ sectors, namely, $F_z=0$, $1$, and $1.99$ (along the columns) are shown in Fig. (\ref{fig1}) for a few different $\alpha$ values along the rows. In the absence of coupling, the system remains integrable. In the largest sector ($F_z=0$), we observe that the chaos slowly builds up with the increase in $\alpha$. However, the sectors far from $F_z=0$ show regular dynamics even when the $\alpha$ is large as evident from Fig. \ref{fig1}g-\ref{fig1}i. Figure (\ref{lyap}a) shows the illustration of maximum LE (averaged over many initial conditions) vs. $\alpha$ calculated using Benettin's algorithm \cite{benettin1980lyapunov, kuznetsov2012hyperbolic} in $F_z=0$ section. The plot suggests that the growth of the average LE versus the coupling strength is logarithmic --- $\lambda\sim\log(\alpha)$. Interestingly, this is very similar to the analytical calculation of LE done for the kicked top in Ref. \cite{PhysRevE.56.5189}. Moreover, as $F_z$ increases, the area of the corresponding phase space sector decreases. As a result, the largest Lyapunov exponent $\lambda_{F_z}$, which roughly scales like $\sim\log(\text{Area})$, decreases. Fig. (\ref{lyap}b) illustrates the same for two different $\alpha$ values. We observe that $\lambda_{F_z}$ decreases linearly with $F_z$. Near $|F_z|=2J$, the area tends to vanish, leading to highly localized dynamics. 
\begin{figure}
\center
\includegraphics[scale=0.5]{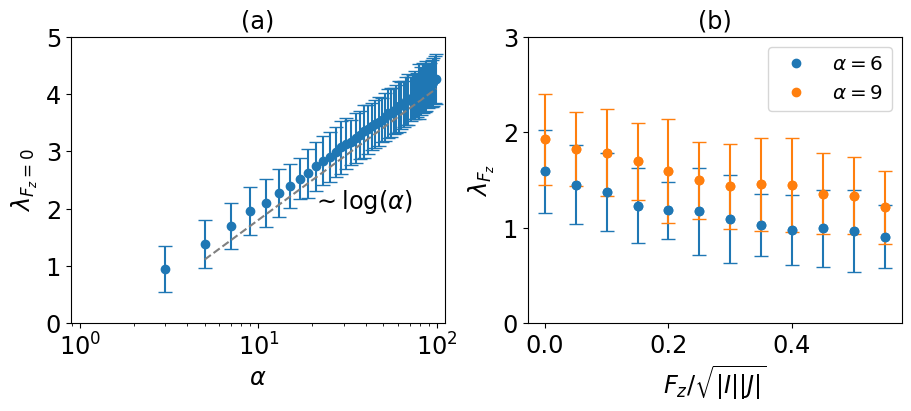}
\caption{\label{lyap} (a) Illustration of the maximum Lyapunov exponent vs. coupling strength $\alpha$ in the $F_z=0$ sector averaged over $10^3$ initial conditions. The plot suggests that the average Lyapunov exponent in the largest invariant subspace ($\lambda_{F_z=0}$) follows an approximately logarithmic law. (b) The Lyapunov exponent is calculated for different values of $F_z$. The interaction strength and the kicking parameter are kept fixed at $6$ and $\pi/2$, respectively. The error bars represent the standard deviation of the maximum Lyapunov exponent.}
\end{figure}

In the quantum regime, the evolution of the system is given by the following Floquet operator:
\begin{eqnarray}\label{KCT}
U=\exp\left\{\dfrac{-i\alpha}{\sqrt{IJ}} \mathbf{I.J}\right\} \exp\left\{-i\beta \hat{J}_{z}\right\},
\end{eqnarray}
where $\mathbf{I.J}=\hat{I}_x\hat{J}_x+\hat{I}_y\hat{J}_y+\hat{I}_z\hat{J}_z$
The quantum version of the KCT is also non-integrable due to insufficient conservation laws. In the uncoupled representation, the elements of $U$ can be written as follows:
\begin{eqnarray}\label{floq}
U_{m_1m_2, m_3m_4}=\sum_{F=|M|}^{2J}e^{-i\frac{\alpha}{2\sqrt{|I||J|}}F(F+1)}e^{-i\beta m_4}C_{I,m_1;J, m_2}^{F, M}C_{I,m_3;J, m_4}^{F, M},
\end{eqnarray}
where $C_{I, m_i;J, m_j}^{F, M}=\langle F, M|I, m_i;J, m_j\rangle$ denote Clebsch-Gordan coefficients. Similar to its classical counterpart, $U$ also admits a decomposition into various invariant subspaces characterized by the quantum number $F_z$ --- $U=\oplus_{F_z}U_{F_z}$, where $F_z$ runs from $-2J$ to $2J$. As a result, $U_{m_1m_2, m_3m_4}=0$ whenever $|m_1+m_2|\neq|m_3+m_4|$. Moreover, the dimension of each subspace is related to $F_z$ as $d=2J+1-|F_z|$.

An interesting feature of quantum chaos is its connections with random matrix theory (RMT). Dyson introduced three random unitary ensembles \cite{dyson1962threefold}, namely circular unitary ensemble (CUE), circular orthogonal ensemble (COE), and circular symplectic ensemble (CSE) based on their properties under time-reversal operation. If a system exhibits chaos in the classical limit ($\hbar \rightarrow 0$), the spectral statistics of its evolution align with one of the three random unitary ensembles. In the present case, we consider the generalized time-reversal operation and see that the system under study is invariant under the time-reversal operation.
\begin{equation}
\label{Eq:Reversal}
T=e^{i \beta \hat{J}_{z}} K,
\end{equation}
where $K$ is complex conjugation. Since both $\hat{I}_{y}$ and $\hat{J}_{y}$ change sign under conjugation, while the $x$ and $z$ components of the angular momentum remain unaffected,
\begin{equation}
K \hat{J}_{z} K = \hat{J}_{z}; \hspace{1 pc} K(\bf{I} \cdot \bf{J})K = \bf{I} \cdot \bf{J}.
\end{equation}

Hence,
\begin{eqnarray}
T U_{\tau} T^{-1} &=& \left( e^{i \beta \hat{J}_{z}} K \right) \left( e^{-i\tilde{\alpha} \bf{I} \cdot \bf{J}}  e^{-i \beta \hat{J}_{z}}  \right) \left( K e^{-i \beta \hat{J}_{z}} \right)  \\
&=&  e^{i \beta \hat{J}_{z}} \left( e^{i \tilde{\alpha} \bf{I}\cdot \bf{J}} e^{i \beta \hat{J}_{z}}\right) e^{-i \beta \hat{J}_{z}} \nonumber\\
&=& e^{i \beta \hat{J}_{z}}e^{i \tilde{\alpha}\bf{I}\cdot \bf{J}}= U_{\tau}^{\dagger}. \nonumber
\end{eqnarray}
As a result, the dynamics are invariant under time reversal operation. Moreover, $T^2=1$. Since no additional discrete symmetries are present, COE is the appropriate unitary ensemble for the KCT model.

\section{Information scrambling in the KCT model}\label{Information scrambling in the KCT model}
Recall from the previous section that the KCT model conserves $\hat{F}_z = \hat{I}_z + \hat{J}_z$, i.e., $[U, \hat{F_z}]=0$. The conservation law causes the decomposition of the dynamics into invariant subspaces. Among these subspaces, the largest one typically remains dominant. Hence, scrambling in this system is mainly driven by the largest subspace dynamics, but the regular dynamics in smaller subspaces somewhat offset it. In the following, we will examine scrambling in the largest subspace, emphasizing OTOC dynamics in the chaotic and mixed phase space regimes. Afterward, we will consider the entire system and analyze the OTOC dynamics for two pairs of initial operators, namely, $(\hat{I}_{z}, \hat{J}_{z})$ --- where both operators commute with the total magnetization operator $\hat{F}_z$, and $(\hat{I}_{x}, \hat{J}_{x})$ --- where neither operator commutes with $\hat{F}_z$.

\subsection{OTOC in the largest invariant subspace}\label{OTOC in the largest invariant subspace}
\subsubsection{Short-time growth}
Unlike smaller subspaces ($|F_z|\gg 0$), at sufficiently strong coupling $(\alpha\gtrsim 6)$, the subspaces close to $F_z=0$ are predominantly chaotic. Here, we numerically study the OTOC for a pair of operators acting exclusively on the largest invariant subspace ($F_z=0$) of the KCT model. 
The corresponding Floquet evolution in this subspace can be written as follows:
\begin{equation}\label{floqsubspace}
U_{F_z=0}=\sum_{F=0}^{2J}\sum_{m_1, m_2}e^{-i(\frac{\alpha}{2J}F(F+1)+\beta m_2)}C_{I,m_1;J, -m_1}^{F, 0}C_{I,m_2;J, -m_2}^{F, 0}|m_1\rangle\langle m_2|.
\end{equation}
Here, the indices $m_1$ and $m_2$ run from $-J$ to $J$, implying that the unitary acts on $(2J+1)$-dimensional Hilbert space.

\begin{figure}
\center
\includegraphics[scale=0.55]{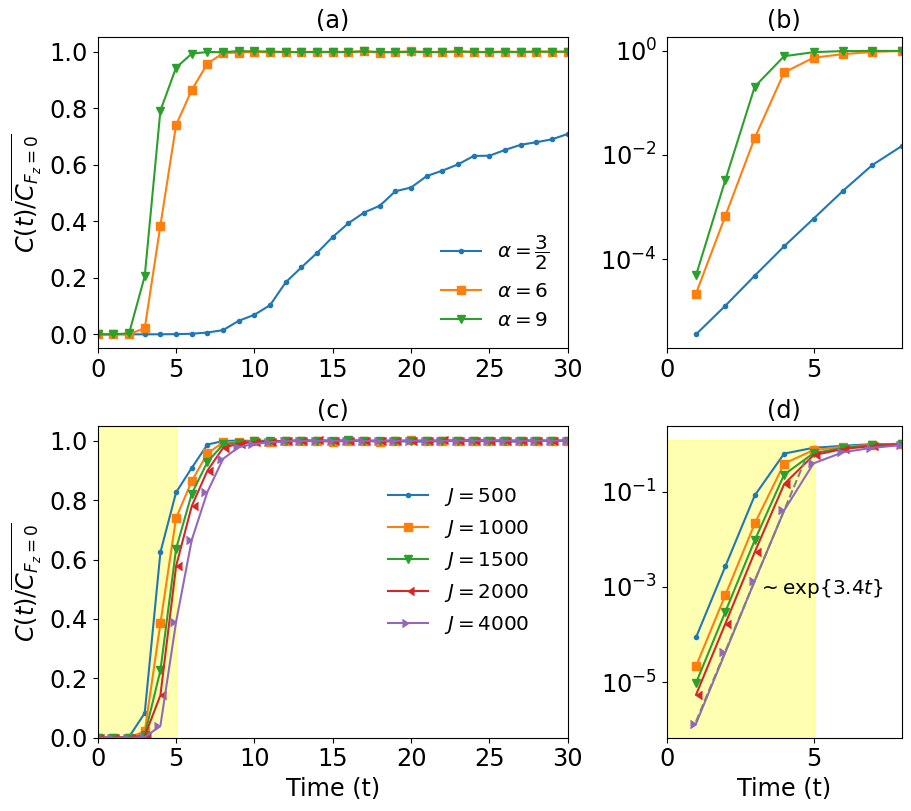}
\caption{\label{largestotoc} Illustration of the OTOC for different $\alpha$ and $J$ values in the largest invariant subspace. The initial operators are $A=B=\hat{S}_z$, the generator of rotations around $z$-axis. These operators act exclusively on the largest subspace $F_z=0$. The initial state is $\rho=\mathbb{I}/(2J+1)$, the maximally mixed state. In panel (a), the OTOC is shown for three different couplings, namely, $\alpha=3/2$, $6$, and $9$. The magnitude of $J$ is fixed at $1000$. Panel (b) shows corresponding early-time exponential growth on a semi-log plot. The plot indicates that the OTOC growth rate increases with an increase in $\alpha$. Panels (c) and (d) demonstrate OTOC growth for varying $J$ values while keeping $\alpha=6$. The plots are normalized by dividing the OTOC $C_{F_z=0}(t)$ with $\overline{C_{F_z=0}}$, the infinite time average of the OTOC.} 
\end{figure}

For the OTOC calculations, we take the initial operators $A=B=\hat{S}_z$, where $S_z$, the generator of rotation along $z$-axis --- $\hat{S}_{z}=\sum_{m=-J}^{J}m|m\rangle\langle m|$, acts non-trivially on the $F_z=0$ subspace. Then, the OTOC is given by 
\begin{eqnarray}
 C_{F_z=0}(t)=\dfrac{1}{2J+1}\left[ \tr\left( \hat{S}^2_z(t)\hat{S}^2_{z} \right)-\tr\left( \hat{S}_{z}(t)\hat{S}_z\hat{S}_z(t)\hat{S}_z \right) \right],  
\end{eqnarray}
where $S_z(t)$ denotes the Heisenberg evolution of $S_z$ under the dynamics given in Eq. (\ref{floqsubspace}). Here, we are fixing the maximally mixed state $\rho=\mathbb{I}/(2J+1)$ as the initial state for the OTOC calculations. In the classical limit, for chaotic systems, the Ehrenfest time is expected to scale as $t_{\text{EF}}\sim\log(\text{dim})/\lambda_{\text{cl}}$, where $\lambda_{\text{cl}}$ is the corresponding classical Lyapunov exponent. Within this timescale, $C_{F_z=0}(t)$ shows an exponential growth if $\alpha$ is sufficiently large. The corresponding numerical results are shown in Fig. \ref{largestotoc}. Figure \ref{largestotoc}a illustrates the operator growth for different values of $\alpha$ by plotting the normalized OTOC --- $C_{F_z=0}(t)/\overline{C_{F_z=0}}$ versus time. Here, $\overline{C_{F_z=0}}$ denotes infinite time average of $C_{F_{z}=0}(t)$ and is given by
\begin{eqnarray}
\overline{C_{F_z=0}}&=&\lim_{t\rightarrow\infty}\dfrac{1}{t}\int_{0}^{t}C_{F_z=0}(s)ds\nonumber\\
&=&\dfrac{1}{2J+1}\left[\sum_{m}[A^2]_{m m}[B^2]_{mm}-[A_{mm}]^2[B_{mm}]^2-\sum_{m\neq n}\left(A_{m m}A_{n n}B_{m n}B_{nm}+A_{mn}A_{n m}B_{m m}B_{nn}\right)\right],\nonumber\\
\end{eqnarray}
where $A_{mn}=\langle E_m|A|E_n\rangle$ etc. and $U_{F_z=0}=\sum_{n}e^{-i\phi_{n}}|E_n\rangle\langle E_n|$ denotes the eigen-decomposition of the subspace Floquet operator. In the figure, the angular momentum is kept fixed at $J=1000$. The figure shows that the growth rate increases with $\alpha$. Moreover, for $\alpha=3/2$, the classical phase space has a mix of regular islands and the chaotic sea [see Fig. \ref{fig1}b]. This results in a slower OTOC growth than the fully chaotic cases (corresponding to $\alpha=6$ and $9$), as illustrated in figure \ref{largestotoc}a. We elaborate more on the mixed-phase space OTOC dynamics in the next subsection. Figure \ref{largestotoc}b demonstrates the corresponding early-time exponential growths of the OTOCs plotted in \ref{largestotoc}a. In Fig. \ref{largestotoc}c and \ref{largestotoc}d, we take the fully chaotic case by fixing $\alpha=6$ and contrast the OTOC growth for different dimensions. The corresponding numerical results yield a quantum exponent of $\lambda_{\text{quant}}\approx 1.67$ for $J=4000$, closely matching the classical LE, $\lambda_{\text{cl}}\approx 1.55$, indicating a good classical-quantum correspondence. The quantum exponent is extracted by fitting the first five data points ($0\leq t<4$) in the curve to an exponential scale. A detailed presentation of the quantum exponents for different $J$-values is given in the table. \ref{table1}.
\begin{figure}
\center
\includegraphics[scale=0.55]{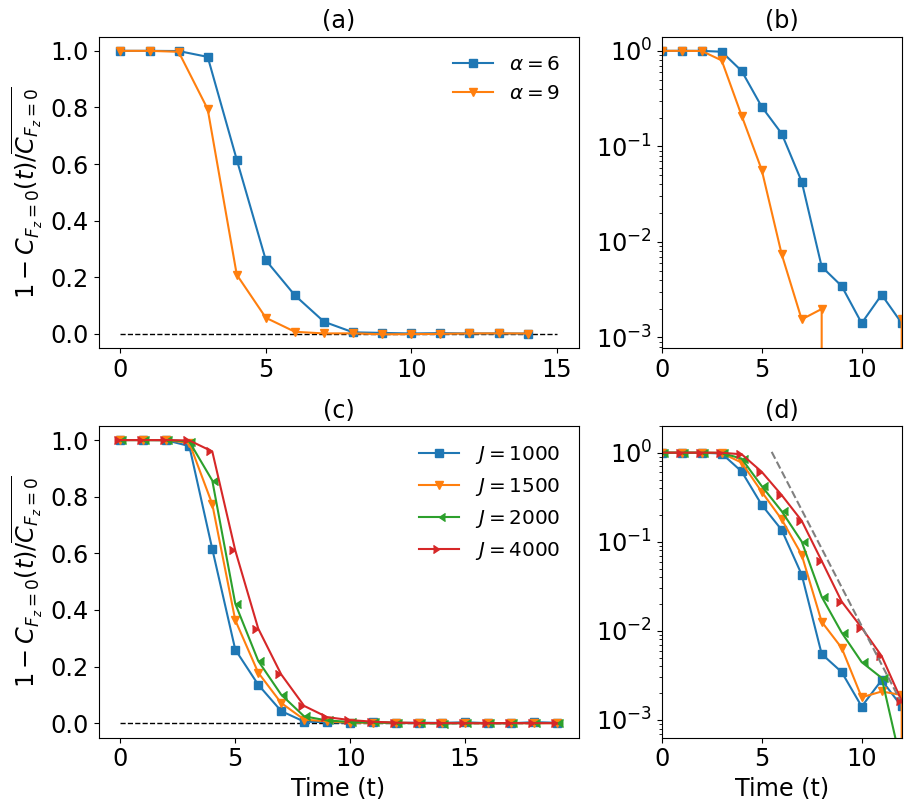}
\caption{\label{largest_relax} Relaxation dynamics of the OTOC as probed by the quantity $1-C_{F_z=0}(t)/\overline{C_{F_z=0}}$. The initial operators are the same as before: $A=B=\hat{S}_z$. Panel (a) and (b) depict relaxation for two different couplings: $\alpha = 6$ and $9$, with fixed $\beta = \pi/2$ and $J = 1000$, shown on linear and semi-log plots, respectively. The plot illustrates that the relaxation to the saturation is faster when $\alpha$ is more. In (c), we fix $\alpha=6$ and vary $J$ from $500$ to $4000$. The plot suggests that the relaxation time scale increases logarithmically with the dimension of the system. Panel (d) illustrates the corresponding relaxation dynamics on a semi-log plot. Note the horizontal black dashed lines in (a) and (c) along the zero on the $Y$-axis. These lines are drawn to demonstrate that the long-time behavior approaches zero.}
\end{figure}

\subsubsection{Relaxation dynamics}
A sharp relaxation to the saturation follows the initial growth. While the short-time behavior of the OTOCs is well-studied and understood, relaxation is relatively unexplored, with only a few studies. For example, exponential relaxation has been found in random quantum circuit models \cite{bensa2022two} and maximally chaotic Floquet systems \cite{claeys2020maximum}. Relaxation in weakly coupled bipartite systems with chaotic subsystems has been studied in Ref. \cite{prakash2020scrambling}. It has been demonstrated that for fully chaotic systems, the relaxation to the equilibrium follows an exponential scaling \cite{polchinski2015chaos, garcia2018chaos}. 
Given a localized state $|\psi\rangle$ and a fully chaotic $U$, randomization of the state implies that $|\langle\psi|U^t|\psi\rangle|^2\sim e^{-\gamma t}\sim 1/(2J+1)$. Then, the information gets fully scrambled in a time window of $t_{\text{sc}}\sim\log(2J+1)/\gamma$. Note that $t_{\text{sc}}$ can not be smaller than $t_{\text{EF}}$, i.e., $t_{\text{EF}}\lesssim t_{\text{sc}}$. Therefore, in strongly chaotic systems, the relaxation takes place over a window of time $t_{\text{relx}}=(t_{\text{sc}}-t_{\text{EF}})\sim\log(2J+1)$. Beyond this, we expect that the OTOC will get saturated. 
To study the relaxation, we take the quantity $1-C_{F_z=0}(t)/\overline{C_{F_z=0}}$. For $t=0$, we have $C_{F_z=0}(t)=0$, implying that $1-C_{F_z=0}(t)/\overline{C_{F_z=0}}=1$. In the long time limit, this quantity approaches $0$. This is because, as $t\rightarrow\infty$, we have $C_{F_z=0}(t)\sim \overline{C_{F_z=0}}$.
The relaxation dynamics are illustrated in Fig. \ref{largest_relax}. Figure \ref{largest_relax}a and \ref{largest_relax}b demonstrate the relaxation for two different $\alpha$ values while keeping $J=1000$. In Fig. \ref{largest_relax}c and \ref{largest_relax}d, we illustrate the relaxation for different $J$ with $\alpha$ fixed. As time progresses, the curves approach zero with fluctuations. These fluctuations are of the order $10^{-3}$ to $10^{-4}$. We further observe that the relaxation follows an exponential decay for the parameters considered. This result is consistent with the results obtained in Ref. \cite{garcia2018chaos}. 
Though not the central focus of our work, we find the relaxation dynamics an interesting observation subject to further exploration.
\begin{figure}
\center
\includegraphics[scale=0.6]{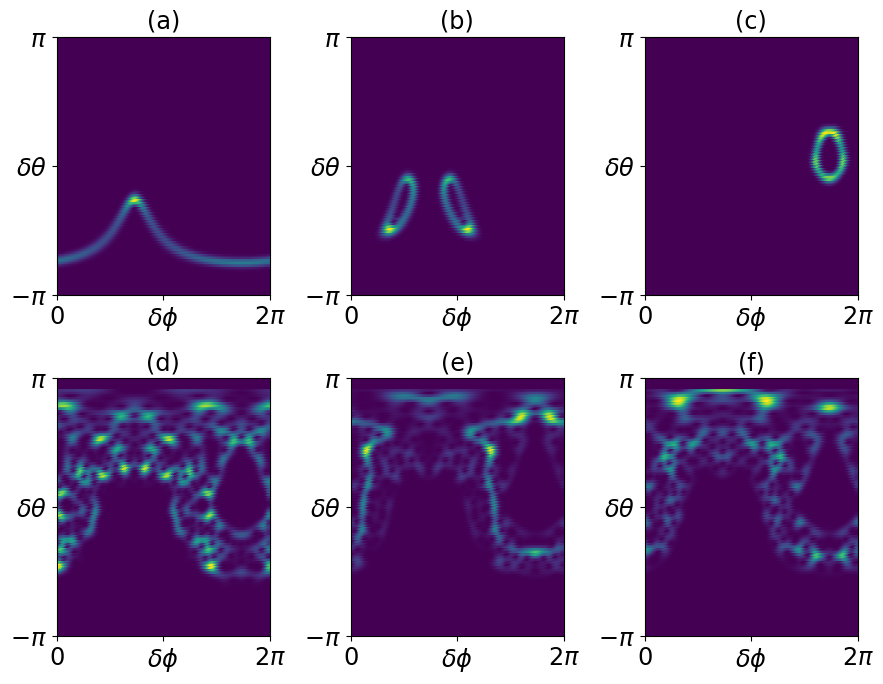}
\caption{\label{fig:eigen_husimi} Husimi plots of the Floquet states of $U_{F_z=0}$. The system parameters correspond to the mixed phase space in the classical limit --- $\alpha=3/2$ and $\beta=\pi/2$. We take $I=J=200$. The Floquet states in (a), (b), and (c) are chosen randomly from the series of points inside the boxes (a), (b), and (c) of Fig. \ref{fig:C_inf_eigen_vs_Sz}, respectively. These three states are localized around different fixed points. (d)-(f) panels correspond to the chaotic Floquet states randomly chosen from the box \ref{fig:C_inf_eigen_vs_Sz}d.  }
\end{figure}

\subsection{Scrambling in the mixed phase space }\label{Scrambling in the mixed phase space}
In fully chaotic systems, the OTOCs exhibit initial exponential growth and subsequent relaxation, followed by saturation. The saturation value can be predicted using an appropriate random matrix ensemble, such as the COE in the case of the KCT model \cite{haake1991quantum}. However, for the systems with mixed phase space classical limit, a direct correspondence with RMT is absent. Moreover, in these systems, the choice of the initial state largely influences the OTOC growth and saturation, a behavior uncommon in fully chaotic cases. To be precise, the initial states with a significant overlap with the coherent states centered near stable fixed points of the phase space limit the degree of operator scrambling. On the contrary, the scrambling is enhanced if the initial state is localized in the chaotic sea. We observe this behavior in both short-term and long-term dynamics of the OTOC.

To begin, we first consider an initial state, which has a non-zero overlap with the coherent states located on the chaotic sea. Percival's conjecture can be used to construct such states \cite{percival1973regular}. The conjecture categorizes Eigenstates or Floquet states of the system into regular (near stable fixed points) and chaotic states (randomly distributed across chaotic regions). To illustrate this in the KCT model, we show the Husimi plots of six randomly chosen Floquet states corresponding to $\alpha=3/2$, $\beta=\pi/2$, and $J=200$ in Fig. \ref{fig:eigen_husimi}. For a given state, the Husimi function is given by $F_H=\langle \delta\theta, \delta\phi |\rho|\delta\theta, \delta\phi\rangle$, where $|\delta\theta, \delta\phi\rangle$ represents the projection of tensor products of spin coherent states $|\theta_I, \phi_I\rangle \otimes |\theta_J, \phi_J\rangle$ onto the largest invariant subspace ($F_z=0$) \cite{trail2008entanglement}:
\begin{align}\label{spin_coherent state}
|\delta\theta, \delta\phi\rangle =\frac{1}{\mathcal{N}}\sum_{m=-J}^{J}\mu^{m}\frac{(2J)!}{(J-m)!(J+m)!}|m,-m\rangle,
\end{align}
where 
\begin{equation*}
\mu=e^{i\delta\phi/2}\left(\dfrac{1+\sin(\delta\theta /2)}{1-\sin(\delta\theta /2)}\right)
\end{equation*} 
and $\mathcal{N}$ denotes the normalizing constant. In Fig. \ref{fig:eigen_husimi}, the top panels (\ref{fig:eigen_husimi}a, \ref{fig:eigen_husimi}a and \ref{fig:eigen_husimi}c) correspond to the Husimi plots of the Floquet states localized near regular islands. On the other hand, the bottom panels (\ref{fig:eigen_husimi}d, \ref{fig:eigen_husimi}e, and \ref{fig:eigen_husimi}f) represent the states delocalized across the chaotic sea. It's important to note that Percival's conjecture is not always accurate, particularly in the deep quantum regime. Nevertheless, quantities such as Husimi entropy of the states can be employed to roughly distinguish the regular states from the chaotic ones \cite{trail2008entanglement}. As chaotic Floquet states are widespread across the chaotic sea, we choose one such state and initialize the system in it. We then compute the OTOC in Eq. (\ref{commutator}) for the largest subspace by fixing the initial operators $A=B=\hat{S}_z$:
\begin{figure}
\center
\includegraphics[scale=0.5]{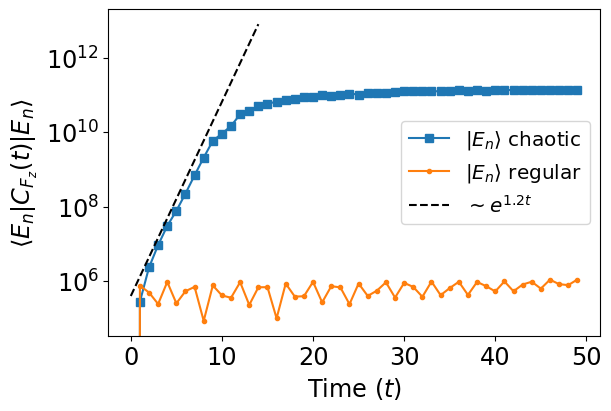}
\caption{\label{fig:OTOC_mixed} Illustration of the OTOC growth in different initial states when the system is associated with the mixed phase space in the $F_z=0$ subsector of the classical phase space. We consider the Floquet states as the initial states. We fix $\alpha=3/2$, $\beta=\pi/2$ and $J=1000$. The initial operators are $A=B=\hat{S}_z$. In the figure, the blue curve represents the OTOC growth in a chaotic Floquet state. The orange curve corresponds to a regular Floquet state localized near a regular island shown in Fig. \ref{fig:eigen_husimi}a. While the former grows at a rate $\sim e^{1.2 t}$ ($\lambda_{\text{otoc}}\approx 0.6$) initially, the latter one fluctuates around a mean value for all the times. The classical Lyapunov exponent for initial conditions in the chaotic sea is $\lambda_{\text{cl}} \approx 0.5$, determined using Benettin's algorithm. }
\end{figure}
\begin{eqnarray}\label{EN}
\langle E_n|C_{F_z=0}(t)|E_n \rangle= \langle E_n|\left[\hat{S}_z(t), \thinspace \hat{S}_z\right]^{\dagger}\left[\hat{S}_z(t), \thinspace \hat{S}_z\right] |E_n\rangle, 
\end{eqnarray}
where $|E_n\rangle$ denotes a randomly chosen chaotic Floquet state of the system. The corresponding results are shown in Fig. \ref{fig:OTOC_mixed}. In the figure, the blue curve represents the OTOC growth in the chaotic initial state. In this case, the OTOC shows an early time exponential growth ($\sim e^{1.2t}$), followed by saturation. To contrast this with the regular states, we also considered a Floquet state localized near a regular island shown in Fig. \ref{fig:eigen_husimi}a and computed the OTOC in it. In the figure, the regular state OTOC is shown in orange. In this case, the OTOC does not grow and shows fluctuations around its mean value. This clearly indicates that the regular states constrain the scrambling of operators in the mixed-phase space. On the other hand, the chaotic states enhance the scrambling. 

\begin{figure}
\center
\includegraphics[scale=0.6]{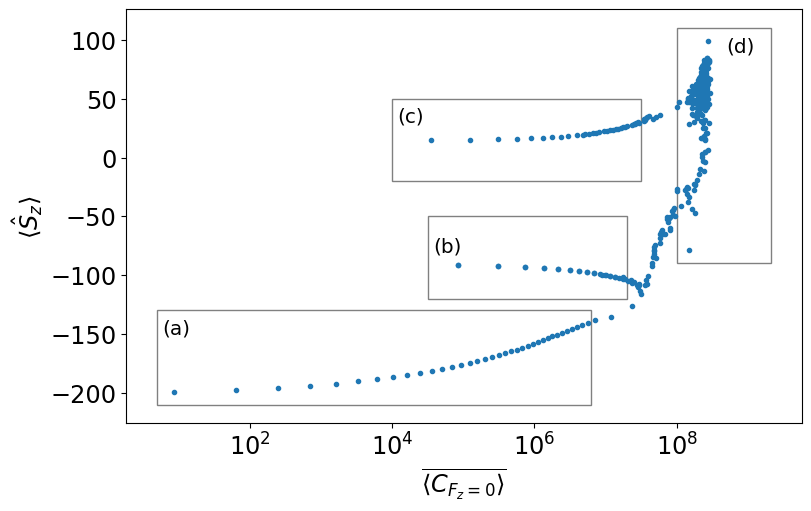}
\caption{\label{fig:C_inf_eigen_vs_Sz} Illustration of $\langle S_z\rangle =\langle E_n| \hat{S}_{z}|E_n\rangle$ versus long-time averaged OTOC with respect the initial state $|E_n\rangle$, when the classical limit of the system is associated with the mixed phase ($\alpha=3/2$ and $\beta=\pi/2$) in the $F_z=0$ sector. The states $\{|E_n\rangle\}$s denote the Floquet states of the operator $U_{F_z=0}$. The angular momentum value is taken to be $I=J=200$. For all the data points, the initial operators are fixed: $A=B=\hat{S}_{z}$. Boxes (a), (b), and (c) correspond to the Floquet states localized on the regular regions of the phase space. Box (d) represents a series of chaotic states that display flat long-time averaged OTOC. Note that the boxes are only representative of the states corresponding to different regions of the phase space and do not characterize the boundaries between these regions. }
\end{figure}

We further elucidate the role of Percival's conjecture in the saturation of mixed-phase space OTOC. To do so, we calculate infinite time averages of the OTOC in the Floquet states and contrast the regular states with the chaotic states:
\begin{eqnarray}\label{floquet_otoc}
\overline{\langle E_n| C_{F_z=0}|E_n\rangle}=\lim_{t\rightarrow\infty}\dfrac{1}{t}\int_{0}^{t}ds\langle E_n|C_{F_z=0}(s)|E_n \rangle,
\end{eqnarray}
where $\langle E_n|C_{F_z=0}(t)|E_n \rangle$ is taken from Eq. (\ref{EN}).
We now cluster the Floquet states based on the infinite time averages and their mean locations in the phase space. 
The mean location can be identified by finding $\langle E_n|\hat{S}_z|E_n\rangle$, which correlates to $\delta\theta$ in the semiclassical limit. Figure \ref{fig:C_inf_eigen_vs_Sz} shows $\overline{\langle E_n|C_{F_z=0}|E_n\rangle}$ vs. $\langle E_n|\hat{S}_z|E_n\rangle$ for a fixed $J=200$ and $\alpha=3/2$. The correlation between $\langle\hat{S}_{z}\rangle$ and $\delta\theta$ causes the Floquet states near stable fixed points to display nearly identical $\langle \hat{S}_z\rangle$ values \cite{gorin1997phase, trail2008entanglement}. In the figure, box (a) corresponds to the series of states localized near the stable south pole [see Fig. \ref{fig:eigen_husimi}a]. The boxes (b) and (c) correspond to the regular islands displayed in Fig. \ref{fig:eigen_husimi}b and Fig. \ref{fig:eigen_husimi}c. Moreover, the chaotic Floquet states typically display higher $\overline{\langle E_n|C_{F_z=0}|E_n\rangle}$ as they are delocalized in the chaotic sea and approximate random quantum states. These states are concentrated vertically along a line within box (d). Note that the plots in Fig. \ref{fig:eigen_husimi}d-\ref{fig:eigen_husimi}f represent three random Floquet states from the box (d). Lesser values of the time-averaged OTOC for the states chosen from boxes (a), (b), and (c) indicate that the operators are less prone to get scrambled if the initial state is localized on a regular island. This analysis shows that the OTOCs are sensitive to the initial state vectors when the classical limit of the system is in a mixed-phase space. Moreover, it is to be noted that the infinite time averages for the regular Floquet states appear to vary smoothly with $\langle \hat{S}_z\rangle$ within the boxes (a), (b), and (c). The points near the extreme right ends of these boxes exhibit infinite time averages close to chaotic states despite being regular. These points represent the states that are localized near the boundaries between the chaotic sea and the regular islands of the phase space. Similar behavior has been previously found for the entanglement entropy \cite{lombardi2011entanglement, madhok2015comment}. Hence, at the boundaries, the operators get scrambled despite the initial states there being regular. Furthermore, our findings complement previous results indicating that hyperbolic fixed points can induce scrambling \cite{hashimoto2020exponential, pilatowsky2020positive, pappalardi2018scrambling, hummel2019reversible, xu2020does, steinhuber2023dynamical}. Conversely, our results demonstrate that regular states located near the boundary of the stable islands with the chaotic region can also induce scrambling.

\begin{figure}
\center
\includegraphics[scale=0.5]{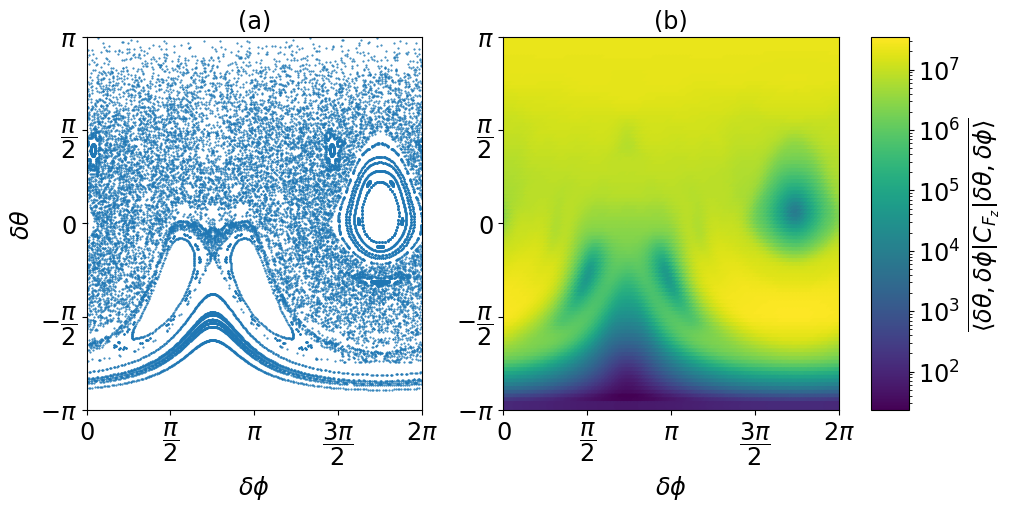}
\caption{\label{fig:husimi_comp} (a) and (b). Side-by-side comparison, showing the infinite-time average of the largest subspace OTOC in the spin coherent states (Fig. (b)) as a superb signature of classical chaos in mixed-phase space ($\alpha =\frac{3}{2}, \beta =\frac{\pi}{2}$) as shown by the Poincaré surface of the section in (a). The angular momentum value is taken to be $I=J=100$. The initial operators are $A=B=\hat{S}_z$. }
\end{figure}

\subsubsection{Quantum-Classical correspondence in the phase space}

Quantum signatures of classical phase space structures appear in various contexts, including the dynamical generation of entanglement \cite{trail2008entanglement, miller1999signatures, ghose2004entanglement, wang2004entanglement, bandyopadhyay2002testing, dogra2019quantum}, tri-partite mutual information \cite{seshadri2018tripartite}, quantum discord \cite{madhok2015signatures} and information gain in quantum state tomography \cite{madhok2014information, madhok2016characterizing}. Does the long-time average of the OTOC in the coherent states display such signatures? To proceed, we replace the Floquet states $\{|E_n\rangle\}$ with the coherent states in the Eq. (\ref{floquet_otoc}) to calculate $\overline{\langle\delta\theta, \delta\phi|C_{F_z=0}|\delta\theta, \delta\phi\rangle}$. Figure \ref{fig:husimi_comp}a and \ref{fig:husimi_comp}b shows the side-by-side comparison between the classical Poincaré section and the scatter plot of the infinite time averaged OTOC (unnormalized) in a mixed-phase space ($\alpha =3/2$). The purpose of this analysis is twofold.
 Firstly, We see a remarkable correlation between the classical phase space and $\overline{\langle\delta\theta, \delta\phi| C_{F_z=0}|\delta\theta, \delta\phi\rangle}$. The latter reproduces classical phase space structures, such as regular islands and the chaotic sea. We see that the values of the quantum calculation attain a fairly uniform value across the entire chaotic sea irrespective of the coordinates of the initial coherent state employed in its computation. Therefore, the information about the initial coordinates of the coherent states in the chaotic sea gets washed away and cannot be recovered from the long-time average OTOC values as seen in the contour plots. This is in contrast with structures seen at the boundary of the chaotic sea and regular islands and also within the regular islands. For example, the darker regions encircle stable fixed points where the operators remain stable and less prone to scrambling. Secondly, while OTOCs have been primarily employed to study the initial rate of divergence of trajectories as captured by Lyapunov exponents, longtime averages of OTOCs provide an understanding of the quantum-classical border.

\begin{table}
\center
\begin{tabular}{ |p{2cm}||p{3cm}|p{3cm}||p{3.cm}|}
\hline
\multicolumn{4}{|c|}{$\alpha=6, \beta=\pi/2$} \\
\hline
Dimension & $\lambda_{F_z=0}$& $\lambda_{\hat{I}_{z}\hat{J}_{z}}$ &$\lambda_{\text{cl}(F_z=0)}$\\
\hline
$J$=500   & 1.7108$\pm O(10^{-4})$ $0\leq t< 4$  &1.47$\pm O(10^{-2})$ $0\leq t< 4$ &{$\approx$ 1.55}\\  
\vspace{.15cm} & \vspace{.15cm}&\vspace{.15cm}&\vspace{.15cm}\\
$J$=1000   & 1.7246$\pm O(10^{-4})$ $0\leq t< 4$ & 1.362$\pm O(10^{-2})$ $0\leq t< 4$&\\
\vspace{.15cm} & \vspace{.15cm}&\vspace{.15cm}&\vspace{.15cm}\\
$J$=2000   & 1.6406$\pm O(10^{-3})$ $0\leq t < 5$ & &\\
\vspace{.15cm} & \vspace{.15cm}&\vspace{.15cm}&\vspace{.15cm}\\
$J$=4000   & 1.6988$\pm O(10^{-3})$ $0\leq t< 5$ & &\\
\hline
\end{tabular}
\caption{The table illustrates the classical Lyapunov exponent of the largest subspace ($\lambda_{\text{cl}(F_z=0)}$) and contrasts it with the OTOC growth rate in the largest subspace ($\lambda_{F_z=0}$) for different $J$-values. Further contrast is made between $\lambda_{F_z=0}$ and the growth rate of $C_{\hat{I}_z\hat{J}_z}(t)$, which is denoted by $\lambda_{\hat{I_z}\hat{J}_z}$. The table also provides the time scales over which the exponential fit is considered. We fix the parameters of the system at $\alpha=6$ and $\beta=\pi/2$. Closeness of $\lambda_{F_z=0}$ and $\lambda_{\text{cl}(F_z=0)}$ indicates good quantum-classical correspondence. Moreover, we observe that $\lambda_{\hat{I}_z\hat{J}_z}$ is always slightly less than $\lambda_{F_z=0}$ (see the main text for more details).}
\label{table1}
\end{table}


\begin{figure}
\center
\includegraphics[scale=0.5]{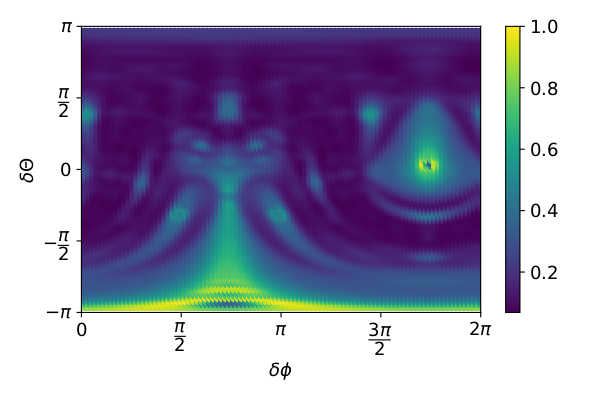}
\caption{\label{fig:otocb} Long time average of the squared commutator when the operators are projectors onto the coherent states. The expectation of the commutator function is taken over a maximally mixed state. The color bar represents the value of the commutator function (normalized).}
\end{figure}

We now consider the initial operators, which are projectors onto the coherent states, i.e., $A=B=|\delta\theta, \delta\phi\rangle\langle\delta\theta, \delta\phi|$. Figure (\ref{fig:otocb}) displays the density plot for the commutator function as a function of the mean coordinate of the coherent states. The plot is reminiscent of the classical Poincaré surface of the section shown in Fig. \ref{fig:husimi_comp}a when $\alpha=3/2$. Intuitively, we can understand the density plot as follows. When the operators are coherent state projectors, the squared commutator becomes 
\begin{equation}
C(t)=|\langle \psi|U(t)| \psi \rangle|^2-|\langle \psi|U(t)| \psi\rangle|^4.
\end{equation}

If the state lies deep inside the regular island close to a stable fixed point, the application of the Floquet operator should not take the state far away from the initial state. This implies that $|\langle \psi|U(t)| \psi \rangle|^2$ is expected to be close to unity. On the other hand, if the state lies in the chaotic sea then the quantity $|\langle \psi|U(t)| \psi \rangle|^2$ is expected to be close to zero as the application of the Floquet operator immediately randomizes the state. In both cases, the squared commutator will be close to zero, and hence, we observe darker contrast near these regions in the density plot. All the other regions, such as regions in the regular islands away from fixed points and stable orbits, display brighter contrast. In other words, the projectors onto the coherent states in the chaotic sea and the states corresponding to the stable fixed points are less prone to get scrambled.

\subsection{The $\hat{I}_{z}\hat{J}_{z}$ OTOC}\label{the IzJzOTOC}
Here, we compute the OTOC for the entire system of KCT by considering the initial operators $A=\hat{I}_{z}\otimes \mathbb{I}$ and $B=\mathbb{I}\otimes \hat{J}_{z}$. Since $A$ and $B$ commute with $F_z$, they can be decomposed as $A=\oplus_{F_z}A_{F_z}$, and $B=\oplus_{F_z}B_{F_z}$. Accordingly, the commutator function $C(t)$ becomes
\begin{equation}\label{F_z_conserving}
C_{\hat{I}_{z}\hat{J}_{z}}(t)=\dfrac{1}{(2J+1)^2} \sum_{F_z}\tr[A_{F_z}^2(t)B_{F_z}^2]-\tr[A_{F_z}(t)B_{F_z}A_{F_z}(t)B_{F_z}]
\end{equation}
where $A_{F_z}(t)=U^{\dagger t}_{F_z}A_{F_z}U^{t}_{F_z}$. Given the absence of cross-terms between different subspaces, each subspace makes an independent contribution to the OTOC. We shall see from the numerical results that $C_{\hat{I}_{z}\hat{J}_{z}}(t)$ shows a growth rate slightly less than that of $C_{F_z=0}(t)$ over a short period. This can be intuitively understood by examining the following quantity:
\begin{equation}
\dfrac{C_{\hat{I}_{z}\hat{J}_{z}}(t)}{C_{F_z=0}(t)}=\dfrac{1}{2J+1}\left[1+\sum_{\substack{F_z=-2J\\ F_z\neq 0}}^{2J}\left(1-\dfrac{|F_z|}{2J+1}\right)\dfrac{C_{F_z}(t)}{C_{F_z=0}(t)}\right], 
\end{equation}
where $C_{F_z}(t)$ denotes OTOC contribution from a subspace with quantum number $F_z$. Due to the fully chaotic nature, the largest subspace will likely have the shortest Ehrenfest time. For $t\lesssim t_{\text{EF} (F_z=0)}$, there is a high chance that in the subspaces close to the largest one, the OTOCs will grow exponentially with the rates $\lambda_{F_z\neq 0}\lessapprox\lambda_{F_z=0}$. Consequently, $C_{F_z\neq 0}(t)/C_{F_z=0}(t)\sim\exp{2(\lambda_{F_z\neq 0}-\lambda_{F_z=0})t}$ either decays exponentially but slowly or stays a constant. On the other hand, for $|F_z|\gg 0$, the quantity $C_{F_z\neq 0}(t)/C_{F_z=0}(t)$ decays quickly to zero. Therefore, $C_{\hat{I}{z}\hat{J}{z}}(t)/C_{F_z=0}(t)$ most likely decays for $t\lesssim t_{\text{EF}(F_z=0)}$. As a result, the initial growth rate of $C_{\hat{I}{z}\hat{J}{z}}(t)$ is expected to be slightly lower than that of $C_{F_z=0}(t)$. We confirm this from the numerical simulations shown in Fig. \ref{fig:zzotoc}. For the numerical simulations, we consider $\alpha=6$, which is strong enough to make a large number of subspaces fully chaotic. 
\begin{figure}
\center
\includegraphics[scale=0.5]{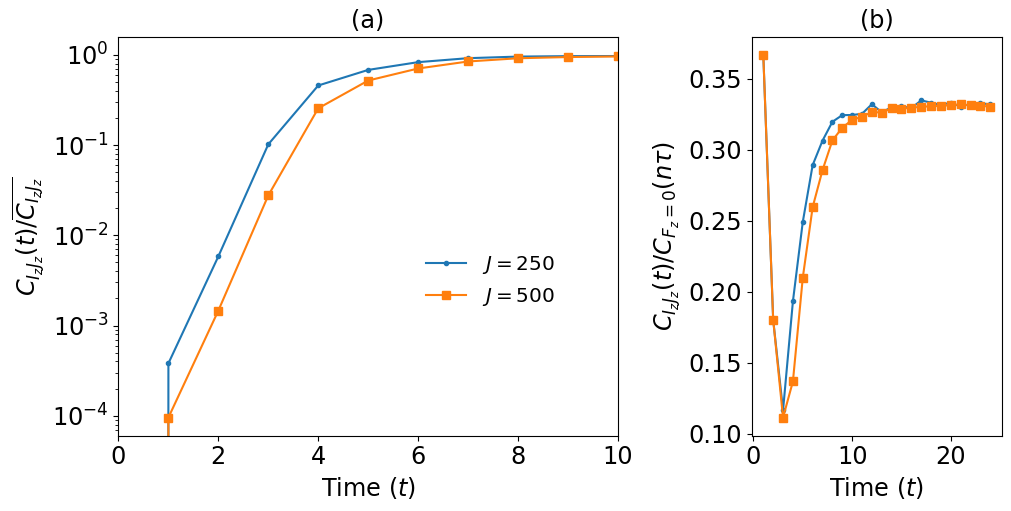}
\caption{\label{fig:zzotoc} (a) Illustration of the early-time exponential growth of $C_{\hat{I}_{z}\hat{J}_{z}}(n\tau)/\overline{C_{\hat{I}_{z}\hat{J}_{z}}}$ for two different $J$ values, where $\overline{C_{\hat{I}_z\hat{J}_{z}}}$ is the long-time average of $C_{\hat{I}_z\hat{J}_z}(n\tau)$. The initial operators are $A=\hat{I}_{z}\otimes \mathbb{I}$ and $B=\mathbb{I}\otimes \hat{J}_{z}$. Here, we fix the parameters $\alpha=6$ and $\beta=\pi/2$. See Table. \ref{table1} for the quantum exponents extracted from the OTOCs. (b) $C_{\hat{I}_{z}\hat{J}_{z}}(n\tau)/C_{F_z=0}(n\tau)$ vs. $n\tau$ for the same $J$ values as before, which aims to elucidate the influence of the largest subspace dynamics on the overall OTOC dynamics. }   
\end{figure} 

In Fig. \ref{fig:zzotoc}a, we plot the short-time exponential growth of the OTOC for two different $J$ values. The plot is normalized by dividing $C_{\hat{I}_{z}\hat{J}_{z}}(t)$ with its infinite time average. Table. \ref{table1} makes a comparison between $\lambda_{\hat{I}_{z}\hat{J}_{z}}$ and $\lambda_{F_z=0}$, and we see that the former is always slightly less than the latter, as inferred earlier. The quantity $C_{\hat{I}_{z}\hat{J}_{z}}(t)/C_{F_z=0}(t)$ is examined for $t\geq 1$ in Fig. \ref{fig:zzotoc}b. The ratio initially decays with time for $1\leq t\leq 3$, indicating that the largest subspace dominates. The decay time scale correlates with $t_{\text{EF} (F_{z}=0)}$. The ratio will grow afterward until the saturation, suggesting non-trivial contributions from the other subspaces.

We now study the relaxation dynamics using the quantity $1-C_{\hat{I}_{z}\hat{J}_{z}}(t)/\overline{C_{\hat{I}_{z}\hat{J}_{z}}}$. Recall that the subspace dynamics transition from chaotic to regular as $|F_z|$ is increased from zero. Consequently, we observe (i) longer Ehrenfest times with slower growth rates and (ii) lesser saturation values as predicted by RMT in those subspaces. Due to the observations (i) and (ii), $C_{\hat{I}_{z}\hat{J}_{z}}(t)$ takes a longer time to saturate compared to $C_{F_z=0}(t)$. The numerical results are shown in Fig. \ref{zz_relax}. The relaxation, as the figure suggests, proceeds in two steps. The first phase (see Fig. \ref{zz_relax}a), which is considerably dominated by the largest subspace dynamics, displays an exponential decay. The decay exponent correlates with that of $1-C_{F_z=0}(t)/\overline{C_{F_z=0}}$. For $J=500$, the decay exponent is obtained to be $\approx 0.55$. At the end of this phase, the OTOC contributions from the subspaces close to $F_z=0$ almost reach a saturation point. The relaxation then shows a crossover to an algebraic law $\sim t^{-0.73}$, where the smaller subspace dynamics dominate. Moreover, due to the regular dynamics of those subspaces, we observe large fluctuations in the quantity $1-C_{\hat{I}_{z}\hat{J}_{z}}(t)/\overline{C_{\hat{I}_{z}\hat{J}_{z}}}$ as $t$ approaches infinity. The corresponding numerics are shown in Fig. \ref{zz_relax}b. As $\alpha$ increases further, smaller subspaces eventually become chaotic, leading to sharper decay in the first phase and flatter dynamics in the second phase.
\begin{figure}
\center
\includegraphics[scale=0.5]{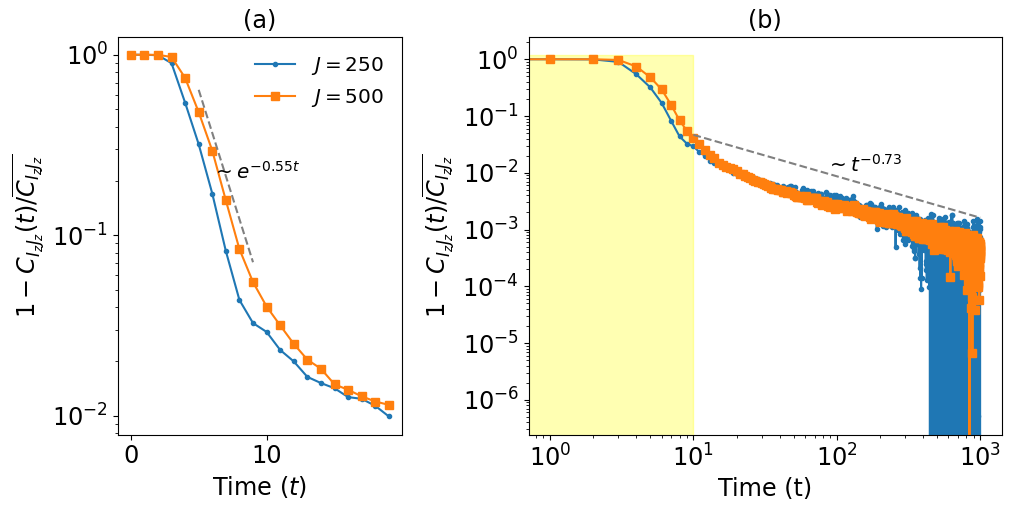}
\caption{\label{zz_relax} Illustration of the two step relaxation of the $I_z-J_z$ OTOC. The parameters are fixed at $\alpha=6$ and $\beta=\pi/2$. The results are shown for the same $J$ values considered in Fig. \ref{fig:zzotoc}. The relaxation appears to proceed in two steps. Panel (a) depicts the short-time exponential decay of $1-C_{\hat{I}_{z}\hat{J}_{z}}(n\tau)/\overline{C_{\hat{I}_{z}\hat{J}_{z}}}$ during the first stage of relaxation. The largest subspace dynamics largely dominate this phase. For $J=500$, we observe a decay that scales as $\sim e^{-0.55 t}$. (b). Long-time algebraic decay of $1-C_{\hat{I}_{z}\hat{J}_{z}}(n\tau)/\overline{C_{\hat{I}_{z}\hat{J}_{z}}}$ in the second-phase.}   
\end{figure}

\subsection{The $\hat{I}_{x}\hat{J}_{x}$ OTOC}\label{the IxJxOTOC}
Here, we study the operator growth for the local operators that do not commute with $\hat{F}_z$. Specifically, we take the initial operators 
\begin{eqnarray}\label{ini}
A=\hat{I}_x\otimes \mathbb{I} \text{ and } B=\mathbb{I}\otimes \hat{J}_x.    
\end{eqnarray}
Unlike $\hat{I}_z$ and $\hat{J}_z$, these operators do not admit the decomposition into disjoint invariant subspaces. Thus, OTOC dynamics are mixed across various subspaces. We aim to contrast the scrambling in this case with the previous one. Before going further, we first derive a useful result. First, we consider $A$ and write it in the computational basis as follows:
\begin{eqnarray}\label{eleA}
A&=&\left\{\dfrac{1}{2}\sum_{m, m'=-J}^{J}\left[P(J, m) \delta_{m', m+1}+ Q(J, m) \delta_{m', m-1} \right]|m\rangle\langle m'|\right\}\otimes \sum_{m''=-J}^{J}|m''\rangle\langle m''|\nonumber\\
&=&\dfrac{1}{2}\sum_{m, m', m''=-J}^{J}\left[P(J, m) \delta_{m', m+1}+ Q(J, m) \delta_{m', m-1} \right] |mm''\rangle\langle m'm''|, 
\end{eqnarray}
where $P(J, m)=\sqrt{(J-m)(J+m+1)}$, and $Q(J, m)=\sqrt{(J+m)(J-m+1)}$. Now, consider an arbitrary bipartite unitary operator $U$ with the symmetry operator $\hat{F}_z$, i.e., $[U, \hat{F}_z]=0$. Then, $U$ can be decomposed as $\oplus_{F_z}U_{F_z}$, where $U_{F_z}$ represents the unitary contribution having supported over the invariant subspace labeled by the charge $F_z$. Moreover, $F_z$ varies from $-2J$ to $2J$. For an arbitrary $F_z$, the corresponding subspace unitary operator can be written in the computational basis as follows:
\begin{eqnarray}
U_{F_z}=\sum_{\substack{m_1, m_2=-J\\m_1+m_2=F_z}}^{J}\sum_{\substack{m_3, m_4=-J\\m_3+m_4=F_z}}^{J}u_{m_1m_2, m_3m_4} |m_1m_2\rangle\langle m_3m_4|. 
\end{eqnarray}
Then, for some other $F'_{z}$, where $-2J\leq F'_{z}\leq 2J$, it follows that 
\begin{align}
U^{\dagger}_{F_z}A U_{F'_z}=&\sum_{\substack{m_1, m_2=-J\\m_1+m_2=F_z}}^{J}\sum_{\substack{m_3, m_4=-J\\m_3+m_4=F_z}}^{J}\sum_{\substack{m'_1, m'_2=-J\\m'_1+m'_2=F'_z}}^{J}\sum_{\substack{m'_3, m'_4=-J\\m'_3+m'_4=F'_z}}^{J}u^*_{m_1m_2, m_3m_4}u_{m'_1m'_2, m'_3m'_4}\nonumber\\
&\hspace{3cm}|m_3m_4\rangle\langle m_1m_2|A|m'_1m'_2\rangle\langle m'_3m'_4|, 
\end{align}
where the elements $\langle m_1m_2|A|m'_1m'_2\rangle$ follow Eq. (\ref{eleA}) and can be written explicitly as follows:
\begin{eqnarray}
\langle m_1m_2|A|m'_1m'_2\rangle=\dfrac{1}{2} \left[P(J, m_1) \delta_{m'_1, m_1+1}+ Q(J, m_1) \delta_{m'_1, m_1-1} \right]\delta_{m_2m'_2}.
\end{eqnarray}
By imposing the condition that $m_1+m_2=F_z$ and $m'_1+m'_2=F'_z$, one can see that $\langle m_1m_2|A|m'_1m'_2\rangle$ returns a non-zero value only when $F_z=F'_z\pm 1$. As a result, the operator $U^{\dagger}_{F_z}A U_{F'_z}$ returns a non-zero matrix only when $F_z= F'_z\pm 1$. A similar reasoning can be extended to the operator $B$, leading to the following:
\begin{eqnarray}\label{inter}
\left. 
  \begin{array}{ c l }
U_{F_z}^{\dagger}AU_{F'_z}\neq\mathbf{0}\\
U_{F_z}^{\dagger}BU_{F'_z}\neq\mathbf{0}
\end{array}
\right\}\quad\text{iff}\quad F'_z=F_z\pm 1, 
\end{eqnarray}
where $\mathbf{0}$ indicates the null matrix or zero matrix, which contains only zeros. The above equation holds for any arbitrary bipartite unitary operator conserving the total magnetization $F_z$ as long as the initial operator $A$ and $B$ are chosen according to Eq. (\ref{ini}). Also note that $U_{F_z}$ remains $(2J+1)^2$-dimensional, but acts non-trivially on the $F_z$-subspace alone. We reemphasize that the $F_z=0$ subspace typically drives the overall OTOC dynamics. However, the observation in Eq. (\ref{inter}) results in vanishing $4$-point correlators at all the subspace levels, i.e., 
\begin{equation}
\tr\left[U^{\dagger t}_{F_z}AU^{t}_{F_z}BU^{\dagger t}_{F_z}AU^{t}_{F_z}B\right]=0\quad \forall F_{z} \text{ and } t\geq 0.
\end{equation}
On the contrary, the $4$-point correlators that connect adjacent subspaces as given by\newline $\tr[U^{\dagger t}_{F_z}AU^{t}_{F_z\pm 1}BU^{\dagger t}_{F_z}AU^{t}_{F_z\pm 1}B]$ remain non-vanishing. The intertwining of adjacent subspace evolutions $U_{F_z}$ and $U_{F_z\pm 1}$ causes a slowdown of the OTOC dynamics whenever the largest subspace is supposedly dominant.

We now numerically calculate $C_{\hat{I}_{x}\hat{J}_{x}}(t)$ for two different values of $J$, namely, $J=150$ and $J=250$. The corresponding results are plotted in Fig. \ref{xxotoc}. Figure \ref{xxotoc}a demonstrates growth of $C_{\hat{I}_{x}\hat{J}_{x}}(t)/\overline{C_{\hat{I}x\hat{J}_x}}$ and the corresponding short-time exponential growth is depicted in Fig. \ref{xxotoc}b. The parameters $\alpha$ and $\beta$ are kept fixed at $6$ and $\pi/2$, respectively, as before. The exponential growth appears within the regime $1\leq t\leq 3$ for both values of $J$. Through the exponential fitting, we find $\lambda_{\hat{I}_{x}\hat{J}_{x} (J=150)}\approx 1.305$ and $\lambda_{\hat{I}_{x}\hat{J}_{x}(J=250)}\approx 1.388$, which are roughly about the same as those of $\lambda_{\hat{I}_{z}\hat{J}_{z}}$, but with a minute difference (see Table. \ref{table1}). However, the OTOC displays slower relaxation in the intermediate regime than in the previous case. For instance, for $J=250$, we observe $(1-C_{\hat{I}_{x}\hat{J}_{x}}(t)/\overline{C_{\hat{I}_{x}\hat{J}_{x}}})\sim \exp\{-0.15 t\}$, which is in contrast with the previous case, where $(1-C_{\hat{I}_{z}\hat{J}_{z}}(t)/\overline{C_{\hat{I}_{z}\hat{J}_{z}}})\sim \exp\{-0.6 t\}$. Hence, the above inference about the slowdown of OTOC dynamics is justified. 
\begin{figure}
\center
\includegraphics[scale=0.5]{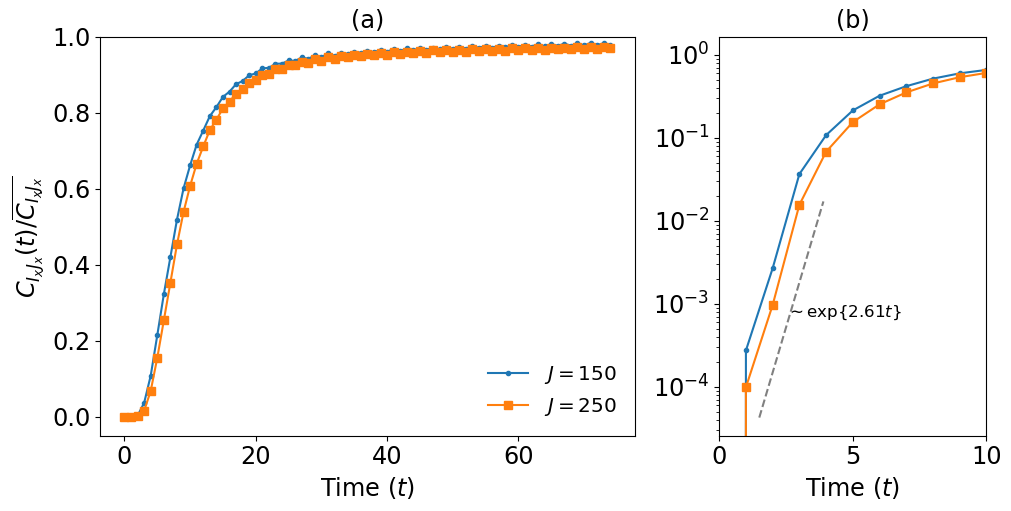}
\caption{\label{xxotoc} (a) The OTOC when the initial operators are $A=\hat{I}_{x}\otimes \mathbb{I}$ and $B=\mathbb{I}\otimes \hat{J}_{x}$ for two different $J$ values. The parameters are kept fixed at $\alpha=6$ and $\beta=\pi/2$. (b) Illustration of the early-time exponential growth on the semi-log scale. The OTOC for $J=250$ grows initially at a rate $\sim e^{2.61 t}$.}   
\end{figure}
 
Beyond the exponential decay, the quantity $1-C_{\hat{I}_x\hat{J}_x}(t)/\overline{C_{\hat{I}_x\hat{J}_x}}$ displays a crossover to an algebraic law. The results are shown in Fig. \ref{xxotoc_relax}. For J=250, through the power law fitting of the data in this regime, we find that the relaxation follows an approximate scaling $\sim t^{-0.85}$, which is slightly faster than the earlier case of $1-C_{\hat{I}_{z}\hat{J}_{z}}(t)/\overline{C_{\hat{I}_{z}\hat{J}_{z}}}\sim t^{-0.75}$. Moreover, due to the regular dynamics associated with the smaller subspaces of the KCT model, the saturation is usually accompanied by fluctuations around the mean value of the OTOC. 

\subsubsection{RMT prediction for the saturation of $\hat{I}_x\hat{J}_{x}$ OTOC}
The OTOCs in fully chaotic systems are expected to saturate after a sufficiently long time. The saturation value follows if we replace the Floquet operator $U^{t}$ with a set of random unitaries drawn from the appropriate ensemble \cite{lakshminarayan2019out}. Recall that the standard COE ensemble is the appropriate random matrix ensemble for the kicked coupled top. This is then followed by performing an average over the COE unitaries.
We denote the resultant OTOC with $C_{\text{RMT}}$ for the given operators $A$ and $B$. Moreover, to incorporate the symmetry operator $\hat{F}_z$, we generate the random COE unitaries of the form $U=\oplus_{F_z}U_{F_z}$, where $U_{F_z}\in \text{COE}(2J+1-|F_z|)$. The observation in Eq. (\ref{inter}) assists in simplifying the COE average of the two-point and four-point correlators as follows:
\begin{eqnarray}\label{23}
\overline{C_{2}}&=&\dfrac{1}{(2J+1)^2}\sum_{F_z}\tr\left[\overline{U^{\dagger}_{F_z}A^2U_{F_z}B^2}\right]\nonumber\\
&=&\dfrac{1}{(2J+1)^2}\sum_{F_z}\dfrac{\tr([A^2]_{F_z}^TB^2)+\tr([A^2]_{F_z})\tr([B^2]_{F_z})}{2J+2-|F_z|}\nonumber\\
\end{eqnarray}
and 
\begin{eqnarray}\label{24}
\overline{C_{4}}&=&\dfrac{1}{(2J+1)^2}\sum_{F_z}\tr\left[\overline{U^{\dagger}_{F_z}AU_{F_z\pm 1}BU^{\dagger}_{F_z}AU_{F_z\pm 1}B}\right]   \nonumber\\ 
&=& 0,
\end{eqnarray}
where $[A^2]_{F_z}=\mathbb{I}_{F_z}A^2\mathbb{I}_{F_z}$ and $I_{F_z}$ is an identity operator acting exclusively on  the $F_z$-subspace, i.e., $\mathbb{I}_{F_z}\equiv\mathbb{I}_{F_z}\oplus\mathbb{O}_{\text{rest}}$, the null matrix $\mathbb{O}_{\text{rest}}$ acts on rest of the subspaces. For a detailed derivative of Eqs. (\ref{23}) and (\ref{24}), refer to Appendix \ref{AppB}. It then follows that $C_{\text{RMT}}=\overline{C_2}$.

\begin{figure}
\center
\includegraphics[scale=0.5]{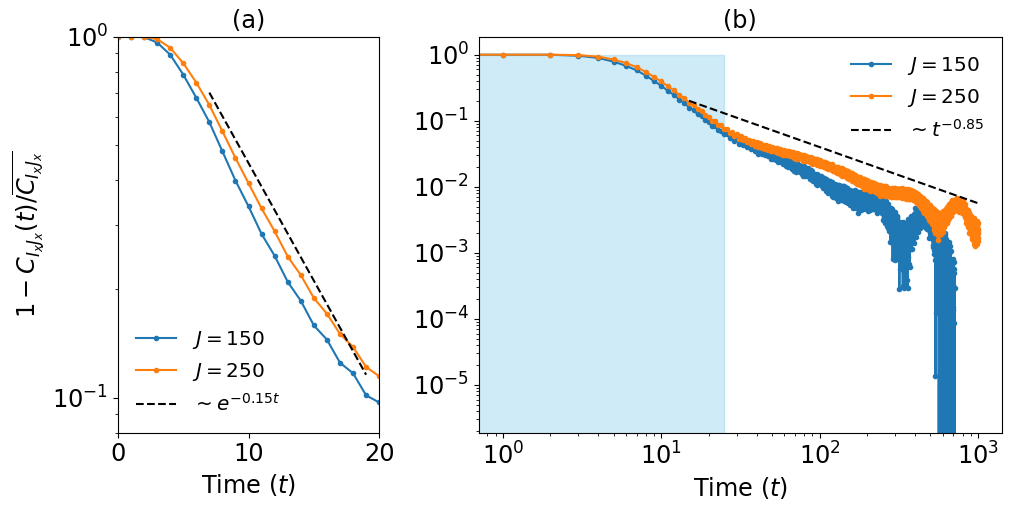}
\caption{\label{xxotoc_relax} Relaxation dynamics of the $\hat{I}_x-\hat{J}_x$ OTOC. We observe that the relaxation, similar to $\hat{I}_z-\hat{J}_z$, proceeds in two steps. In (a), the short-time exponential decay of $1-C_{\hat{I}_{x}\hat{J}_{x}}(n\tau)/\overline{C_{\hat{I}_{x}\hat{J}_{x}}}$ is shown, which is slower than its $I_z-J_z$ counterpart (see Fig. \ref{zz_relax}a). (b). Long-time relaxation following algebraic decay $\sim t^{-1.065}$ and $\sim t^{-0.85}$ respectively for $J=150$ and $J=250$.}   
\end{figure}

\section{Scrambling, operator entanglement and coherence}\label{Scrambling, operator entanglement and coherence}
The preceding section focused on the analysis of OTOCs in the KCT system for various choices of initial operators. In this section, we present a general result concerning the OTOCs and the operator entanglement of arbitrary quantum evolutions. The results in this section are independent of those presented in the previous section. Recent studies have demonstrated that in a bipartite quantum system, when initial operators are local, independent, and uniformly chosen at random from the unitary group, the averaged OTOC equates to the linear operator entanglement entropy of the corresponding quantum evolution \cite{styliaris2021information, pawan}. This result primarily relies on selecting Haar-random local unitaries as initial operators. However, it remains unclear whether such a relationship persists when the initial operators are different from Haar random unitaries. Here, we address this question by considering random Hermitian operators drawn from the unitarily invariant Gaussian ensemble (GUE) as the initial operators. We first establish that the connection between the OTOC and the operator entanglement holds up to a scaling factor for the GUE operators. Additionally, we examine diagonal GUE matrices as initial operators and demonstrate their relationship with the coherence-generating power introduced in Ref. \cite{anand2021quantum}.

\subsection{GUE initial operators}\label{GUE operators}   
We first consider a bipartite quantum system governed by the evolution $U$. We label the subsystems with $A$ and $B$ and assume that both are equal dimensional, i.e., $d_1=d_2=d$. We denote with $C_{PQ}(t)$, the OTOC for two initial operators $P=O_1\otimes\mathbb{I}$ and $Q=\mathbb{I}\otimes O_2$, where $O_1, \thinspace O_2\in\text{GUE}(d)$. We are interested in evaluating the quantity $C_{\text{GUE}}(t)=\overline{C_{PQ}}(t)=\overline{C_{2}}(t)-\overline{C_{4}}(t)$, where the overline indicates the averaging over all $O_1$ and $O_2\in \text{GUE}(d)$. In the following, we first evaluate $\overline{C_{2}}(t)$:
\begin{eqnarray}
\overline{C_{2}}(t)&=&\dfrac{1}{d^2}\int_{O_1\in\text{GUE}}d\mu(O_1)\int_{O_2\in\text{GUE}}d\mu(O_2)\tr\left( U^{\dagger t}O^2_1U^{t}O^{2}_{2} \right)\nonumber\\
&=&\dfrac{1}{d^2}\tr\left[U^{\dagger t}\left(\int_{O_1} d\mu(O_1)O^{2}_1\right)U^{t}\left(\int_{O_2} d\mu(O_2)O^{2}_{2}\right)\right]
\end{eqnarray}
where $d\mu(O_{1(2)})$ represents the normalized measure over GUE($d$). Then, the following intergal follows from the unitary invariance property of the GUE measure. 
\begin{eqnarray}\label{osqre}
\int_{O_1\in \text{GUE}} O^{2}_{1} d\mu(O_1)&=&\int_{O_1\in \text{GUE}} d\mu(O_1)\int_{u\in U(d)} d\mu(u) u^{\dagger}O^{2}_{1}u,\nonumber\\ 
&=&\int_{O_1\in\text{GUE}}d\mu(O_1)\left(\dfrac{\tr(O^{2}_{1})}{d}\right)\mathbb{I}_{d}\nonumber\\
&=&d\left(\mathbb{I}_{d}\right),
\end{eqnarray}
where $d\mu(u)$ is the invariant Haar measure over the unitary group $U(d)$. 
It then follows that
\begin{eqnarray}
\overline{C_2}(t)=\dfrac{1}{d^2}\tr(U^{\dagger t}U^t)=d^2
\end{eqnarray}

We now evaluate the average four-point correlator.  
\begin{eqnarray}\label{4poin}
\overline{C_4}(t)&=&\dfrac{1}{d^2}\int_{O_1}d\mu(O_1)\int_{O_2}d\mu(O_2)\tr\left(O_1(t)O_2O_1(t)O_2\right)\nonumber\\
&=&\dfrac{1}{d^2}\int_{O_1}d\mu(O_1)\int_{O_2}d\mu(O_2)\tr(S U^{\dagger t\otimes 2}O^{\otimes 2}_1U^{t\otimes 2}O_2^{\otimes 2}),\nonumber\\ 
\end{eqnarray}
where, in the second equality, a replica trick given in Ref. \cite{styliaris2021information} is used --- for any $d$-dimensional $p$ and $q$ matrices, the replica trick implies $\tr(pq)=\tr(S(p\otimes q))$. Also note that $S$ in Eq. (\ref{4poin}) swaps quantum states on the Hilbert space $\mathcal{H}_{AB}\otimes \mathcal{H}_{AB}$. This allows us to write $S$ as $S=S_{AA}S_{BB}$. To evaluate $\overline{C_{4}}(t)$, we again invoke the unitary invariance property of the probability measure of GUE.
\begin{equation}
\int_{O_1}d\mu(O_1) O^{\otimes 2}_{1} =\int_{O_1}d\mu(O_1)\int_{u\in U(d)} d\mu(u) \left(u^{\dagger \otimes2}O^{\otimes 2}_1u^{\otimes 2}\right),
\end{equation}
where the second integral on the right-hand side can be solved using the Schur-Weyl duality \cite{zhang2014matrix}. 
\begin{eqnarray*}
\int\limits_{u} d\mu(u) \left(u^{\dagger \otimes2}O^{\otimes 2}_{1}u^{\otimes 2}\right)= \left( \frac{\tr(O_1)^2}{d^2-1}-\frac{\tr(O^{2}_1)}{d(d^2-1)} \right)\mathbb{I}_{d^2} - \left(\frac{\tr(O_1)^2}{d(d^2-1)}-\frac{\tr(O^{2}_1)}{d^2-1}\right)S_{AA} .
\end{eqnarray*}
We are now left with the integrals of $\tr(O_1)^2$ and $\tr(O^{2}_{1})$ over GUE, which are related to the mean and variance of the elements of the GUE matrices. Upon performing these integrals, we obtain 
\begin{eqnarray}\label{31}
\int_{O_1\in \text{GUE}}d\mu(O_1) O^{\otimes 2}_{1}  = S_{AA}. 
\end{eqnarray}
Similarly, 
\begin{eqnarray}\label{32}
\int_{O_2\in \text{GUE}}d\mu(O_2) O^{\otimes 2}_{2}  = S_{BB}. 
\end{eqnarray}
Resultantly, the four-point correlator averaged over GUE local operators can be written as 
\begin{eqnarray}
\overline{C_4}(t)&=&\dfrac{1}{d^2}\tr[S (U^{\dagger t})^{\otimes 2}S_{AA}(U^{t})^{\otimes 2}S_{BB}]    \nonumber\\
&=& \dfrac{1}{d^2}\tr\left( S_{AA}(U^{\dagger t})^{\otimes 2} S_{AA}(U^{\dagger t})^{\otimes 2}\right)
\end{eqnarray}
It then follows that  
\begin{eqnarray}
C_{\text{GUE}}(t)=d^2-\frac{1}{d^2}\tr[S_{AA} (U^{\dagger t})^{\otimes 2}S_{AA}(U^{t})^{\otimes 2}], 
\end{eqnarray}
where the second term is related to the linear entanglement entropy of quantum evolutions \cite{styliaris2021information}. In deriving the above relation, we utilized the unitary invariance of the measure over GUE. Additionally, the mean ($\mu=0$) and variance ($\sigma^2=1$) of the elements of the GUE operators were crucial for deriving this relation. We note that the unitary invariance is sufficient to establish a relation between the operator entanglement and the OTOC. On the other hand, zero mean and unit variance ensure that Eqs. (\ref{31}) and (\ref{32}) hold. Therefore, as long as the operators are chosen from ensembles with the listed properties, such as unitary invariance, zero mean, and unit variance, the above relation is expected to remain valid.

\subsection{Random diagonal initial operators}\label{Random diagonal initial operators}
Here, we take $O_1$ and $O_2$ to be diagonal operators with the elements chosen randomly from the Gaussian distribution with the mean $\mu=0$ and the standard deviation $\sigma=1$, i.e., 
\begin{eqnarray}
O_1=\sum_{i=0}^{d-1}x_i|i\rangle\langle i|,\quad \text{and}\quad O_2=\sum_{i=0}^{d-1}y_i|i\rangle\langle i|,  
\end{eqnarray}
where $x_i$ and $y_i$ are independent Gaussian random variables. It then follows that
\begin{eqnarray}
\int_{O_{1(2)}}O^2_{1(2)}d\mu(O_{1(2)})&=&\mathbb{I}_{1(2)}, \nonumber\\
\text{and}\quad\int_{O_{1(2)}}(O_{1(2)}\otimes O_{1(2)})d\mu(O_{1(2)})&=&\sum_{i, j=0}^{d-1}\delta_{ij}|ij\rangle\langle ij|. 
\end{eqnarray}
The two-point and the four-point correlators in this case take the following forms:
\begin{eqnarray}
\overline{C_2}(t)&=& d^2\nonumber\\
\overline{C_4}(t)&=&\sum_{ijkl}\langle ji|U^{\dagger t\otimes 2}|kk\rangle\langle kk|U^{t\otimes 2}|ll\rangle\delta_{li}\delta_{lj}\nonumber\\
&=&\sum_{kl}\langle kk|U^{t\otimes 2}(t)|ll\rangle\langle ll|U^{\dagger t\otimes 2}|kk\rangle\nonumber\\
&=&\sum_{kl}|\langle k|U^t|l\rangle|^4
\end{eqnarray}
Hence the commutator function is given by
\begin{eqnarray}
C_{\text{dGUE}}(t)=1-\dfrac{1}{d^2}\sum_{kl}|\langle k|U^t|l\rangle|^4.
\end{eqnarray}
This is closely related to the coherence generating power (CGP) of the time evolution operator \cite{anand2021quantum}.

\section{Summary and Discussion}\label{SUMMARY AND DISCUSSIONS}
The central question concerning quantum chaos is how classically chaotic dynamics inform us about specific properties of quantum systems, e.g., the energy spectrum, nature of eigenstates, correlation functions, and, more recently, entanglement. Alternatively, what features of quantum systems arise since their classical description is chaotic? It is now well understood that if the quantum system has an obvious classical limit, then the dynamics of quantities like entanglement entropy and the OTOC are characterized by the corresponding classical Lyapunov exponents (see for example \cite{lerose2020bridging} and references therein). Recently, the study of OTOC and scrambling of quantum information as quantified by operator growth has witnessed a surge in interest. This is supplemented by the role of non-integrability, symmetries, and chaos in the above as well as in the study of thermalization and statistical aspects of many-body quantum systems. Our work is an attempt to explore operator scrambling and OTOCs and the role of symmetries and conserved quantities in the process.
For this purpose, we have employed a model system of KCTs that is simple enough yet exhibits rich dynamics.
Many-body quantum systems have an intimate connection with quantum chaos. There has been a significant push to understand the issues involving thermalization, irreversibility, equilibration, coherent backscattering, and the effects of quantum interference in such systems. The Eigenstate thermalization hypothesis (ETH) offers a rich platform for understanding several of these phenomena. However, one needs to be extra careful in attributing chaos and non-integrability to ETH, as recent studies have shown that even a classical Lipkin-Meshkov-Glick (LMG) model displays ETH properties \cite{kelly2020thermalization, lambert2021quantum}. In our work, we have considered the simplest many-body system, the kicked coupled top, which is rich enough to allow us to study these aspects. We have taken up this line of study and explored how information propagates with operator scrambling and how the incompatibility of two operators can quantify this with time.

We have numerically studied the OTOCs in the KCT model, which exhibits chaos in the classical limit and conserves the angular momentum along the $z$-axis. The conservation law enables us to partition the system into distinct invariant subspaces with constant magnetization. We first considered the largest invariant subspace and examined the scrambling in the fully chaotic regime. We observed a good quantum-classical correspondence between the OTOC growth rate and the maximum classical Lyapunov exponent. We then studied the mixed phase space OTOC dynamics with the help of Percival's conjecture. We have observed that the operators are more likely to get scrambled if the initial state corresponds to the chaotic sea. This is reflected both in the short-time and the long-time behavior. Additionally, we have studied the infinite time averages of the OTOCs in the Floquet states against their mean location in the phase space. We observed that states corresponding to the same regular islands form clusters. Interestingly, these averages appear to vary smoothly with respect to the mean location of these states. Conversely, the OTOC time averages in chaotic Floquet states take nearly uniform values. Also, the Floquet states that are localized near the boundaries exhibit infinite time averages close to chaotic states despite being regular. It is to be noted that similar results exist for the entanglement entropy \cite{lombardi2011entanglement, madhok2015comment}. Our findings complement previous results indicating that hyperbolic fixed points can induce scrambling without chaos \cite{hashimoto2020exponential, pilatowsky2020positive, pappalardi2018scrambling, hummel2019reversible, xu2020does, steinhuber2023dynamical}. In particular, our results demonstrate that regular states located near the boundary of the stable islands with the chaotic region can also induce scrambling. We have also studied the OTOCs by taking coherent states as initial states. It is interesting that while the short-time growth of OTOCs captures the Lyapunov divergence, the long-term behavior of OTOCs remarkably reproduces classical phase space structures. 

Subsequently, we extended our analysis to the entire KCT system, investigating OTOCs for two distinct choices of initial operators. We found that the scrambling behavior differs depending on the choice of initial operators. Furthermore, conservation laws have been observed to impede operator dynamics, regardless of the choice of the initial operators. We also proved an independent result concerning the OTOCs and operator entanglement entropy of arbitrary time evolution operators by taking the Gaussian unitary operators as the initial operators \cite{styliaris2021information}. Moreover, the OTOC is related to the coherence generating power for the diagonal Gaussian operators \cite{anand2021quantum}. It is worth noting that in this work, we have considered equal spin sizes for the OTOC analysis. Hence, it is interesting to see if considering unequal magnitudes results in different dynamics \cite{haake1991quantum}. Nevertheless, we do not expect significant differences as the time-reversal symmetry is preserved in the system. There has been debate whether or not ``scrambling" is necessary or sufficient for chaos \cite{dowling2023scrambling, xu2020does} as OTOCs can show rapid exponential growth from unstable fixed points. We feel one needs both short-term and long-term average growth to capture the complete dynamics of the system. Not only does this give a more complete picture when using OTOCs as probes, but it also settles the question that OTOCs can be used as probes for chaos when studied in both limits.

One of the central focuses in quantum information science is to simulate these systems on a quantum device and on the related issues involving quantifying the complexity of these simulations, benchmarking these simulations in the presence of errors, and exploring connections to quantum chaos in these systems. How is information scrambling related to the complexity of simulating many-body quantum systems \cite{sieberer2019digital}? We hope our study paves for the exploration of these intriguing directions.


\chapter{Deep thermalization and Emergent state designs in systems with symmetry}
\label{deepthchap}

\section{Introduction}
Quantum state designs \cite{renes2004symmetric, klappenecker2005mutually, benchmarking2}, by enabling an efficient sampling of random quantum states, play a quintessential role in devising and benchmarking various quantum protocols \cite{benchmarking1, knill2008randomized, benchmarking2} with broad applications ranging from circuit designs \cite{harrow2009random, brown2010convergence} to black hole physics \cite{sekino2008fast}. Symmetries, on the other hand, are expected to reduce the randomness of a state. Symmetries play decisive roles in constraining the static and dynamical characteristics of quantum many-body systems, ranging from stabilizing novel quantum phases and phase transitions, many-body localization, and quantum chaos to thermalization in closed and open systems \cite{ope4, ope5, friedman2019spectral, yunger2016microcanonical, nakata2023black, bhattacharya2017syk, balachandran2021eigenstate, kudler2022information, chen2020many, paviglianiti2023absence, agarwal2023charge, varikuti2022out}. Despite being ubiquitous, the effects of symmetry on the quantum state designs remain an outstanding question. The recently introduced projected ensemble framework generates efficient approximate state t-designs by using many-body quantum chaos as a resource \cite{choi2023preparing, cotler2023emergent, ho2022exact, ippoliti2022solvable, lucas2023generalized, shrotriya2023nonlocality, versini2023efficient, claeys2022emergent, ippoliti2023dynamical, bhore2023deep, mcginley2023shadow, choi2023preparing, mark2024maximum, liu2024deep, chan2024projected, vairogs2024extracting, varikuti2024unraveling}. 
In this chapter, we examine the emergence of state designs from the random generator states exhibiting symmetries \cite{varikuti2024unraveling}. The latter restricts the arbitrary choice of measurement basis for the projective measurements, which are at the heart of the framework. Initially relying on the translation symmetric case, we analytically establish a sufficient condition for the basis leading to the emergence of state designs. We then demonstrate an intriguing interplay between the measurements and the symmetry. To probe, we employ the trace distance measure between the moments of projected and Haar ensembles, which we analyze numerically to illustrate the convergence to the state designs. Subsequently, we investigate the violation of the sufficient condition in order to identify bases that fail to converge. Finally, we consider the evolution of a quantum state under the Hamiltonian of a chaotic tilted field Ising chain with periodic boundary conditions to display the emergence of state designs for physical systems with translation symmetry. By contrasting the results with the open boundary case, we show that the trace distance displays faster convergence in the initial time; however, it saturates to finite values, deviating from random matrix predictions at late time. To delineate the general applicability of our results, we extend our analysis to other symmetries, particularly by choosing $Z_2$ and reflection symmetries as representative cases. We expect our findings to pave the way for further exploration of deep thermalization and equilibration of closed and open quantum many-body systems.

\begin{figure}
\begin{center}
\includegraphics[scale=0.6]{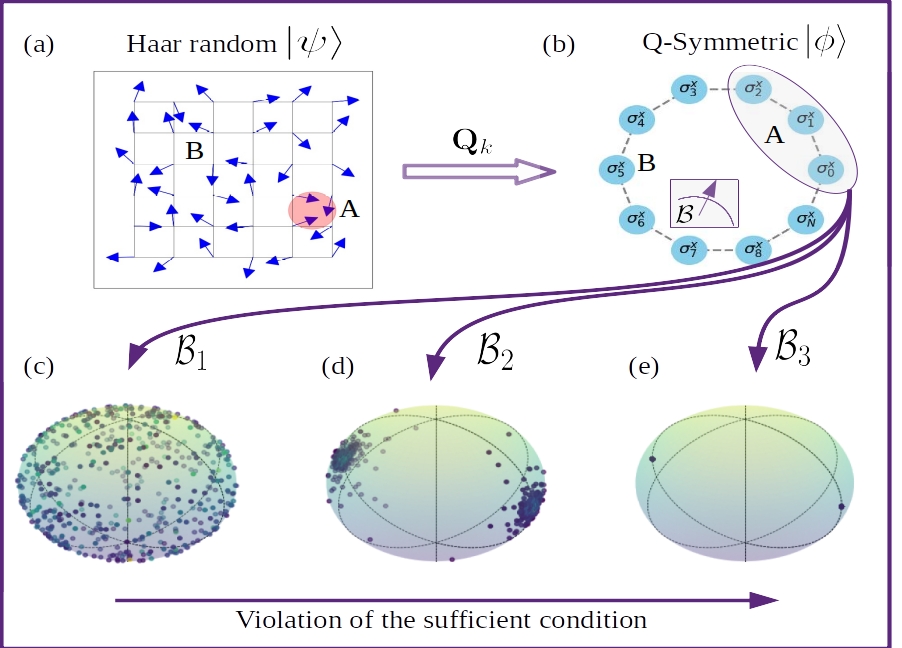}
\end{center}
\caption{\label{fig:sch} Schematic representation of the projected ensemble framework for a $Q$-symmetric state (an eigenstate of a symmetry operator $Q$), showcasing the interplay between measurement bases and symmetry. (a) Haar random state, wherein the spins are randomly aligned and entangled. An application of the subspace projector $Q_k$ takes the Haar random state to an invariant subspace with charge $k$, which is depicted in (b). For simplicity, we are showing a $Z_2$ symmetric state (see Sec. \ref{z2s}) in this schematic. The labels $\sigma^x_j$ indicate that the state is invariant under spin-flip operations. The resultant state is subjected to projective measurements over the subsystem-$B$. These measurements result in the so-called projected ensembles supported over $A$ 
(see Sec. \ref{Frame}).
(c)-(d) 
distributions of the resultant projected ensembles for three different measurement bases. While in (c), the measurements in $\mathcal{B}_1$ yield a uniformly distributed projected ensemble, $\mathcal{B}_2$ and $\mathcal{B}_3$ result in ensembles of pure states localized near $|\pm\rangle$ states. We understand that the latter cases largely violate a sufficient condition we derive in this text for the emergence of state designs. }   
\end{figure}

This chapter is organized as follows. In Sec. \ref{Frame}, we briefly outline the central question of this work and summarize our key results. In Sec. \ref{T-invariant}, we consider the ensembles of translation symmetric states and provide an analytical expression for their moments. We carry out the analysis of emergent state designs with symmetries in sec. \ref{emergent_designs} and provide a sufficient condition on the measurement bases. We then consider a chaotic Ising chain with periodic boundary conditions in Sec. \ref{sising} and examine deep thermalization in a state evolved under this Hamiltonian. In Sec. \ref{genn}, we generalize the results to other discrete symmetries such as $Z_2$ and reflection symmetries. Finally, we conclude this chapter in Sec. \ref{discussion}.

\section{Summary of main results}\label{Frame}
Recall that the projected ensemble framework aims to generate quantum state designs from a single chaotic or random many-body quantum state. First, consider a generator quantum state $|\psi\rangle\in\mathcal{H}^{\otimes N}$, where $\mathcal{H}$ denotes the local Hilbert space of dimension $d$ and $N=N_A+N_B$ denotes the size of the system constituting subsystems-$A$ and $B$. Then, projectively measuring the subsystem-$B$ in a basis $\mathcal{B}$ gives a statistical mixture of pure states (or projected ensemble) corresponding to the subsystem-$A$:
\begin{eqnarray}
|\phi(b)\rangle=\dfrac{|\Tilde{\psi}(b)\rangle}{\sqrt{p_b}}=\dfrac{\left(\mathbb{I}_{2^{N_A}}\otimes\langle b|\right)|\psi\rangle}{\sqrt{\langle\psi |b\rangle\langle b|\psi\rangle}},     
\end{eqnarray}
where the probabilities are given by $p_b=\langle\Tilde{\psi}(b)|\Tilde{\psi}(b)\rangle=\langle\psi |b\rangle\langle b|\psi\rangle$. Then, the projected ensemble on $A$ given by $\mathcal{E}(|\psi\rangle, \mathcal{B})\equiv\left\{p_b, |\phi(b)\rangle\right\}$ approximate higher-order quantum state designs if $|\psi\rangle$ is sufficiently chaotic or random \cite{ho2022exact, cotler2023emergent}. We use the following trace distance metric as a figure of metric to quantify how close the projected ensembles are to the Haar distribution: 
\begin{eqnarray}\label{Deltat}
\Delta^{(t)}_{\mathcal{E}}\equiv\left\|\sum_{|b\rangle\in\mathcal{B}}\dfrac{\left(\langle b|\psi\rangle\langle\psi |b\rangle\right)^{\otimes t}}{\left(\langle\psi |b\rangle\langle b|\psi\rangle\right)^{t-1}}-\int_{\phi\in\mathcal{E}^{A}_{\text{Haar}}}d\phi \left[|\phi\rangle\langle\phi|\right]^{\otimes t}\right\|_1\leq\varepsilon.
\end{eqnarray}
If the generator state $|\psi\rangle$ is Haar random, the trace distance $\Delta^{(t)}_{\mathcal{E}}$ exponentially converges to zero with $N_B$ for any $t\geq 1$, as demonstrated in Ref. \cite{cotler2023emergent}. In this case, the measurement basis can be arbitrary, and the behavior is generic to the choice of basis. On the other hand, for the generator state abiding by symmetry, the choice of measurement basis becomes crucial. In this chapter, we address this particular aspect.

Our general result can be summarized as follows: Given a symmetry operator $Q$ and a measurement basis $\mathcal{B}$, the projected ensembles of generic $Q$-symmetric quantum states approximate state $t$-designs if for all $|b\rangle\in\mathcal{B}$, $\langle b|\mathbf{Q}_{k}|b\rangle=\mathbb{I}_{2^{N_A}}$, where $\mathbf{Q}_{k}$ represents the projector onto an invariant subspace with charge $k$.
Here, it is to be noted that $\mathbf{Q}_{k}$ can be constructed by taking an appropriate linear combination of the elements of the corresponding symmetry group. We employ $\Delta^{t}_{\mathcal{E}}$ as a figure of merit for the state designs. We explicitly derive this condition for the ensembles of translation invariant states and provide arguments to show its generality to other symmetries. We further elucidate the emergence of state designs from translation symmetric states by explicitly considering a physical model.

\section{T-invariant quantum states}\label{T-invariant}
In this section, we construct the ensembles of translation symmetric (T-invariant) quantum states from the Haar random states and calculate the moments associated with those ensembles. The T-invariant states are well studied in condensed matter physics and quantum field theory for the ground state properties, such as entanglement and phase transitions. Due to the non-onsite nature of the symmetry, the generic T-invariant states (excluding the product states of the form $|\psi\rangle^{\otimes N}$) are necessarily long-range entangled\footnote{An operational definition of the long-range entanglement is as follows: A state $|\phi\rangle\in \mathcal{H}^{\otimes N}$ is long-range entangled if it can not be realized by the application of a finite-depth local circuit on a trivial product state $|0\rangle^{\otimes N}$.} \cite{gioia2022nonzero}. Thus, in a generic T-invariant state, the information is uniformly spread across all the sites, somewhat mimicking the Haar random states. This makes them ideal candidates for generating state designs besides Haar random states. Moreover, in the context of ETH, generic T-invariant systems have been shown to thermalize local observables \cite{santos2010localization, mori2018thermalization, sugimoto2023eigenstate}. Hence, studying deep thermalization in these systems is of profound interest.

Let $T=e^{iP}$ denote the lattice translation operator on a system with a total of $N$ sites, each having a local Hilbert space dimension of $d$, where $P$ is the lattice momentum operator. Then, $T$ can be defined by its action on the computational basis vectors as follows:
\begin{equation}
T|i_1, i_2, ..., i_N\rangle =|i_N, i_1, ..., i_{N-1}\rangle, 
\end{equation}
with $N$-th roots of unity as eigenvalues. A system is considered T-invariant if its Hamiltonian $H$ commutes with $T$, i.e., $[H, T]=0$. A T-invariant state $|\psi\rangle$ is an eigenstate of $T$ with an eigenvalue $e^{-2\pi i k/N}$, i.e., $T|\psi\rangle =e^{-2\pi ik/N}|\psi\rangle$, where $k\in\mathbb{Z}_{\geq 0}$ and $0\leq k\leq N-1$, characterizes the lattice momentum charge. We then represent the set of all pure T-invariant states having the momentum charge $k$ with $\mathcal{E}^{k}_{\text{TI}}$. 
To compute their moments, we first outline their construction from the Haar random states, followed by the Haar average. 

The translation operator, $T$, generates the translation group given by $\{T^j\}_{j=0}^{N-1}$. Then, a Haar random state $|\psi\rangle\in\mathcal{H}^{\otimes N}$ can be projected onto a translation symmetric subspace by taking a uniform superposition of the states $\{e^{2\pi ijk/N}T^j|\psi\rangle\}_{j=0}^{N-1}$:
\begin{eqnarray}\label{Tsym}
|\psi\rangle\rightarrow |\phi\rangle=\mathbb{T}_{k}(|\psi\rangle)=\dfrac{1}{\mathcal{N}}\mathbf{T}_{k}|\psi\rangle, 
\end{eqnarray}
where
\begin{eqnarray}\label{Tstate}
\mathbf{T}_{k}=\sum_{j=0}^{N-1}e^{2\pi ijk/N}T^j\quad\text{and }
\mathcal{N}=\sqrt{\langle\psi|\mathbf{T}^{\dagger}_{k}\mathbf{T}_{k}|\psi\rangle}.
\end{eqnarray}
Here, the mapping from $|\psi\rangle$ to the T-invariant state $|\phi\rangle$ is denoted with $\mathbb{T}_{k}(|\psi\rangle)$. 
The Hermitian operator $\mathbf{T}_{k}$ projects any generic state onto an invariant momentum sector with the charge $k$. Therefore, $\mathbf{T}_{k}\mathbf{T}_{k'}=N\mathbf{T}_{k}\delta_{k, k'}$, yielding the normalizing factor $\mathcal{N}=\sqrt{N\langle\psi|\mathbf{T}_{k}|\psi\rangle}$. The resultant state $|\phi\rangle$ is an eigenstate of $T$ with the eigenvalue $e^{-2\pi ik/ N}$, i.e., $T|\phi\rangle =e^{-2\pi ik/N}|\phi\rangle$. In this way, we can project the set of Haar random states to an $N$-number of disjoint sets of random T-invariant states, each characterized by the momentum charge $k$. Introducing the translation invariance causes partial de-randomization of the Haar random states. This is because a generic quantum state with support over $N$ sites can be described using $n_{\text{Haar}}\approx d^{N}$ independent complex parameters \cite{linden1998multi}. The translation invariance, however, reduces this number by a factor of $N$, i.e., $n_{\text{TI}}\approx d^N/N$. As we shall see, this results in more structure of the moments of the T-invariant states,  $\mathbb{E}_{\phi\in\mathcal{E}^{k}_{\text{TI}}}\left[|\phi\rangle\langle\phi|^{\otimes t}\right]$.

Before evaluating the moments of $\mathcal{E}^{k}_{\text{TI}}$, it is useful to state the following result: 
\begin{theorem}
Let $U_{\text{TI}}(d^N)$ be a subset of the unitary group $U(d^N)$, which contains all the unitaries that are T-invariant, i.e., $[v, T]=0$ for all $v\in U_{\text{TI}}(d^N)$, then $U_{\text{TI}}(d^N)$ is a compact subgroup of $U(d^N)$.  
\end{theorem}
\begin{proof}
Consider the subset $U_{\text{TI}}(d^N)$ of $U(d^N)$ containing all the T-invariant unitaries --- for every $v\in U_{\text{TI}}(d^N)$, we have $T^{\dagger}vT=v$. Clearly, $U_{\text{TI}}(d^N)$ is a subgroup of $U(d^N)$. As the operator norm of any unitary matrix is bounded, $U_{\text{TI}}(d^N)$ is bounded. Moreover, we can define $U_{\text{TI}}(d^N)$ as the preimage of the null matrix $\mathbf{0}$ under the operation $v-T^{\dagger}vT$ for $v\in U(d^N)$, hence it is necessarily closed. This implies that $U_{\text{TI}}(d^N)$ is a compact subgroup of $U(d^N)$. 
\end{proof}

A method for constructing random translation invariant unitaries using the polar decomposition is outlined in Appendix \ref{polar}. 
An immediate consequence of the above result is that there exists a natural Haar measure on the subgroup $U_{\text{TI}}(d^N)$. It is to be noted that projecting Haar random states onto T-invariant subspaces creates uniformly random states within those subspaces. This means that the distribution of states in $\mathcal{E}^{k}_{\text{TI}}$ is invariant under the action of $U_{\text{TI}}(2^N)$.
To see this, sample $|\phi\rangle$ and $|\phi'\rangle$ from $\mathcal{E}^{k}_{\text{TI}}$ such that they are related to each other via the left invariance of the Haar measure over $U(d^N)$, i.e., $|\phi\rangle=\mathbf{T}_{k}u|0\rangle/\mathcal{N}$ and $|\phi'\rangle=\mathbf{T}_{k}vu|0\rangle/\mathcal{N}$, where $u\in U(d^N)$ and $v\in U_{\text{TI}}(d^N)$. Since $v$ and $\mathbf{T}_{k}$ commute, we can write $|\phi'\rangle=v|\phi\rangle$. 
Let $v$ be sampled according to the Haar measure in $U_{\text{TI}}(d^N)$, then, the state $v|\phi\rangle$ must be uniformly random in $\mathcal{E}^{k}_{\text{TI}}$. Now, $|\phi\rangle$ and $|\phi'\rangle$ being sampled through the projection 
of $\mathbf{T}_{k}$, we can conclude that all the states similarly projected to $\mathcal{E}^{k}_{\text{TI}}$ will be uniformly random. 
We use this result to derive the moments of $\mathcal{E}^{k}_{\text{TI}}$. 

\begin{theorem}\label{Tran_designs}
Let $|\psi\rangle$ be a pure quantum state drawn uniformly at random from the Haar ensemble, then for the mapping $|\psi\rangle\rightarrow |\phi\rangle=\mathbb{T}_{k}(|\psi\rangle)$, it holds that 
\begin{eqnarray}
\mathbb{E}_{|\phi\rangle\in\mathcal{E}^{k}_{\text{TI}}}\left[\left[|\phi\rangle\langle\phi|\right]^{\otimes t}\right]=\dfrac{1}{\alpha_{t}}\mathbf{T}^{\otimes t}_{k}\mathbf{\Pi}_{t}, 
\end{eqnarray}
where $\alpha_t$ denotes the normalizing constant, which is given by
\begin{eqnarray}\label{Tmoments}
\alpha_t=\text{Tr}\left(\mathbf{T}^{\otimes t}_{k}\mathbf{\Pi}_{t}\right).    
\end{eqnarray}    
\end{theorem}

\begin{proof}
We first note that given a Haar random pure state $|\psi\rangle$, under the map $\mathbb{T}_{k}$, becomes an eigenstate of $T$ with the eigenvalue $e^{-2\pi ik/N}$, i.e., $|\phi\rangle =\mathbf{T}_{k}|\psi\rangle/\sqrt{\langle\psi|\mathbf{T}^{\dagger}_{k}\mathbf{T}_{k}|\psi\rangle}$. This generates an ensemble $\{|\phi\rangle \}$, denoted with $\mathcal{E}^{k}_{\text{TI}}$, when $|\psi\rangle\in\mathcal{E}_{\text{Haar}}$. We are interested in finding the moments associated with $\mathcal{E}^{k}_{\text{TI}}$. For any $k$, the $t$-th moment can be evaluated as follows: 
\begin{eqnarray}
\mathbb{E}_{\phi\in\mathcal{E}^{k}}\left[\left[|\phi\rangle\langle\phi|\right]^{\otimes t}\right]=\int_{\psi\in\mathcal{E}_{\text{Haar}}}d\psi\dfrac{\mathbf{T}^{\otimes t}_{k}\left[|\psi\rangle\langle\psi|\right]^{\otimes t}\mathbf{T}^{\dagger\otimes t}_{k}}{\langle|\psi|\mathbf{T}^{\dagger}_{k}\mathbf{T}_{k}|\psi\rangle^{t}},
\end{eqnarray}
where the integral on the right-hand side is performed over the Haar random pure states. Since the Haar random states can be generated through the action of Haar random unitaries on a fixed fiducial quantum state, we can write
\begin{eqnarray}
\mathbb{E}_{\phi\in\mathcal{E}^{k}}\left[\left[|\phi\rangle\langle\phi|\right]^{\otimes t}\right]=\int_{u\in U(d^N)}d\mu(u)\dfrac{\mathbf{T}^{\otimes t}_{k}\left[u|0\rangle\langle 0|u^{\dagger}\right]^{\otimes t}\mathbf{T}^{\dagger\otimes t}_{k}}{\langle 0|u^{\dagger}\mathbf{T}^{\dagger}_{k}\mathbf{T}_{k}u|0\rangle^{t}}.
\end{eqnarray}
Since the integrand in the above equation has $u$-dependence in both the numerator and the denominator, direct evaluation of the Haar integral is challenging. To circumvent it, let us now consider the following ensemble average:
\begin{eqnarray}\label{A3}
\mathbb{E}\left[\left[|\phi\rangle\langle\phi|\right]^{\otimes t} \langle\psi|\mathbf{T}^{\dagger}_{k}\mathbf{T}_{k}|\psi\rangle^{t}\right]=\int_{u\in U(d^N)}d\mu(u)\dfrac{\mathbf{T}^{\otimes t}_{k}\left[u|0\rangle\langle 0|u^{\dagger}\right]^{\otimes t}\mathbf{T}^{\dagger\otimes t}_{k}}{\langle 0|u^{\dagger}\mathbf{T}^{\dagger}_{k}\mathbf{T}_{k}u|0\rangle^{t}}\langle 0|u^{\dagger}\mathbf{T}^{\dagger}_{k}\mathbf{T}_{k}u|0\rangle^{t}.\nonumber\\
\end{eqnarray}
By taking advantage of the left and the right invariance of the Haar measure over the unitary group $U(d^N)$, we replace $u$ in the above equation with $vu$, where $v\in U_{\text{TI}}(d^N)\subset U(d^N)$. Under this action, the term $\langle 0|u^{\dagger}\mathbf{T}^{\dagger}_{k}\mathbf{T}_{k}u|0\rangle^{t}$ remains independent of $v$ as $[v, \mathbf{T}_{k}]=0$ for all $v\in U_{\text{TI}}(d^N)$ and $k$.
We then perform the Haar integration over $U_{\text{TI}}(d^N)$, which corresponds to the following:

\begin{align}\label{uni_inv}
\mathbb{E}\left[\left[|\phi\rangle\langle\phi|\right]^{\otimes t} \langle\psi|\mathbf{T}^{\dagger}_{k}\mathbf{T}_{k}|\psi\rangle^{t}\right]
&=\int_{u\in U(d^N)}d\mu(u)\langle 0|u^{\dagger}\mathbf{T}^{\dagger}_{k}\mathbf{T}_{k}u|0\rangle^{t}\underbrace{\int_{v\in U_{\text{TI}}(d^N)}d\mu_{\text{TI}}(v)\dfrac{\mathbf{T}^{\otimes t}_{k}\left[vu|0\rangle\langle 0|u^{\dagger}v^{\dagger}\right]^{\otimes t}\mathbf{T}^{\dagger\otimes t}_{k}}{\langle 0|u^{\dagger}\mathbf{T}^{\dagger}_{k}\mathbf{T}_{k}u|0\rangle^{t}}}_{=\mathbb{E}_{\phi\in\mathcal{E}^{k}_{\text{TI}}}\left[\left[|\phi\rangle\langle\phi|\right]^{\otimes t}\right]}\nonumber\\
&=\mathbb{E}_{\phi\in\mathcal{E}^{k}_{\text{TI}}}\left[\left[|\phi\rangle\langle\phi|\right]^{\otimes t}\right] \int_{u\in U(d^N)}d\mu(u)\langle 0| u^{\dagger}\mathbf{T}^{\dagger}_{k}\mathbf{T}_{k}u|0\rangle^{t}\nonumber\\
&=\mathbb{E}_{\phi\in\mathcal{E}^{k}_{\text{TI}}}\left[\left[|\phi\rangle\langle\phi|\right]^{\otimes t}\right] \int_{|\psi\rangle\in \mathcal{E}_{\text{Haar}}}d\psi
\langle\psi|\mathbf{T}^{\dagger}_{k}\mathbf{T}_{k}|\psi\rangle^{t}.
\end{align}
where $d\mu_{\text{TI}}$ denotes the Haar measure over the subgroup $U_{\text{TI}}(d^N)$. Since $v$ is uniformly random in $U_{\text{TI}}(d^N)$, $v|\phi\rangle$ is also uniformly random in $\mathcal{E}^{k}_{\text{TI}}$
for any $|\phi\rangle\in\mathcal{E}^{k}_{\text{TI}}$. Therefore, $\mathbb{E}_{|\phi\rangle\in\mathcal{E}^{k}_{\text{TI}}}\left[\left(|\phi\rangle\langle\phi|\right)^{\otimes t}\right]=E_{v\in U_{\text{TI}}(d^N)}\left[\left(v|\phi\rangle\langle\phi|v^{\dagger}\right)^{\otimes t}\right]$. This is substituted in the second equality above. Equation (\ref{uni_inv}) implies that $\left[|\phi\rangle\langle\phi|\right]^{\otimes t}$ and $\langle\psi|\mathbf{T}^{\dagger}_{k}\mathbf{T}_{k}|\psi\rangle^{t}$ are independent random variables. Then, combining Eq. (\ref{A3}) and (\ref{uni_inv}), we get
\begin{eqnarray}
\mathbb{E}_{\phi\in\mathcal{E}^{k}_{\text{TI}}}\left[\left[|\phi\rangle\langle\phi|\right]^{\otimes t}\right]&=&\dfrac{\int_{u\in U(d^N)}d\mu(u)\mathbf{T}^{\otimes t}_{k}\left[u|0\rangle\langle 0|u^{\dagger}\right]^{\otimes t}\mathbf{T}^{\dagger\otimes t}_{k} }{\int_{u\in U(d^N)}d\mu(u)\langle 0| u^{\dagger}\mathbf{T}^{\dagger}_{k}\mathbf{T}_{k}u|0\rangle^{t}} \nonumber\\
&=&\dfrac{\mathbf{T}^{\otimes t}_{k}\mathbf{\Pi}_{t}}{\text{Tr}\left( \mathbf{T}^{\otimes t}_{k}\mathbf{\Pi}_{t} \right)},
\end{eqnarray}
implying the result. Our analysis does not require an explicit form for $\alpha^{t}_{k}=\text{Tr}\left( \mathbf{T}^{\otimes t}_{k}\mathbf{\Pi}_{t} \right)$. Hence, we leave it unchanged.
\end{proof}

Therefore, the T-invariance imposes an additional structure to the moments through the product of $\mathbf{T}^{\otimes t}_{k}$, where the Haar moments are uniform linear combinations of the permutation group elements. 
In the following section, we identify a sufficient condition for obtaining approximate state designs from the generic T-invariant generator states sampled from $\mathcal{E}^{k}_{\text{TI}}$. 

\section{Quantum state designs from T-invariant generator states}\label{emergent_designs}
In this section, we construct the projected ensembles from the T-invariant generator states sampled from $\mathcal{E}^{k}_{\text{TI}}$and verify their convergence to the quantum state designs. In particular, if $|\phi\rangle$ is drawn uniformly at random from $\mathcal{E}^{k}_{\text{TI}}$, we intend to verify the following identity:
\begin{eqnarray}\label{mean_t}
\mathbb{E}_{|\phi\rangle\in\mathcal{E}^{k}_{\text{TI}}} \left(\sum\limits_{|b\rangle\in\mathcal{B}}\dfrac{\left(\langle b|\phi\rangle\langle\phi|b\rangle\right)^{\otimes t}}{\left(\langle\phi|b\rangle\langle b|\phi\rangle\right)^{t-1}}\right)=\dfrac{\mathbf{\Pi}^{A}_{t}}{\mathcal{D}_{A, t}}, 
\end{eqnarray}
where $\mathcal{D}_{A, t}=2^{N_A}(2^{N_A}+1)...(2^{N_A}+t-1)$.  

The term on the left-hand side evaluates the $t$-th moment of the projected ensemble of a T-invariant state $|\phi\rangle$, with an average taken over all such states in $\mathcal{E}^{k}_{\text{TI}}$. It has been shown that for the Haar random states, the $t$-th moments of the projected ensembles are Lipshitz continuous functions with a Lipschitz constant $\eta=2(2t-1)$ \cite{cotler2023emergent}. Being $\mathcal{E}^{k}_{\text{TI}}\subset\mathcal{E}_{\text{Haar}}$, $\eta$ is also a Lipschitz constant for the case of $\mathcal{E}^{k}_{\text{TI}}$. Note that the number of independent parameters of a T-invariant state with momentum $k$ is $l=2\thinspace\text{rank}(\mathbf{T}_{k})-1$. Since $\mathbf{T}_k$ is a subspace projector, its eigenvalues assume values of either $N$ or zero. Therefore, $\text{rank}(\mathbf{T}_{k})=\text{Tr}(\mathbf{T}_{k})/N$. Hence, a random T-invariant state can be regarded as a point on a hypersphere of dimension $l=2\text{Tr}(\mathbf{T}_{k})/N-1$. Then, if the above relation in Eq. (\ref{mean_t}) holds, Levy's lemma guarantees that any typical state drawn from $\mathcal{E}^{k}_{\text{TI}}$ will form an approximate state design. In the following, we first verify the identity in Eq. \eqref{mean_t} for $t=1$, followed by the more general case of $t>1$. We then compare the results against the Haar random generator states.

\subsection{Emergence of first-order state designs ($t=1$)}
\label{first-design}
From Eq. (\ref{Tmoments}) of the last section, the first moment of $\mathcal{E}^{k}_{\text{TI}}$ is
$\mathbb{E}_{|\phi\rangle\in\mathcal{E}^{k}_{\text{TI}}}\left(|\phi\rangle\langle\phi|\right) = \mathbf{T}_{k}/\alpha_{1}$, where $\alpha_{1}$ typically scales exponentially with $N$. If $N$ is a prime number, we can write $\alpha^{1}_{k}$ explicitly as follows \cite{nakata2020generic}:
\begin{equation}\label{trace_T_k}
\alpha_{1}=\text{Tr}(\mathbf{T}_{k})=
\begin{cases}
    d^N+d(N-1) & \quad \textrm{if } k=0 \\
    d^N-d & \textrm{otherwise. }
\end{cases}
\end{equation}
To construct the projected ensembles, we now perform the local projective measurements on the $B$-subsystem. For $t=1$, the measurement basis is irrelevant. Then, for some generator state $|\phi\rangle\in\mathcal{E}^{k}_{\text{TI}}$, the first moment of the projected ensemble is given by $\sum_{|b\rangle\in\mathcal{B}}\langle b|\phi\rangle\langle\phi|b\rangle=\text{Tr}_{B}(|\phi\rangle\langle\phi|)$. Typically $\text{Tr}_{B}(|\phi\rangle\langle\phi |)$ approximates the maximally mixed state in the reduced Hilbert space $\mathcal{H}^{\otimes N_A}$. We verify this by averaging the partial trace over the ensemble $\mathcal{E}^{k}_{\text{TI}}$:
\begin{align}\label{N_A_1stmoment_gen}
\mathbb{E}_{\phi\in\mathcal{E}^{k}_{\text{TI}}}\left[\text{Tr}_{B}\left(|\phi\rangle\langle\phi |\right)\right]&=\text{Tr}_{B}\left[\mathbb{E}_{|\phi\rangle\in\mathcal{E}^{k}_{\text{TI}}}\left( |\phi\rangle\langle\phi| \right)\right]=\dfrac{\text{Tr}_{B}(\mathbf{T}_{k})}{{\alpha_{1}}}.
\end{align}
To examine the closeness of $\text{Tr}_{B}(\mathbf{T}_{k})/\alpha_{1}$ to the maximally mixed state, we first expand $\mathbf{T}_{k}$ given in Eq. (\ref{Tstate}) as
\begin{align}\label{Ptrace}
\text{Tr}_{B}(\mathbf{T}_{k})=2^{N}\left[\dfrac{\mathbb{I}_{A}}{2^{N_A}}+\dfrac{1}{2^N}\sum_{j=1}^{N-1}e^{2\pi ijk/N}\text{Tr}_{B}(T^j)\right].
\end{align}
For $j=1$, it can be shown that $\text{Tr}_{B}(T)=T_{A}$, where $T_A$ is the translation operator acting exclusively on the subsystem-$A$. For $j\geq 2$, $\text{Tr}_{B}(T^j)=\sum_{b}\langle b|T^j|b\rangle$ outputs a random permutation operator supported over $A$ whenever $N_A\geq \gcd{(N, j)}$. If $N_A<\gcd{(N, j)}$, the partial trace would result in a constant times identity operator $\mathbb{I}_{2^{N_A}}$. Interested readers can find more details in Appendix \ref{ptrace}. If $N$ is prime and $2\leq j<N$, we have $\gcd{(N, j)}=1$, which is less than $N_A$ whenever $N_A>2$. Then, we can explicitly show that the trace norm of the second term on the right side in Eq. (\ref{Ptrace}) is bounded from above as follows: 
\begin{eqnarray}\label{1-design}
\dfrac{1}{2^N}\left\| \sum_{j=1}^{N-1}e^{2\pi ijk/N}\text{Tr}_{B}(T^j) \right\|_{1}\leq \sum_{j=1}^{N-1}\dfrac{\left\| \text{Tr}_{B}(T^j) \right\|_{1}}{2^N}=\dfrac{(N-1)}{2^{N_B}},
\end{eqnarray}
implying that it converges to a null matrix exponentially with $N_B$. Hence, Eq. (\ref{N_A_1stmoment_gen}) converges exponentially with $N_B$ to the maximally mixed state ($\rho_{A}\propto\mathbb{I}$) in the reduced Hilbert space $\mathcal{H}^{\otimes N_A}$ --- for $N_B\gg\log_2(N)$, we have $\mathbb{E}_{\phi}[\sum_{b}\langle b|\phi\rangle\langle\phi|b\rangle]\approx\mathbb{I}_{A}/2^{N_A}$. Then, as mentioned before, we can invoke Levy's lemma to argue that a typical $|\phi\rangle\in\mathcal{E}^{k}$ approximately generates a state-$1$ design. 

While we initially assumed that $N$ is prime, the results also hold qualitatively for non-prime $N$. In the latter case, for $N_A<\gcd(N, j)$, partial traces may yield identity operators, i.e., $\sum_{b}\langle b|T^j|b\rangle\propto\mathbb{I}_{2^{N_A}}$. These instances introduce slight deviations from Eq. (\ref{1-design}) but have a minor impact on the overall result. Since the error is exponentially suppressed, it is natural to expect that $\Delta^{(1)}$ for these states typically shows identical behavior as that of the Haar random states, and we confirm this with the help of numerical simulations. 

\begin{figure}
\begin{center}
\includegraphics[scale=0.55]{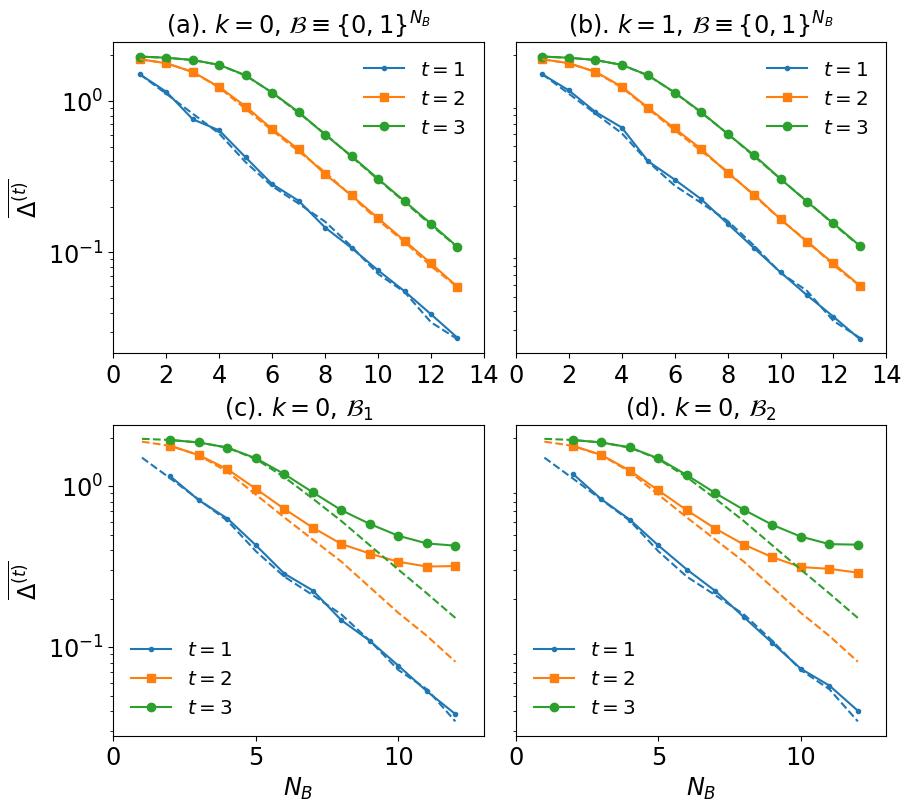}
\end{center}
\caption{\label{fig:haar_vs_tran} Illustration of $\overline{\Delta^{(t)}}$ versus $N_B$ for the projected ensembles of T-invariant generator states sampled uniformly at random from $\mathcal{E}^{k}_{\text{TI}}$. The results are shown for the first three moments. In the top panels (a) and (b), the measurements are performed in the computational basis: $\{\Pi_b = |b\rangle\langle b|$ for all $b\in \{0, 1\}^{N_B}\}$, where $\{0, 1\}^{N_B}$ denotes the set of all $N_B$-bit strings. While (a) represents the numerical computations in the $k=0$ momentum sector, $k=1$ is considered in panel (b). The results are averaged over ten initial generator states. The results appear qualitatively similar for both the momentum sectors. The trace distance falls to zero with an exponential scaling $\sim 2^{-N_B/2}$ for all the moments calculated. The calculations in (c) and (d) are performed in the momentum sector $k=0$ for two different measurement bases. In (c), we perform the measurements in the eigenbasis of the local translation operator supported over $B$ --- $T_B$. In (d), we break the local translation symmetry of $T_B$ by applying a local Haar random unitary. We perform measurements in the eigenbasis of the resulting operator. See the main text for more details. }   
\end{figure}

Figure \ref{fig:haar_vs_tran} illustrates the decay of the trace distance versus $N_B$ for the first three moments, considering three different bases (see the description of the Fig. \ref{fig:haar_vs_tran} and the following subsection) and two momentum sectors. The blue-colored curves correspond to the first moment. The blue curves, irrespective of the choice of basis and the momentum sector, always show exponential decay, i.e., $\overline{\Delta^{(1)}}\sim 2^{-N/2}$, where the overline indicates that the quantity is averaged over a few samples. We further benchmark these results against the case of Haar random generator states. For the Haar random states, the results are shown in dashed curves with the same color coding used for the first moment. We observe that the results are nearly identical in both cases, with minor fluctuations attributed to the averaging over a finite sample size.

\subsection{Higher-order state designs ($t>1$)}
\label{second-design}
Here, we provide a condition for producing approximate higher-order state designs from the random T-invariant generator states. 
    
\begin{theorem}\label{sufficient}
{\textup{(Sufficient condition for the identity in Eq. (\ref{mean_t}))}} 
Given a measurement basis $\mathcal{B}$ having supported over the subsystem-$B$, then the identity in Eq. (\ref{mean_t}) holds if for all $|b\rangle\in\mathcal{B}$, $\langle b|\mathbf{T}_{k}|b\rangle =\mathbb{I}_{2^{N_A}}$. 
\end{theorem}

\begin{proof}
Here, we seek to obtain a sufficient condition for the emergence of state designs from randomly chosen T-invariant generator states from $\mathcal{E}^{k}_{\text{TI}}$. In particular, for a randomly chosen $|\phi_{AB}\rangle\in\mathcal{E}^{k}_{\text{TI}}$, we aim to establish a condition on the measurement basis $|\mathcal{B}\rangle\equiv\{|b\rangle\}$ for the following identity:
\begin{align}
\mathbb{E}_{|\phi_{AB}\rangle\in\mathcal{E}^{k}_{\text{TI}}} \left(\sum_{|b\rangle\in\mathcal{B}}\dfrac{\left[\langle b|\phi_{AB}\rangle\langle\phi_{AB}|b\rangle\right]^{\otimes t}}{\left(\langle\phi_{AB}|b\rangle\langle b|\phi_{AB}\rangle\right)^{t-1}}\right)=\dfrac{\mathbf{\Pi}^{A}_{t}}{d_A(d_A+1)...(d_A+t-1)},
\end{align}
where $d_A=2^{N_A}$, the total Hilbert space dimension of the susbsyetm $A$. Since the expectation ($\mathbb{E}_{|\phi\rangle\in\mathcal{E}^{k}_{\text{TI}}}$) commutes with the summation ($\sum_{|b\rangle\in\mathcal{B}}$), we write 
\begin{align}
\mathbb{E}_{|\phi_{AB}\rangle\in\mathcal{E}^{k}_{\text{TI}}} \left(\sum_{|b\rangle\in\mathcal{B}}\dfrac{\left[\langle b|\phi_{AB}\rangle\langle\phi_{AB}|b\rangle\right]^{\otimes t}}{\left(\langle\phi_{AB}|b\rangle\langle b|\phi_{AB}\rangle\right)^{t-1}}\right)=\sum_{|b\rangle\in\mathcal{B}}\mathbb{E}_{|\phi_{AB}\rangle\in\mathcal{E}^{k}_{\text{TI}}}\left( \dfrac{\left[\langle b|\phi_{AB}\rangle\langle\phi_{AB}|b\rangle\right]^{\otimes t}}{\left(\langle\phi_{AB}|b\rangle\langle b|\phi_{AB}\rangle\right)^{t-1}} \right) .   
\end{align}
We note that for any $|\phi_{AB}\rangle\in\mathcal{H}^{\otimes N}$ and any $|b\rangle\in\mathcal{H}^{\otimes N_B}$, the scalar quantity $\langle\phi_{AB}|b\rangle\langle b|\phi_{AB}\rangle$ is always less than or equal to $1$, i.e., $\langle\phi_{AB}|b\rangle\langle b|\phi_{AB}\rangle\leq 1$. Thus, by writing $(1-(1-\langle\phi_{AB}|b\rangle\langle b|\phi_{AB}\rangle))$ in the denominator, we make use of the infinite series expansion of $1/(1-x)^{t-1}$ to evaluate the above expression. It then follows that 
\begin{align}\label{B3}
\dfrac{\left(\langle b|\phi_{AB}\rangle\langle\phi_{AB}|b\rangle\right)^{\otimes t}}{\left(\langle\phi_{AB}|b\rangle\langle b|\phi_{AB}\rangle\right)^{t-1}}&= \left(\langle b|\phi_{AB}\rangle\langle\phi_{AB}|b\rangle\right)^{\otimes t}\sum_{n=0}^{\infty}\binom{n+t-2}{t-2}\sum_{r=0}^n \binom{n}{r} (-1)^r \text{Tr}\left[\left(\langle b|\phi_{AB}\rangle\langle\phi_{AB}|b\rangle\right)^{\otimes r}\right]
\end{align}
Note that the (unnormalized) state $\langle b|\phi_{AB}\rangle\langle\phi_{AB}|b\rangle$ has support solely over the subsystem $A$. For computational convenience, we write it as follows:
\begin{eqnarray}\label{unn}
\langle b|\phi_{AB}\rangle\langle\phi_{AB}|b\rangle &=&\left(\sum_{m_i=0}^{2^{N_A}-1}|m_i\rangle \langle m_i|\right)\left(\langle b|\phi_{AB}\rangle\langle\phi_{AB}|b\rangle\right)\left(\sum_{n_i=0}^{2^{N_A}-1}|n_i\rangle \langle n_i|\right)\nonumber\\
&=&\sum_{m_i=0}^{2^{N_A}-1}\sum_{n_i=0}^{2^{N_A}-1}|m_i\rangle\langle n_i| \text{Tr}\left[ |n_i\rangle\langle m_i|\left(\langle b|\phi_{AB}\rangle\langle\phi_{AB}|b\rangle\right) \right],
\end{eqnarray}
where 
\begin{equation*}
\sum_{m_i=0}^{2^{N_A}-1}|m_i\rangle \langle m_i|=\sum_{n_i=0}^{2^{N_A}-1}|n_i\rangle \langle n_i|= \mathbb{I}_{2^{N_A}}.  
\end{equation*}
Incorporating Eq. (\ref{unn}) into Eq. (\ref{B3}) gives
\begin{align}
\dfrac{\left[\langle b|\phi_{AB}\rangle\langle\phi_{AB}|b\rangle\right]^{\otimes t}}{\left(\langle\phi_{AB}|b\rangle\langle b|\phi_{AB}\rangle\right)^{t-1}}&=\sum_{\substack{m_1, m_2, ..., m_t \\ n_1, n_2, ..., n_t}}|m_1m_2..., m_t\rangle\langle n_1, n_2, ...n_t|\sum_{n=0}^{\infty}\binom{n+t-1}{t-1}\sum_{r=0}^n \binom{n}{r}  \nonumber\\
&(-1)^r\text{Tr}\left[|n_1, n_2, ..., n_t\rangle\langle m_1, m_2, ..., m_t|\left(\langle b|\phi_{AB}\rangle\langle\phi_{AB}|b\rangle\right)^{\otimes (t+r)}\right],
\end{align}
In this expression, all the replicas of $\langle b|\phi_{AB}\rangle\langle\phi_{AB}|b\rangle$ are stacked together within the trace operation. This allows us to perform the invariant integration over the states $|\phi_{AB}\rangle\in\mathcal{E}^{k}_{\text{TI}}$, which is evaluated as 
\begin{eqnarray}\label{B6}
\mathbb{E}_{|\phi_{AB}\rangle\in\mathcal{E}^{k}_{\text{TI}}}\left[\left(\langle b|\phi_{AB}\rangle\langle\phi_{AB}|b\rangle\right)^{\otimes (t+r)}\right]=\dfrac{\langle b|\mathbf{T}_{k}|b\rangle^{\otimes (t+r)}\mathbf{\Pi}^{A}_{t+r}}{\text{Tr}(\mathbf{T}_{k}\mathbf{\Pi}^{AB}_{t+r})},    
\end{eqnarray}
where $\mathbf{\Pi}^{A}_{t+r}$ and $\mathbf{\Pi}^{AB}_{t+r}$ denote projectors onto the permutation symmetric subspaces of $t+r$ copies. While the former acts only on the replicas of the subsystem $A$, the latter acts on the replicas of the entire system $AB$. It is now useful to write $\mathbf{\Pi}^{A}_{t+r}$ as follows:
\begin{eqnarray}
\mathbf{\Pi}^{A}_{t+r}=\mathcal{D}_{A, t+r}\int_{|\psi\rangle\in\mathcal{E}_{\text{Haar}}}d\psi \left(|\psi\rangle\langle\psi |\right)^{\otimes (t+r)},\quad\text{where }|\psi\rangle\in\mathcal{H}^{\otimes N_A}
\end{eqnarray}
Where, $\mathcal{D}_{A, t+r}=d_A(d_A+1)...(d_A+t+r-1)$ and $d_A=2^{N_A}$. It then follows that 
\begin{align}\label{sing}
\sum_{|b\rangle\in\mathcal{B}}\mathbb{E}_{\phi_{AB}\in\mathcal{E}^{k}_{\text{TI}}}\left( \dfrac{\left[\langle b|\phi_{AB}\rangle\langle\phi_{AB}|b\rangle\right]^{\otimes t}}{\left(\langle\phi_{AB}|b\rangle\langle b|\phi_{AB}\rangle\right)^{t-1}} \right)&=
\sum_{|b\rangle\in\mathcal{B}}\langle b|\mathbf{T}_{k}|b\rangle^{\otimes t}\int_{\psi\in\mathcal{E}_{\text{Haar}}}d\psi \left(|\psi\rangle\langle\psi |\right)^{\otimes t} \nonumber\\
&\sum_{n=0}^{\infty}\binom{n+t-2}{t-2}\sum_{r=0}^n \binom{n}{r} (-1)^r \mathcal{D}_{A, t+r}\dfrac{\langle\psi b|\mathbf{T}_{k}|\psi b\rangle^{r}}{\text{Tr}(\mathbf{T}_{k}\mathbf{\Pi}^{AB}_{t+r})}.
\end{align}
A sufficient condition, $\langle b|\mathbf{T}_{k}|b\rangle=\mathbb{I}_{2^{N_A}}$ for all $|b\rangle\in\mathcal{B}$, ensures the convergence of the right-hand side of the above expression to the Haar moments. If satisfied, for all $|b\rangle\in\mathcal{B}$, we will have $\langle \psi b|\mathbf{T}_{k}|\psi b\rangle=1$. As a result, both the integral and the integrand can be decoupled from the infinite series. Then, the infinite series can be understood as the normalizing factor, which necessarily converges to $1/2^{N_B}$. Therefore, we get
\begin{align}
\mathbb{E}_{\phi_{AB}\in\mathcal{E}^{k}}\left(\sum_{|b\rangle\in\mathcal{B}} \dfrac{\left[\langle b|\phi_{AB}\rangle\langle\phi_{AB}|b\rangle\right]^{\otimes t}}{\left(\langle\phi_{AB}|b\rangle\langle b|\phi_{AB}\rangle\right)^{t-1}} \right)&= \int_{\psi\in\mathcal{E}^{A}_{\text{Haar}}}d\psi \left(|\psi\rangle\langle\psi |\right)^{\otimes t}\nonumber\\
&=\dfrac{\mathbf{\Pi}^{A}_{t}}{2^{N_A}(2^{N_A}+1)...(2^{N_A}+t-1)},
\end{align}
implying that the moments of the projected ensembles, on average, converge to the Haar moments. 
\end{proof}

For a given basis vector $|b\rangle\in\mathcal{B}$, the condition is maximally violated if it can be extended to have support over the full system such that it becomes an eigenstate (with momentum charge $k$) of the translation operator. That is, if there exists an arbitrary $|a\rangle\in\mathcal{H}^{\otimes N_A}$ such that $T\left(|a\rangle\otimes |b\rangle\right) =e^{-2\pi ik/N}\left(|a\rangle\otimes |b\rangle\right)$. Then, the expectation of $\mathbf{T}_{k}$ in this state becomes $\langle ab|\mathbf{T}_{k}|ab\rangle =N$, which is in maximal violation of the sufficient condition \footnote{Note that the sufficient condition would require $\langle ab|\mathbf{T}_{k}|ab\rangle =1$}. For example, when $k=0$, the basis states $|0\rangle^{\otimes N_B}$ and $|1\rangle^{\otimes N_B}$ of the standard computational basis can be extended to $|0\rangle^{\otimes N}$ and $|1\rangle^{\otimes N}$ respectively. So, $\langle 0|^{\otimes N}\mathbf{T}_0|0\rangle^{\otimes N}=\langle 1|^{\otimes N}\mathbf{T}_0|1\rangle^{\otimes N}=N$, correspond to the maximal violation. On the other hand, most of the basis vectors of the computational basis satisfy the sufficient condition. It is also worth noting that if $|ab\rangle$ becomes a T-invariant state with a different momentum charge ($k'\neq k$), then $\langle ab| \mathbf{T}_{k} |ab\rangle=0$. 

We now calculate the trace distance $\Delta^{(t)}$ and average it over a few sample states taken from $\mathcal{E}^{k}_{\text{TI}}$. We illustrate the results in Fig. \ref{fig:haar_vs_tran} for the second (orange color) and third (green color) moments by keeping $N_A$ and the measurement bases as before. 
Similar to the case of the first moment, we contrast the results with the case of Haar random generator states represented by the dashed curves. 
In the panels \ref{fig:haar_vs_tran}a and \ref{fig:haar_vs_tran}b corresponding to generator states with $k=0$ and $1$,  projective measurements in the computational basis show an exponential decay of the average trace distance with $N_B$ for both the moments. Additionally, on a semilog scale, the decays for all the moments appear to align along parallel lines at sufficiently large $N_B$ values. From the comparison between the Figs.\ref{fig:haar_vs_tran}a and \ref{fig:haar_vs_tran}b, it is evident that there are no noticeable differences when the generator states are chosen from different momentum sectors. We also consider the eigenbasis of the local translation operator $T_B$ for the measurements, of which a representative case for $k=0$ is shown in Fig. \ref{fig:haar_vs_tran}c. We observe the average trace distances deviate from the initial exponential decay and approach non-zero saturation values. In Fig. \ref{fig:haar_vs_tran}d, we consider the case where the translation symmetry is broken weakly by applying a single site Haar random unitary to the left of $T_B$, i.e., $T'_B=(u\otimes\mathbb{I}_{2^{N_B-1}})T_{B}$. We then consider the eigenbasis of the resultant operator $T'_B$ for the measurements on $B$. Surprisingly, the trace distance still saturates to a finite value for both the moments despite the broken translation symmetry. 
In the following subsection, we elaborate more on the interplay between the measurement bases and the sufficient condition derived in Result. \ref{sufficient}.

\subsection{Overview of the bases violating the sufficient condition} 
\label{violation}
From Fig. \ref{fig:haar_vs_tran}, it is evident that not all measurement bases furnish higher-order state designs. Here, we analyze the degree of violation of the sufficient condition by different measurement bases. 
Some, like the computational basis, exhibit mild violations, while others significantly deviate from the condition. Given a measurement basis $\mathcal{B}$, we quantify the average violation of the sufficient condition using the quantity $\mathbf{\Delta}(\mathbf{T}_{k}, \mathcal{B})/2^{N_B}$, where 
\begin{eqnarray}
\mathbf{\Delta}(\mathbf{T}_{k}, \mathcal{B})=\sum_{|b\rangle\in\mathcal{B}}\|\langle b|\mathbf{T}_{k}|b\rangle -\mathbb{I}_{2^{N_A}}\|_{1}.     
\end{eqnarray}
In general, finding bases that fully satisfy the condition, implying $\mathbf{\Delta}(\mathbf{T}_{k}, \mathcal{B})=0$, is hard.
Depending upon $\mathcal{B}$, this quantity will display a multitude of behaviors. Also, note that for a single site unitary $u$, the local transformation of $|b\rangle$ to $|b'\rangle=u^{\otimes N_B}|b\rangle$ leaves $\mathbf{\Delta}(\mathbf{T}_{k}, \mathcal{B})$ invariant. 
To see the nature of the violation in a generic basis, we numerically examine $\mathbf{\Delta}(\mathbf{T}_{k}, \mathcal{B})/2^{N_B}$ versus $N_B$ for three different bases, namely, the computational basis, a Haar random product basis, and a Haar random entangling basis, all supported over $B$.

\begin{figure}
\begin{center}
\includegraphics[scale=0.55]{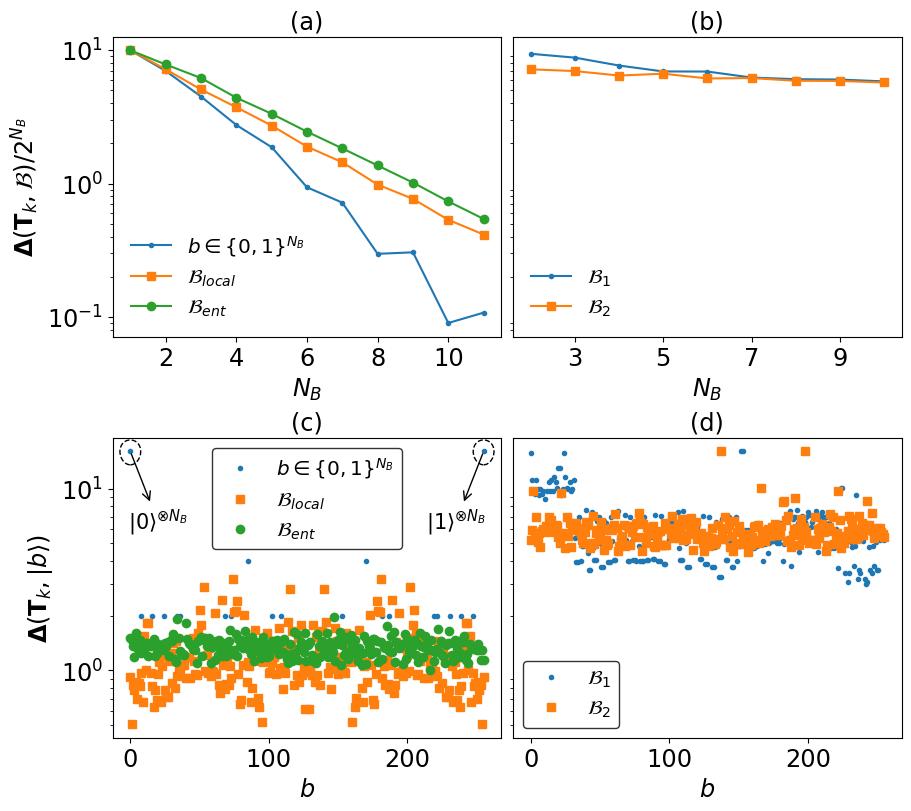}
\end{center}
\caption{\label{fig:viol}  The figure illustrates the average violation of the sufficient condition by different bases as quantified by $\mathbf{\Delta}(\mathbf{T}_{k}, \mathcal{B})/2^{N_B}$. Here, we fix $N=3$. In panel (a), $\mathbf{\Delta}(\mathbf{T}_{k}, \mathcal{B})/2^{N_B}$ versus $N_B$ is plotted for three different bases, namely, (i) the computational basis or $\sigma^z$ basis (blue), (ii) basis obtained by applying local Haar random unitaries on the computational basis (orange), and (iii) an entangled basis obtained by applying a Haar unitary of dimension $2^{N_B}$ on the computational basis (green). For all three bases, $\mathbf{\Delta}(\mathbf{T}_{k}, \mathcal{B})$ decays exponentially with $N_B$. In panel (b), the measurements are performed in the eigenbases of the operators $T_B$ and $u_{N_A+1}T_B$, where $T_B$ is the local translation operator supported over the subsystem $B$ and $u_{N_A+1}$ denotes a local Haar unitary acting on a site labeled with $N_A+1$. The violation stays nearly constant with $N_B$ for these two bases. In the bottom panels (c) and (d), the violation is quantified for each basis vector by fixing $N=11$. Here, we plot $\|\langle b| \mathbf{T}_{k} |b\rangle\|_1$ for each $|b\rangle\in\mathcal{B}$ for all the bases considered in the above panels. See the main text for more details. }   
\end{figure}

The results are shown in Fig. \ref{fig:viol}. In Fig. \ref{fig:viol}a, the blue curve represents the violation for the computational basis. Clearly, the decay of the violation is exponential and faster than the other cases considered. The orange curve represents the case of local random product basis. This can be obtained by applying a tensor product of Haar random local unitaries on the computational basis vectors. 
As the figure depicts, the violation still decays exponentially but slower than in the case of computational basis. Finally, we consider a random entangling basis by applying global Haar unitaries on the computational basis vectors. The violation still decays exponentially but slower than the previous two. In Fig. \ref{fig:viol}c, we plot the violation for each basis vector of the above bases considered while keeping $N_B$ fixed. We see that, except for a few vectors, the violation stays concentrated near a value of order $O(1)$. In the computational basis, the only vectors $|0\rangle^{\otimes N_B}$ and $|1\rangle^{\otimes N_B}$ show the maximal violation, which are encircled/marked in Fig. \ref{fig:viol}c. We consider the violation is significant if $\mathbf{\Delta}(\mathbf{T}_{k}, \mathcal{B})/2^{N_B}$ does not decay with $N_B$. If this happens, the projected ensembles may fail to converge to the state designs even in the large $N_B$ limit. Note that, while the exponential decay of the violation may appear generic, there exist bases that show nearly constant violation as $N_B$ increases, which are depicted in Fig. \ref{fig:haar_vs_tran}b and \ref{fig:haar_vs_tran}d.

To explore such measurement bases with significant violations, we consider the following.
\begin{align}\label{tria}
\mathbf{\Delta}(\mathbf{T}_{k}, \mathcal{B})&=\sum_{|b\rangle\in\mathcal{B}}\left\|\langle b|\sum_{j=0}^{N-1}e^{2\pi ijk/N}T^j|b\rangle -\mathbb{I}_{2^{N_A}}\right\|_{1}\nonumber\\  &=\sum_{|b\rangle\in\mathcal{B}} \left\| \sum_{j=1}^{N-1}e^{2\pi ijk/N} \langle b| T^j |b\rangle \right\|_{1} \nonumber\\
&=\sum_{|b\rangle\in\mathcal{B}}\left(\left\|e^{2\pi ir/N} \langle b|T^{r}|b\rangle +\sum_{j\neq r}e^{2\pi ijk/N} \langle b| T^j |b\rangle \right\|_{1}\right).
\end{align}
In the second line, we used the fact $\langle b|\mathbb{I}_{2^N}|b\rangle =\mathbb{I}_{2^{N_A}}$ and subtracted it from $e^{2\pi i 0/N}\langle b|T^0|b\rangle$. 
In the third equality, a term corresponding to an arbitrary integer $r$ in the summation, where $1 \leq r \leq N-1$, has been isolated from the remaining terms. This enables the analysis of violations with respect to each element of the translation group, facilitating the identification of the violating bases. To illustrate this, we consider the specific case where $r=1$:
\begin{eqnarray}
\langle b|T|b\rangle&=&\langle b|S_{1, 2}S_{2, 3}\cdots S_{N-1, N}|b\rangle\nonumber\\    
&=&\left(S_{1, 2}\cdots S_{N_A-1, N_A}\right)\langle b|S_{N_A, N_A+1}\cdots S_{N-1, N}|b\rangle\nonumber\\
&=&\int_{u}d\mu(u)\left(T_Au^{\dagger}_{N_A}\right)\langle b|u_{N_A+1}T_B|b\rangle, 
\end{eqnarray}
where $S_{i,i+1}$ denotes the swap operator between $i$ and $i+1$ sites, $T_A$ and $T_B$ are translation operators locally supported over the subsystems $A$ and $B$. In the third equality, $S_{N_A, N_A+1}$ is replaced by the unitary Haar integral expression of the swap operator --- $S_{N_A, N_A+1}=\int_{u}d\mu(u)(u^{\dagger}_{N_A}\otimes u_{N_A+1})$, where $d\mu(u)$ represents the invariant Haar measure over the unitary group $U(2)$ \cite{zhang2014matrix}. Then, we heuristically argue that the measurements in the eigenbasis of the operator $u_{N_A+1}T_B$ for an arbitrary $u$ would lead to $\mathbf{\Delta}(\mathbf{T}_{k}, \mathcal{B})\sim O(2^{N})$. Consequently, the average violation $\mathbf{\Delta}(\mathbf{T}_{k}, \mathcal{B})/2^{N_B}$ stays nearly a constant of order $O(2^{N_A})$ even in the large $N_B$ limit. We illustrate this by considering the eigenbases of the operator $u_{N_A+1}T_B$ in Fig. \ref{fig:viol}b and \ref{fig:viol}d for two cases of $u_{N_B}$, namely, $u_{N_B}=\mathbb{I}_{2}$ and a Haar random $u_{N_B}$. Infact, we considered the same measurement bases in Figs. \ref{fig:haar_vs_tran}c and \ref{fig:haar_vs_tran}d and observed that the projected ensembles do not converge to the higher-order state designs.
It is interesting to notice that the measurement bases that do not respect the translation symmetry can also hinder the design formation [see Fig. \ref{fig:haar_vs_tran}d]. In the following, we elaborate on this aspect further by considering $r=2$ in Eq. (\ref{tria}).

\subsection{Violation of the sufficient condition for $r=2$}
\label{vior2}
In this appendix, we examine Eq. (\ref{tria}) for $r=2$ and seek to identify the measurement bases that strongly violate the sufficient condition. So far, we have carried out the analysis for $r=1$. We have observed that the eigenbases of the operators $u_{N_A+1}T_B$ for all $u_{N_A+1}\in U(d)$ significantly violate the condition, where $T_B$ denotes the translation operator supported over $B$. The subscript $N_A+1$ denotes that the unitary acts on the site labeled $N_A+1$. When $u_{N_A+1}=\mathbb{I}_{2}$, the operator is simply a translation operator over $B$. The eigenbasis of this operator is locally translation invariant. However, for a random $u_{N_A+1}$, the translation symmetry gets weakly broken. As $r$ is increased further, the local translation symmetry of the measurement basis gradually disappears. To see this for $r=2$, we consider the following equality: 
\begin{eqnarray}
\mathbf{\Delta}(\mathbf{T}_{k}, \mathcal{B})=\sum_{|b\rangle\in\mathcal{B}}\left(\left\|e^{2\pi ir/N} \langle b|T^{2}|b\rangle +\sum_{j\neq 2}e^{2\pi ijk/N} \langle b| T^j |b\rangle \right\|_{1}\right).  
\end{eqnarray}
We now write $T^2$ as 
\begin{eqnarray}
 T^2&=&TT\nonumber\\
 &=&\left(S_{12}S_{23}\cdots S_{N-1, N}\right)\left(S_{12}S_{23}\cdots S_{N-1, N}\right)
\end{eqnarray}
For simplicity, we take $N_A=3$. The partial expectation of $T^2$ with respect to a basis vector $|b\rangle\in\mathcal{B}$ can be written as 
\begin{eqnarray}
\langle b|T^2|b\rangle &=& S_{12}S_{23}S_{12}\langle b| S_{34}S_{23}\left( S_{45}S_{34}S_{56}\cdots S_{N-1, N} \right)\left( S_{45}S_{56}\cdots S_{N-1, N} \right) |b\rangle\nonumber\\
\end{eqnarray}
We now substitute the integral expression of the swap operators corresponding to $S_{34}$ and $S_{23}$. It then follows that 
\begin{eqnarray}
\langle b|T^2|b\rangle&=&S_{13}\int_{u\in U(d)}d\mu(u)\int_{v\in U(d)} d\mu(v)\int_{w\in U(d)}d\mu(w)u_{2}v_{3}u_{3}w_{3}\nonumber\\
&&\hspace{3.5cm}\langle b| v_{4} \left( S_{45}w_{4}S_{56}\cdots S_{N-1, N} \right)\left( S_{45}S_{56}\cdots S_{N-1, N} \right)|b\rangle \nonumber\\
&=&S_{13}\int_{u\in U(d)}d\mu(u)\int_{v\in U(d)} d\mu(v)\int_{w\in U(d)}d\mu(w)u_{2}v_{3}u_{3}w_{3}\langle b| v_{4}T_Bw_{4}T_{B} |b\rangle\nonumber\\
\end{eqnarray}
If the measurement basis is the eigenbasis of the operator $( v_4T_Bw_4T_B)$ for some $v$ and $w$ being local Haar random unitaries, one can expect that $\sum_{|b\rangle\mathcal{B}}\|\langle b| T^2 |b\rangle\|_{1}\sim O(2^{N_B})$. For $v=w=\mathbb{I}_{2}$, the above operator becomes $T^2_B$. The eigenbasis of this operator is invariant under translations by two sites. Likewise, one can show that for an arbitrary integer $r$, the eigenbasis of $T^r$ strongly violates the condition. Figures \ref{r2r3}a-\ref{r2r3}c demonstrate the decay of trace distance by considering the measurements in the eigenbases of $T^2_{B}$ and $T^3_{B}$ operators. From the figure, it is evident that the design formation is obstructed. This suggests that the sufficient condition we derived could potentially be necessary as well for the emergence of higher-order state designs.  

\begin{figure*}
\includegraphics[scale=0.425]{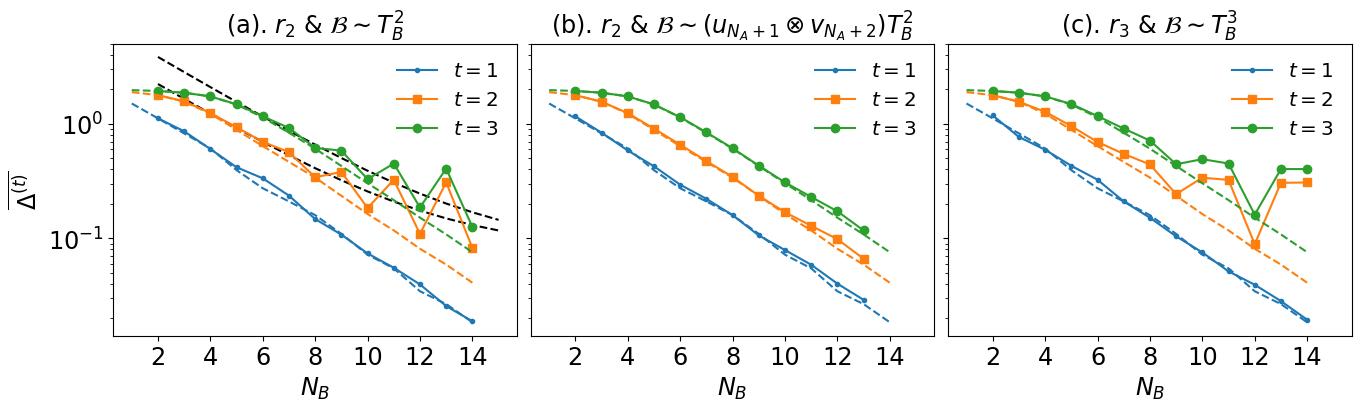}
\caption{\label{r2r3} Illustration of $\overline{\Delta^{(t)}}$ versus $N_B$ for the first three moments when the measurements are performed in bases largely violating the sufficient condition. In panels (a) and (b), the measurement bases are the eigenbases of $T^2_{B}$ and $(u_{N_A+1}\otimes v_{N_A+1})T^2_{B}$, where $T_B$ denotes the local translation operator over the subsystem-$B$. We observe that in the former case, the trace distances for higher order moments initially decay and acquire oscillatory behavior around exponential curves decaying to finite non-zero values for larger values of $N_B$. Due to the applications of local Haar random unitaries, the behavior in the latter case appears to decay for smaller $N_B$ values. However, for larger values of $N_B$, we anticipate that the trace distance approaches a finite non-zero value. Panel (c) corresponds to the measurements in the eigenbasis of $T^3_{B}. $  
} 
\end{figure*}

\section{Generalization to other symmetries}
\label{genn}
In the preceding sections, we examined the emergence of state designs from the translation symmetric generator states. In this section, we extend these findings to other symmetries, specifically considering $Z_{2}$ and reflection symmetries as representative examples.

\subsection{$Z_2$-symmetry}
\label{z2s}
The group associated with $Z_2$-symmetry consists of two elements $\{\mathbb{I}_{2^N}, \Sigma\}$, where $\Sigma=\otimes_{i=1}^{N}\sigma^{x}_{i}$. If a system is $Z_2$-symmetric, its Hamiltonian will be invariant under the action of $\Sigma$, i.e., $\Sigma  H\Sigma=H$. On the other hand, a quantum state $|\psi\rangle$ is considered $Z_2$-symmetric if it is an eigenstate of the operator $\Sigma$ with an eigenvalue $\pm 1$. Here, similar to the case of translation symmetry, we first construct an ensemble of $Z_2$-symmetric states by projecting the Haar random states onto a $Z_2$-symmetric subspace. We then follow the analysis of state designs using the projected ensemble framework.  

Given a Haar random state $|\psi\rangle\in\mathcal{E}_{\text{Haar}}$ that has support over $N$-sites, then
\begin{eqnarray}
|\psi\rangle\rightarrow |\phi\rangle=\dfrac{1}{\mathcal{N}}\mathbf{Z}_{k}|\psi\rangle
\end{eqnarray}
is a $Z_2$-symmetric state with an eigenvalue $(-1)^k$ with $k\in\{0, 1\}$, 
where $\mathbf{Z}_k=\mathbb{I}+(-1)^k\Sigma\quad\text{and }\mathcal{N}=\sqrt{\langle\psi|\mathbf{Z}^{\dagger}_{k}\mathbf{Z}_{k}|\psi\rangle}=\sqrt{2\langle\psi|\mathbf{Z}_{k}|\psi\rangle}$.
Introducing this symmetry reduces the randomness of the Haar random state by a factor of $2$. Specifically, the number of independent complex parameters needed to describe the state scales like $\sim O(2^{N-1})$. Using similar techniques employed for the T-invariant states, the moments of $Z_2$-symmetric ensembles can be evaluated as
\begin{eqnarray}
 \mathbb{E}_{\phi\in\mathcal{E}^{k}_{Z_{2}}}\left[\left[ |\phi\rangle\langle\phi| \right]^{\otimes t}\right] = \dfrac{\mathbf{Z}^{\otimes t}_{k}\mathbf{\Pi}_{t}}{\text{Tr}\left(\mathbf{Z}^{\otimes t}_{k}\mathbf{\Pi}_{t}\right)}.
\end{eqnarray}

The on-site nature of the $\mathbb{Z}_{2}$-symmetry allows us to write closed-form expressions for the moments of the projected ensembles irrespective of the choice of the measurement basis. Through detailed analytical derivation [see Appendix \ref{app-moments}], we show that the $t$-th order moment of the projected ensembles averaged over initial generator states takes the following form:
\begin{eqnarray}\label{z2moment}
\mathcal{M}^{t}_{\mathbb{Z}_{2}}= \dfrac{1}{\mathcal{N}}\sum_{|b\rangle\in\mathcal{B}}\langle b|\mathbf{Z}_{k}|b\rangle^{\otimes t}\mathbf{\Pi}^{A}_{t},    
\end{eqnarray}
where $\mathcal{N}$ denotes the normalization constant and is given by $\text{Tr}\left( \sum_{|b\rangle\in\mathcal{B}}\langle b|\mathbf{Z}_{k}|b\rangle^{\otimes t}\mathbf{\Pi}^{A}_{t} \right)$. Whenever $\langle b|\mathbf{Z}_{2}|b\rangle =\mathbb{I}_{2^{N_A}}$ for all $|b\rangle\in \mathcal{B}$, as required by the sufficient condition, right-hand side of the Eq. (\ref{z2moment}) equates to the $t$-th order Haar moment. In the following, we examine the projected ensembles for a few different choices of measurement bases.

To analyze the projected ensembles, we first fix the measurements on $N_B$ sites in the computational basis, i.e., $\{|b\rangle\langle b|\}$ for all $|b\rangle\in\{0, 1\}^{N_B}$. To get approximate state designs, it is sufficient to have $\langle b|\mathbf{Z}_{k}|b\rangle\approx\mathbb{I}_{2^{N_A}}$ for a sufficiently large number of $b\in\{0, 1\}^{N_B}$. In the case of $Z_2$-symmetry, this is indeed satisfied as we have $\langle b|\mathbf{Z}_{k}|b\rangle=\mathbb{I}_{2^{N_A}}+(-1)^{k}\langle b|\Sigma|b\rangle$, where the second term can be simplified as 
\begin{eqnarray}
\langle b|\Sigma|b\rangle =\left(\otimes_{j=1}^{N_A}\sigma^{x}_{j}\right)\left(\langle b_1|\sigma^x|b_1\rangle ...\langle b_{N_B}|\sigma^x|b_{N_B}\rangle\right).     
\end{eqnarray}
The terms within the brackets are the diagonal elements of $\sigma^x$ operator, which are zeros in the computational basis, implying that $\langle b|\Sigma|b\rangle=0$. Therefore, in this case, the sufficient condition is exactly satisfied. Consequently, the projected ensembles converge to the state designs for large $N_B$. The numerical results for the average trace distance are shown in Fig. \ref{fig:ZZ2}. We notice that the results nearly coincide with the case of Haar random generator states [see Fig. \ref{fig:ZZ2}a and \ref{fig:ZZ2}b]. 

\begin{figure}
\begin{center}
\includegraphics[scale=0.55]{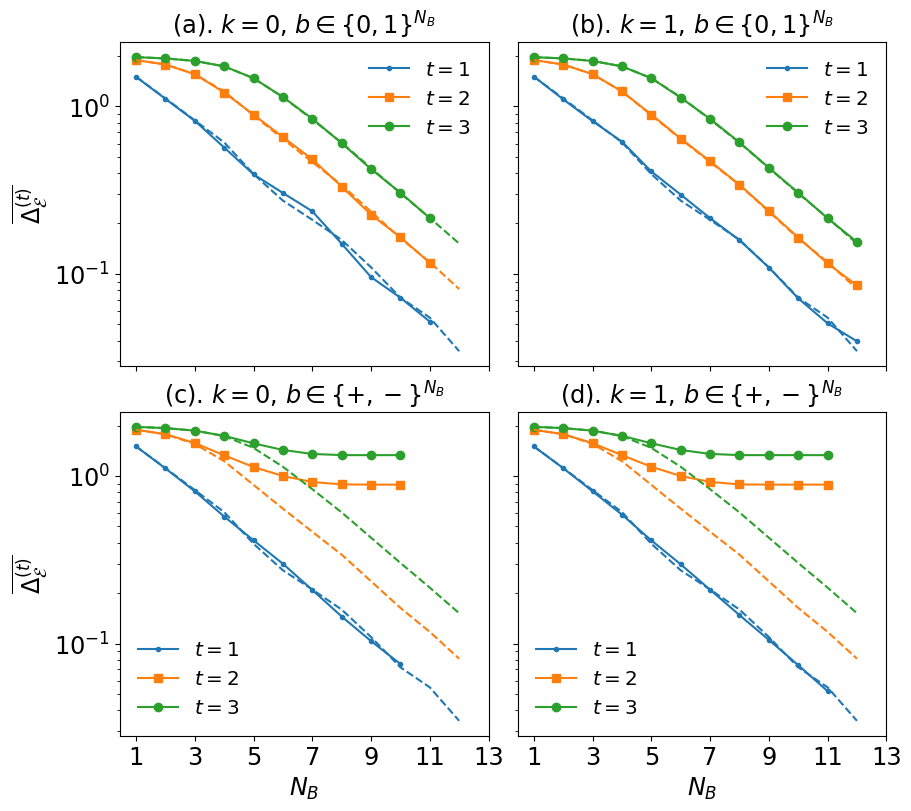}
\end{center}
\caption{\label{fig:ZZ2}  Average trace distance ($\overline{\Delta^{(t)}}$) between the moments of the projected ensembles and the moments of the Haar random states supported over $N_A$ sites, plotted against $N_B$. The initial states are chosen uniformly at random from the ensemble of $Z_2$-symmetric quantum states. The average is computed over ten samples of the initial generator states. In (a) and (b), the measurements are performed in the computational ($\sigma^{z}$) basis --- $\{|b\rangle\langle b|\}$ for all $b\in\{0, 1\}^{N_B}$. While the states chosen in (a) have the eigenvalue $1$, the other panel is plotted for the states with eigenvalue $-1$. We repeat the same calculation in the panels (c) and (d) with the measurement basis given by $\{|b\rangle\langle b|\}$ for all $b\in\{+, -\}^{N_B}$, where $|+\rangle$ and $|-\rangle$ represent the eigenstates of $\sigma^x$ and are connected to $|0\rangle$ and $|1\rangle$ through the Hadamard transform. 
} 
\end{figure}

\begin{figure}
\begin{center}
\includegraphics[scale=0.55]{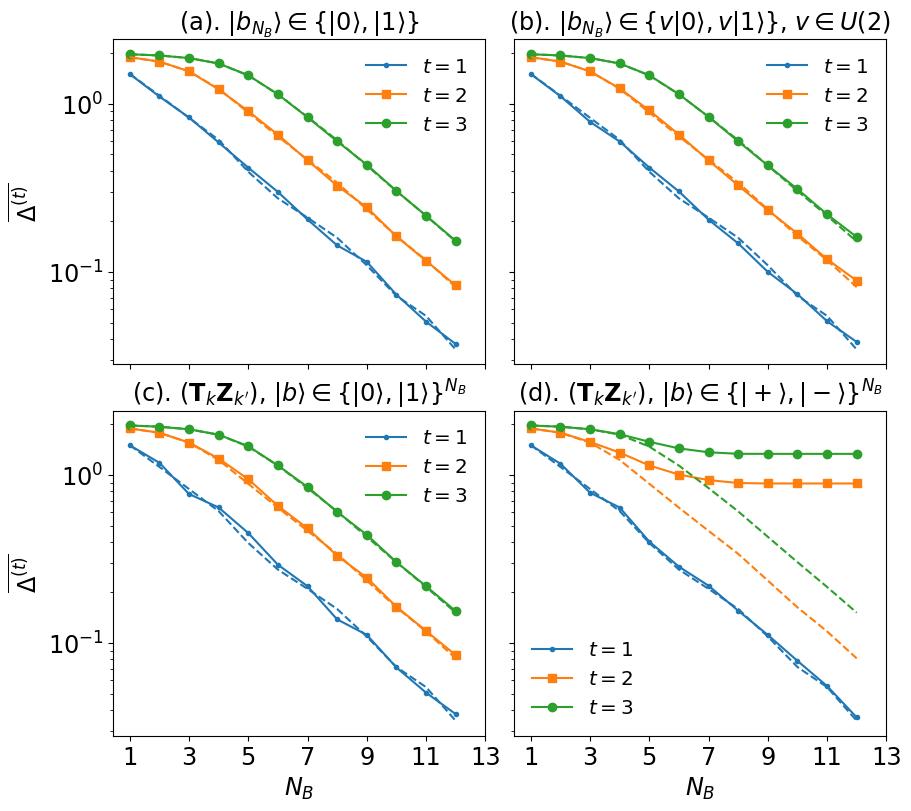}
\end{center}
\caption{\label{fig:ZZ2plustrans} The figure illustrates $\overline{\Delta^{(t)}}$ versus $N_B$. In (a), the measurements are performed in the eigenbasis of the operator $\sigma^{x}\otimes\cdots\otimes\sigma^x\otimes \sigma^z$, where the tensor product of $\sigma^x$ operators have support over $N_B-1$ sites and $\sigma^z$ is supported over $N_B$-th site. In (b), we replace the eigenbasis of $\sigma^z$ on $N_B$-th site with an eigenbasis of single site Haar random unitary. Here, we fix the charge $k=0$. In (c) and (d), we take the random generator states that are simultaneous eigenstates of both $T$ and $\Sigma$. While the measurement basis in $c$ is the computational basis, the eigenbasis of $\sigma^x$ is considered for measurements in (d). Here, the charges are fixed at $k=k'=0$. } 
\end{figure}

However, if the measurements are performed in $\sigma^x$ basis, given by $\{|b\rangle\langle b|\}$ for all $b\in\{+, -\}^{N_B}$, where $|+\rangle=(|0\rangle + |1\rangle)/\sqrt{2}$ and $|-\rangle=(|0\rangle - |1\rangle)/\sqrt{2}$, then $\langle \pm|\sigma^x|\pm\rangle=\pm 1$. As a result, the projected ensembles deviate significantly from the quantum state designs. The violation from the sufficient condition in this case, as quantified by $\mathbf{\Delta}(\mathbf{Z}_{k}, \mathcal{B})/2^{N_B}$, remains a constant for any $N_B$: 
\begin{align}
\dfrac{\mathbf{\Delta}(\mathbf{Z}_{k}, \mathcal{B})}{2^{N_B}} &=\dfrac{1}{2^{N_B}}\sum_{b\in\{+, -\}^{N_B}}\left\| \langle b|\mathbf{Z}_{k}|b\rangle -\mathbb{I}_{2^{N_A}} \right\|_{1}\nonumber\\
&=\dfrac{1}{2^{N_B}}\sum_{b\in\{+, -\}^{N_B}}\left\| (-1)^{k+\sum_{i=1}^{N_B}\text{sgn}(b_i)}\otimes_{j=1}^{N_A} \sigma^{x}_{j} \right\|\nonumber\\
&=2^{N_A}.
\end{align}
We demonstrate the numerical results of the average distance in Figs. \ref{fig:ZZ2}c and \ref{fig:ZZ2}d. In contrast to the computational basis measurements, here, only the first moment coincides with the case of Haar random generator states, while the higher moments appear to saturate to a finite value of $\overline{\Delta^{(t)}}$. {Interestingly, in this case, the average $t$-th order moment of the projected ensembles as obtained in Eq. (\ref{z2moment}) admits the following simple form [see Appendix \ref{app-moments} for details]:
\begin{eqnarray}
\mathcal{{M}}^{t}_{\mathbb{Z}_{2}}=\dfrac{1}{\mathcal{N}}\left( \mathbf{Z}^{\otimes t}_{0, N_A} + \mathbf{Z}^{\otimes t}_{1, N_A}\right) \mathbf{\Pi}^{t}_{A}    
\end{eqnarray}
with an appropriate normalizing constant $\mathcal{N}$. Here, $\mathbf{Z}_{0(1), N_A}$ denotes the $\mathbb{Z}_{2}$-symmetric subspace projector with corresponding charge when the system size is $N_A$. 
}


We now consider a case where we perform the local measurements on $N_B$-th site in $\sigma^z$ basis while keeping $\sigma^x$ measurement basis for the remaining $N_B-1$ sites. We represent the resultant basis with $b'\in\{+, 1\}^{N-1}\times \{0, 1\}$. Then,
\begin{eqnarray}
\langle b'|\mathbf{Z}_k|b'\rangle=\mathbb{I}_{2^{N_A}}+(-1)^{k}\langle b'|\Sigma|b'\rangle,     
\end{eqnarray}
where the second term vanishes as $\langle b'_{N_B}|\sigma^{x}_{N_B}|b'_{N_B}\rangle =0$. Therefore, we have $\langle b'|\Sigma|b'\rangle =0$ for all $|b'\rangle$, implying the required condition for the convergence of projected ensembles to the quantum state designs.
The corresponding results for the average trace distances are shown in Fig. \ref{fig:ZZ2plustrans}a. In Fig. \ref{fig:ZZ2plustrans}b, we replace the $\sigma^x$ basis on $N_B$-th site with a local Haar random basis. We find no significant differences between \ref{fig:ZZ2plustrans}a and \ref{fig:ZZ2plustrans}b within the range of $N_B$ considered for the numerical simulations.  
This demonstrates that a mild modification of the measurement basis can retrieve the state design if the sufficient condition is satisfied. So, the sufficient condition allows us to infer suitable measurement bases for obtaining higher-order state designs, which might not be apparent otherwise. Moreover, by gradually switching the local measurement basis (over $N_B$-th site) from $\sigma^x$ to $\sigma^z$, we can observe a transition in the randomness of the projected ensembles as characterized by $\overline{\Delta^{(t)}}$. In particular, we can choose the eigenbasis of $\alpha \sigma^z+(1-\alpha)\sigma^x$ for the measurements on $N_B$-th site. As the parameter varies, we observe a transition in $\overline{\Delta^{(t)}}$ from a finite constant value toward a system size-dependent value. We show the corresponding results in Fig. \ref{fig:z2trans}. From the figure, we observe that near $\alpha=0$, the trace distance $\overline{\Delta^{(t)}}$ remains system size independent. In contrast, as $\alpha$ approaches $1$, the trace distance becomes sensitive to the system size $N$. The violation of the sufficient condition, in this case, can be quantified as follows:
\begin{figure*}
\includegraphics[scale=0.43]{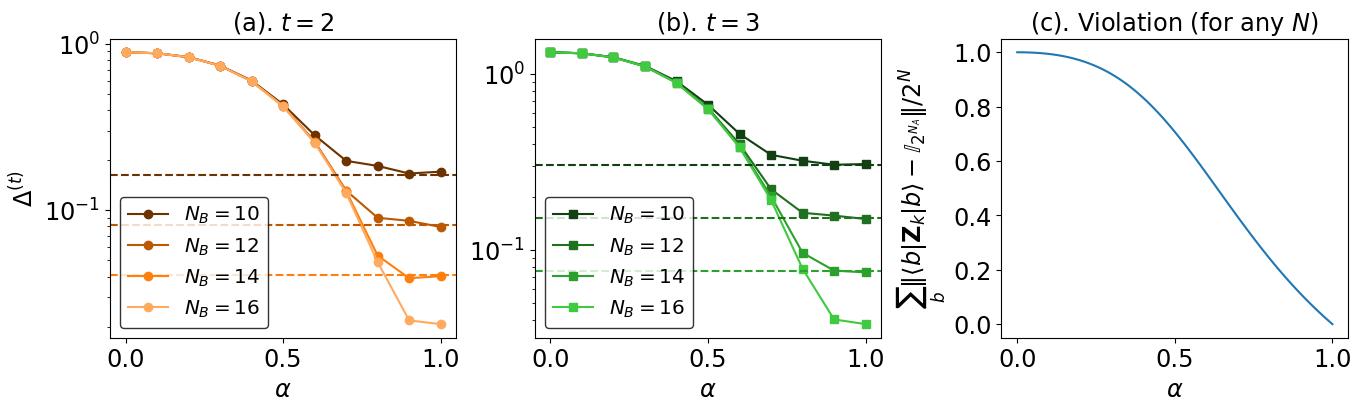}
\caption{\label{fig:z2trans} The figure illustrates the transition in the randomness of the projected ensemble when the initial generator states are generic states with $Z_2$ symmetry. Local $\sigma^x$ basis measurements are fixed for $N_B-1$ sites. The measurements on $N_B$-th site are performed in the eigenbasis of $\alpha\sigma^z+(1-\alpha)\sigma^x$. The trace distance between the moments of the Haar ensemble and the projected ensemble, $\overline{\Delta^{t}}$, is plotted against the parameter $\alpha$ for $t=2$ and $3$ for different system sizes. The dashed lines correspond to $\overline{\Delta^{t}}$ of that of Fig. \ref{fig:ZZ2plustrans}a. Note that the case of $t=1$ is trivial and stays nearly a constant for any $\alpha$, as it is independent of the measurement basis considered. } 
\end{figure*}
\begin{align}\label{G1}
\dfrac{\mathbf{\Delta}(\mathbf{Z}_{k}, \mathcal{B})}{2^{N_B}} &=\dfrac{1}{2^{N_B}}\sum_{b\in\mathcal{B}}\left\| \langle b|\mathbf{Z}_{k}|b\rangle -\mathbb{I}_{2^{N_A}} \right\|_{1}\nonumber\\
&=\dfrac{1}{2^{N_B}}\sum_{b_{1}\cdots b_{N_B-1}\in\{+, -\}^{N_B-1}}\sum_{b_{N_B}}\left\| (-1)^{k+\sum_{i=1}^{N_B-1}\text{sgn}(b_i)}\langle b_{N_B}| \sigma^x |b_{N_B}\rangle\otimes_{j=1}^{N_A} \sigma^{x}_{j} \right\|\nonumber\\
&=\dfrac{1}{2^{N_B}}\left\|\otimes_{j=1}^{N_A} \sigma^{x}_{j} \right\|\sum_{b_{1}\cdots b_{N_B-1}\in\{+, -\}^{N_B-1}}\sum_{b_{N_B}} |\langle b_{N_B}| \sigma^x |b_{N_B}\rangle|\nonumber\\
&=2^{N_A-1}\sum_{b_{N_B}} \left|\langle b_{N_B}| \sigma^x |b_{N_B}\rangle\right|,
\end{align}
where $\{|b_{N_B}\rangle\}$ denotes the eigenbasis of the operator $\alpha\sigma^z+(1-\alpha)\sigma^x$. From Eq. (\ref{G1}), we notice that the violation remains independent of the system size ($N$) and depends only on the parameter $\alpha$. Near $\alpha=0$, the violation stays nearly as constant ($\approx 1$) as depicted in Fig. \ref{fig:z2trans}c. Since $\dfrac{\mathbf{\Delta}(\mathbf{Z}_{k}, \mathcal{B})}{2^{N_B}}$ remains independent of $N$, the projected ensembles do not converge to the designs even in the limit of large $N$ when $\alpha$ is close to $0$. Hence, $\Delta^{(t)}$ remains nearly constant for all $N$. On the contrary, as $\alpha$ approaches $1$, the violation decays to zero, implying the convergence of the projected ensembles to the designs in the large $N$ limit. This may be understood as the transition of the projected ensemble from a localized distribution to a uniform distribution over the Hilbert space.


For completeness, we also demonstrate the emergence of state designs from the generator states that respect both translation and $Z_2$-symmetry. Since $T$ and $\Sigma$ commute, we can easily construct the states that respect both the symmetries by applying corresponding projectors consecutively on an initial state. Let $|\psi\rangle$ denote a Haar random state, then $|\phi\rangle =\mathbf{T}_{k_1}\mathbf{Z}_{k_2}|\psi\rangle /\sqrt{\mathcal{N}}$, where $\mathcal{N}=\sqrt{2N\left(\mathbf{T}_{k_1}\mathbf{Z}_{k_2}\right)}$, is a random vector that is simultaneously an eigenvector of both $T$ and $\Sigma$ with respective charges $k_1$ and $k_2$. Then, the condition to get quantum state designs from the projected ensembles would be $\langle b|\mathbf{T}_{k_1}\mathbf{Z}_{k_2}|b\rangle =0$ for all $|b\rangle\in\mathcal{B}$. The projective measurements in the standard computational basis mildly violate this condition, which results in the convergence towards state designs. 
The numerical results are shown in \ref{fig:ZZ2plustrans}c and \ref{fig:ZZ2plustrans}d. 
While the former is plotted by taking the computational basis measurements, the latter represents the results for the local measurements in $\sigma^x$ basis, i.e., $b\in\{+, -\}^{N_B}$.

\subsection{Reflection symmetry}
\begin{figure}
\begin{center}
\includegraphics[scale=0.45]{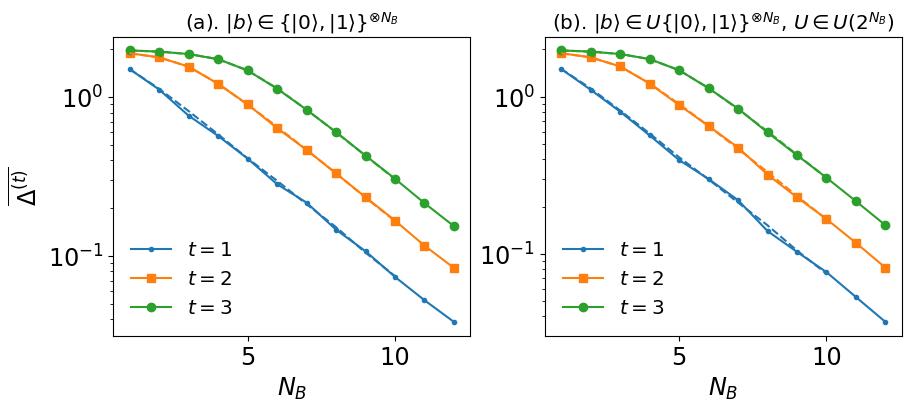}
\end{center}
\caption{\label{fig:ref} Illustration of $\overline{\Delta^{(t)}}$ vs $N_B$ for the random generator states with the reflection symmetry for the first three moments. We fix the charge $k=0$. In (a), the computational basis measurements are considered. In (b), the measurements are performed in a random product basis. We find no noticeable differences between these two. Moreover, the decay nearly coincides with the case when the generator states are Haar random. 
} 
\end{figure}

Here, we employ the projected ensemble framework for the generator states having reflection or mirror symmetry. In a system exhibiting reflection symmetry, the Hamiltonian remains invariant under swapping of mirrored sites around the center.
Let $R$ denote the reflection operation. Then $R$ generates a cyclic group of two elements, namely, the identity $\mathbb{I}$ and $R$ itself. In an $N$-qubit system, $R$ is defined as 
\begin{eqnarray}
R=
\begin{cases}
    S_{1, N}S_{2, N-1}....S_{N/2, N/2+1}      & \quad \textrm{if } N\textrm{ is even} \\
    S_{1, N}S_{2, N-1}....S_{(N-1)/2, (N+3)/2} & \quad \textrm{if } N\textrm{ is odd.}
\end{cases}
\end{eqnarray}
The reflection operator has $\{-1, 1\}$ as its eigenvalues. Hence, the total Hilbert space admits a decomposition into two invariant sectors. Then, the Hermitian operators $\mathbf{R}_{\pm}=\mathbb{I}\pm R$ project arbitrary states onto the respective subspaces. 
Engineering state designs from these generator states would require $\langle b|\mathbf{R}_{\pm}|b\rangle =\mathbb{I}_{2^{N_A}}$ for all $|b\rangle\in\mathcal{B}$. 

We consider a product basis $\mathcal{B}\equiv\{u|0\rangle, u|1\rangle\}^{\otimes N_B}$ for the measurements, where $u$ is an arbitrary unitary operator. 
For some $|b\rangle\in\mathcal{B}$, we have $\langle b|\mathbf{R}_{\pm}|b\rangle =\mathbb{I}_{2^{N_A}}\pm \langle b|R|b\rangle$. Since $\mathcal{B}$ is assumed to be a local product basis, we can write $|b\rangle =|b_{N_A+1}b_{N_A+2}...b_{N}\rangle$. To be explicit in calculating $\langle b|R|b\rangle$, let us consider $N_A=3$ and  $N_B=N-3\geq N_A$. Then,  
\begin{align}
 \langle b|R|b\rangle=\underbrace{|b_{N}b_{N-1}b_{N-2}\rangle\langle b_{N}b_{N-1}b_{N-2}|}_{\text{Supported on }A}  \underbrace{\delta_{b_4, b_{N-3}}  \delta_{b_5, b_{N-4}}...\delta_{b_{(N-1)/2}, b_{(N+3)/2}}}_{\text{palindrome condition}}.
\end{align}
Thus, $\langle b|R|b\rangle$ remains a non-zero operator only when the palindrome condition on the first $N_B-N_A$ bits of the string-$b$ is satisfied. If $N_B-N_A$ is even (odd), then we have a total of $n_{\text{even}}=2^{(N_B-N_A)/2}$ ($n_{\text{odd}}=2^{(N_B-N_A+1)/2}$) distinct palindromes. Then, the total number of violations of the condition for the given measurement basis will be $n_{\text{even}(\text{odd})}2^{N_A}$. For even-$N$, this number is exactly $2^{N/2}$. 
Moreover, the violation of the sufficient condition in the considered basis is
\begin{eqnarray}
\dfrac{\mathbf{\Delta}(\mathbf{R}_{k}, \mathcal{B})}{2^{N_B}}&=&\dfrac{1}{2^{N_B}}\sum_{|b\rangle\in\mathcal{B}}\| \langle b|R|b\rangle \|_{1}\nonumber\\
&=&
\begin{cases}
    2^{-N/2},& \text{if } N \text{ is even}\\
    2^{-(N-1)/2},              & \text{otherwise.}
\end{cases}
\end{eqnarray}
Since $\mathbf{\Delta}(\mathbf{R}_{k}, \mathcal{B})/2^{N_B}$ is exponentially suppressed as $N$ increases, the moments of the projected ensembles converge to the Haar moments. 
We plot the average $\Delta^{(t)}$ versus $N_B$ in Fig. \ref{fig:ref}. To illustrate, we consider the computational basis and random entangling basis for the measurement in  \ref{fig:ref}a and \ref{fig:ref}b, respectively. The later basis states can be obtained by the application of a fixed Haar random unitary supported over $B$ on the computational basis vectors. The results in both panels coincide with the case of Haar random generator states, implying the generation of higher-order state designs.

\subsection{Brief comment on the continuous symmetric cases}
So far, we have focused on the chaotic generator quantum states with discrete symmetry group structures and examined the emergence of state designs with respect to various measurement bases. Constructing projectors onto the subspaces that conserve the charge of the symmetry operators lies at the heart of our formalism. We highlight that our formalism equally applies to the cases involving continuous symmetries, provided one can construct projectors onto the charge-conserving subspaces. For instance, the total magnetization conservation $Q=\sum_{j}\sigma^z_{j}$ (or equivalently $U(1)$ symmetry) generates continuous symmetry group. In this case, the projector onto the charge conserving sector with the total charge $s$ can be written in the computational basis as 
\begin{eqnarray}
\mathbf{Q}_{s}= \sum_{f\in\{0, 1\}^{N}/|f|=s}|f\rangle\langle f|,\quad \text{where }|f|=\sum_{j=0}^{N}(-1)^{f_j+1}
\end{eqnarray}
Therefore, the sufficient condition for the emergence of state designs from the random eigenstates of $Q$ with the charge $s$ is $\langle b| \mathbf{Q}_{s} |b\rangle =\mathbb{I}_{2^{N_A}}$ for all $|b\rangle\in\mathcal{B}$. In this case, the computational basis measurements are unsuitable for extracting state designs from the projected ensembles. One can then extract the state designs from the projected ensembles by carefully choosing the measurement basis.

\section{Deep thermalization in a chaotic Hamiltonian}
\label{sising}

\begin{figure*}
\includegraphics[scale=0.44]{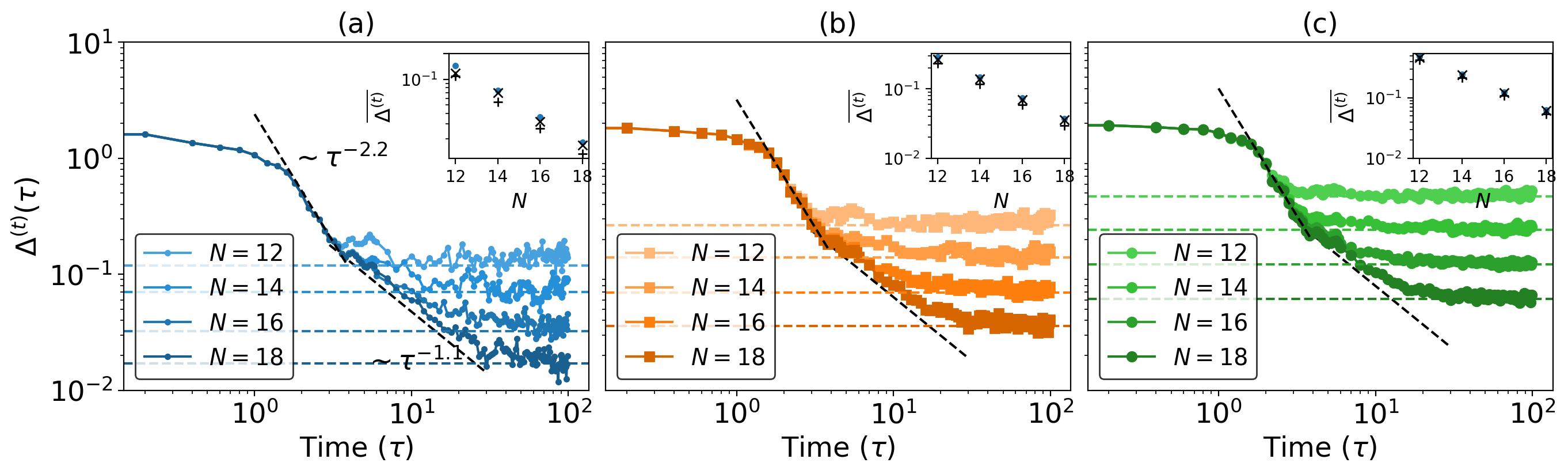}
\caption{\label{fig:ising-phy} Deep thermalization or dynamical generation of state designs (as characterized by $\Delta^{(t)}(\tau)$) of a quantum state $|\psi(\tau)\rangle$ evolved under the dynamics of a chaotic Ising Hamiltonian with periodic boundary conditions. Here, the evolution time is denoted with $\tau$, and the initial state is taken to be $|0\rangle^{\otimes N}$. The results are sequentially shown for the first three moments ($t=1, 2$, and $3$) in the panels along the row. The size of the projected ensembles $N_A$ is fixed at $3$. For each moment, the numerics are carried out for different system sizes varying from $N=12$ to $N=18$. Note the color scheme: 
darker to lighter shading of the colors represents larger to smaller system sizes. 
The dashed horizontal lines in all the panels represent the value attained on average by a typical random state, which is both translation symmetric (with momentum charge $k=0$) and a common eigenstate of the reflection operators about every site, with eigenvalue $1$ [see Appendix \ref{benchmark_latetim}]. From the numerical results, we observe a two-step relaxation of $\Delta^{(t)}(\tau)$ towards the saturation. (insets) {Comparison of the long-time averages of $\Delta^{t}(\tau)$ with the average trace distance when the generator states are random simultaneous eigenvectors of the translation and reflection operators (shown with $X $-markers) and random translation symmetric states (as shown with $+$-markers)}. The insets reveal that the long-time averages of $\Delta^{(t)}(\tau)$ coincide well with the former case while the latter case shows slight deviations. See the main text for more details. } 
\end{figure*}

In the preceding sections, our analysis focused primarily on obtaining state designs from Haar random states with symmetry. Here, we examine the dynamical generation of the state designs in a tilted field Ising chain with periodic boundary conditions (PBCs). The corresponding Hamiltonian is given by 
\begin{eqnarray}\label{ising}
H=\sum_{i=1}^{N}\sigma^{x}_{i}\sigma^{x}_{i+1}+h_x\sum_{i=1}^{N}\sigma^{x}_{i}+h_y\sum_{i=1}^{N}\sigma^{y}_{i}, 
\end{eqnarray}
where the PBCs correspond to $\sigma^{x, y, z}_{N+i}=\sigma^{x, y, z}_{i}$. {The periodicity, along with the homogeneity of the interactions and the magnetic fields, makes the system translation invariant, i.e., $[H, T]=0$.} In addition, the Hamiltonian is invariant under reflections about all the sites. For the parameters $h_x=(\sqrt{5}+1)/4$ and $h_y=(\sqrt{5}+5)/8$, the system is chaotic, and the ETH has been thoroughly verified in Ref. \cite{kim2014testing}. Furthermore, deep thermalization has also been investigated in this model with open boundary in Ref. \cite{cotler2023emergent} and \cite{bhore2023deep}, which does not respect the translation symmetry. Here, we explore this aspect for the Hamiltonian in Eq. (\ref{ising}) and contrast the results with those of the open boundary condition (OBC). To proceed, we consider a trivial product state $|\psi\rangle=|0\rangle^{\otimes N}$ as the initial state and evolve it under the  Hamiltonian. It is to be noted that the initial state is a common eigenstate of all the reflection operators (with eigenvalue $1$) and the translation operator (with eigenvalue $1$). Since the Hamiltonian commutes with $T$, the final state $|\psi(\tau)\rangle=e^{-i\tau H}|0\rangle^{\otimes N}$ will remain a common eigenstate of the translation and reflection operators with corresponding eigenvalues, where $\tau$ denotes the time of evolution. As this state evolves, we construct and examine its projected ensembles at various times. The computational basis is considered for the projective measurements on the subsystem-$B$. The corresponding numerical results are shown in Fig. \ref{fig:ising-phy}. 

Figures \ref{fig:ising-phy}a-\ref{fig:ising-phy}c demonstrate the decay of $\Delta^{(t)}$ as a function of evolution time ($\tau$) for the first three moments,  $t=1$, $2$, and $3$, respectively. We show this evolution for different system sizes $N$ by considering $N_A=3$ fixed and changing $N_B$.
The evolution of $\Delta^{(t)}(\tau)$ suggests a two-step relaxation towards the saturation. Initially, over a short period, the trace distance $\Delta^{(t)}(\tau)$ scales like $\sim \tau^{-2.2}$ for all three moments. For the largest considered system size, $N=18$, this power law behavior spans across the region $1\lesssim \tau\lesssim 4$. Moreover, this time scale appears to grow with $N_B$, which can be read off from the plots. Note that in Ref. \cite{cotler2023emergent}, the initial decay has been observed to be $\sim \tau^{-1.2}$ for the model with OBC [see also Appendix \ref{symbrok_comparision}]. In contrast, the present case exhibits a nearly doubled exponent of the power law behavior. Interestingly, similar behavior characterized by exponential decay has been observed in Ref. \cite{shrotriya2023nonlocality}, where the authors focused on dual unitary circuits and contrasted the case of PBCs with OBCs {while keeping the fields and interactions homogeneous}. The doubled exponent noticed in the current case can be attributed to the periodicity as well as the homogeneity of interactions and fields in the Hamiltonian and, in turn, the translation symmetry. To elucidate it further, in Appendix \ref{symbrok_comparision}, we contrast these dynamics with the symmetry-broken cases obtained through the modification of boundary conditions and the introduction of disorders. In the present case, the entanglement across the bi-partition $AB$ grows at a rate twice the rate in the case of OBC \cite{mishra2015protocol, pal2018entangling}. To elucidate the role of entanglement growth at initial times, we examine the trace distance for the simplest case $t=1$, i.e., $\Delta^{(1)}(\tau)$. Since the first moment of the projected ensemble is simply the reduced density matrix of $A$, the trace distance can be written as 
\begin{eqnarray}
\Delta^{(1)}(\tau)=\left\|\rho_{A}(\tau)-\dfrac{\mathbb{I}}{2^{N_A}}\right\|_{1}
=\sum_{j=0}^{2^{N_A}-1}\left| \gamma^{2}_{j}(\tau)-\dfrac{1}{2^{N_A}} \right|,    
\end{eqnarray}
where $\rho_{A}(t)=\text{Tr}_{B}(|\psi(t)\rangle\langle\psi(t)|)$ and $\{\gamma_{j}\}$ denote the Schmidt coefficients of $|\psi(t)\rangle$ across the bipartition. Above expression relates $\Delta^{(1)}(\tau)$ to the fluctuations of the Schmidt coefficients around $1/2^{N_A}$, corresponding to the maximally mixed value. At the time $\tau=0$, as the considered initial state is a product state, the only non-zero Schmidt coefficient is $\gamma_0(\tau=0)=1$. Hence, it is fair to say that $\Delta^{(1)}(0)$ is largely dominated by the decay of $\gamma_{0}$ during the early times. This regime usually witnesses inter-subsystem scrambling of the initial state mediated by the entanglement growth.  
Moreover, the initial decay appears across all three panels with the same power law scaling, indicating a similar early-time dependence of $\Delta^{t}(\tau)$ over $\gamma_{0}(\tau)$.

Beyond the initial power-law regime, we observe a crossover to an intermediate-time power-law decay regime with a smaller exponent, which is followed by saturation at a large time. The power-law exponent in the intermediate decay regime depends on the system size $N$, yielding a non-universal characteristic. For $N=18$, we obtain $\Delta^{(t)}\sim t^{-1.2}$, which is to be contrasted with the early-time decay exponent. At the onset of this scaling, the largest Schmidt coefficient $\gamma_0$ becomes comparable with the other $\gamma_j$s. The corresponding timescale is referred to as collision time \cite{vznidarivc2012subsystem}. At the collision time, $\gamma_0$ comes close to $\gamma_1$, the second largest Schmidt-coefficient. Hence, $\gamma_0$ does not solely determine the decay of $\Delta^{(1)}(\tau)$. The two-step relaxation of quantum systems has been recently studied in systems with two or more symmetries and also in quantum circuit models \cite{bensa2022two, vznidarivc2023two}. Finally, we benchmark the late time saturation values for each $N$ using the random matrix theory (RMT) predictions for the appropriate ensembles [see Appendix \ref{benchmark_latetim}]. We do this by plotting horizontal (dashed) lines corresponding to the RMT values. We notice that the saturation matches well with the corresponding RMT predictions. The same is also illustrated in the insets of Fig. \ref{fig:ising-phy}, where the long-time averages of $\Delta^{t}(\tau)$ for different system sizes are denoted with dots. Whereas the RMT values are shown with the marker-$X$. It is evident from the insets that the saturation values and the RMT values nearly coincide. On the other hand, it is interesting to note for the random translation symmetric states with no other symmetries present, the average trace distance $\overline{\Delta^{(t)}}$ deviates slightly from the former case. However, in the case of higher-order moments, these differences appear to become smaller. 
We intuitively expect the two-step relaxation observed in the present case to arise from its two competing features: initial faster decay and late time saturation above random matrix prediction.

\section{Summary and Discussion}\label{discussion}

In summary, we have investigated the role of symmetries on the choice of measurement basis for quantum state designs within the projected ensemble framework. 
By employing the tools from Lie groups and measure theory, we have evaluated the higher-order moments of the symmetry-restricted ensembles. Using these, we have derived a sufficient condition on the measurement basis for the emergence of higher-order state designs. The condition reads as follows: Given an arbitrary measurement basis $\mathcal{B}\equiv \{|b\rangle\}$ over a subsystem-$B$, for a typical $Q$-symmetric state $|\psi_{AB}\rangle\in\mathcal{E}^{k}_{\text{Q}}$ with a charge $k$, $\langle b|\mathbf{Q}_{k}|b\rangle=\mathbb{I}_{2^{N_A}}$ for all $|b\rangle\in\mathcal{B}$ implies that the projected ensembles approximate higher-order state designs. Moreover, the approximation improves exponentially with $N_B$, the bath size.
While the condition is sufficient for the emergence of state designs, the necessity of it remains an open question.
We demonstrate its versatility by considering measurement bases violating the condition mildly. Our analysis further suggests that a significant violation of the condition likely prevents the convergence of projected ensembles to the designs even in the limit of large $N_B$. To elucidate it, we have quantified the extent to which a basis violates the sufficient condition using the quantity $\sum_{|b\rangle\in\mathcal{B}}\|\langle b|\mathbf{Q}_{k}|b\rangle -\mathbb{I}_{2^{N_A}}\|_{1}/2^{N_B}$. This quantity allows us to identify the bases that violate the condition significantly. We have shown that the measurements in these bases result in a finite value for the trace distance $\Delta^{(t)}$ even when $N_B$ is large. Surprisingly, these include bases that do not adhere to the symmetry in the generator states.

To begin with, we have chosen random T-invariant states as the generator states. In constructing these states, we projected the Haar random states onto the momentum-conserving subspaces to reconcile both randomness and symmetry. This allows one to construct distinct ensembles of T-invariant states, each with a different momentum. Thanks to the inherent Haar measure in these ensembles, the states in them are uniformly distributed. Given a suitable measurement basis, Levy's lemma then ensures that the projected ensemble of a typical T-invariant state well approximates a state design. Equipped with this argument, we have numerically verified the emergence of designs for different measurement bases. These bases include the standard computational ($\sigma^z$) basis and the eigenbasis of $T_B$. While the former nearly satisfies the sufficient condition, the latter violates it significantly. Accordingly, the trace distance $\Delta^{(t)}$ decays exponentially with $N_B$ for the computational basis. Whereas, for the eigenbasis of $T_B$, $\Delta^{(t)}$ converges to a non-zero value. To further contextualize our results in a more physical setting, we have focused on deep thermalization in a tilted field Ising chain with PBCs, respecting the translation symmetry. The results indicate that the decay of the trace distance with time occurs in two steps. The initial decay is observed to be twice the rate of the case of the same model with OBCs. In the intermediate time, the decay trend is a system-dependent power law. Whereas, in a long time, the trace distance saturates to a value slightly larger than RMT prediction, a reminiscence of other symmetries.

Due to the generality of our formalism, the results can be extended to other discrete symmetries, and are expected to hold for continuous symmetries. In particular, generalization to other cyclic groups is straightforward. To illustrate this, we have examined the projected ensembles from the generator states with $Z_2$ and reflection symmetries. A crucial implication of our results is that the sufficient condition plays a pivotal role in identifying appropriate measurement bases, even when their suitability is not immediately apparent.

If one considers two or more non-commuting symmetries, they do not share common eigenstates. In such systems, the equilibrium states have been shown to approximate non-abelian thermal states \cite{majidy2023noncommuting, yunger2016microcanonical}. These states have been experimentally realized recently in Ref. \cite{kranzl2023experimental}. Hence, an extensive study of deep thermalization and emergent state designs in these systems is a topic of our immediate future investigation. Additionally, measurement-induced phase transitions (MIPTs) occur due to an interplay between the measurements and the dynamics in many-body chaotic systems \cite{skinner2019measurement}.
Our results can offer insights into the mechanism of the MIPTs whenever the dynamics and the measurements are chosen to respect symmetries \cite{majidy2023critical}.

\chapter{Conclusion and Outlook}\label{conclusionthesis}

\section{Summary of the thesis}
This thesis focuses on investigating information scrambling in quantum chaotic systems and the role of symmetries and other constraints, such as mixed phase space dynamics, in limiting it at different time scales \cite{varikuti2022out, dileep2024}. In particular, we have examined the scrambling dynamics using OTOCs in two different systems that follow complementary routes to the chaos, namely, the KAM route \cite{varikuti2022out} and the non-KAM route \cite{dileep2024}. The KAM route implies that the system under consideration gradually transits from an integrable to a chaotic regime as the parameters that characterize the chaos in the system slowly vary. On the other hand, when perturbed by weak time-dependent fields, the non-KAM systems offer a fast route to the chaos through an abrupt breaking of invariant phase-space tori \cite{dileep2024}. Understanding the dynamics of the non-KAM systems is paramount as they offer prime examples where noisy quantum simulators fail to simulate due to the exceptionally high dynamical sensitivity of these systems to the variations in the parameters \cite{chinni2022trotter}. 

In our study on scrambling in non-KAM systems, we considered the kicked harmonic oscillator model and examined the OTOCs at different time scales for both resonance and non-resonance scenarios. We have observed that the initial growth of the OTOCs is insensitive to the resonant to non-resonant transitions. On the other hand, the differences between the resonant and non-resonant cases become evident in the long-term regime. When the system is in resonance, the lang time growth of the OTOC is quadratic, which contrasts with the non-resonant case where the OTOC growth is largely suppressed. Moreover, the OTOC, when averaged over all the coherent states, shows indefinite linear growth for strongly non-resonant cases. We have further extended this result to the case of finite-dimensional systems. 

As an application of the non-KAM systems, we have demonstrated their usefulness in quantum metrology by considering quantum Fisher information (QFI) as a figure of metric. Under the non-resonance condition, QFI exhibits quadratic growth over time, as expected for regular systems. Conversely, QFI shows at most hexic ($\sim t^6$) scaling over time under resonance conditions. The different scalings of QFI follow directly from the mean energy growth in the considered state. Specifically, if the mean energy growth follows a power law with an exponent $\alpha$, QFI will show power-law scaling with the exponent of $2\alpha+2$. In addition to numerical results, we also provided analytical arguments concerning translation-invariant and quantum-resonant cases to demonstrate that the maximum possible growth of QFI is indeed hexic. 

In our work on scrambling in kicked coupled tops, we have examined the scrambling for fully chaotic and mixed phase space dynamics. In the former case, we have primarily examined the classical-quantum correspondence of the OTOCs. In the latter case, analyzing the correspondence principle is challenging because of the coexistence of the regular and chaotic islands. For a comprehensive understanding of the mixed-phase space scrambling, we invoked Percival’s conjecture, which characterizes the eigenstates of the Hamiltonian or time-evolution operator into regular and chaotic. We showed that the OTOCs in states localized near regular islands saturate to lesser values than those localized on the chaotic sea.  

This thesis further branches out from scrambling studies and examines the role of symmetries in constraining a measurement-induced phenomenon called \textit{emergence of quantum state designs} in many-body quantum systems within the projected ensemble framework \cite{varikuti2024unraveling}. This phenomenon has been called the deep thermalization of many-body systems, a notion closely related to the eigenstate thermalization hypothesis (ETH). However, it generalizes the notion of thermalization to the level of subsystem wave functions. In this framework, projectively measuring a chaotic quantum state yields an ensemble of projected states associated with the unmeasured subsystem (called a projected ensemble) approximating higher-order designs \cite{cotler2023emergent}. Our work \cite{varikuti2024unraveling} has established a sufficient condition on the measurement basis for the emergence of designs from the many-body systems with symmetries. Under this condition, the measurements in a given basis ensure the convergence of the projected ensembles to quantum state designs in the limit of large measured subsystem sizes. Moreover, we observed that the projected ensembles undergo \textit{randomness-transitions} for certain symmetries as the measurement basis slowly varies. 

\section{Future prospepts}
In this thesis, we have tackled fundamental questions concerning information scrambling and the effect of the system's route to chaos in the quantum regime on this process. Additionally, we have explored how symmetries can impose constraints on deep thermalization in many-body quantum systems. As is customary in any scientific inquiry, addressing one question often uncovers numerous other open problems, and our work is no exception. Here, we outline a few open directions that are relevant to the current thesis.

One central focus in quantum many-body physics is to simulate quantum systems on a quantum device and benchmark these simulations in the presence of hardware errors \cite{sahu2022quantum, trivedi2022quantum}. In the semiclassical limit, the assurance will be provided for the stability of simulation if the classical counterpart of the target Hamiltonian is KAM and the error that scales extensively with the number of particles is small enough \cite{bulchandani2022onset}. However, for systems that have no classical analogue, the quantum KAM theorem remains elusive. Nevertheless, recent progress in this direction has studied the stability analysis for the symmetries of quantum systems --- see, for instance, Refs. \cite{brandino2015glimmers, burgarth2021kolmogorov} and the references therein. Quantum simulations of systems whose classical limit is non-KAM can pose a significant challenge to experimental implementations. Any slight perturbation in the form of hardware noise can give rise to dynamics far from the target dynamics. Moreover, digital quantum simulation of these systems is also challenging due to structural changes near resonances even when the hardware error is negligible \cite{chinni2022trotter}. Hence, more careful methods must be devised to simulate this class of systems.

Another intriguing direction concerning mixed-phase space dynamics is to explore the properties of Floquet or eigenstates of the system situated at the boundaries of regular islands. Despite their regularity, these states exhibit properties reminiscent of chaotic states. Hence, a detailed study of these states is an immediate and interesting future direction. Additionally, there has been ongoing debate regarding whether chaos is necessary or sufficient for scrambling. While some studies have argued that chaos is sufficient for scrambling, the necessary conditions remain ambiguous. Consequently, investigations in this direction are of significant interest.

In our study of deep thermalization in quantum systems with symmetries, we have derived a sufficient condition for the emergence of higher-order state designs. While the necessity of this condition has been confirmed through numerical results, conducting a thorough analytical investigation of the same would be an intriguing next step. In addition, if one considers two or more non-commuting symmetries, they do not share common eigenstates. In such systems, the equilibrium states have been shown to approximate non-abelian thermal states \cite{majidy2023noncommuting, yunger2016microcanonical}. These states have been experimentally realized recently in Ref. \cite{kranzl2023experimental}. Hence, an extensive study of deep thermalization and emergent state designs in these systems is a topic of our immediate future investigation. Additionally, measurement-induced phase transitions (MIPTs) occur due to an interplay between the measurements and the dynamics in many-body chaotic systems \cite{skinner2019measurement}.
Our results can offer insights into the mechanism of the MIPTs whenever the dynamics and the measurements are chosen to respect symmetries \cite{majidy2023critical}.

\appendix
\chapter{Appendix for Chapter 3}
\section{Computation of $\overline{\langle D(\beta)\rangle }$}\label{appendix:a}
In this appendix, we find the average expectation value of an arbitrary displacement operator over the space of all pure states, which facilitates the derivations of the average state-OTOCs discussed in the main text. Consider a displacement operator $D(\beta),$ where $\beta\in\mathbb{C}$. We are interested in evaluating $\overline{\langle D(\beta)\rangle}=\int_{\psi}d\psi\langle\psi|D(\beta)|\psi\rangle$, where $d\psi$ represents the normalized uniform measure over the space of pure states in the infinite-dimensional Hilbert space. Since the Fock states of the quantum harmonic oscillator form an orthonormal basis and constitute continuous variable state 1-designs, it suffices to average $\langle D(\beta)\rangle$ over the set of Fock states.
\begin{eqnarray}\label{dispavg}
\overline{\langle D(\beta)\rangle }=\int_{\psi}d\psi\langle\psi|D(\beta)|\psi\rangle&\equiv&\lim_{N\rightarrow \infty}\dfrac{1}{N} \sum_{n=0}^{N-1}\langle n|D(\beta)|n\rangle
\end{eqnarray}
The elements of the displacement operator in the Fock state basis can be written in terms of the associated Laguerre polynomials. Specifically, the diagonal entries, $\langle n|D(\beta)|n\rangle$ for all $n\geq 0$, can be obtained as 
\begin{equation}
\langle n|D(\beta)|n\rangle=e^{-|\beta|^2/2}L^0_n(|\beta|^2), \text{ where }L^0_n\left(|\beta|^2\right)=\sum_{k=0}^{n}(-1)^k\dfrac{n!}{(n-k)!k!k!}|\beta|^{2k}. 
\end{equation}
After incorporating this into Eq. (\ref{dispavg}), we get
\begin{eqnarray}\label{dispavgfinal}
\overline{\langle D(\beta)\rangle }&=&\lim_{N\rightarrow\infty}\dfrac{1}{N}\sum_{n=0}^{N-1}e^{-|\beta|^2/2}L^0_{n}(|\beta|^2)\nonumber\\
&=&e^{|\beta|^2/2}\lim_{N\rightarrow\infty}\dfrac{L^1_{N-1}(|\beta|^2)}{N} \nonumber\\
&=&e^{|\beta|^2/2}\delta_{\Re\left(\beta\right), 0}\delta_{\Im\left(\beta\right), 0}\nonumber\\
&=&\delta_{\Re\left(\beta\right), 0}\delta_{\Im\left(\beta\right), 0}.
\end{eqnarray}
In the second equality, we used the recursive relation $\sum_{n=0}^{k}L^i_{n}(x)=L^{i+1}_{k}$. Therefore, the average vanishes unless $\beta$ is zero, in which case $D(\beta)$ reduces to an infinite dimensional identity operator.

\section{Small $K$ Approximation of OTOCs}
\label{appendix:b}
In this appendix we aim to achieve two objectives. First, we provide analytical arguments to demonstrate that for small $K$, the commutator function exhibits quadratic growth at resonances. We then give an explicit derivation for Eq. (\ref{irr-otoc}) --- the average state-OTOC for irrational $R$ and small $K$. Here, we take $\hbar=1$. The derivation involves approximating the commutator $[\hat{a}(t), \hat{a}]$ up to the terms of order $O(K^2)$. 
\begin{eqnarray*}
[\hat{a}(t), \hat{a}]e^{2\pi it/R}=1+\dfrac{iK}{\sqrt{2\omega}}\sum_{j=0}^{t-1}e^{2\pi ij/R}\left[\sin\hat{X}(j), \hat{a}^{\dagger}\right], \text{ where } X=\dfrac{\hat{a}+\hat{a}^{\dagger}}{\sqrt{2\omega}}.
\end{eqnarray*}
To approximate the commutator up to the second order in $K$, we only require to retain $\sin\hat{X}(j)$ to the zeroth and the first order terms in $K$ for all $j>0$. First, a single application of $\hat{U}_{\tau}$ on $\sin\hat{X}$ gives
\begin{eqnarray}
\sin\hat{X}(1)&=&e^{iK\cos\hat{X}}e^{i\omega\tau\hat{a}^{\dagger}\hat{a}}\left(\sin\hat{X}\right)e^{-i\omega\tau\hat{a}^{\dagger}\hat{a}}e^{-iK\cos\hat{X}}\nonumber\\
&=&\sin\left(\hat{X}_{2\pi /R}\right)+iK\left[\cos\hat{X}, \sin\left(\hat{X}_{2\pi/R}\right)\right]+O(K^2),
\end{eqnarray}
where $\hat{X}_{\theta}=(\hat{a}e^{-i\theta}+\hat{a}^{\dagger}e^{i\theta})/\sqrt{2\omega}$. After $j$-number of repeated applications, the time evolved operator $\sin\hat{X}(j)$ can be approximated as shown below:
\begin{equation}
\sin\hat{X}(j)= \sin\left(\hat{X}_{2\pi j/R}\right)+iK\sum_{n=0}^{j-1}\left[\cos\left(\hat{X}_{2\pi n/R}\right), \sin\left(\hat{X}_{2\pi j/R}\right)\right]+O(K^2).
\end{equation}
Consequently, the Heisenberg evolution of $\hat{a}$ can be approximated as
\begin{equation}
\hat{a}(t)\approx \hat{a}+\dfrac{iK}{\sqrt{2\omega}}\sum_{j=0}^{t-1}e^{ij\omega\tau}\left\{\sin\left(\hat{X}_{2\pi j/R}\right)+iK\sum_{n=0}^{j-1}\left[\cos\left(\hat{X}_{2\pi n/R}\right),\sin\left(\hat{X}_{2\pi j/R}\right) \right]\right\}.
\end{equation}
From Eq. (\ref{com}), it follows that
\begin{eqnarray}\label{com-app}
\left[\hat{a}(t), \hat{a}^{\dagger}\right]e^{2\pi it/R}&\approx& 1+\dfrac{K}{2\omega}\sum_{j=0}^{t-1}\cos\left(\hat{X}_{2\pi j/R}\right)\nonumber\\
&&\hspace{-1cm}-\dfrac{K^2}{\sqrt{2\omega}}\sum_{j=0}^{t-1}\sum_{n=0}^{j-1}e^{2\pi ij/R}\left[\left[\cos\left(\hat{X}_{2\pi n/R}\right),\sin\left(\hat{X}_{2\pi j/R}\right) \right], \hat{a}^{\dagger}\right]
\end{eqnarray}
This approximation is valid for any $R$ as long as $K$ is small. In the following, we will focus on two separate instances: $R=4$ and an irrational value of $R$.

\textit{Instance-1 :} 
In the main text, we argued that the resonances usually result in coherent summations, leading to the quadratic growth of OTOCs. To illustrate this further, we examine Eq. (\ref{com-app}) for the case of $R=4$. To simplify the computation, we consider $t=4s$ where $s$ is a non-negative integer. By doing so, we can observe that the first summation on the right-hand side exhibits a clear linear dependence on $t$.
\begin{eqnarray}
\sum_{j=0}^{t-1}\cos\left(\hat{X}_{2\pi j/R}\right)= \dfrac{t}{4}\sum_{j=0}^{3}\cos\left(\hat{X}_{2\pi j/R}\right). 
\end{eqnarray}
on the other hand, the third term that contains the double summations has an implicit quadratic time dependence. Therefore, in the limit of weak perturbations, the OTOCs will grow quadratically at resonances. The proof of the argument is now complete.

\textit{Instance-2: }
We now consider $R$ to be an irrational number. The approximate OTOC in the limit of small $K$ can be written as
\begin{eqnarray}\label{irr-comm-comp}
C_{\hat{a}\hat{a}^{\dagger}}(t)&\approx & 1+\dfrac{K^2}{4\omega^2}\sum_{j, j'=0}^{t-1}\cos\left(\hat{X}_{2\pi j/R}\right)\cos\left(\hat{X}_{2\pi j'/R}\right)\nonumber\\
&&\hspace{-1cm}-\dfrac{K^2}{\sqrt{2\omega}}\sum_{j=0}^{t-1}\sum_{n=0}^{j-1}\left\{e^{2\pi ij/R}\left[\left[\cos\left(\hat{X}_{2\pi n/R}\right),\sin\left(\hat{X}_{2\pi j/R}\right) \right], \hat{a}^{\dagger}\right]+\textbf{h.c.}\right\}
\end{eqnarray}
Upon performing the average over the pure states, the third term vanishes for all $j<t$ and $n<j$. The second term also vanishes unless $j= j'$. Therefore, we finally obtain
\begin{eqnarray}
\overline{C_{|\psi\rangle,\thinspace \hat{a}\hat{a}^{\dagger}}}(t)&\approx &1+\dfrac{K^2}{4\omega^2}\sum_{j=0}^{t-1}\overline{\langle\psi| \cos^2\left(\hat{X}_{2\pi j/R}\right)|\psi\rangle}\nonumber\\
&=&1+\dfrac{K^2t}{8\omega^2} \text{ for } K\ll 1.
\end{eqnarray}
This concludes the derivation of Eq. (\ref{irr-otoc}) in the main text. 

\section{Linear growth of OTOCs in finite-dimensional integrable quantum systems}
\label{appendix:c}
It is well-known that a generic quantum system, whose classical limit is integrable, possesses uncorrelated eigenspectrum (or eigenphases if the system is time-periodic). If $V$ denotes the time evolution of a typical integrable system, then its eigenphases can be viewed as complex phases drawn uniformly at random from a complex unit circle. Without further loss of generality, we assume $V$ to be a random diagonal unitary acting on a $d$-dimensional Hilbert space. The elements of $V$ can be characterized as follows:
\begin{eqnarray}
V_{i j}=
\begin{cases}
 e^{2\pi i\phi}, \phi\in \left[-0.5, 0.5\right]& \text{if }  i=j\\
    0,              & \text{otherwise},
\end{cases}
\end{eqnarray}
where $\phi$ is a uniform random variable. After applying the perturbation, we assume the following Floquet form for the total system evolution: 
\begin{eqnarray}
U=Ve^{-i\varepsilon H} \text{ and }\varepsilon\ll 1, 
\end{eqnarray}
where $H$ is the perturbation and $\varepsilon$ is the kicking strength. By choosing the uncorrelated eigenphases, we already invoked the first condition that led to the linear growth of the OTOCs under the non-resonance condition. We now choose the initial operators that are conserved up to a phase under the action of $V$ --- let $A$ be an operator that satisfies $V^{\dagger}AV=e^{-i\theta}A$. For simplicity, we take $\theta=0$. We are now interested in computing the quantity given by 
\begin{eqnarray}
C(t)=\dfrac{1}{d}\text{Tr}\left[\left[A(t), A\right]^{\dagger}\left[A(t), A\right]\right], 
\end{eqnarray}
where $A(t)=U^{\dagger t}AU^t$. Ignoring all the higher order terms in $\varepsilon$, the commutator, $\left[A(t), A\right]$, can be written as 
\begin{eqnarray}
\left[A(t), A\right]\approx i\varepsilon\sum_{j=0}^{t-1}\left[\left[V^{\dagger j}HV^j, A\right], A\right].
\end{eqnarray}
It then follows that
\begin{eqnarray}
\text{Tr}\left[\left[A(t), A\right]^{\dagger}\left[A(t), A\right]\right] &\approx&\varepsilon^2\sum_{j, j'=0}^{t-1}\left[\left[V^{\dagger j}HV^j, A\right], A\right]^\dagger\left[\left[V^{\dagger j'}HV^{j'}, A\right], A\right]\nonumber\\
&=&\varepsilon^2\sum_{j, j'=0}^{t-1}\text{Tr}\left[6H_jA^2H_{j'}A^2-4H_jA^3H_{j'}A-4H_jAH_{j'}A^3\right.\nonumber\\
&&\left.+H_jH_{j'}A^4+H_{j'}H_jA^4\right]\nonumber\\
&=&\varepsilon^2\sum_{j, j'=0}^{t-1}\text{Tr}\left[6H_{j-j'}A^2HA^2-4H_{j-j'}A^3HA-4H_{j-j'}AHA^3\right.\nonumber\\
&&\left.+H_{j-j'}HA^4+HH_{j-j'}A^4\right]
\end{eqnarray}
where $H_{j-j'}=V^{\dagger j-j'}HV^{j-j'}$. Now, the process of averaging over the random diagonal unitaries gives
\begin{eqnarray}\label{C6}
\int_{V}dVC(t)&=&\dfrac{1}{d} \int_{V}\text{Tr}\left[\left[A(t), A\right]^{\dagger}\left[A(t), A\right]\right]dV\nonumber\\
&=&\dfrac{\varepsilon^2}{d}\sum_{j, j'=0}^{t-1}\int_{V}dV\text{Tr}\left[6H_{j-j'}A^2HA^2-4H_{j-j'}A^3HA-4H_{j-j'}AHA^3\right.\nonumber\\
&&\left.+H_{j-j'}HA^4+HH_{j-j'}A^4\right]\nonumber\\
&=&\dfrac{\varepsilon^2}{d}t\text{Tr}\left[6HA^2HA^2+2H^2A^4-8HAHA^3\right]\nonumber\\
&&+\dfrac{\varepsilon^2}{d}\sum_{j\neq j'}\int_{V}dV\text{Tr}\left[6H_{j-j'}A^2HA^2-4H_{j-j'}A^3HA-4H_{j-j'}AHA^3\right.\nonumber\\
&&\left.+H_{j-j'}HA^4+HH_{j-j'}A^4\right].
\end{eqnarray}
The integrals over the random diagonal unitaries can be solved as follows:
\begin{eqnarray}
\int_{V}dV \text{ Tr}\left[H_{j-j'}A^2HA^2\right]&=&\text{ Tr}\left[\text{diag}\left(H\right)\right]\text{Tr}\left[HA^4\right],\nonumber\\
\int_{V}dV\text{ Tr}\left[H_{j-j'}A^3HA\right]&=&\text{ Tr}\left[\text{diag}\left(H\right)\right]\text{Tr}\left[HA^4\right],\nonumber\\
\int_{V}dV\text{ Tr}\left[H_{j-j'}HA^4\right]&=&\text{ Tr}\left[\text{diag}\left(H\right)\right]\text{Tr}\left[HA^4\right], \nonumber\\
\int_{V}dV\text{ Tr}\left[HH_{j-j'}A^4\right]&=&\text{ Tr}\left[\text{diag}\left(H\right)\right]\text{Tr}\left[HA^4\right].\nonumber\\
\end{eqnarray}
All the above integrals yield identical results. As a result, the second term in Eq. (\ref{C6}) involving the double summation vanises. Therefore, we finally obtain 
\begin{eqnarray}
\int_{V}dVC(t)=\dfrac{\varepsilon^2 t}{d}\text{ Tr}\left[6HA^2HA^2+2H^2A^4-8HAHA^3\right]
\end{eqnarray}
In conclusion, we have showed analytically that the OTOC grows linearly with time, given that the aforementioned conditions regarding the initial operators and eigenphases are satisfied.

\chapter{Appendix for Chapter 5}

\section{Derivation of Eqs. (\ref{23}) and (\ref{24})}
\label{AppB}
We first compute the two point correlator $\tr\left( U^{\dagger}A^2UB^2 \right)$, where $U=\oplus_{\substack{F_{z}}}U_{F_z}$ and $F_z$ varies from $-2J$ to $+2J$. The unitaries $U_{F_z}$s have non-trivial support over the invariant subspaces labeled by the quantum number $F_z$. For the RMT value, we choose these unitaries randomly from the circular orthogonal ensemble. We now write 
\begin{eqnarray}
\tr\left( U^{\dagger}A^2UB^2 \right)&=&\tr\left[\left(\underset{F_z}{\oplus} U_{F_z}\right)^{\dagger}A^2\left(\underset{F'_z}{\oplus}U_{F'_z}\right)B^2 \right] \nonumber\\
&=&\sum_{F_z, F'_z=-2J}^{2J}\tr\left( U^{\dagger}_{F_z}A^2U_{F'_z}B^2 \right)\nonumber\\
&=&\sum_{F_z=-2J}^{2J}\tr\left( U^{\dagger}_{F_z}A^2U_{F_z} B^2\right). 
\end{eqnarray}
The third equality follows directly from Eq. (\ref{inter}). We now compute the average two-point correlator over all subspace unitaries $U_{F_z}\in \text{CUE}(2J+1-|F_z|)$. 
\begin{eqnarray}
\overline{C_{2}}=\dfrac{1}{2J+1}\sum_{F_z=-2J}^{2J}\tr\left( \overline{U^{\dagger}_{F_z}A^2U_{F_z}} B^2 \right). 
\end{eqnarray}
For an arbitrary matrix $P$ and a random $W\in \text{COE}$, the operator $\overline{W^{\dagger}PW}$ can be evaluated as \cite{brouwer1996diagrammatic}
\begin{eqnarray}
\overline{W^{\dagger}PW}=\dfrac{1}{d+1}\left( P^{T}+\tr(P)\mathbb{I} \right),  
\end{eqnarray}
where $d$ is the dimension of the Hilbert space over which the unitary $W$ acts, and $P^T$ denotes the transpose matrix of $P$. In the above expression, all the operators have been assumed to have equal dimensions. We now replace $P$ with $A^2$ and consider the average over the subspace unitaries belonging to the COE. In this case, the above equation gets modified as follows:
\begin{eqnarray}
\overline{U^{\dagger}_{F_z}A^2U_{F_z}} = \dfrac{1}{2J+2-|F_z|}\left[ \tr\left( [A^2]_{F_z}\right)\mathbb{I}_{F_z}+[A^2]^{T}_{F_z} \right], 
\end{eqnarray}
where $[A^2]_{F_z}$ denotes the projection of $A^2$ onto the invariant subspace with the quantum number $F_z$ and $\mathbb{I}_{F_z}$ is the identity operator supported non-trivially over the same subspace. It then follows that 
\begin{eqnarray}
\overline{C_{2}}=\dfrac{1}{(2J+1)^2}\sum_{F_z}\dfrac{\tr([A^2]_{F_z}^TB^2)+\tr([A^2]_{F_z})\tr([B^2]_{F_z})}{2J+2-|F_z|}. 
\end{eqnarray}
We now show that the four-point correlator, on the other hand, vanishes for the $\hat{I}_x\hat{J}_x$ OTOC. To do so, we consider 
\begin{align}
\overline{C_4}=&\dfrac{1}{(2J+1)^2}\tr\left[\overline{\left(\underset{F^{a}_z}{\oplus}U_{F^{a}_z}\right)^{\dagger}A\left(\underset{F^{b}_z}{\oplus}U_{F^{b}_z}\right)B\left(\underset{F^{c}_z}{\oplus}U_{F^{c}_z}\right)^{\dagger}A\left(\underset{F^{d}_z}{\oplus}U_{F^{d}_z}\right)B}\right] \nonumber\\
=&\dfrac{1}{(2J+1)^2}\sum_{F^{a}_zF^{b}_zF^{c}_zF^{d}_z}\tr\left(\overline{U^{\dagger}_{F^{a}_z}AU_{F^{b}_z}BU^{\dagger}_{F^{c}_z}AU_{F^{d}_z}}B\right).
\end{align}
Equation (\ref{inter}) implies that the four-point correlator vanishes unless $F^{a}_{z}=F^{b}_{z}\pm 1$, and $F^{c}_{z}=F^{d}_{z}\pm 1$. Hence, 
\begin{eqnarray}
\overline{C_{4}}=\dfrac{1}{(2J+1)^2}\sum_{F^{a}_z, F^{c}_{z}}\tr\left(\overline{U^{\dagger}_{F^{a}_z}AU_{F^{a}_z\pm 1}BU^{\dagger}_{F^{c}_z}AU_{F^{c}_z\pm 1}}B  \right).
\end{eqnarray}
Interestingly, the above expression can be further simplified by noting that the terms inside the summation vanish unless $F^{a}_{z}=F^{c}_{z}$. It then follows that 
\begin{eqnarray}
\overline{C_{4}}=\dfrac{1}{(2J+1)^2}\sum_{F_{z}=-2J}^{2J}\tr\left(\overline{U^{\dagger}_{F_z}AU_{F_z\pm 1}BU^{\dagger}_{F_z}AU_{F\pm 1}}B\right).    
\end{eqnarray}
Here, we replaced $F^{a}_z$ with $F_z$. A few remarks are in order. Firstly, the subspace operators $U_{F_z}$ and $U_{F_z\pm 1}$ act non-trivially on the invariant subspaces with the total magnetization $F_z$ and $F_z\pm 1$ respectively, hence independent. Hence, the COE average over the unitaries $U_{F_z}$ and $U_{F_z\pm 1}$ should be performed independently. See Ref. \cite{brouwer1996diagrammatic} for more details concerning the averages over COE matrices. Upon performing the average, we get
\begin{eqnarray}
\overline{C_{4}}=0.    
\end{eqnarray}
Therefore, $C_{\text{RMT}}=\overline{C_{2}}$. 
\chapter{Appendix for Chapter 6}
\section{Construction of random T-invariant unitaries}
\label{polar}
The QR decomposition is traditionally used to generate Haar random unitary operators from the initial random Gaussian matrices. However, QR decomposition can not produce uniformly distributed unitaries from the subgroups such as $U_{\text{TI}}(d^N)$ as the decomposition does not preserve the symmetries of the initial operator. Here, we use polar decomposition as an alternative to the QR decomposition to generate random unitary operators. For a given initial operator $Z$, the polar decomposition is given by $Z=UP$, where $P$ is a positive semi-definite operator, $P=\sqrt{Z^{\dagger}Z}$. If $Z$ is a full rank matrix, $U=Z(Z^{\dagger}Z)^{-1/2}$ can be  uniquely computed. If $Z$ is a complex Gaussian matrix with mean $\mu=0$ and standard deviation $\sigma=1$, the polar decomposition will yield the ensemble of unitaries whose moments match those of the Haar random unitaries. To see this, consider $A$, an arbitrary operator acting on $t$-replicas of the same Hilbert space $\mathcal{H}^{d}$. Then, we have 
\begin{eqnarray}
\langle U^{\dagger\otimes t}AU^{\otimes t}\rangle &=& \int_{Z}d\mu(Z) \left(Z(Z^{\dagger}Z)^{-1/2}\right)^{\dagger\otimes t}A\left(Z(Z^{\dagger}Z)^{-1/2}\right)^{\otimes t}\nonumber\\
&=&\int_{Z}d\mu(Z) \left((Z^{\dagger}Z)^{-1/2} Z^{\dagger}\right)^{\otimes t}A\left(Z(Z^{\dagger}Z)^{-1/2}\right)^{\otimes t},  
\end{eqnarray}
where $d\mu(Z)$ denotes the invariant measure over the Ginebre ensemble. 
Since the Ginibre ensemble is unitarily invariant, we replace $Z$ with $VZ$ for some $V\in U(d^N)$ and perform Haar integral over $V$. This action keeps the overall integral in the above equation invariant. 
\begin{eqnarray}
\langle U^{\dagger\otimes t}AU^{\otimes t}\rangle &=&\int_{Z}d\mu(Z)\int_{V\in U(d^N)}d\mu(V) \left((Z^{\dagger}Z)^{-1/2} Z^{\dagger}V^{\dagger}\right)^{\otimes t}A\left(VZ(Z^{\dagger}Z)^{-1/2}\right)^{\otimes t}\nonumber\\
&=&\int_{Z}d\mu(Z)\left((Z^{\dagger}Z)^{-1/2} Z^{\dagger}\right)^{\otimes t}\left(\int_{V}d\mu(V)V^{\dagger\otimes t}AV^{\otimes t}\right)\left(Z(Z^{\dagger}Z)^{-1/2}\right)^{\otimes t}. \nonumber\\
\end{eqnarray}
By the Schur-Weyl duality, $\int_{V\in U(d^N)}d\mu(V)V^{\dagger\otimes t}AV^{\otimes t}=\sum_{i=1}^{t!}c_i \pi_i$, where $\{\pi_i\}$s are permutation operators acting on $t$-replicas of the Hilbert space. It then follows that 
\begin{eqnarray}
\langle U^{\dagger\otimes t}AU^{\otimes t}\rangle =\left(\int_{V}d\mu(V)V^{\dagger\otimes t}AV^{\otimes t}\right)\left(\int_{Z}d\mu(Z)\left((Z^{\dagger}Z)^{-1/2} Z^{\dagger}\right)^{\otimes t}\left(Z(Z^{\dagger}Z)^{-1/2}\right)^{\otimes t}\right). \nonumber\\ 
\end{eqnarray}
Since $Z(Z^{\dagger}Z)^{-1/2}=U$ is a unitary operator, the integrand of the second integral becomes the Identity operator. Therefore, 
\begin{eqnarray}
\langle U^{\dagger\otimes t}AU^{\otimes t}\rangle= \int_{V\in U(d^N)}d\mu(V)V^{\dagger\otimes t}AV^{\otimes t}. 
\end{eqnarray}
This equation implies that the moments of the ensemble of unitaries from the polar decomposition are identical to those of the Haar ensemble of unitaries. 

We now show that if the initial operator commutes with an arbitrary unitary operator, then the resulting unitary from the polar decomposition necessarily commutes with the same. For our purpose, we take the commuting unitary to be $T$, the translation operator. Let $Z$ be randomly drawn from the Ginibre ensemble. Then, the operator $Z'=\sum_{j=0}^{N-1}T^{\dagger j}ZT^{j}$ is translation invariant as $T^{\dagger}Z'T=Z'$. Moreover, $P'=\sqrt{Z^{'\dagger}Z'}$ is also $T$-invariant whenever $Z'$ is a full rank matrix. Consequently, the resulting unitary $U'$ commutes with $T$. One can also show that the distribution of $Z'$ is invariant under the action of elements of $U_{\text{TI}}(d^N)$. Therefore, the resulting ensemble of unitaries has the same moments as those of the unitary subgroup $U_{\text{TI}}(d^N)$.

\section{Partial trace of $T^j$}
\label{ptrace}
This appendix shows that the particle trace of $T^j$ results in some permutation operator whenever $N_A\geq \gcd(N, j)$. We first consider $j=1$. Then, $\text{Tr}_{B}(T)$ is still a translation operator, acting on the subsystem-$A$ as shown in the following:
\begin{eqnarray}
\text{Tr}_{B}\left(T\right)&=&\sum_{b\in\{0, 1\}^{N_B}}\langle b|T|b\rangle    \nonumber\\
&=&\sum_{\substack{a \in \{0, 1\}^{N_A} \\ a' \in \{0, 1\}^{N_A} \\ b \in \{0, 1\}^{N_B}}} 
\langle a_1...a_{N_A}b_1...b_{N_B}|T|a'_1...a'_{N_A}b_1...b_{N_B}\rangle |a_1...a_{N_A}\rangle\langle a'_1...a'_{N_A}|\nonumber\\
&=&\sum_{\substack{a \in \{0, 1\}^{N_A} \\ a' \in \{0, 1\}^{N_A} \\ b \in \{0, 1\}^{N_B}}} 
\left(\langle a_1...a_{N_A}b_1...b_{N_B}|b_{N_B}a'_1...a'_{N_A}b_1...b_{N_B-1}\rangle\right) |a_1...a_{N_A}\rangle\langle a'_1...a'_{N_A}|\nonumber\\
&=&\sum_{\substack{a \in \{0, 1\}^{N_A} \\ a' \in \{0, 1\}^{N_A} \\ b \in \{0, 1\}^{N_B}}} 
\delta_{a_1, b_{N_B}}\delta_{a_2, a'_1}\delta_{a_3, a'_2}...\delta_{a_{N_A}, a'_{N_A-1}}\delta_{b_1, a'_{N_A}}\delta_{b_2, b_1}...\delta_{b_{N_B}, b_{N_B-1}}\nonumber\\
&&\hspace{8cm}\left(|a_1...a_{N_A}\rangle\langle a'_1...a'_{N_A}|\right)\nonumber\\
&=&\sum_{a\in\{0, 1\}^{N_A}}|a_1...a_{N_A}\rangle\langle a_2...a_{N_A}a_1|\nonumber\\
&=&T_A.
\end{eqnarray}
In the fourth equality, on the right-hand side, the product of Kronecker deltas results in the following chains of equalities:
\begin{align}
a_1=b_{N_B}=b_{N_B-1}=...=b_{2}=b_{1}=a'_{N_A} \quad\text{and}\quad a_{i}=a'_{i-1}\quad\text{for all }N\leq i\leq 2.
\end{align}
The first chain contains equalities of all the bits of the $b$-strings. Therefore, the summation over $b\in\{0, 1\}^{N_B}$ disappears. Besides, the sum involving $a'$ strings disappears due to the remaining equalities, finally leading to the translation operator on $A$. 

For any $j>1$, the partial trace of $T^j$ also forms the product of Kronecker deltas. Every equality chain starting with $a_i$ of the string $a$ must end with $a'_j$ of $a'$ for some $i, j \leq N{A}$. We denote this chain as $[a_i-a'_j]$. Constructing a sequence of distinct chains $[a_i-a'_j][a_j-a'_k]...[a_l-a'_i]$, where subscripts of the last and first elements of consecutive chains match, forms a complete cycle if it covers all bits of $a$, $a'$, and $b$. Then, for any $j>1$, we observe the following implications:
\begin{enumerate}[i]
    \item If $N_A < \gcd(N, j)$, chains starting with $a_i$s always end with $a'_i$s, preventing a complete cycle. However, there will be exactly $\gcd(N, j)$ number of equality chains, each forming an incomplete cycle. Since the endpoints of the chains share the same subscripts, the resulting operator is a constant multiple of $\mathbb{I}_{2^{N_A}}$. 
    
    \item The second possibility is that all the chains can be stacked together to form a full cycle, mapping all $a'_i$s to distinct $a_j$s. Consequently, the resulting operator becomes a permutation operator on subsystem-$A$. A complete cycle can only be formed if $N_A\geq \gcd(N, j)$ (See also Lemma 3.8 in Ref. \cite{sugimoto2023eigenstate}). 
\end{enumerate}       
For $N_A>1$, a full cycle will always form if $N$ assumes a prime number as $N_A\geq \gcd(N, j)=1$ for any $j$.

\section{Applicability of Levy's lemma for the moments of the projected ensembles with typical generator states}
{Equation (6.26) of the main text tells us that the $t$-th moment operator of a projected ensemble, when averaged over many random symmetric generator states, equates to the $t$-th Haar moment whenever the sufficient condition holds. Then, Levy's lemma can be used to argue that the moment operator for a typical generator state also approximates the Haar moments, and the distance between them shrinks exponentially with the Hilbert space dimension. In this appendix, we provide some basic details concerning Levy's lemma and its applicability in the context of the projected ensembles.} 

{\textbf{Definition} (Lipschitz continuous functions). A function $f: X\rightarrow Y$ is Lipschitz continuous with Lipschitz constant $\eta$, if for any $x_1, x_2\in X$, it holds that 
\begin{eqnarray}
d_y(f(x_1), f(x_2))\leq \eta d_x(x_1, x_2),     
\end{eqnarray}
where $d_x$ and $d_y$ indicate the distance metrics associated with the spaces $X$ and $Y$,  respectively. The Lipschitz continuity is a stronger form of the uniform continuity of $f$ \cite{o2006metric}, and $\eta$ upper bounds the slope of $f$ in $X$ \cite{milman1986asymptotic, ledoux2001concentration}}. 

\textbf{Levy's lemma} \cite{milman1986asymptotic, ledoux2001concentration}
{Let $f: \mathbb{S}^{d-1} \rightarrow \mathbb{R}$ be a Lipschitz function defined over a $(d-1)$-sphere $\mathbb{S}^{d-1}$, equipped with a natural Haar measure. Suppose a point $x \in \mathbb{S}^{d-1}$ is drawn uniformly at random from $\mathbb{S}^{d-1}$. Then, for any $\varepsilon > 0$, the following concentration inequality holds:
\begin{eqnarray}
\text{Pr}\left[ \left|f(x)-\mathbb{E}_{x\in \mathbb{S}^{d-1}}(f(x))\right|\geq \varepsilon \right]\leq 2\exp\left\{ \dfrac{-d\varepsilon^2}{9\pi^3\eta^2} \right\}, 
\end{eqnarray}
where $\eta$ is the Lipschitz constant of $f$ and $c$ is a positive constant.}

{Proof of Levy's lemma can be found in Ref. \cite{gerken2013measure}. Levy's lemma guarantees that the value of a Lipschitz continuous function at a typical $x\in \mathbb{S}^{d-1}$ is always close to its mean value, as given by $\mathbb{E}_{x\in \mathbb{S}^{d-1}}(f(x))$. The difference between the mean and a typical value is exponentially suppressed with the Hilbert space dimension}.

{Let us now consider the moments of a projected ensemble for an arbitrary generator state ($|\phi\rangle$) with symmetry: 
\begin{eqnarray}\label{levylemma}
\mathcal{M}^{t}(|\phi\rangle)=\sum_{|b\rangle\in\mathcal{B}} \dfrac{\left[\langle b|\phi\rangle\langle\phi|b\rangle\right]^{\otimes t}}{\left(\langle\phi|b\rangle\langle b|\phi\rangle\right)^{t-1}} .   
\end{eqnarray}
For Haar random generator state, $[\mathcal{M}^{t}(|\phi\rangle)]_{ij}$, the elements of the moment operator, have been shown to be Lipschitz continuous functions of $|\phi\rangle$ with the Lipschitz constant $\eta\leq 2(2t-1)$. Then, Levy's lemma, as given in Eq. (\ref{levylemma}), implies that the differences between the elements of the projected ensemble moments and their respective means are exponentially suppressed with the total Hilbert space dimension $2^{N}$. A detailed explanation for the same can be found in Ref. \cite{cotler2023emergent}.
Since the ensembles of uniform random states with symmetry are subsets of the Haar ensemble, the corresponding $[\mathcal{M}^{t}(|\phi\rangle)]_{ij}$s remain Lipschitz continuous with the same Lipschitz constant $\eta$ (or smaller than $\eta$)}.

\section{Characterization of the distribution of the projected ensembles }
\label{app-moments}
Equation (\ref{sing}) provides an expression for the moments of the projected ensembles with respect to a given measurement basis, averaged over the initial generator states from the symmetry-restricted ensembles. These moments characterize the underlying distribution of the ensembles to some extent. The resultant distribution depends both on the initial symmetry and the measurement basis considered. When the sufficient condition is satisfied, the resultant distribution displays the moments of the Haar ensemble. It's interesting to examine the distribution of projected ensembles when the sufficient condition does not hold. In this appendix, we provide closed-form expressions for the moments of the projected ensembles. To illustrate, we take the generator states from $\mathbb{Z}_2$-symmetric ensembles and perform the measurements in the local product basis given by $\mathcal{B}\equiv\{+, -\}^{N_B}$. 
For the $\mathbb{Z}_2$-symmtric case, the resulting distribution of the projected ensembles from the measurements in $\sigma^x$ basis indeed displays the moments, which are linear combinations of the symmetry-restricted moments corresponding to both parities. In the following, we shall show this analytically.

{To begin with, let us recall that the moments of the projected ensembles averaged over the ensemble of $\mathbb{Z}_2$-symmetric states are given by 
\begin{eqnarray}\label{moment}
\mathcal{M}^{t}_{\mathbb{Z}_{2}}&=&\mathbb{E}_{|\phi\rangle\in\mathcal{E}^{k}_{\mathbb{Z}_{2}}}\left(\sum_{|b\rangle\in\mathcal{B}}\dfrac{\left[\langle b|\phi\rangle\langle\phi|b\rangle\right]^{\otimes t}}{\left(\langle\phi|b\rangle\langle b|\phi\rangle\right)^{t-1}}\right)\nonumber\\
&=& \int_{u\in U_{\mathbb{Z}_2}(2^N)} d\mu(u) \left(\sum_{|b\rangle\in\mathcal{B}}\dfrac{\left[\langle b|u|\phi\rangle\langle \phi|u^{\dagger}|b\rangle\right]^{\otimes t}}{\left(\langle \phi|u^{\dagger}|b\rangle\langle b|u|\phi\rangle\right)^{t-1}}\right), 
\end{eqnarray}
where $U_{\mathbb{Z}_{2}}(2^N)\subset U(2^N)$ denotes the set of all unitaries that commute with the operator $\otimes_{i=1}^{N}\sigma^x_{i}$. One can easily show that $U_{\mathbb{Z}_{2}}(2^N)$ is a compact subgroup of $U(2^N)$.} 

\begin{figure}
    \centering
    \includegraphics[width=0.5\linewidth]{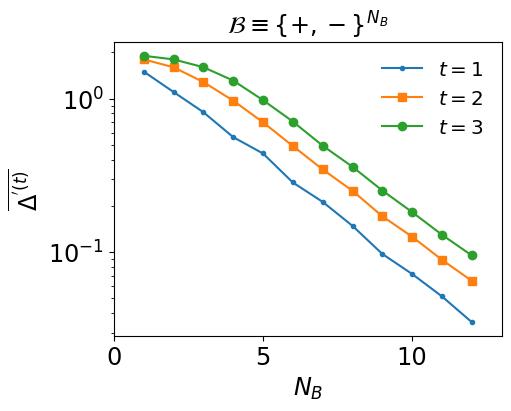}
    \caption{{Illustration of the average trace distance $\overline{\Delta^{'(t)}}$ vs $N_B$ for the random generator states with the $\mathbb{Z}_{2}$-symmetry for the first three moments. We fix the charge $k=0$ and $N_A=3$. The measurements are performed in local product basis $\mathcal{B}\equiv \{+, -\}^{N_B}$. For numerical purposes, the average trace distance is evaluated by considering $10$ samples of the initial generator states.}}
    \label{fig1}
\end{figure}

{Let $U_{\mathbb{Z}_{2}}(2^{N_A})$ denotes the set of all unitaries that commute with $\bigotimes_{i=1}^{N_A}\sigma^{x}_{i}$. Then, one can verify that the elements of $U_{\mathbb{Z}_{2}}(2^{N_A})$ also commute with the symmetry operator \( \bigotimes_{i=1}^{N} \sigma^x_i \). This will imply that $U_{\mathbb{Z}_{2}}(2^{N_{A}})\subset U_{\mathbb{Z}_{2}}(2^{N})$ \footnote{Note that for non-on-site symmetries, such as translation and reflection symmetries, this statement does not hold, so the following analysis cannot be extended straightforwardly.}. By making use of invariance of Haar measure associated with $U_{\mathbb{Z}_{2}}(2^{N})$, we replace $u$ in Eq. (\ref{moment}) with $vu$, where $v\in U_{\mathbb{Z}_{2}}(2^{N_{A}})$. It then follows that
\begin{eqnarray}
\mathcal{M}^{t}_{\mathbb{Z}_{2}}=\sum_{|b\rangle\in\mathcal{B}} \int_{u\in U_{\mathbb{Z}_2}(2^N)} d\mu(u)  \left(\dfrac{\left[\langle b|vu|\phi\rangle\langle \phi|u^{\dagger}v^{\dagger}|b\rangle\right]^{\otimes t}}{\left(\langle \phi|u^{\dagger}|b\rangle\langle b|u|\phi\rangle\right)^{t-1}}\right)=v^{\otimes t}\mathcal{M}^{t}_{\mathbb{Z}_{2}} v^{\dagger \otimes t}. 
\end{eqnarray}
The unitary freedom in the above equation can be used to show that the integrand and the denominator are independent random variables. This can be done by using arguments similar to those used while deriving the moments of the symmetry-restricted ensembles in the main text. It is then straightforward to write the moments of the projected ensemble as
\begin{eqnarray}\label{ind-meas}
\mathcal{M}^{t}_{\mathbb{Z}_{2}} \propto \sum_{|b\rangle\in\mathcal{B}}\langle b| \mathbf{Z}_{k} |b\rangle^{\otimes t}\mathbf{\Pi}^{t}_{A}, 
\end{eqnarray}
where $\mathbf{\Pi}^{t}_{A}=\sum_{j}\pi_{j}$ is the projector onto the permutation symmetric subspace of $t$-copies of the Hilbert space spanning $N_A$ sites. Note that Eq. (\ref{ind-meas}) holds for any measurement basis. We now fix the measurement basis to be $\mathcal{B}\equiv\{+, -\}^{N_B}$. For any $|b\rangle\in \mathcal{B}$, the partial inner product $\langle b| \mathbf{Z}_{k} |b\rangle$ can be evaluated as follows: 
\begin{eqnarray}
\langle b| \mathbf{Z}_{k} |b\rangle = \mathbb{I}_{2^{N_A}}+(-1)^{k+\sum_{i=1}^{N_B}\text{sgn}(b_i)}\left(\otimes_{i=1}^{N_A}\sigma^{x}_{i}\right)=\mathbf{Z}_{k', N_A},   
\end{eqnarray}
where $k'=k+\sum_{i=1}^{N_B}\text{sgn}(b_i)$. The subscript $N_A$ in $\mathbf{Z}_{k', N_A}$ distinguishes the subspace projector acting on all the sites ($\mathbf{Z}_{k}$) from the one that acts only on $N_A$-number of sites ($\mathbf{Z}_{k', N_A}$). Depending on the values taken by $k$ and the parity of the basis vector $|b\rangle$, r.h.s of the above expression becomes a subspace projector onto one of the eigenspaces of the $\mathbb{Z}_{2}$-symmetry operator. Moreover, one can verify that half of the basis vectors have even parity and the remaining half have odd parity. It then follows that 
\begin{eqnarray}
\mathcal{M}^{t}_{\mathbb{Z}_{2}}=\dfrac{1}{\mathcal{N}} \left( \mathbf{Z}^{\otimes t}_{0, N_A} + \mathbf{Z}^{\otimes t}_{1, N_A}\right) \mathbf{\Pi}^{t}_{A}, 
\end{eqnarray}
where $\mathcal{N}$ denotes the normalization constant and is given by $\mathcal{N}=\text{Tr}\left[ \left( \mathbf{Z}^{\otimes t}_{0, N_A} + \mathbf{Z}^{\otimes t}_{1, N_A}\right) \mathbf{\Pi}^{t}_{A}\right]$.}

{We now numerically calculate $\overline{\Delta^{'(t)}}$, the average trace distance between the moments of the projected ensembles of a typical $\mathbb{Z}_{2}$-symmetric generator state and $\mathcal{M}_{\mathbb{Z}^{t}_{2}}$, as the system size varies. 
\begin{eqnarray}
    \Delta'^{(t)}=\left\| \left(\sum_{|b\rangle\in\mathcal{B}}\dfrac{\left[\langle b|\phi_{AB}\rangle\langle\phi_{AB}|b\rangle\right]^{\otimes t}}{\left(\langle\phi_{AB}|b\rangle\langle b|\phi_{AB}\rangle\right)^{t-1}}\right) - \mathcal{M}^{t}_{\mathbb{Z}_{2}}\right\|_{1}.
\end{eqnarray}
We evaluate $\Delta^{'(t)}$ and average it over many samples of the initial generator states, which is denoted by $\overline{\Delta^{'t}}$. Supporting numerical results are shown in Fig. \ref{fig1}. We observe that the average trace distance $\overline{\Delta^{'t}}$ exponentially converges to zero as the measured subsystem size $N_B$ increases. }

\section{Further details on deep thermalization in Ising chain}\label{comparision}
\subsection{Contrasting deep thermalization with and without translation symmetry}\label{symbrok_comparision}
{In the main text, we have examined the deep thermalization in the Ising chain having homogeneous model parameters and with periodic boundary conditions (PBCs). Recall that PBCs, together with homogeneous interaction and fields, imply translation symmetry in the considered system. Due to this symmetry, the entanglement builds up in the time-evolved state at a rate double that of the case with an open boundary. As we notice here, the same is reflected in the phenomenon of emergent state designs. 
In this appendix, we study deep thermalization in the absence of translation symmetry and contrast it with the translation-symmetric case. }
\subsubsection{Through modification of boundary conditions}
\begin{figure}
\includegraphics[scale=0.43]{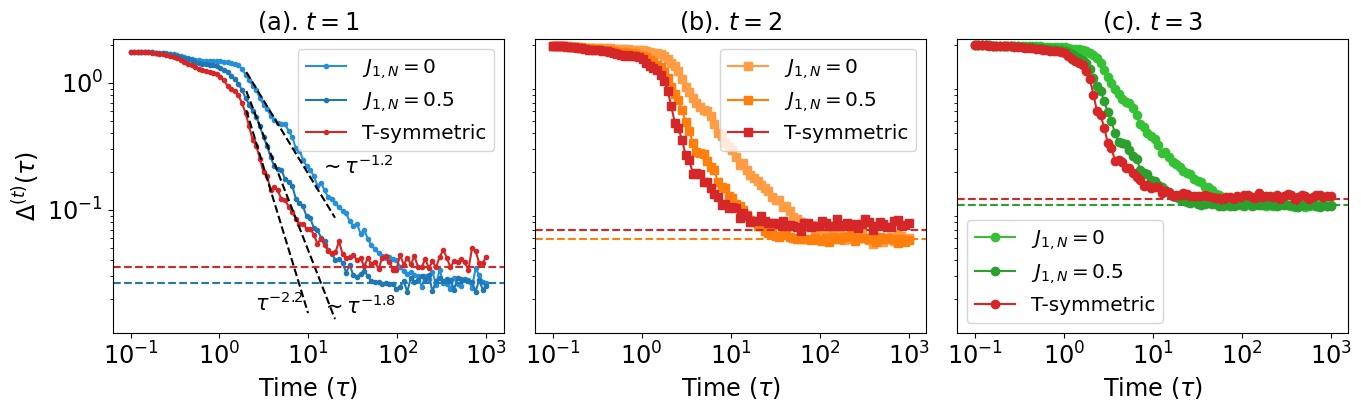}
\caption{\label{fig:ising} The figure contrasts the decay of $\Delta^{(t)}$ for the state $|0\rangle^{\otimes N}$ evolved under the Hamiltonian of the Ising chain with translation symmetry (red curves) with that of moderately broken translation symmetry and the one with open ends (blue, orange, and green curves for $t=1$, $2$, and $3$). The latter cases are characterized by the interaction strength between the first and $N$-th spins, $J_{1, N} = 0.5$ and $0$, respectively. The data is presented for $N=16$ and $N_A=3$. The panels correspond to the first three moments $t=1$, $2$, and $3$, respectively. We show the power-law scaling to compare the convergence rate in the early time regime. At late times, the data saturate to the predicted RMT values with and without symmetry, which are marked by the horizontal dashed lines.} 
\end{figure}

{Here, we break the translation symmetry of the model by considering (i) open boundary condition (OBC) where interaction between the last and first spins of the chain is absent and (ii) inhomogeneous interaction between a pair of spins with closed boundary. 
 Specifically, we contrast the initial decay of the trace distance measure for these two cases with the translation symmetric case. The corresponding results are shown in Fig. \ref{fig:ising}. 
 The OBC implies $J_{1,N} = 0$, where $J_{1, N}$ represents the interaction strength between the first and $N$-th spins.
 In this case, \(\Delta^{(t)}(\tau)\) displays a power-law decay \(\sim \tau^{-1.2}\) in the early time regime. This has nearly half the exponent of $\sim \tau^{-2.2}$ scaling observed in the translation-symmetric case. In order to interpolate between these two cases, we consider the system with a closed boundary but with an inhomogeneous interaction between a pair of spins. This is implemented by considering $J_{1, N}=0.5$. Notice that though the system is now closed, the translation symmetry still remains absent, and we refer to this case as moderately broken translation symmetry. 
 Here, we observe that the power-law decay $\sim \tau^{-1.8}$ faster than in the OBC case, however, it is still slower than in the PBCs case with homogeneous model parameters. In all the cases, the trace distance measure saturates at late times to appropriate RMT predicted values, which are marked by the horizontal dashed lines. }


\begin{figure}
\includegraphics[scale=0.43]{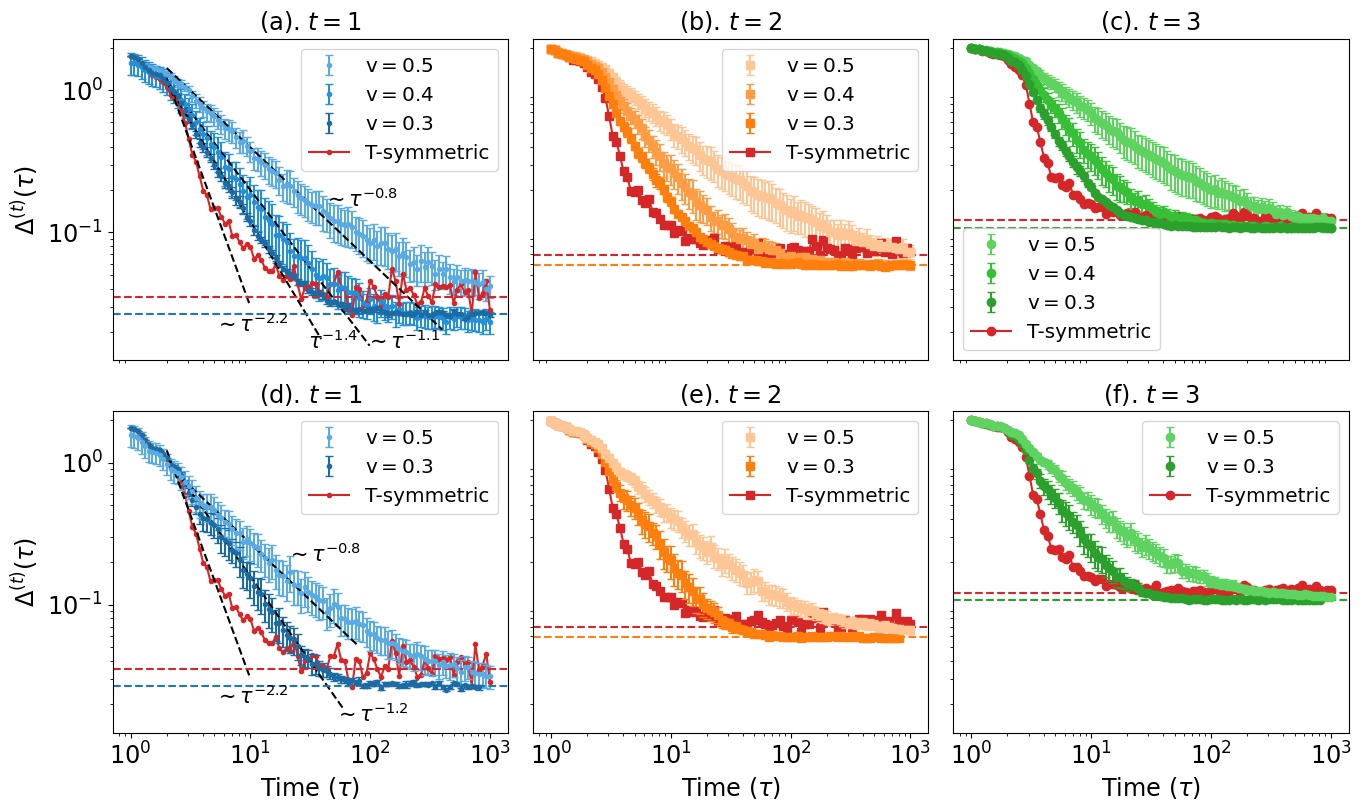}
\caption{\label{fig:ising-disorder} 
The figure contrasts the decay of $\Delta^{(t)}$ for the state $|0\rangle^{\otimes N}$ evolved under the dynamics of translation-symmetric Ising model (red curves) and disordered Ising model given in Eq. \eqref{inhomo_ising} (blue, orange, and green curves for $t=1$, $2$, and $3$).  Like in the main text, we employ PBCs, and the data is presented for $N = 16$ and $N_A = 3$. (a)-(c) and (d)-(f) illustrate the results when the disorder is introduced by randomizing the strengths of the interactions and transverse fields, respectively. The variance $v$ of the randomized parameters controls the strength of the disorder (see Appendix \ref{comparision_disord}), and lighter shading corresponds to a stronger disorder. 
For the disordered cases, ensemble-averaged (over $10$ realizations) data with standard ensemble error (shown as error bars) is presented. We show the power-law scaling to compare the convergence rate in the early time regime. 
At late times, the data saturate to the predicted RMT values with and without symmetry, which are marked by the horizontal dashed lines.} 
\end{figure}

\subsubsection{Through introduction of disorder}
\label{comparision_disord}
{In this section, we consider the Ising chain with inhomogeneous model parameters. We employ this by introducing diagonal and off-diagonal disorders. In particular, we consider two cases by randomizing the strengths of the (i) interactions and (ii) transverse fields. Thus, the Hamiltonian of the Ising model with these disorders can be written as  
\begin{eqnarray}\label{inhomo_ising}
H=\sum_{i=1}^{N} (J+\eta_{i})\sigma^{x}_{i} \sigma^{x}_{i+1}+\sum_{i=1}^{N}h_x\sigma^{x}_{i}+\sum_{i=1}^{N}(h_y+\xi_i)\sigma^{y}_{i}, 
\end{eqnarray}
where $\eta_{i}$ and $\xi_i$ are independent and identically distributed random variables chosen from Gaussian distribution $\mathcal{N}(0, v)$ with zero mean and variance $v$.  Similar to the main text, here we study the model with a closed boundary. The case with $\eta_{i} = \xi_i = 0$ for all $i=1, 2, \cdots, N$, and $J = 1$ corresponds to the Hamiltonian considered in Eq. \eqref{ising} of the main text. Note that the model needs to be chaotic in order to obtain the deep-thermalization characteristics. It is known that the Ising model with both transverse and longitudinal fields is nonintegrable for non-zero values of $h_x$, $h_y$, and $J$ \cite{kim2023nonintIsing, kim2014testing, sharma2015quenches, cotler2023emergent}. But, in order to stay close to the point where the model is robustly chaotic, we choose $\left[J, h_{x}, h_y\right] = \left[1, (\sqrt{5}+1)/4, (\sqrt{5}+5)/8\right]$, as also considered in the main text. Then, the random variables $\eta_{i}$ and $\xi_i$ make the model inhomogeneous, thereby breaking the translation symmetry. The variance $v$ of the random variables controls the strength of the disorder introduced in the otherwise translation symmetric system. In other words, $v$ quantifies the degree of translation symmetry breaking in the considered model.}

{In Fig. \ref{fig:ising-disorder}, we illustrate the decay of $\Delta^{(t)}(\tau)$ for the disordered Ising Hamiltonian presented in Eq. (\ref{inhomo_ising}), and contrast it with the clean, homogeneous case. Figures \ref{fig:ising-disorder}a-\ref{fig:ising-disorder}c show the evolution when the disorder is introduced by randomizing the interaction strengths. We consider a few values of $v$ to demonstrate the effects of disorder with increasing strength. However, we keep $v$ considerably small such that the model remains chaotic, which can also be inferred from the decreasing trend and long-time saturation of $\Delta^{(t)}(\tau)$. We notice that in the early time regime, the trace distance measure shows slower convergence in comparison to the clean case, and this rate decreases as the strength of the disorder is enhanced. In particular, the numerical results depict that the disorder ensemble averaged $\Delta^{(t)}(\tau)$ has the power-law scaling as $\sim\tau^{-1.4},$ $\sim\tau^{-1.1},$ and $\sim\tau^{-0.8}$ for $v = 0.3$, $0.4$ and $0.5$, respectively. This can be contrasted with the clean case, where the decay follows the $\sim \tau^{-2.2}$ law. For all the cases, $\Delta^{(t)}(\tau)$ saturate at late time to appropriate RMT predicted values, which are marked by the horizontal dashed lines. Similarly, Figs. \ref{fig:ising-disorder}d-\ref{fig:ising-disorder}f illustrate the evolution when the disorder is present only in the transverse fields. As previously, we observe that in the early time-regime, the decay of trace distance measure gets slower with increasing disorder strength. Consequently, the above analysis unveils that the Ising model with translation symmetry exhibits faster convergence of the trace distance measure during the early time evolution in comparison to the cases when the symmetry is broken.}

\subsection{Benchmarking the late time evolution with translation plus reflection symmetric random states}\label{benchmark_latetim}
\begin{figure}
    \centering
    \includegraphics[width=0.4\linewidth]{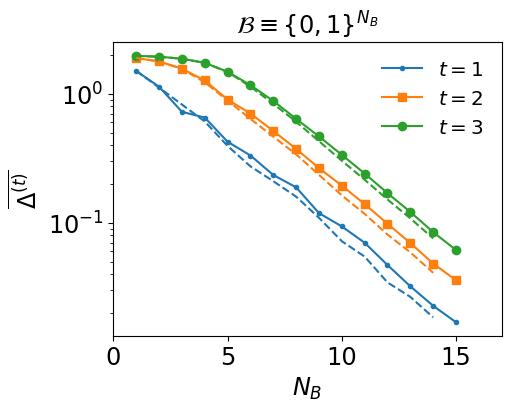}
    \caption{Comparison of the average trace distance $\overline{\Delta^{t}}$ versus $N_B$ for the first three moments when the initial generator states are simultaneous eigenvectors of the translation and reflection symmetries (shown with thick lines) with the case of Haar random generator states (shown with dashed lines). The measurements are performed in $\sigma^z$ basis, and $N_A$ is fixed at $3$. Additionally, note that we have considered $10$ samples of the initial generator states to numerically evaluate the average trace distance. }
    \label{ref-tran}
\end{figure}
{The Ising Hamiltonian considered in this work displays reflection symmetries about every site in addition to the translation symmetry. Moreover, the initial state $|0\rangle^{\otimes N}$ is a common eigenvector of all the aforementioned symmetry operators. Hence, to obtain the RMT predictions, it is imperative to evaluate the trace distance $\overline{\Delta^{t}}$ for the ensemble of states that are common eigenvectors of the translation and reflection symmetries with the eigenvalue $1$. It is to be noted that the symmetry operators corresponding to distinct reflection operations do not commute with each other, nor do they commute with the translation operator. Nevertheless, they all share an overlapping eigenspace. Therefore, under the dynamics of the chaotic Ising chain, the initial state evolves and equilibrates to one of the states belonging to the ensemble spanned by the common eigenspace of all the symmetry operators. Starting from a global Haar random state $|\psi\rangle$, one can construct uniform random states with the aforementioned symmetries as follows:
\begin{eqnarray}\label{tran_ref}
   |\phi\rangle=\dfrac{1}{\mathcal{N}}\mathbf{R}^{N-1}_{0}\mathbf{R}^{N-2}_{0}\cdots \mathbf{R}^{0}_{0}\mathbf{T}_{0}|\psi\rangle,  
\end{eqnarray}
where $\mathbf{R}^{j}_{0}$ denotes the projector onto the reflection (around $j$-th site) symmetric subspace with the eigenvalue $(-1)^{0}=1$. It is worth noting that if a state $|\psi\rangle$ is a simultaneous eigenvector of a reflection operator about an arbitrary site and also the translation operator, then it will also be an eigenvector of other reflection operators corresponding to remaining $N-1$ sites. Having constructed the ensemble of states as given in Eq. (\ref{tran_ref}), one can proceed with the computation of $\overline{\Delta^{t}}$. The corresponding results for the first three moments are shown In Fig. \ref{ref-tran}. These results are compared with the case when the initial states are completely Haar random (shown with dashed lines). We notice slight differences between both cases. These differences may be attributed to the non-commuting nature of the translational and reflection symmetries. The RMT values obtained in this figure, by taking all the symmetries into account, can be used to compare the saturation values in Fig. \ref{fig:ising-phy} of the main text.}

\bibliographystyle{plain}
\bibliography{thesis.bbl} 

\end{document}